\documentclass[a4,11pt,twoside,onecolumn,final]{book} %

\def\techreport{}

\ifdefined\techreport
\fi

\usepackage[utf8]{inputenc}
\usepackage{url}
\usepackage{verbatim}
\usepackage{graphicx}
\usepackage{paralist}
\usepackage{xspace} 
\usepackage[algoruled,vlined,english]{algorithm2e}
\usepackage{latexsym}
\usepackage{amsmath}
\usepackage{amssymb}
\usepackage{epsfig}
\usepackage{algorithm2e}
\usepackage{bm}
\usepackage{mathtools}
\usepackage{paralist}
\usepackage{wrapfig}
\usepackage{color}
\usepackage[dvipsnames]{xcolor}
\usepackage{listings}
\usepackage{tikz}
\usetikzlibrary{calc,decorations,arrows,arrows.meta,shapes,patterns,decorations.pathreplacing,calligraphy}
\usepackage{multirow}
\usepackage{arrayjobx} 
\usepackage{subcaption}

\ifdefined\techreport
  \usepackage{amsthm}

  \usepackage{titlesec}

  \usepackage[colorlinks    = true,
              allcolors     = blue,
              plainpages    = false,
              pdfpagelabels = true
              ]{hyperref}

  \usepackage{bookmark}
  \usepackage[backend=bibtex,
              style=authoryear,
              dashed=false,
              natbib=true,
              giveninits,
              uniquename=false,
              uniquelist=false,
              maxcitenames=2,
              maxbibnames=99,
              isbn=false]{biblatex}
  \AtEveryBibitem{%
    \clearfield{note}%
  }
              
  \usepackage[vcentering,hcentering]{geometry}

  \usepackage{fancyhdr}
  \fancyheadoffset{0pt}%
  \newenvironment{abstract}
  {\begin{center}\large\bfseries Abstract\end{center}}
  {}
\fi

\usepackage[draft]{fixme}
\fxsetup{theme=color}
\definecolor{fxtarget}{rgb}{0.8000,0.0000,0.0000}

\title{Termination Analysis of Linear-Constraint Programs}

\def\amirn{Amir M. Ben-Amram}
\def\amiri{Qiryat Ono, Israel}
\def\amire{benamram.amir@gmail.com}
\def\samirn{Samir Genaim}
\def\samiri{Complutense University of Madrid, Spain}
\def\samire{sgenaim@ucm.es}
\def\joeln{Jo\"{e}l Ouaknine}
\def\joeli{Max Planck Institute for Software Systems, Saarland, Germany}
\def\joele{joel@mpi-sws.org}
\def\benn{James Worrell}
\def\beni{University of Oxford, UK}
\def\bene{jbw@cs.ox.ac.uk}

\unless\ifdefined\techreport
  \maintitleauthorlist{%
    \amirn\\
    \amiri \\
    \amire
    \and
    \samirn\\
    \samiri \\
    \samire
    \and
    \joeln\\
    \joeli\\
    \joele
    
    \and
    \benn\\
    \beni\\
    \bene
  }

  \issuesetup
  {%
    copyrightowner={Amir M. Ben-Amram, Samir Genaim, Jo\"{e}l Ouaknine and James Worrell},
    volume        = xx,
    issue         = xx,
    pubyear       = 2024,
    isbn          = xxx-x-xxxxx-xxx-x,
    eisbn         = xxx-x-xxxxx-xxx-x,
    doi           = 10.1561/XXXXXXXXX,
    firstpage     = 1, %
    lastpage      = 18
  }

  \author[1]{\amirn}
  \author[2]{\samirn}
  \author[3]{\joeln}
  \author[4]{\benn}
  
  \affil[1]{\amiri; \amire}
  \affil[2]{\samiri; \samire}
  \affil[3]{\joeli; \joele}
  \affil[4]{\beni; \bene}

  \articledatabox{\nowfntstandardcitation}
\fi

\newcommand{\tuple}[1]{\langle #1 \rangle}
\newcommand{\set}[1]{\{ #1 \}}
\newcommand{\cbox}[1]{\colorbox{yellow}{\ensuremath{#1}}}

\newcommand{\lp}{LP\xspace}
\newcommand{\poly}[1]{{\mathcal #1}}
\newcommand{\intpoly}[1]{{I(#1)}}
\newcommand{\inthull}[1]{{#1}_I}
\newcommand{\trans}[0]{\top} %
\newcommand{\cone}[0]{\ensuremath{\mathtt{cone}}}
\newcommand{\ccone}[0]{\ensuremath{\mathtt{rec.cone}}}
\newcommand{\conv}[0]{\ensuremath{\mathtt{conv}}}
\newcommand{\convhull}[0]{\conv}
\newcommand{\aff}{\ensuremath{\mathtt{aff}}}
\newcommand{\proj}[2]{{\mathtt{proj}_{#1}{(#2)}}}
\newcommand{\size}[1]{\Vert #1 \Vert}
\newcommand{\bitsize}[1]{\Vert #1 \Vert} %
\newcommand{\transitions}{\poly{Q}}
\newcommand{\inv}{I}
\newcommand{\pinv}{\poly{I}}
\newcommand{\trres}[2]{{#1}_{#2}}
\newcommand{\tr}[2]{\ensuremath{\bigl(\begin{smallmatrix}{#1}\\{#2}\end{smallmatrix}\bigr)}}
\newcommand{\mlc}[0]{\ensuremath{\mathit{MLC}}\xspace}
\newcommand{\slc}[0]{\ensuremath{\mathit{SLC}}\xspace}
\newcommand{\while}[0]{\ensuremath{\mathit{while}}\xspace}
\newcommand{\wdo}[0]{\ensuremath{\mathit{do}}\xspace}
\newcommand{\cfg}[0]{\mbox{\upshape{CFG}}\xspace}
\newcommand{\cfgs}[0]{\mbox{\upshape{CFG}s}\xspace}
\newcommand{\reach}[2]{\ensuremath{\mathsf{RCH}(#1,#2)}}
\newcommand{\pc}[0]{\ensuremath{\mathit{pc}}\xspace}
\newcommand{\errl}[0]{\ensuremath{\ell_{\mathit{err}}}}
\newcommand{\pre}[0]{\ensuremath{\mathsf{pre}}}
\newcommand{\post}[0]{\ensuremath{\mathsf{post}}}
\newcommand{\clang}[0]{\textsf{C}\xspace}
\newcommand{\mc}[0]{MC\xspace}
\newcommand{\sct}[0]{SCT\xspace}
\newcommand{\dsct}[0]{\ensuremath{\delta}\sct}
\newcommand{\sctcom}[0]{\ensuremath{\bullet}\xspace}
\newcommand{\mccom}[0]{\ensuremath{\diamond}\xspace}

\newcommand{\scsg}[0]{SCSG\xspace}
\newcommand{\scsgs}[0]{\scsg{s}\xspace}

\newcommand{\scc}[0]{SCC\xspace}
\newcommand{\sccs}[0]{\scc{s}\xspace}

\let\vect=\vec
\renewcommand{\vec}[1]{\ensuremath{\bm{#1}}}
\newcommand{\bgamma}{{\vec{\gamma}}}
\newcommand{\bmu}{{\vec{\mu}}}

\newcommand{\bv}{\vec{v}}

\newcommand{\bz}{\vec{z}}

\newcommand{\bb}{\vec{b}}
\newcommand{\bc}{\vec{c}}

\newcommand{\p}[0]{\ensuremath{\mathtt{P}}\xspace}
\newcommand{\ptime}[0]{\ensuremath{\mathtt{PTIME}}\xspace}
\newcommand{\exptime}[0]{\ensuremath{\mathtt{EXPTIME}}\xspace}
\newcommand{\np}[0]{\ensuremath{\mathtt{NP}}\xspace}
\newcommand{\nph}[0]{\ensuremath{\mathtt{NP}}{-}hard\xspace}
\newcommand{\npc}[0]{\ensuremath{\mathtt{NP}}{-}complete\xspace}
\newcommand{\conp}[0]{\ensuremath{\mathtt{coNP}}\xspace}
\newcommand{\conph}[0]{\ensuremath{\mathtt{coNP}}{-}hard\xspace}
\newcommand{\conpc}[0]{\ensuremath{\mathtt{coNP}}{-}complete\xspace}
\newcommand{\pspace}[0]{\ensuremath{\mathtt{PSPACE}}\xspace}
\newcommand{\pspaceh}[0]{\ensuremath{\mathtt{PSPACE}}{-}hard\xspace}
\newcommand{\expspace}[0]{\ensuremath{\mathtt{EXPSPACE}}\xspace}

\newcommand{\ackh}[0]{\ensuremath{\mathtt{Ackermann}}{-}hard\xspace}
\newcommand{\re}[0]{\ensuremath{\mathtt{RE}}\xspace}
\newcommand{\core}[0]{\ensuremath{\mathtt{coRE}}\xspace}
\newcommand{\ints}{\ensuremath{\mathbb Z}\xspace}
\newcommand{\nats}{\ensuremath{\mathbb N}\xspace}
\newcommand{\rats}{\ensuremath{\mathbb Q}\xspace}
\newcommand{\reals}{\ensuremath{\mathbb R}\xspace}
\newcommand{\areals}{\ensuremath{\reals_{\mathbb A}}\xspace}
\newcommand{\complex}{\ensuremath{\mathbb C}\xspace}
\newcommand{\noneg}[1]{\ensuremath{{#1}_{\ge0}}\xspace}
\newcommand{\numdom}[0]{\ensuremath{R}\xspace}
\newcommand{\TT}{{\mathbb{T}}}

\newcommand{\llrfsym}[0]{\ensuremath{\tau}\xspace}
\newcommand{\diff}[1]{\ensuremath{\Delta #1}}
\newcommand{\rfcoeff}[0]{\lambda}
\newcommand{\ti}[0]{TI\xspace}

\newcommand{\dti}[0]{DTI\xspace}
\newcommand{\dtis}[0]{DTIs\xspace}
\newcommand{\lrfdti}[0]{{LRF-DTI}\xspace}
\newcommand{\lrfdtis}[0]{{LRF-DTIs}\xspace}
\newcommand{\nollrf}[0]{\mbox{\textsc{None}}\xspace}
\newcommand{\lrf}[0]{\mbox{\upshape LRF}\xspace}
\newcommand{\lrfs}[0]{\mbox{{\upshape LRF}s}\xspace}
\newcommand{\llrf}[0]{\mbox{\upshape LLRF}\xspace}
\newcommand{\llrfs}[0]{\mbox{{\upshape LLRF}s}\xspace}
\newcommand{\qlrf}[0]{\mbox{\upshape QLRF}\xspace}
\newcommand{\qlrfs}[0]{\mbox{{\upshape QLRF}s}\xspace}
\newcommand{\mlrf}[0]{\mbox{\upshape{M$\Phi$RF}}\xspace}
\newcommand{\mlrfs}[0]{\mbox{\upshape{M$\Phi$RFs}}\xspace}

\newcommand{\nlrf}[0]{\mbox{\upshape{NLRF}}\xspace}
\newcommand{\nlrfs}[0]{\mbox{\upshape{NLRF}s}\xspace}

\newcommand{\paramllrf}[1]{{#1}{-}\llrf}
\newcommand{\paramllrfs}[1]{{#1}{-}\llrfs}
\newcommand{\adfg}[0]{\mbox{\upshape ADFG}\xspace}
\newcommand{\bg}[0]{\mbox{\upshape BG}\xspace}
\newcommand{\bms}[0]{\mbox{\upshape BMS}\xspace}

\newcommand{\mphi}[0]{\mbox{\upshape{M$\Phi$}}\xspace}
\newcommand{\bmsllrf}[0]{\paramllrf{\bms}}
\newcommand{\bmsllrfs}[0]{\paramllrfs{\bms}}
\newcommand{\bgllrf}[0]{\paramllrf{\bg}}
\newcommand{\bgllrfs}[0]{\paramllrfs{\bg}}
\newcommand{\adfgllrf}[0]{\paramllrf{\adfg}}
\newcommand{\adfgllrfs}[0]{\paramllrfs{\adfg}}
\newcommand{\gnta}[0]{GNTA\xspace}
\newcommand{\gntas}[0]{\gnta{s}\xspace}

\newcommand{\ie}{\emph{i.e.}\xspace}

\newcommand{\eg}{\emph{e.g.}\xspace}

\newcommand{\etc}{\emph{etc}\xspace}
\newcommand{\wrt}{wrt.\xspace}

\newcommand{\survey}{survey\xspace} %
\newcommand{\Survey}{Survey\xspace} %
\ifdefined\techreport
  \newcommand{\chp}{chapter\xspace} %
  \newcommand{\Chp}{Chapter\xspace} %
\else
  \newcommand{\chp}{section\xspace} %
  \newcommand{\Chp}{Section\xspace} %
\fi
\newcommand{\cegar}[0]{CEGAR\xspace}
\newcommand{\usedin}[1]{%
  \noindent
  \colorbox{yellow}{Used in #1}
  \medskip%
}

\newcommand{\tikzgrid}[4]{
  \def\gridwidth{#2}   %
  \def\gridheight{#3}  %
  \def\gridstep{#4}  %
  \def\gridcolor{#1} %

  \foreach \y in {0,\gridstep,...,\gridheight} {
    \draw[color=\gridcolor] (0,\y) -- (\gridwidth,\y);
  }

  \foreach \x in {0,\gridstep,...,\gridwidth} {
    \draw[color=\gridcolor] (\x,0) -- (\x,\gridheight);
  }
}

\tikzset{
  loc/.style={
    draw=red,
    thick,
    fill=yellow!50,
    rounded corners=5pt,
    font=\footnotesize,
    inner sep=2pt
  },
  inloc/.style={
    loc,
    fill=green
  },
  tr/.style={
    color=blue,
    font=\footnotesize,
    inner sep=2pt
  }
,
  tre/.style={ %
    color=black,
    ->
  },
  sctloc/.style={
    font=\footnotesize,
    inner sep=1pt,
  },
  scttitle/.style={
    font=\footnotesize,
    inner sep=1pt
  },
  sctw/.style={
    font=\tiny,
    inner sep=1pt,
    color=red
  },
  scte/.style={
    ->,
    draw
  },
  sctbe/.style={
    ->,
    dashed,
    draw,
    color=blue
  }
}

\ifdefined\techreport
  \newtheorem{theorem}{Theorem}[chapter]
  \newtheorem{lemma}[theorem]{Lemma}
  \newtheorem{corollary}[theorem]{Corollary}
  \newtheorem{proposition}{Proposition}[chapter]
  \newtheorem{remark}{Remark}[chapter]
  \newtheorem{definition}{Definition}[chapter]
  \newtheorem{example}{Example}[chapter]
\fi

\newtheorem{observation}[theorem]{OBSERVATION}
\newtheorem{problem}[]{OPEN PROBLEM}
\newtheorem{problems}[problem]{OPEN PROBLEMS}
\newtheorem{conjecture}[theorem]{CONJECTURE}

\newcommand{\gridName}[2]{n#1n#2}
\tikzset{
  nborder/.style={circle,inner sep=0.75mm,minimum size=0mm},
  ngrid/.style={nborder,fill=black!30,draw=black},
  matchingsep/.style={shape=circle,fill=red!10,inner sep=0mm,minimum size=0mm},
  matching/.style  = {matchingsep,midway,left}
}

\tikzset{
  zrBorder/.style={circle,inner sep=0.5mm,minimum size=0mm},
  zrNode/.style={zrBorder,fill=black!40} %
}

\tikzset{
  nborder/.style={circle,inner sep=0.5mm,minimum size=0mm},
  ngrid/.style={nborder,fill=black!40},
  matchingsep/.style={shape=circle,fill=red!10,inner sep=0mm,minimum size=0mm},
  matching/.style  = {matchingsep,midway,left}
}

\newcommand{\TermGridGenNoCap}[8]
{
  \begin{scope}
    \def\xIndent{#7}
    \def\yIndent{#8}
    \def\yBase{#4+(#6-1)*#5}
    \foreach \x in {1,...,#3} {
      \pgfmathsetmacro{\aa}{\xIndent+#1+(\x-1)*#2}
      \foreach \y in {1,...,#6} {
        \pgfmathsetmacro{\bb}{\yBase-(\y-1)*#5}
        \def\nName{\gridName{\x}{\y}}
        \node (\nName) at (\aa cm,\bb cm) [ngrid] {};
      }
    }
  \end{scope}
}

\newcommand{\TermGridGenNoCapBoxed}[8]
{
    \TermGridGenNoCap{#1}{#2}{#3}{#4}{#5}{#6}{#7}{#8}

    \pgfmathsetmacro{\cXleft}{#1+(#7)-0.2}
    \pgfmathsetmacro{\cXright}{#1+(#2)*(#3-1)+#7+0.2}
    \pgfmathsetmacro{\cYbot}{#4-0.2}
    \pgfmathsetmacro{\cYtop}{#4+(#5)*(#6-1)+0.2}
    \draw [rounded corners=4pt,dotted,very thick] (\cXleft cm,\cYbot cm) rectangle (\cXright cm,\cYtop);
}

\newcommand{\TermGridGen}[9]
{

  \begin{scope}

    \def\xIndent{#7}
    \def\yIndent{#8}

    \pgfmathsetmacro{\aa}{#1}
    \foreach \y in {1,...,#6} {
      \pgfmathsetmacro{\bb}{#4+(\y-1)*#5}
      \def\nName{\gridName{0}{\y}}
      \pgfmathtruncatemacro\inx{#6-\y+1}
      \node (\nName) at (\aa cm,\bb cm) [nborder] {$x_{\inx}$};
    }
    \pgfmathsetmacro{\bb}{#4+(#6-1)*#5+\yIndent}
    \foreach \x in {1,...,#3} {
      \pgfmathsetmacro{\aa}{\xIndent+#1+(\x-1)*#2}
      \def\nName{\gridName{\x}{0}}

      \pgfmathtruncatemacro\inxvv{mod(\x-1,#3-1)}
      \pgfmathtruncatemacro\inxww{\x-1}
      \def\inx{\ifthenelse{\equal{#9}{1}}{\inxvv}{\inxww}}

      \node (\nName) at (\aa cm,\bb cm) [nborder] {$\vec{x}^{(\inx)}$};
    }

    \TermGridGenNoCap{#1}{#2}{#3}{#4}{#5}{#6}{#7}{#8}

  \end{scope}
}
\newcommand{\TermGridGenDifferentCap}[9]
{

  \begin{scope}

    \def\xIndent{#7}
    \def\yIndent{#8}

    \pgfmathsetmacro{\aa}{#1}
    \pgfmathsetmacro{\aap}{#1+(#2)*(#3-1)+2*#7}
    \foreach \y in {1,...,#6} {
      \pgfmathsetmacro{\bb}{#4+(\y-1)*#5}
      \def\nName{\gridName{0}{\y}}
      \pgfmathtruncatemacro\inx{#6-\y+1}
      \node (\nName) at (\aa cm,\bb cm) [nborder] {$x_{\inx}$};
      \node (\nName) at (\aap cm,\bb cm) [nborder] {$x'_{\inx}$};
    }

    \TermGridGenNoCap{#1}{#2}{#3}{#4}{#5}{#6}{#7}{#8}

  \end{scope}
}

\newcommand{\TermGridEdgeC}[6]
{
  \draw [->,thick] (\gridName{#1}{#2}) -- (\gridName{#3}{#4})
     node[midway,#6] {#5}; 
}

\tikzset{
  zrBorder/.style={circle,inner sep=0.5mm,minimum size=0mm},
  zrNode/.style={zrBorder,fill=black!40} %
}

\newcommand{\TermZrName}[2]{n#1n#2}
\newcommand{\TermZrNameP}[2]{np#1n#2}
\newcommand{\TermZrNameC}[2]{nc#1n#2}

\newcommand{\TermZrCapGraph}[6]
{
  \pgfmathsetmacro{\capStart}{(#1-1)*(#4+#5)+#4+(#5/2)} 
  \pgfmathsetmacro{\capStep}{#4+#5}
  \pgfmathsetmacro{\capHeight}{(#3-1)*#6+0.5}
  \pgfmathtruncatemacro\aSize{#2-#1+1}
  \foreach \y in {1,...,\aSize} {
     \pgfmathsetmacro{\capPos}{\capStart+(\y-1)*\capStep}
     \node at (\capPos,\capHeight) {\aCaption(\y)};
  }
}

\newcommand{\TermZrBase}[9]
{
  \begin{scope}
    \pgfmathtruncatemacro\xStart{#1-1}
    \foreach \x in {#1,...,#2} {
      \pgfmathsetmacro{\aa}{#4+(\x-1)*(#4+#5)}
      \pgfmathsetmacro{\aaP}{\aa+#5}
      \foreach \y in {1,...,#3} {
        \pgfmathsetmacro{\bb}{(#3-\y)*#6}
        \node (\TermZrName{\x}{\y}) at (\aa cm,\bb cm) [zrNode] {};  %
        \node (\TermZrNameP{\x}{\y}) at (\aaP cm,\bb cm) [zrNode] {}; %
      }
      \pgfmathsetmacro{\cXleft}{\aa-0.1}
      \pgfmathsetmacro{\cXright}{\aaP+0.1}
      \pgfmathsetmacro{\cYbot}{-0.2}
      \pgfmathsetmacro{\cYtop}{(#3-1)*#6+0.2}
      \draw [rounded corners=4pt,dotted,very thick] (\cXleft cm,\cYbot cm) rectangle (\cXright cm,\cYtop);
    }

    \def\TermZrXGap{0.2}
    \def\TermZrYExtra{0.3}
    \pgfmathsetmacro{\yBase}{(#3-1)*#6} %
    \foreach \x in {\xStart,...,#2} {
      \pgfmathsetmacro{\aaL}{(\x)*(#4+#5)+\TermZrXGap} %
      \pgfmathsetmacro{\aaR}{\aaL+#4-(2*\TermZrXGap)} %
      \pgfmathsetmacro{\bbB}{0-\TermZrYExtra} %
      \pgfmathsetmacro{\bbT}{\yBase+\TermZrYExtra} %
      \draw [rounded corners=4pt] (\aaL cm,\bbB cm) rectangle (\aaR cm,\bbT cm);
    }
    \foreach \x in {\xStart,...,#2} {
      \pgfmathsetmacro{\aaC}{\x*(#4+#5)+(#4/2)} %
      \foreach \y in {1,...,#3} {
        \pgfmathsetmacro{\bb}{(#3-\y)*#6}
        \node (\TermZrNameC{\x}{\y}) at (\aaC cm,\bb cm) {}; %
      }
    }
    \foreach \x in {1,...,#9} {
      \pgfmathsetmacro{\aaL}{#7*(#4+#5)+(\x-1)*(#4+#5)*#8+(\TermZrXGap/2)} %
      \pgfmathsetmacro{\aaR}{\aaL+(#4+#5)*#8} %
      \pgfmathsetmacro{\bbB}{0-\TermZrYExtra-0.6} %
      \pgfmathsetmacro{\bbT}{\yBase+\TermZrYExtra+0.6} %
      \draw [rounded corners=4pt,dotted,very thick] (\aaL cm,\bbB cm) rectangle (\aaR cm,\bbT cm);
      \pgfmathsetmacro{\capX}{(\aaL+\aaR)/2}
      \pgfmathsetmacro{\capY}{\bbT+0.3};
    }

  \end{scope}
}

\newcommand{\TermZrStateElem}[3]
{
    \path let \p1 = (\TermZrNameC{#1}{#2}) in node at (\x1,\y1) [nborder] {$#3$};
}

\newcommand{\TermZrEdgeFWLab}[4]
{
  \draw [->,thick] (\TermZrName{#1}{#2}) -- (\TermZrNameP{#1}{#3}) node[midway,above] {#4};
}

\newcommand{\TermZrEdgeBWLab}[4]
{
  \draw [->,thick] (\TermZrNameP{#1}{#2}) -- (\TermZrName{#1}{#3}) node[midway,above] {#4};
}

\ifdefined\techreport
\fi

\hypersetup{bookmarksdepth=3} %

\begin{document}

\ifdefined\techreport
  \frontmatter

\def\auth#1#2#3{
  \begin{flushleft}
    {\Large\bfseries #1}\\
    {\large #2}\\
    {\large #3}
  \end{flushleft}
}

\begin{titlepage}
  \makeatletter

  \hspace*{-2.5cm}
  \begin{minipage}{19cm}
    \begin{flushleft}
      {\huge\sc\bfseries\@title}
      \rule{\textwidth}{2pt}\\[1pt]
    \end{flushleft}
    \medskip
    \auth{\amirn}{\amiri}{\amire}
    \auth{\samirn}{\samiri}{\samire}
    \auth{\joeln}{\joeli}{\joele}
    \auth{\benn}{\beni}{\bene}
    
    \vspace*{9cm}
    \rule{\textwidth}{2pt}\\[0pt]
    \today \hfill FTPGL version: \url{https://doi.org/10.1108/FTPGL-07-2025-0071}
  \end{minipage}
  \makeatother
  
\end{titlepage}
 \else
  \makeabstracttitle
\fi

\begin{abstract}
  This paper provides an overview of techniques in termination
  analysis for programs with numerical variables and transitions
  defined by linear constraints.
  This subarea of program analysis is challenging due to the existence
  of undecidable problems, and this paper systematically explores
  approaches that mitigate this inherent difficulty.
  These include foundational decidability results, the use of ranking
  functions, and disjunctive well-founded transition invariants. The
  paper also discusses non-termination witnesses, used to prove that a
  program will not halt.
  We examine the algorithmic and complexity aspects of these methods,
  showing how different approaches offer a trade-off between
  expressive power and computational complexity.
  The paper does not discuss how termination analysis is performed on
  real-world programming languages, nor does it consider more
  expressive abstract models that include non-linear arithmetic,
  probabilistic choice, or term rewriting systems.
\end{abstract}

\ifdefined\techreport
  \tableofcontents
\fi

\listoffixmes

\ifdefined\techreport
  \mainmatter 
\fi

\chapter{Introduction} 

Proving termination is a basic building block of establishing program
correctness, or analysing the behaviour of systems modelled by
programs. The topic of this \survey is the termination problem for
programs with numerical variables (storing integers, rationals, or
reals) whose transitions are specified by linear equations and
inequalities. To make this notion concrete, here is an example of a
loop whose termination we may want to prove:
\begin{center}
\lstinline!while (x2-x1<=0 && x1+x2>=1)$~$x2=x2-2*x1+1;!
\end{center}
While this loop is written in \clang syntax, we prefer to abstract
from any particular programming language and model the loop body as a
relation between values $x_1,x_2$ of the program variables before its
execution and their values $x_1',x_2'$ after its execution.
We thus express the above loop as:
\[
\while\; ( x_2{-}x_1 \le 0, x_1{+}x_2 \ge 1 ) \; \wdo \; x_2' = x_2{-}2x_1{+}1, x_1'=x_1 \, .
\]
This, more mathematical, expression generalises easily by allowing
inequalities as well as equations in the specification of the ``loop
body'', for example we might consider
\[
\while\; ( x_2{-}x_1 \le 0, x_1{+}x_2 \ge 1 ) \; \wdo \; x_2' = x_2{-}2x_1{+}1, x_1' \le x_1 \,.
\]
This is what we call a \emph{simple loop}, or a \emph{single-path
  loop}.  Note that such a loop is, in general, non-deterministic.  In
the above example, in any execution of the loop body any value of
$x_1'$ that satisfies the constraint may be chosen.  We will also
consider \emph{multi-path loops}, that model branching in the loop
body, so that the iteration is represented by several alternatives,
each one with its set of constraints; and the most general form, a
\emph{control-flow graph} which can represent a branching structure,
nested loops \etc. We sometimes group all these types under the
heading \emph{linear-constraint programs}.

Where do such termination problems come from?  As stated before, the
main motivation is program analysis.  In many programs the variables
whose behaviour is relevant to program termination are numerical, and
in this case the program can be often faithfully modelled by
linear-constraint programs, possibly abstracting away operations that
are not relevant to termination. Our model is also abstract in the
sense that we consider the domain of variables to be either $\ints$,
$\rats$, or $\reals$ --- we do not model the finite universe of
machine integers, or the finite precision of floating-point numbers.

There are, of course, computer programs that manipulate non-numerical
data; but in many such programs the proof of termination relies on
numbers related to these data, \eg, the length of lists
constructed or consumed by the program. Thus several tools for testing
the termination of programs abstract structured values into numbers
and in essence reduce the problem to the analysis of numerical
programs.

The termination of numerical programs defined by linear constraints is
a challenging area, since it includes undecidable problems---so it is
important to break the area into subproblems, and attempt to
understand the decidability and complexity of each subproblem. In
\Chp~\ref{chp:dec} we provide the complete solution for one
subproblem, the termination of simple loops whose body is a linear
transformation (thus defined by linear equations and not
inequalities).
We also present a couple of results that illustrate the limitations of
decidability in the termination analysis of programs of the kind we
consider, namely sub-classes of programs for which termination is
undecidable.

Other subproblems arise by weakening the goal from
determining termination \emph{tout court}, to that of determining
whether termination can be established by a specific method.
The best-known example is the principle of \emph{ranking program
  states}: if we can associate with each program state a \emph{rank}
such that ranks are bound to decrease during computation (but can not
decrease forever, \eg, because they are natural numbers), then the
program terminates. When we fix the set of admissible functions for
ranking states (also know as \emph{termination witnesses}), we get a
well-defined subproblem of the termination problem that may well be
solvable, and in fact this is one of the approaches extensively used
by termination tools. In \Chp~\ref{chp:rfs} we survey algorithmic
results for ranking-function problems, specifically we consider
\emph{linear} ranking functions and \emph{lexicographic-linear}
ranking functions.
In \Chp~\ref{chp:dti} we consider the \emph{disjunctive transition
  invariant} technique, which breaks the termination proof for a
program into multiple sub-proofs, intuitively for different cycles in
the program.  This technique is too general to allow for a complete
solution for all types of programs, but we survey classes of programs
for which it is both known that the technique is sufficient to prove
termination, and there are effective techniques of implementing it.

Just as there are witnesses that ensure termination, there are also
witnesses to non-termination: a trivial example is a state that is
repeated. We discuss certain more involved non-termination witnesses
in \Chp~\ref{chp:nt}.

Termination analysis of programs is a broad field and this \survey is
necessarily limited in scope. In particular, we leave out all
discussion of how termination analysis is done in actual programming
languages and how the abstract programs we are dealing with are
extracted from real code.  We leave out certain more expressive
abstract models, encompassing for example non-linear arithmetic or
probabilistic choice. Furthermore, we do not discuss termination
analysis of term rewriting systems, a field that has generated a
considerable amount of research.
The results we present attempt to show the state of the art for the
subproblems we consider---giving complete solutions wherever possible,
leaving out partial solutions and heuristic techniques, that may have
their own merits.  We also focus on presenting algorithms, examples and complexity results,
rather than on giving proofs.  The latter can be found in the given
references. Throughout the survey, we also list 16 \emph{open problems}
that may be the subject of further research.

\paragraph{Organisation of this \Survey.}
\Chp~\ref{chp:pre} provides the necessary mathematical background and
defines the programs we use. The other \chp{s} are independent of each other and can be read in any order, except for \Chp~\ref{chp:dti} that has some dependence on \Chp~\ref{chp:rfs}. \Chp~\ref{chp:dec} overviews results on
the decidability and undecidability of termination for
linear-constraint programs, and is mostly dedicated to the decidability of termination of so-called linear loops. \Chp~\ref{chp:rfs} discusses ranking
functions. We then overview works on disjunctive well-founded
invariants in \Chp~\ref{chp:dti} and witnesses for non-termination in
\Chp~\ref{chp:nt}. \Chp~\ref{chp:conc} concludes the
discussion.

\chapter{Preliminaries}
\label{chp:pre}

This \chp provides the mathematical
background~(Section~\ref{sec:mathbg}), over\-views definitions related
to polyhedra and linear programming (Section~\ref{sec:polyhedra}), and
defines the programs (Section~\ref{sec:programs}) we use in this
\survey.

\section{Mathematical Background}
\label{sec:mathbg}
  
This section provides the mathematical background used throughout the
\survey.

\subsection{Notations}
\label{sec:mathbg:notation}

For a set $A$, $x \in A$ means that $x$ is an element of $A$, and
$x \not\in A$ means that $x$ is not an element of $A$.
The empty set is denoted by $\emptyset$.
The cardinality of a set $A$, denoted by $|A|$, is the number of
elements in $A$.
For sets $A$ and $B$, $A \subseteq B$ means that $A$ is a subset of
$B$, $A \subset B$ means that $A$ is a strict subset of
$B$, $A \cup B$ is their union, $A \cap B$ is their intersection, and
$A \setminus B$ is their difference.
The Cartesian product of two sets $A$ and $B$, denoted by
$A \times B$, is the set of all ordered pairs $(a,b)$ where $a \in A$
and $b \in B$.
The $n$th Cartesian power of $A$ is $A^n = A \times \cdots \times A$
($n$ times).
The power set of $A$, denoted by $\wp(A)$, is the sets of all
subsets of $A$.

The set of real, rational, integer and non-negative integer numbers
are denoted respectively by $\reals$, $\rats$, $\ints$, and
$\nats$. Note that some literature uses $\nats$ to denote the set of
positive integers. We also use \areals to denote the set of real
algebraic numbers.
For $\numdom \in \{\reals,\rats,\ints\}$, we use
$\noneg{\numdom}$ for the corresponding subset of non-negative values.
We use $\vect{x}=(x_1,\ldots,x_n)$, where $x_i\in \numdom$, to
represent a row vector, and $\vec{x}=(x_1,\ldots,x_n)^\trans$
to represent a column vector.
The elements of $\numdom^n$ are column vectors, however, abusing
notation we might write $\vec{x}\in\numdom^n$ or
$\vect{x}\in\numdom^n$.
The addition of two vectors $\vec{a}=(a_1,\ldots,a_n)^\trans$ and
$\vec{b}=(b_1,\ldots,b_n)^\trans$ is defined as
$\vec{a}+\vec{b}=(a_1+b_1,\ldots,a_n+b_n)^\trans$.
Let $A\subseteq\numdom^{n+m}$, and let us write the vectors of $A$ as
$(\vect{x},\vect{y}) \in A$ where $\vect{x}\in\numdom^{n}$ and
$\vect{y}\in\numdom^{m}$.
The \emph{projection} of $A$ onto the $\vect{x}$-space is defined as
$\proj{\vect{x}}{A}=\{\vect{x} \in\numdom^n \mid \exists
\vect{y}\in\numdom^{m} ~\mbox{such that}~
(\vect{x},\vect{y}) \in A\}$.

The set of complex numbers is denoted by $\complex$. For
$c = a+bi \in \complex$, we use $\bar{c}=a-bi$ for its complex
conjugate. A complex number $c$ is said to be a root of the unity if
$c^n = 1$ for some integer $n>0$.

\subsection{Eigenvectors and Eigenvalues}
\label{sec:mathbg:eigenv}

\usedin{\Chp \ref{chp:dec}}

For a given square matrix $A \in \complex^{n \times n}$, a non-zero
vector $\vec{v}$ is an eigenvector if it satisfies the relationship
$A\vec{v} = \lambda\vec{v}$, where $\lambda$ is a scalar known as the
eigenvalue corresponding to $\vec{v}$.
The eigenvalues of a matrix are the roots of its characteristic
polynomial, $\det(A - \lambda I) = 0$, where $I$ is the identity
matrix. Note that the eigenvalues may be complex numbers even if all
entries of $A$ are real numbers.
The number of times an eigenvalue $\lambda$ is a root of the
characteristic polynomial is called its \emph{algebraic multiplicity}.
The concepts of eigenvalues and eigenvectors are essential for a wide
range of applications, including stability analysis of dynamical
systems and termination analysis.

\subsection{Exponential Polynomials}
\label{sec:exp-poly}
\label{sec:mathbg:exp-poly}

\usedin{\Chp \ref{chp:dec}}

Let \(\lambda_1,\ldots,\lambda_m \in \complex\) be distinct complex
numbers and \(e_1,\ldots,e_m\) positive integers.  Then the family of
\emph{exponential-polynomial} functions
\(p_{i,j} : \nats\rightarrow \complex\), for \(j\in\{1,\ldots,m\}\)
and \(i\in\{0,\ldots,e_j-1\}\), given by
\(p_{i,j}(n) = \binom{n}{i}\lambda_j^n\) is linearly independent over
\(\complex\).  Moreover if \(p:\nats\rightarrow\complex\) is a
\(\complex\)-linear combination of the \(p_{i,j}\), then \(p\) is
identically zero if and only if \(p(n)=0\) for \(e_1+\cdots+e_m\)
consecutive values \(n\in \nats\).  Both of the above facts can be
proved using generalised Vandermonde determinants~\citep[Proposition
2.11]{TUCS05}.

\subsection{Convexity}
\label{sec:mathbg:conv}

\usedin{\chp{s} \ref{chp:dec}--\ref{chp:rfs}}

The \emph{affine hull} of \(S\subseteq \reals^n\) is the smallest
affine set that contains \(S\), where an affine set is the translation
of a vector subspace of \(\reals^n\).
The affine hull of \(S\) can be characterised as follows:
\begin{equation*}
  \aff(S) := \left\{ \sum_{i=1}^k \alpha_i \vec{x}_i
    \mid k>0,\vec{x}_i \in S,\alpha_i \in \reals,
    \sum_{i=1}^k \alpha_i =1 \right \} \, .
\end{equation*}
The \emph{convex hull} of \(S\subseteq \reals^n\) is the smallest
convex set that contains \(S\).  The convex hull of \(S\) can be
characterised as follows:
\begin{equation*}
  \conv(S) := \left\{ \sum_{i=1}^k \alpha_i\vec{x}_i
    \mid k>0,\vec{x}_i \in S,\alpha_i \in \reals_{\geq 0},
    \sum_{i=1}^k \alpha_i =1 \right \} \, .
\end{equation*}
Clearly \(\conv(S)\subseteq \aff(S)\).

The \emph{relative interior} of a convex set \(S\subseteq \reals^n\)
is its interior \wrt the restriction of the Euclidean topology to
\(\aff(S)\). For example, the relative interior of a line segment in
three dimensions is the line segment minus its endpoints.  We have the
following easy proposition, characterising the relative interior.

\begin{proposition}
  Let
  \(S = \{\vec{a}_1,\ldots,\vec{a}_n\} \subseteq
  \reals^n\).  Then \(\vec{u}\) lies in the relative interior
  of \(\conv(S)\) if and only if  there exist \(\alpha_1,\ldots,\alpha_n > 0\)
  such that \(\vec{u}=\sum_{i=1}^n \alpha_i \vec{a}_i\) and
  \(\sum_{i=1}^n \alpha_i = 1\).
\label{prop:relative-int}
\end{proposition}

The \emph{conic hull} of \(S\subseteq \reals^n\) is the smallest
conic set that contains \(S\).  The conic hull of \(S\) can be
characterised as follows:
\begin{equation*}
  \cone(S) := \left\{ \sum_{i=1}^k \alpha_i\vec{x}_i
    \mid k>0,\vec{x}_i \in S,\alpha_i \in \reals_{\geq 0}\} \right \}.
\end{equation*}

\subsection{Lattices}
\label{sec:mathbg:lattices}

\usedin{\Chp \ref{chp:dec}}

A \emph{lattice of rank \(r\)} in \(\reals^n\) is a set
\begin{equation*}
  \Lambda :=\{ z_1 \vec{v}_1 + \cdots + z_r \vec{v}_r \mid z_1,\ldots,z_r \in \ints\} \, ,
\end{equation*}
where \(\vec{v}_1,\ldots,\vec{v}_r\) are linearly independent vectors
in \(\reals^n\). Given a convex set \(C\subseteq \reals^n\), define the
\emph{width} of \(C\) along a vector \(\vec{u} \in \reals^n\) to be
\begin{equation*}
  \sup\{ \vec{u}^\trans (\vec{x} - \vec{y})  \mid \vec{x},\vec{y} \in C \} \, .
\end{equation*}
Furthermore the \emph{lattice width} of \(C\) is the infimum over all
non-zero vectors \(\vec{u} \in \Lambda\) of the width of \(C\)
along \(\vec{u}\).

The following result~\citep{flatness99, Khinchin48} captures the
intuition that a convex set that contains no lattice point in its
interior must be ``thin'' in some direction.

\begin{theorem}[Flatness Theorem]
\label{thm:flatness}
Given a full-rank lattice \(\Lambda\) in \(\reals^n\), there exists
\(W\) such that any convex set \(C\subseteq \reals^n\) that has
non-empty interior and lattice width at least \(W\) contains a lattice
point in its interior.
\end{theorem}

Recall that \(C\subseteq \reals^n\) is said to be \emph{semi-algebraic}
if it is definable by a boolean combination of polynomial constraints
\(p(x_1,\ldots,x_n) > 0\), where \(p \in \ints[x_1,\ldots,x_n]\).

\begin{theorem}[\citealp{DinZ10,KP97}]
\label{thm:KP}
It is decidable whether a given convex semi-algebraic set
\(C\subseteq \reals^n\) contains an integer point, that is, whether
\(C \cap \ints^n \neq \emptyset\) and whether it contains a rational
point, that is, whether \(C \cap \rats^n\neq \emptyset\).
\end{theorem}

\subsection{Multiplicative Relations}
\label{subsec:grp_of_mul_rel}
\label{sec:mathbg:grp_of_mul_rel}

\usedin{\Chp \ref{chp:dec}}

Next, we introduce some concepts concerning groups of multiplicative
relations among algebraic numbers.

Let \(\TT = \set{z \in \complex : |z|=1 }\). We define the
\(s\)-dimensional torus to be \(\TT^s\), considered as a group under
component-wise multiplication.  Given a tuple of algebraic numbers
\(\bgamma = (\gamma_1, \cdots, \gamma_s) \in \TT^s\), the orbit
\(\set{\bgamma^n : n\in\nats}\), where $\bgamma^n$ is defined to be
$(\gamma_1^n,\ldots,\gamma_s^n)$, is a subset of \(\TT^s\).  In the
following we characterise the topological closure of the orbit as an
algebraic subset of \(\TT^s\).

The \emph{group of multiplicative relations} of \(\bgamma\in \TT^s\)
is defined to be the following additive subgroup of \(\ints^s\):
\begin{equation*}
L(\bgamma):=\set{\vec{v} \in \ints^s: \bgamma^{\vec{v}}=1},
\end{equation*}
where \(\bgamma^{\bv}\) is defined to be
\(\gamma_1^{v_1}\cdots\gamma_s^{v_s}\) for \(\bv\in\ints^s\), that is,
exponentiation acts coordinate-wise.  Since \(L(\bgamma)\) is a
subgroup of \(\ints^s\), it is a free Abelian group and hence has a
finite basis.

Corresponding to \(L(\bgamma)\), we consider the following
multiplicative subgroup of \(\TT^s\):
\begin{equation*}
  T(\bgamma)=\set{\bmu \in \TT^s : \forall\bv \in L(\bgamma),
    \bmu^{\bv}=1}.
\end{equation*}
If \(\mathcal{B}\) is a basis of \(L(\bgamma)\), we can equivalently
characterise \(T(\bgamma)\) as
\(\set{\bmu \in \TT^s : \forall \bv \in \mathcal{B}, \bmu^{\bv}=1}\).
Crucially, this finitary characterisation allows us
to represent \(T(\bgamma)\) as an algebraic set in \(\TT^s\) by the following result:
\begin{theorem}
\label{thm:dense}
Let \(\bgamma \in \TT^s\). Then the orbit
\(\set{\bgamma^k : k \in \nats}\) is a dense subset of \(T(\bgamma)\).
\end{theorem}
As shown by~\citet[Proposition 3.5]{OuaknineW14a},
Theorem~\ref{thm:dense} is a (relatively straightforward) consequence
of Kronecker's theorem on simultaneous Diophantine approximation.

\begin{example}
  Let   $\gamma_1:=\frac{1+2i}{\sqrt{5}},\gamma_2:=\frac{-3+4i}{5},
  \gamma_3:=\frac{-3-4i}{5}$.  Then the multiplicative relations
  $\gamma_1^2\gamma_3=1$ and $\gamma_2\gamma_3=1$ determine a basis
  $(2,0,1), (0,1,1)$ of the lattice
  $L(\vec{\gamma}) \subseteq \mathbb Z^3$ of all multiplicative
  relations among $\gamma_1,\gamma_2,\gamma_3$.  By Theorem~\ref{thm:dense}, we thus have that
  $\{(\gamma_1^n,\gamma_2^n,\gamma_3^n) : n \in \mathbb N\}$ is dense
  in the set
  \[ T(\vec{\gamma}) := \{ \vec{\mu}\in\mathbb T^3 : \mu_1^2\mu_3 = 1 \text{ and } \mu_2\mu_3=1 \} \, .\]
\end{example}

\section{Polyhedra and Linear Programming}
\label{sec:polyhedra}

We recall some definitions related to polyhedra, integer polyhedra and
linear programming (\lp), mostly as presented
by~\citet{BG14}. \citet{Schrijver86} is a useful reference for the
theory of polyhedra and \lp.

\begin{figure}

\begin{center}
\begin{tikzpicture}[scale=0.5]

\begin{scope}[shift={(0,0)}]

  \draw [->,thick] (0,0) node  {} -- (9,0) node (xaxis) [right] {$x_1$};
  \draw [->,thick] (0,0) node  {} -- (0,8) node (yaxis) [above] {$x_2$};

  \coordinate (a) at (0.5,3.5);
  \coordinate (b) at (4.5,7.5);
  \coordinate (c) at (3.333333333,0.6666666);
  \coordinate (e) at (8.5,3.25);

  \fill[fill=black!20] (a) -- (b) -- (e) -- (c) -- cycle;

  \fill[] (a) circle (3pt);
  \fill[] (c) circle (3pt);

  \draw [thick]  (a) -- (b) {};
  \node[rotate=45] () at (2.5,6) {\scalebox{0.8}{\textcolor{blue}{$x_2{-}x_1{\le}3$}}}; 

  \draw [thick]  (a) -- (c) {};
  \node[rotate=-45] () at (1.7,1.8) {\scalebox{0.8}{\textcolor{blue}{${-}x_1{-}x_2{\le}{-}4$}}}; 

  \draw [thick]  (c) -- (e) {};
  \node[rotate=27] () at (6.5,1.7) {\scalebox{0.8}{\textcolor{blue}{$\frac{1}{2}x_1{-}x_2{\le}1$}}}; 

  \draw[dotted] (yaxis |- a) node[left] {\scalebox{0.8}{$\mathbf{\frac{7}{2}}$}}
            -| (xaxis -| a) node[below] {\scalebox{0.8}{$\mathbf{\frac{1}{2}}$}};

  \draw[dotted] (yaxis |- c) node[left] {\scalebox{0.8}{$\mathbf{\frac{2}{3}}$}}
            -| (xaxis -| c) node[below] {\scalebox{0.8}{$\mathbf{\frac{10}{3}}$}};

  \node[] at (8,7) {\textcolor{purple}{$\poly{P}$}};

  \draw [opacity=0.8,red,very thick,dashed,->]  (3,3) -- (5.5,5.5) {};
  \draw [opacity=0.8,red,very thick,dashed,->]  (3,3) -- (6.3,4.4) {};

\end{scope}

\begin{scope}[shift={(12,0)}]

  \draw [->,thick] (0,0) node  {} -- (9,0) node (xaxis) [right] {$x_1$};
  \draw [->,thick] (0,0) node  {} -- (0,8) node (yaxis) [above] {$x_2$};

  \coordinate (a) at (0.5,3.5);
  \coordinate (b) at (4.5,7.5);
  \coordinate (c) at (3.333333333,0.6666666);
  \coordinate (e) at (8.5,3.25);
  \coordinate (f) at (3,1);
  \coordinate (g) at (1,3);
  \coordinate (h) at (1,4);
  \coordinate (i) at (4,1);

  \fill[pattern=dots]  (b) -- (h) -- (g) -- (f) -- (i) -- (e) -- cycle;
  \fill[fill=black!20] (a) -- (h) -- (g) -- cycle;
  \fill[fill=black!20] (f) -- (c) -- (i) -- cycle;

  \fill[] (f) circle (3pt);
  \fill[] (g) circle (3pt);
  \fill[] (h) circle (3pt);
  \fill[] (i) circle (3pt);

  \draw [thick]  (f) -- (i) {};
  \node[inner sep =0] () at (3.55,1.45) {\scalebox{0.8}{\textcolor{blue}{$x_2{\ge}1$}}}; 

  \draw [thick]  (g) -- (h) {};
  \node[rotate=90, inner sep =0] () at (1.5,3.45) {\scalebox{0.8}{\textcolor{blue}{$x_1{\ge}1$}}}; 

  \draw [thick]  (h) -- (b) {};
  \node[rotate=45] () at (2.5,6) {\scalebox{0.8}{\textcolor{blue}{$x_2{-}x_1{\le}3$}}}; 

  \draw [thick]  (g) -- (f) {};
  \node[rotate=-45] () at (1.7,1.8) {\scalebox{0.8}{\textcolor{blue}{${-}x_1{-}x_2{\le}{-}4$}}}; 

  \draw [thick]  (i) -- (e) {};
  \node[rotate=27] () at (6.5,1.7) {\scalebox{0.8}{\textcolor{blue}{$\frac{1}{2}x_1{-}x_2{\le}1$}}};

  \draw[dotted] (yaxis |- f) node[left] {\scalebox{0.8}{$\mathbf{1}$}}
            -| (xaxis -| f) node[below] {\scalebox{0.8}{$\mathbf{3}$}};

  \draw[dotted] (yaxis |- g) node[left] {\scalebox{0.8}{$\mathbf{3}$}}
            -| (xaxis -| g) node[below] {\scalebox{0.8}{$\mathbf{1}$}};

  \draw[dotted] (yaxis |- h) node[left] {\scalebox{0.8}{$\mathbf{4}$}}
            -| (xaxis -| h) node[below] {\scalebox{0.8}{$\mathbf{1}$}};

  \draw[dotted] (yaxis |- i) node[left] {\scalebox{0.8}{$\mathbf{1}$}}
            -| (xaxis -| i) node[anchor=west,below] {\scalebox{0.8}{$\mathbf{4}$}};

  \phantom{\draw[dotted] (yaxis |- a) node[left] {\scalebox{0.8}{$\mathbf{\frac{7}{2}}$}}
            -| (xaxis -| a) node[below] {\scalebox{0.8}{$\mathbf{\frac{1}{2}}$}};}

  \phantom{\draw[dotted] (yaxis |- c) node[left] {\scalebox{0.8}{  $\mathbf{\frac{2}{3}}$}}
            -| (xaxis -| c) node[below] {\scalebox{0.8}{$\mathbf{\frac{10}{3}}$}};}

  \node[] at (8,7) {\textcolor{purple}{$\inthull{\poly{P}}$}};
  \draw [opacity=0.8,red,very thick,dashed,->]  (3,3) -- (5.5,5.5) {};
  \draw [opacity=0.8,red,very thick,dashed,->]  (3,3) -- (6.3,4.4) {};
 \end{scope}
\end{tikzpicture}

\end{center}
 
\caption{A polyhedron $\poly{P}$ and its integer hull $\inthull{\poly{P}}$ (Figure from \citep{BG14}).}
\label{fig:poly}

\end{figure}

\subsection{Polyhedra}

\usedin{\chp{s} \ref{chp:rfs}--\ref{chp:nt}}

For $\numdom \in \{\reals,\rats\}$, a \emph{convex polyhedron}
$\poly{P} \subseteq \numdom^n$ (\emph{polyhedron} for short) is the
set of solutions of a set of inequalities $A\vec{x} \le \vec{b}$,
namely $\poly{P}=\{ \vec{x}\in\numdom^n \mid A\vec x \le \vec b \}$,
where $\vec x\in\numdom^n$, $A \in \rats^{m \times n}$ is a rational
matrix of $n$ columns and $m$ rows, and $\vec b \in \rats^m$ is a
column vector of $m$ rational values.
We say that $\poly{P}$ is specified by $A\vec{x} \le \vec{b}$.
We use calligraphic letters, such as $\poly{P}$ and $\poly{Q}$ to
denote polyhedra.
We sometimes write $\poly{P}$ as a set that includes the inequalities
of $A\vec{x} \le \vec{b}$.

The set of \emph{recession directions} of a polyhedron $\poly{P}$
specified by $A\vec{x} \le \vec b$ is the set
$\ccone(\poly{P}) = \{ \vec{y}\in\numdom^n \mid A\vec{y} \le
\vec{0}\}$.
$\poly{P}$ is said to be bounded if $\ccone(\poly{P})=\{\vec{0}\}$.

\begin{example}
\label{ex:poly:polyhedron}
Consider the polyhedron $\poly{P}$ of Figure~\ref{fig:poly} (on the
left). The points defined by the gray area and the black borders are
solutions to the system of linear inequalities
$\{x_2-x_1\le 3, \; -x_1-x_2\le -4 ,\; \frac{1}{2}x_1-x_2\le 1\}$.
\end{example}

The projection of a set onto a sub-space as defined in
Section~\ref{sec:mathbg:notation} applies to polyhedra as
well. Namely, let $\poly{P}\subseteq\numdom^{n+m}$ be a polyhedron,
and let $\tr{\vec{x}}{\vec{y}} \in \poly{P}$ be such that
$\vec{x}\in\numdom^{n}$ and $\vec{y}\in\numdom^{m}$.
The \emph{projection} of $\poly{P}$ onto the $\vec{x}$-space is
defined as
$\proj{\vec{x}}{\poly{P}}=\{\vec{x}\in\numdom^n \mid \exists
\vec{y}\in\numdom^{m} ~\mbox{such that}~ \tr{\vec{x}}{\vec{y}} \in
\poly{P}\}$.

\subsection{Integer Polyhedra}

\usedin{\Chp \ref{chp:rfs}}

For a given polyhedron $\poly{P} \subseteq \numdom^n$ we let
$\intpoly{\poly{P}}$ be $\poly{P} \cap \ints^n$, \ie, the set of
integer points of $\poly{P}$. The \emph{integer hull} of $\poly{P}$,
commonly denoted by $\inthull{\poly{P}}$, is defined as the convex
hull of $\intpoly{\poly{P}}$, \ie, every rational point of
$\inthull{\poly{P}}$ is a convex combination of integer points. This
property is fundamental to the results presented in the next sections.
It is known that $\inthull{\poly{P}}$ is also a polyhedron.  An
\emph{integer polyhedron} is a polyhedron $\poly{P}$ such that
$\poly{P} = \inthull{\poly{P}}$, and in such case we say that
$\poly{P}$ is \emph{integral}.

\begin{example}
\label{ex:poly:inthull}
The integer hull $\inthull{\poly{P}}$ of polyhedron $\poly{P}$ of
Figure~\ref{fig:poly} (on the left) is given in the same figure (on
the right). It is defined by the dotted area and the black border, and
is obtained by adding the inequalities $x_1 \ge 1$ and $x_2 \ge 1$ to
$\poly{P}$. The two gray triangles next to the edges of
$\inthull{\poly{P}}$ are subsets of $\poly{P}$ that were eliminated
when computing $\inthull{\poly{P}}$.
\end{example}

The integer hull of a polyhedron $\poly{P}$ can be computed in
exponential time~\citep{Hartmann88,CharlesHK09}. Note that this
algorithm supports only bounded polyhedra, the integer hull of an
unbounded polyhedron is computed by considering a corresponding
bounded one~\citep[Th.~16.1, p.~231]{Schrijver86}.

\subsection{Generator Representation} 

\usedin{\chp{s} \ref{chp:rfs} and \ref{chp:nt}}

Polyhedra also have a \emph{generator representation} in terms of
vertices and rays%
\footnote{Technically, the $\vec x_1,\ldots,\vec x_n$ are only
  vertices if the polyhedron is \emph{pointed} (\ie, its recession
  cone does not contain lines).}%
, written as
\[
\poly{P} = \convhull\{\vec x_1,\dots,\vec x_m\} + \cone\{\vec
y_1,\dots,\vec y_t\} \,.
\]
This means that $\vec x\in \poly{P}$ if and only if $\vec x =
\sum_{i=1}^m a_i \vec x_i + \sum_{j=1}^t b_j \vec y_j$ for some
rationals $a_i,b_j\ge 0$, where $\sum_{i=1}^m a_i = 1$. 
An important property is that if $\poly{P}$ is integral, then there is a generator representation in
which all $\vec{x}_i$ and $\vec{y}_j$ are integer.

\begin{example}
\label{ex:poly:genrep}
The generator representations of $\poly{P}$ and $\inthull{\poly{P}}$
of Figure~\ref{fig:poly} are
\[
\begin{array}{rl}
\poly{P} = & \convhull\{(\frac{1}{2},\frac{7}{2}),(\frac{10}{3},\frac{2}{3})\}+ \cone\{(1,1),(7,3)\}\\[1ex]
\inthull{\poly{P}} = & \convhull\{(1,3),(1,4),(3,1),(4,1)\}+ \cone\{(1,1),(7,3)\}\\
\end{array}
\]
The points in $\convhull$ are vertices, they correspond to the points
marked with $\bullet$ in Figure~\ref{fig:poly}.  The rays are the
vectors $(1,1),(7,3)$; they describe a direction, rather than a
specific point, and are therefore represented in the figure as arrows.
Note that the vertices of $\inthull{\poly{P}}$ are integer points,
while those of $\poly{P}$ are not.
The point $(3,2)$, for example, is defined as
$\frac{5}{17}\cdot(\frac{1}{2},\frac{7}{2}) +
\frac{12}{17}\cdot(\frac{10}{3},\frac{2}{3})+
\frac{1}{2}\cdot(1,1)+0\cdot(7,3)$ in $\poly{P}$, and as
$0\cdot(1,3)+\frac{1}{3}\cdot(1,4)+0\cdot(3,1)+\frac{2}{3}\cdot(4,1)+
0\cdot(1,1)+0\cdot(7,3)$ in $\inthull{\poly{P}}$.
\end{example}

\subsection{Size of Polyhedra} 

\usedin{\Chp \ref{chp:rfs}}

Complexity of algorithms on polyhedra is measured in this \survey by
running time, on a conventional computational model (polynomially
equivalent to a Turing machine), as a function of the \emph{bit-size}
of the input.  Following~\citet[Sec. 2.1]{Schrijver86}, we define the
bit-size of an integer $x$ as
$\size{x} = 1 + \lceil \log (|x|+1)\rceil$; the bit-size of an
$n$-dimensional vector $\vec a$ as
$\size{\vec{a}} = n+\sum_{i=1}^n \size{a_i}$; and the bit-size of an
inequality $\vec{a}^\trans \vec{x} \le c$ as
$1+\size{c}+\size{\vec{a}}$.
For a polyhedron $\poly{P} \subseteq \numdom^n$ defined by
$A\vec x \le \vec b$, we let $\bitsize{\poly{P}}$ be the bit-size of
$A\vec x \le \vec b$, which we can take as the sum of the sizes of the
inequalities.

\subsection{Farkas' Lemma}
\label{sec:farkas}

\usedin{\chp{s} \ref{chp:rfs}--\ref{chp:nt}}

Many of the techniques presented in this \survey heavily rely on (a
variation) of Farkas' Lemma~\citep[p.~94]{Schrijver86}, which states
that a polyhedron $\poly{P}\subseteq\numdom^n$, with
$\numdom \in \{\rats,\reals\}$, specified by $A{\vec x} \le \vec{c}$,
entails an inequality $\vect{\lambda}\vec{x} \le \lambda_0$ if and
only if there is a vector of non-negative coefficients $\vect{\mu}$,
of appropriate dimension, such that the following holds:
\begin{align}
\vect{\mu} A =& \vect{\lambda}\label{farkas:coeff} \\
\vect{\mu}\vec{c} \le & \lambda_0\label{farkas:const}
\end{align}
The vector $\vect{\mu}$ will be called the Farkas' coefficients in the
rest of this \survey.
It is also easy to show that
$\vect{\lambda}\vec{x} \le \lambda_0$ is entailed by
$\intpoly{\poly{P}}$, \ie, by the set of integer points of
$\poly{P}$, if and only if it is entailed by $\inthull{\poly{P}}$.
This follows from the fact that if the inequality holds for points
$\vec{x}_1\in\intpoly{\poly{P}}$ and $\vec{x}_2\in\intpoly{\poly{P}}$,
then it holds for their convex combinations.
Note that (\ref{farkas:coeff},\ref{farkas:const}) are linear
constraints when considering $\lambda_0$, $\vect{\lambda}$, and
$\vect{\mu}$ as unknowns, and thus synthesising entailed inequalities
can be done in polynomial time by seeking a solution
for~(\ref{farkas:coeff},\ref{farkas:const}).
Note also that for some techniques, such as those based on templates,
$A$ and $\vec{c}$ might also include unknowns, and thus
(\ref{farkas:coeff},\ref{farkas:const}) are non-linear in such case.

\begin{example}
\label{ex:farkas}
Consider a polyhedron $\poly{P}$ defined by the following set of
inequalities (those of Figure~\ref{fig:poly} on the left)
\begin{equation}
\{ -x_1-x_2  \le -4, \; x_2-x_1 \le 3,\; \frac{1}{2}x_1-x_2 \le 1 \}
\end{equation}
and its matrix representation $A\vec{x}\le \vec{c}$ where
\begin{align*}
    A=
  \begin{pmatrix}
          -1 & -1 & \\
          -1 &  1 & \\
 \frac{1}{2} & -1 & 
  \end{pmatrix}
  &
~
  &
    \vec{c}=
    \begin{pmatrix}
      -4 \\
       3 \\
       1 \\
    \end{pmatrix}
\end{align*}
Let $\lambda_1x_1+\lambda_2x_2 \le \lambda_0$ be an implied inequality
template, and $\vect{\mu}=(\mu_0,\mu_1,\mu_2)$.
Note that $\vect{\mu}$ has components like the number of inequalities
(which coincide with the number of rows of $A$). To synthesise
inequalities implied by $A\vec{x}\le \vec{c}$, we
use~(\ref{farkas:coeff},\ref{farkas:const}) to generate the following
constraints system:
\begin{equation}
\label{eq:poly:farkas:ex}
\begin{array}{r}
  -\mu_0 - \mu_1 + \frac{1}{2}\mu_2 = \lambda_1,\;
  -\mu_0 + \mu_1 - \mu_2 = \lambda_2,\;\\
  -4\mu_0 + 3\mu_1 + \mu_2 \le \lambda_0 \\
  \mu_0\ge0,\; \mu_1\ge0,\;\mu_2\ge0 \\
\end{array}
\end{equation}
The constraints in the first line come from~\eqref{farkas:coeff}, and
correspond to multiplying $\vect{\mu}$ by the columns of $A$.
The constraint in the second line comes from~\eqref{farkas:const}, and
correspond to multiplying $\vect{\mu}$ by $\vec{c}$.
The third line is used to require the coefficients $\vect{\mu}$ to be
non-negative.

The valuation
$\{
\lambda_1\mapsto -1, \lambda_2\mapsto0, \lambda_0\mapsto-\frac{1}{2},
\mu_0 \mapsto \frac{1}{2}, \mu_1 \mapsto \frac{1}{2}, \mu_2\mapsto 0\}$
is a solution for~\eqref{eq:poly:farkas:ex}, and thus $-x_1 \le -\frac{1}{2}$
is an implied inequality.

If we are interested in an implied inequality of a specific form, \eg,
one in which $\lambda_1=\lambda_2$ or $\lambda_1\le\lambda_2$, we can
add a corresponding constraint to~\eqref{eq:poly:farkas:ex}.
If we are interested in several implied inequalities, that share some
coefficients, we can solve several instances
of~\eqref{eq:poly:farkas:ex} at the same time (even if each is implied
by a different $A\vec{x}\le\vec{c}$).
Finally, if we are interested in inequalities that are implied only by
$\intpoly{\poly{P}}$, \ie, the integer points of $\poly{P}$, we can use the
constraints that represent its integer-hull $\inthull{\poly{P}}$ (the
polyhedron of Figure~\ref{fig:poly} on the right).
\end{example}

\subsection{Linear Programming}

\usedin{\Chp \ref{chp:rfs}}

A linear programming (\lp) problem concerns the maximisation or
minimisation of a linear objective function, such as $\vec{a}\vec{x}$,
subject to a system of linear inequalities, typically represented as
$A\vec{x} \le \vec{c}$. It can also refer to the problem of finding a
solution that satisfies the inequalities.
When the variables are restricted to take real or rational values, an
\lp problem can be solved in polynomial time. However, if the
variables are restricted to be integers, the problem is known as an
integer linear programming problem, which is \nph.

\section{Programs}
\label{sec:programs}
A program is often modelled as a transition relation
$T \subseteq S\times S$, where $S$ is a set of possible program
states.
An execution, or a trace, is a (possibly infinite) sequence
$s_0,s_1,\ldots$ where $(s_i,s_{i+1})\in T$.
A transition relation $T \subseteq S\times S$ or a set of states
$S'\subseteq S$ are often defined by predicates (formulas whose
models define the elements of the set), and thus we write $T(s,s')$
and $S'(s)$ instead of $(s,s') \in T$ and $s \in S'$.
The successors operator $\post_T:S \to S$ is
$\post_T(X) = \{ s' \in S \mid s \in X, (s,s') \in T\}$, and the
predecessors operator $\pre_T:S \to S$ is
$\pre_T(X) = \{ s \in S \mid s' \in X, (s,s') \in T\}$.
For an initial set of states $S_0$, the set of reachable states
$\reach{T}{S_0}$ contains the states that can be reached from $S_0$ by
a finite trace; this is the least fixpoint of
$F(X) = S_0 \cup \post_T(X)$.
The restriction of $T$ to the reachable states $\reach{T}{S_0}$ is
defined as
$\trres{T}{S_0} = \{ (s,s') \in T \mid s \in \reach{T}{S_0} \}$.
The composition of transition relations $T_1,T_2 \subseteq S\times S$
is defined as
$T_1 \circ T_2=\{ (s,s'') \in S\times S \mid (s,s') \in T_1, (s',s'')
\in T_2\}$.

We say that $T$ is \emph{terminating} for an initial state $s_0\in S$,
if there are no infinite traces starting with $s_0$, and
\emph{non-terminating} if such an infinite trace exists.
We say that $T$ is \emph{universally terminating} if it is terminating
for any initial state. Equivalently, $T$ is universally terminating if
and only if it is well-founded (when considered as a ``greater than''
relation). 
Note that termination of $T$ \wrt $S_0$ is equivalent to universal
termination of $\trres{T}{S_0}$.
As in much of the literature, the unqualified term \emph{termination}
means universal termination if no reference to particular initial
states is made, and \emph{non-termination} means the negation of universal termination.
The problem of deciding whether $T$ is terminating for a given \emph{single}
initial state $s_0\in S$ is known as the \emph{halting} problem.
%

\subsection{Linear-Constraint Control-Flow Graphs}

Structured program representations, such as the \emph{Control-Flow
  Graph} (\cfg), are often employed for practical reasons since they
are easily derived from real-world programming languages. Furthermore,
our focus is restricted to program states that involve only numerical
variables.

A \cfg is a tuple $P=(V,\numdom,L,\ell_{0},E)$, where:
\begin{enumerate}[(i)]
\item $V=\{x_1,\ldots,x_n\}$ is a finite set of program variables
  taking values from a numerical domain
  $\numdom \in \{\reals,\rats,\ints\}$;
\item $L=\{\ell_0,\ldots,\ell_k\}$ is a finite set of locations, where
  $\ell_0\in L$ represents the initial location; and
\item $E \subseteq L \times \wp(\numdom^n\times \numdom^n) \times L$ is a set
  of edges annotated with transition relations over $\numdom^n$.
\end{enumerate}
An edge $(\ell,T,\ell') \in E$ defines how an execution step can move
from location $\ell$ to $\ell'$: if the execution is at location
$\ell$, the variables have values $\vec{x}\in \numdom^n$, and
$(\vec{x},\vec{x}') \in T$ then we can move to location $\ell'$ and
set the program variables to $\vec{x}'$.
Sometimes we write $T_{\ell,\ell'} \in E$ to refer to the transition
relation directly. We can also write $(\ell,T,\ell') \in P$ and
$\ell \in P$ instead of referring to the sets of edges and locations.
Viewing states as tuples $(\ell,\vec{x}) \in L\times \numdom^n$, it is
easy to see that a \cfg $P$ induces a transition relation
$T_{P} \subseteq (L\times \numdom^n)\times (L\times \numdom^n)$.
When the location is known from context, we sometimes omit the location and
refer to the variables $\vec{x}$ as ``the state''.

A common way of representing a numerical transition relation
$T \subseteq \numdom^n \times \numdom^n$ is as a conjunction of linear
constraints, where the $i$th constraint is of the form
$\sum_{j=1}^n a_{ij} x_j + \sum_{j=1}^n a'_{ij} x_j' \le
c_i$.
Here, $(x_1,\ldots,x_n)^\top$ represents the current state and
$(x_1',\ldots,x_n')^\top$ represents a possible successor.
Such a transition relation is a polyhedron, and is specified by
$A''\vec{x}'' \le \vec{c}''$ where $\vec{x}''=\tr{\vec{x}}{\vec{x}'}$,
$A''\in\rats^{m\times 2n}$, and $\vec{c}''\in\rats^{m\times 1}$ for
some $m\ge 1$ (the number of constraints in the conjunction). Note
that all coefficients are rational, but in some settings we will assume
that they are integer.
We call this polyhedron a \emph{transition polyhedron} and denote it
by $\transitions \subseteq \reals^{2n}$. Note that if the domain is the integers,
the set of transitions is
$\intpoly{\transitions} \subseteq \ints^{2n}$.

\begin{remark}
  For simplicity, this \survey always uses  non-strict linear
  inequalities (\ie, $\le$). Many of the results presented here can be
  generalised to include strict inequalities, a point we will
  explicitly note. This distinction is crucial only for rational and
  real variables; for integers, strict inequalities can be converted
  into equivalent non-strict ones, so we may use both in our examples.
\end{remark}

We sometimes write $A''\vec{x}'' \le \vec{c}''$ as
$A\vec{x}+A'\vec{x}' \le \vec{c}''$ for appropriate
$A,A' \in \rats^{m\times n}$, or as
$B\vec{x}\le \vec{b} \wedge A\vec{x}+A'\vec{x}' \le \vec{c}$ when we
are explicitly interested%
\footnote{If the transition relation is obtained from a \clang-like
  program, then it is easy to extract the condition, otherwise we can
  project $A''\vec{x}'' \le \vec{c}''$ onto $\vec{x}$.}
in the condition ($B\vec{x}\le \vec{b}$) that allows taking the
corresponding edge (the \emph{guard} of the edge).
We may also use $=$ and $\ge$ instead of $\le$ when writing
constraints, as such constraints can be naturally converted to use
$\le$ only. We also write a conjunction of inequalities as a set, in
which case the empty set represents the constraint $\mathit{true}$
(\ie, the whole space). We also write $\poly{P}_1\land\poly{P}_2$ to
refer to the polyhedron specified by the constraints of both
$\poly{P}_1$ and $\poly{P}_2$ (even if they use different variables,
in which case the dimension of $\poly{P}_1\land\poly{P}_2$ is larger
than those of $\poly{P}_1$ and $\poly{P}_2$).
We call a transition polyhedron \emph{deterministic} if, for a given
state $\vec x \in \numdom^n$ there is at most one state
$\vec{x}'\in\numdom^n$ such that $(\vec{x},\vec{x}')\in\transitions$.

A linear-constraint \cfg is a \cfg where edges are annotated with
transition polyhedra. In this \survey, the term \cfg will refer to a
linear-constraint \cfg unless otherwise specified.

\begin{remark}
  Linear-constraint \cfgs can also represent programs that manipulate
  data structures. This is usually done by abstracting the data
  structures into numerical representations---for example, the length
  of a list, the depth of a tree,
  \etc.~\citep{LS97,LJB01,BCGGV07,SMP10,MTLT10}. While these
  abstractions are typically sound for proving termination, they are
  not always sound for proving non-termination.
\end{remark}

When proving termination and non-termination for \cfgs, we are
primarily interested in executions that start from the initial
location $\ell_0$. We may also restrict the input variables to a given
set of values $S_0 \subseteq \numdom^n$; we sometimes omit $S_0$
because it can be represented by adding an initial transition out of
$\ell_0$.
Universal termination for \cfgs allows starting at any location with
any values for the variables.

\begin{remark}
  While universal termination for \cfgs typically refers to starting
  at the entry location $\ell_0$ with arbitrary variable values, the
  property of terminating from \emph{any} arbitrary location is more
  accurately termed \emph{mortality}. In this \survey, however, we use
  the term universal termination to refer to mortality.
  The standard case starting from $\ell_0$ is treated as a specific
  instance of termination with respect to an initial set of states
  $S_0$, which in this case corresponds to the entire domain
  $\numdom^n$.
\end{remark}

Many of the termination and non-termination techniques in this \survey
rely on local, edge-level reasoning. Consequently, they cannot easily
account for information from preceding edges or assumptions about the
initial state unless that information is propagated to each location
using invariants.

\begin{definition}
\label{def:inv}
We call $\inv_\ell \subseteq \numdom^n$ an invariant for a location
$\ell$ if, for any execution starting from $(\ell_0,\vec{x})$ where
$\vec{x} \in S_0$, all reachable states
$(\ell, \vec{x}) \in \reach{T_P}{S_0}$ satisfy
$\vec{x} \in \inv_\ell$.
\end{definition}

In this \survey, we focus on polyhedral invariants. 

\begin{remark}
  Inferring polyhedral invariants is outside the scope of this
  \survey; we assume they have been inferred beforehand and are
  provided as input. However, some techniques combine invariant
  inference with the search for termination (or non-termination)
  witnesses, and we will explicitly comment on those.
\end{remark}

\begin{figure}[t]

\begin{center}
\begin{tikzpicture}[>=latex,line join=bevel,]

\begin{scope}[shift={(3,2.25)}]
  \node [inner sep=0pt, align=left, font=\ttfamily] at (0,0) {
\begin{minipage}{6.5cm}
\begin{lstlisting}[escapechar=\#, xleftmargin=0pt]
while(x >= 0 && y >= 0) {
  if (nondet()) {
    while (y <= z && nondet())
      y++;
    x--;
  } else {
    y--;
  }
}
\end{lstlisting}
\end{minipage}
};
\end{scope}

\begin{scope}[shift={(2.2,0.65)}]
  
  \node (l4) at (3,0) [loc] {$\ell_{4}$};
  \node (l3) at (0,0) [loc] {$\ell_{3}$};    
  \node (l5) at (0,1) [loc] {$\ell_{5}$};
  \node (l2) at (1.5,0) [loc] {$\ell_{2}$};
  \node (l1) at (1.5,1) [loc] {$\ell_{1}$};
  \node (l0) at (3,1) [inloc] {$\ell_0$};
  
  \draw [tre] (l0) to[] node[tr,above] {\scalebox{0.8}{$\transitions_0$}} (l1);
  \draw [tre] (l1) to[] node[tr,right] {\scalebox{0.8}{$\transitions_1$}} (l2);
  \draw [tre] (l2) to[] node[tr,below] {\scalebox{0.8}{$\transitions_2$}} (l3);
  \draw [tre, bend right=40] (l3) to[] node[tr,above] {\scalebox{0.8}{$\transitions_5$}} (l4);
  \draw [tre] (l4) to[] node[tr,right] {\scalebox{0.8}{$\transitions_6$}} (l1);
  \draw [tre,bend right=20] (l1) to[] node[tr,above] {\scalebox{0.8}{$\transitions_7$}} (l5);
  \draw [tre,bend left=20] (l1) to[] node[tr,below] {\scalebox{0.8}{$\transitions_8$}} (l5);
  \draw [tre] (l3) to[out=70,in=110,loop] node[tr,right] {\scalebox{0.8}{$\transitions_4$}}  (l3);
  \draw [tre] (l2) to[] node[tr,below] {\scalebox{0.8}{$\transitions_3$}}  (l4);
\end{scope}

\begin{scope}[shift={(8,2.25)}]
  \node [align=center, inner sep=0pt, font=\ttfamily\footnotesize] at (0.5,0) {
    \begin{minipage}{5.75cm}
\begin{center}
\(
  \begin{array}{|@{}r@{\hskip 2pt}l@{}|}
    \hline
\transitions_0{:}&  \{x'=x,y'=y,z'=z\} \\
\transitions_1{:}&  \{x \ge 0, y \ge 0,x'=x,y=y,z'=z\} \\
\transitions_2{:}&  \{x'=x,y'=y,z'=z\} \\
\transitions_3{:}&  \{x'=x,y'=y-1,z'=z\} \\
\transitions_4{:}&  \{y \le z, x'=x,y'=y+1,z'=z\} \\
\transitions_5{:}&  \{x'=x-1,y=y,z'=z\} \\
\transitions_6{:}&  \{x'=x,y=y,z'=z\} \\
\transitions_7{:}&  \{x \le -1, x'=x,y=y,z'=z\} \\
\transitions_8{:}&  \{y \le -1, x'=x,y=y,z'=z\} \\
    \hline
  \end{array}
  \)
  
\(
\begin{array}{|@{}l@{}|}
\hline
\poly{S}_0 = \pinv_{\ell_0}= \pinv_{\ell_2}= \pinv_{\ell_3}= \{x \ge 0, y \ge 0\} \\
\pinv_{\ell_2}=  \pinv_{\ell_3}=\{x \ge 0, y \ge 0\} \\
\pinv_{\ell_1}= \pinv_{\ell_4} = \pinv_{\ell_5}= \{x \ge -1, y \ge -1\} \\
\hline
\end{array}
\)
\end{center}
\end{minipage}
};
\end{scope}
\end{tikzpicture}
\end{center}

\caption{A program, its corresponding \cfg, and invariants.}%
\label{fig:cfg:1}

\end{figure}
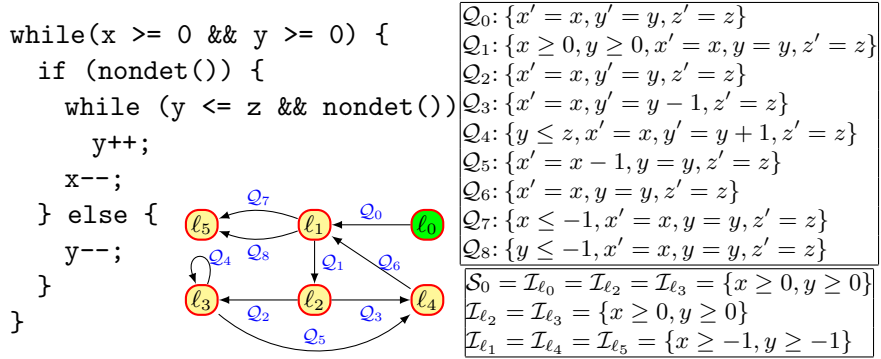

\begin{example}
\label{ex:cfg:1}  
Figure~\ref{fig:cfg:1} presents an imperative program (in a
\clang-like language), along with a possible corresponding \cfg and
its invariants. The \lstinline{nondet()} instruction produces an
arbitrary (integer) value and is typically used to abstract
expressions that cannot be modelled with linear arithmetic.
\end{example}

\subsection{Linear-Constraint Loops}
\label{sec:loops}

This section presents special cases of \cfgs that are in the
form of loops.

\subsubsection{Multi-path Linear-Constraint Loops}

A \cfg with a single node and $k$ edges is called a \emph{multipath}
linear-constraint loop (\mlc for short), and can be represented by a
set of polyhedra $\transitions_1,\ldots,\transitions_k$, each
specified by $A_i''\vec{x}'' \le \vec{c}_i''$ (the location need not be specified).
This kind of \cfgs arise in program analysis as an abstraction of an
iterative (or recursive) code that includes branching in the loop
body.
When we are interested in the conditions%
\footnote{If the loop is obtained from a \clang-like program, then it
  is easy to extract the condition, otherwise we can project
  $A_i''\vec{x}'' \le \vec{c}_i''$ onto $\vec{x}$.}
that allows the corresponding edge to be taken, we rewrite
$A_i''\vec{x}'' \le \vec{c}_i''$ as
$B_i\vec{x} \le \vec{b}_i \,\land\, A_i\vec{x}'' \le \vec{c}_i$
where, for some $p_i,q_i>0$, $B_i \in {\rats}^{p_i\times n}$,
$A_i\in {\rats}^{q_i\times 2n}$, $\vec{b}_i\in {\rats}^{p_i}$,
$\vec{c}_i\in {\rats}^{q_i}$.
For a path $i$, the constraint $B_i\vec{x} \le \vec{b}_i$ is called
\emph{the path guard}, and the other constraint is called \emph{the
  update}.
We say that the loop is a \emph{real}, \emph{rational}, or
\emph{integer} loop depending on the domain of the variables.
We say that there is a transition from a state $\vec{x}\in\numdom^n$ to
a state $\vec{x}'\in\numdom^n$, if there is a path $i$ such that
$\vec{x}$ satisfies its guard and $\vec{x}$ and $\vec{x}'$ satisfy its
update.
We also consider \mlc loops with an initial polyhedral set of states
$\poly{S}_0$.

\begin{example}
\label{ex:mlc:0}
Let $\transitions_1=\{ x_1 \geq 0, x_1' = x_1-1 \}$ and
$\transitions_2=\{ x_2 \geq 0, x_2'=x_2-1,  x_1' \leq x_1 \}$. Then
$\transitions_1,\transitions_2$ is an \mlc loop with two paths.
\end{example}

\begin{remark}
  \mlc loops are as expressive as linear-constraint \cfgs, because a \cfg
  $(V,\numdom,L,\ell_{0},E)$ can be transformed to \mlc loops by
  adding an extra variable that stores the value of the location as
  follows:
  we first introduce an injective mapping
  $\mathit{loc}:L \mapsto \nats$ that maps locations to natural
  numbers, and then every edge $(\ell,\transitions,\ell') \in E$ induces a path
  $\transitions \land x_{\mathit{loc}}=\mathit{loc}(\ell) \land
  x'_{\mathit{loc}}=\mathit{loc}(\ell')$.
\end{remark}

\subsubsection{Single-path Linear-Constraint Loops}

A \emph{single-path} linear-constraint loop (\slc for short) is a
special case of \mlc loop with a single path, \ie, the corresponding
\cfg has a single edge.
We represent such a loop by a single transition polyhedron
$\transitions$ specified by $A''\vec{x}'' \le \vec{c}''$.
If we are explicitly interested in the condition that allows the
edge to be taken, we write it as a \emph{while} loop of
the following form:
\begin{equation} \label{eq:slc-loop} 
  \while\; (B\vec{x} \le \vec{b}) \; \wdo \; A\vec{x}'' \le \vec{c}
\end{equation}

\begin{example}
  Consider the \slc loop
  $\transitions=\{4x_1 \ge x_2, x_2 \ge 1, 5x_1' \le 2x_1+1, 5x_1' \ge
  2x_1-3, x_2'=x_2\}$. We can also write this as follows to make the
  condition and the update explicit:
\begin{equation}
\label{eq:bg:loop2}
\begin{small}
\begin{array}{l}
\while\; (4x_1 \ge x_2, x_2 \ge 1) \; \wdo \; 5x_1' \le 2x_1+1, 5x_1' \ge 2x_1-3, x_2'=x_2
\end{array}
\end{small}
\end{equation}
This loop, interpreted over the integers, represents the \clang
language loop
\begin{center}
\normalfont\lstinline!while (4*x1>=x2 && x2>=1)$~$x1=(2*x1+1)/5;!
\end{center}
Note that if Loop \eqref{eq:bg:loop2} is interpreted over the
rationals, it becomes nondeterministic.
\end{example}

\subsubsection{Affine Single-path Linear-Constraint Loops}

An affine \slc loop is a special case of \slc loops where the update
can be described as a linear transformation, and is written as:
\begin{equation} \label{eq:slc-lin-loop} 
  \while\; (B\vec{x} \le \vec{b}) \; \wdo \; \vec{x}'=A\vec{x}+\vec{c}
\end{equation}
where $\vec{x}=(x_1,\ldots,x_n)^\trans$ and
$\vec{x}'=(x_1',\ldots,x_n')^\trans$ are column vectors, and for some
$m>0$, $B \in {\rats}^{m\times n}$, $A\in {\rats}^{n\times n}$,
$\vec{b}\in {\rats}^m$, $\vec{c}\in {\rats}^n$.
When it is convenient, we also write such loops as an imperative loop
\begin{equation} \label{eq:slc-lin-loop-imp} 
  \while\; (g_1(\boldsymbol x) \geq  0 \wedge \ldots \wedge
  g_m(\boldsymbol x) \geq 0) \; \wdo \; \boldsymbol
  x:=f(\boldsymbol x) \, ,
\end{equation}
where $g_i(\vec{x})=-\vec{b}_i^\trans\vec{x} + b_i$ with $\vec{b}_i$
being the $i$th row of $B$ and $b_i$ the $i$th element of $\vec{b}$,
and $f(\vec{x})=A\vec{x}+\vec{c}$.
The term \emph{linear loops} is frequently used in the literature to
refer to affine \slc loops.

\subsection{Counter Programs}

Counter programs (also known as counter machines) are a universal
computational model~\citep{Minsky:1967} used in this \survey to study
the decidability of classes of linear programs through reduction.

A (deterministic) counter program $P_C$ with $n$ (integer) counters
$X_1,\dots,X_n$ is a list of labelled instructions
$1{:}I_1,\ldots,m{:}I_m,m{+}1{:}\bot$ where each instruction ${I_k}$
is one of the following:
$$ 
incr(X_j) \mid decr(X_j) \mid
\mathit{if}~X_j>0~\mathit{then}~k_1~\mathit{else}~k_2
$$
with $1 \le k_1,k_2 \le m{+}1$ and $1 \le j \le n$.

A state is of the form $(k,(a_1,\ldots,a_n)^\trans)$ which indicates
that instruction $I_k$ is to be executed next, and the current values
of the counters are $X_1=a_1,\ldots,X_n=a_n$.
In a valid state, $1\le k\le m+1$ and all $a_i \in \nats$.
Any state in which $k=m+1$ is a halting state. For any other valid
state ${(k,\tuple{a_1,\ldots,a_n})}$, the successor state is defined
as follows:

\begin{itemize}
\item If $I_k$ is $decr(X_j)$ (resp. $incr(X_j)$), then $X_j$ is
  decreased (resp. increased) by $1$ and the execution moves to label
  $k+1$.

\item If $I_k$ is
  ``$\mathit{if}~X_j>0~\mathit{then}~k_1~\mathit{else}~k_2$'', then
  the execution moves to label $k_1$ if $X_j$ is positive, and to
  $k_2$ if it is $0$. The values of the counters do not change.
 \end{itemize}

 Since counter programs are a universal computational model, they have
 an undecidable halting problem (termination from a provided initial
 state). We know that the (universal) termination problem is
 undecidable as well.

\begin{theorem}[\citealp{BlondelBKPT01}]
\label{thm:counter}
Universal termination of counter programs is undecidable, even restricted to 2-counter programs.
\end{theorem}

\chapter{Decidability of Termination of Linear-Constraint Programs}
\label{chp:dec}

In this \chp, we overview decidability and undecidability results for
termination of the different linear-constraint program types
introduced in Section~\ref{sec:programs}, both with and without
initial states. This is a crucial and challenging research area
because it establishes the fundamental limits of termination analysis.

From a theoretical perspective, determining if such programs always
terminate is a non-trivial problem that often requires sophisticated
mathematical tools from areas like linear algebra, number theory, and
geometry. Furthermore, the decidability of termination for
linear-constraint programs is highly dependent on the variable domain
(integers, rationals, or reals). A loop that terminates for integer
variables might not terminate for reals. Typically, integer
linear-constraint programs are the most difficult to analyse.

For at least two decades, the decidability of termination for linear
programs has received considerable attention. Much of the progress in
this area has focused on affine \slc loops, for which many
decidability results have been established over $\reals$, $\rats$ and
$\ints$. The main part of this section provides an overview of these
results.
The more complex case of general \slc loops remains a significant open
problem, though some special cases and extensions of this model have
been considered. For \mlc loops, the problem becomes even more difficult,
and the research has primarily yielded undecidability results, even
for a small number of paths or variables.

\paragraph{Organisation of this \Chp.}
Section~\ref{sec:dec:affine} discusses termination of affine \slc
loops, Section~\ref{sec:dec:slc} discusses termination of \slc loops, and
Section~\ref{sec:dec:mlc} discusses termination of \mlc loops.

\section{Termination of Affine Single-path Linear-Constraint Loops}
\label{sec:dec:affine}

In this section, we consider the termination of affine \slc loops (like
Loop \eqref{eq:slc-lin-loop}), where the loop body has a single
control path that performs a simultaneous affine update of the program
variables. Analysing these loops, including acceleration and
termination, can be part of the analysis for more complex
programs~\citep{Boigelot03,JeannetSS14,KincaidBCR19}.

We are primarily interested in universal termination---that is,
determining whether these loops terminate for all initial values of
the program variables, regardless of whether the domain of variables
is $\reals$, $\rats$, or $\ints$. We also discuss termination from a
specific set of initial states in Section~\ref{sec:aff:init}.

The following examples, taken from~\citet{Bra06}, illustrate several
relevant phenomena, including how termination depends on the domain
of the loop variables.

\begin{example}
\label{ex:bra06:1}
Consider the loop:
\begin{align*}
\while\; (4x + y \geq 1) \; \wdo \;
\begin{pmatrix}
x \\
y
\end{pmatrix}
:=
\begin{pmatrix}
-2 & 4 \\
4 & 0
\end{pmatrix}
\begin{pmatrix}
x \\
y
\end{pmatrix}\, .
\end{align*}
The matrix in the loop body has two eigenvectors:

\[
\vec{v}_1:=(-1 - \sqrt{17}, 4) \quad \text{and} \quad\vec{v}_2:=(-1 + \sqrt{17}, 4) \,  ,
\]
respectively corresponding to the eigenvalues:
\[
\lambda_1:=-1 - \sqrt{17} \quad \text{and} \quad \lambda_2:=-1 + \sqrt{17} \,  .
\]
The eigenvector \( \vec{v}_2 \) satisfies the loop guard and
corresponds to a positive eigenvalue.  Hence the loop does not
terminate over \( \reals \).
However, the line through the origin parallel to \( \vec{v}_2 \) does
not contain any rational points other than \(0\), and the loop outside
this line is dominated by the negative eigenvalue \( \lambda_1 \),
which is larger in absolute value than $\lambda_2$.
At the limit, the orbit of \( (x, y) \) alternates between the
directions \( \vec{v}_1 \) and \( -\vec{v}_1 \). Hence, the loop
terminates on $\rats$.
\end{example}

\begin{example}
\label{ex:bra06:2}
Consider the loop:
\begin{align*}
\while\; (4x - 5y \geq 1) \; \wdo \;
\begin{pmatrix}
x \\
y
\end{pmatrix}
:=
\begin{pmatrix}
2 & 4 \\
4 & 0
\end{pmatrix}
\begin{pmatrix}
x \\
y
\end{pmatrix} \, .
\end{align*}
The matrix has two eigenvectors:
\[
  \vec{v}_1:=(1 + \sqrt{17}, 4) \quad \text{and} \quad 
  \vec{v}_2 :=(1 - \sqrt{17}, 4) \,  ,
\]
respectively corresponding to the eigenvalues:
\[
\lambda_1:=1 + \sqrt{17} \quad \text{and} \quad \lambda_2:=1 - \sqrt{17} \, .
\]
The eigenvalue $\lambda_1$ is positive and dominant and so all points
on the half-line $L = \{ r\vec{v}_1 : 4(r-1)\geq \sqrt{17} \}$ in the
direction of $\vec{v}_1$ are non-terminating (note that the lower
bound on $r$ ensures that the points satisfy the loop guard).
The
half-line $L$ does not contain any rational points, however a suitably
small perturbation of a point on $L$ remains non-terminating since such
a point converges to $L$ as the loop unfolds.  Thus there is a cone of
non-terminating points around $L$ that contains rational points and
even integer points.  For example, the point \((9,7)\) is non-terminating.
\end{example}

\begin{example}
\label{ex:bra06:3}
  The following loop terminates over the integers, but not over the
  rationals:
\begin{gather*}
  \while \; (x\geq 0) \; \wdo \;
  x:=-2x+1 \,  .
\end{gather*}
The only non-terminating initial value is
  $x=\frac{1}{3}$.
\end{example}

When considering termination over $\reals$ and $\rats$, we assume all
numerical constants in the loops are rational. Similarly, for
termination over $\ints$, we assume all numerical constants are
integers.
Despite the simplicity of affine \slc loops, the question of deciding
termination has proven challenging. \citet{Tiwari:04} showed that
termination for these loops is decidable over $\reals$.
Subsequently, \citet{Bra06}, using a more refined analysis, showed
that termination is decidable over \(\rats\) and noted that
termination on \(\ints\) can be reduced to termination on \(\rats\) in
the homogeneous case, \ie, when \(\bb, \bc\)
in~\eqref{eq:slc-lin-loop-imp} are both all-zero vectors~(this result
is for loops with strict inequalities, for non-strict ones the loop
obviously does not terminate with this change).
Finally, \citet{HOW19} gave a procedure for deciding termination over
the integers without restriction.

\subsection*{Overview of the Section}

The rest of this section presents a uniform framework, based on the
work of~\citet{HOW19}, that shows how to decide termination over
$\reals$, $\rats$, and $\ints$.
The high-level idea is that for a given linear loop with $n$
variables, one computes a convex semi-algebraic set
$\mathit{PN}\subseteq \reals^n$ of \emph{potentially non-terminating
  points}.
The key properties of $\mathit{PN}$ are that
\begin{inparaenum}[\upshape(i\upshape)]
\item it contains all non-terminating initial values in $\reals^n$;
\item it is a loop invariant;
\item all points in the relative interior of $\mathit{PN}$ are
  non-terminating.
\end{inparaenum}
These properties can be used to show that for each ring
$\numdom \in \{\reals, \rats, \ints\}$, the loop is non-terminating
over $\numdom$ if and only if $\mathit{PN}$ contains a point in
$\numdom^n$.
Then termination of the given loop over $\reals$ reduces to checking
non-emptiness of $\mathit{PN}$.  Thus, termination over $\rats$
or $\ints$ can respectively be determined using procedures of
Khachiyan and Porkolab~\citep{DinZ10,KP97} for determining whether a
given convex semi-algebraic set contains a rational point and whether
it contains an integer point.

The construction of the set $\mathit{PN}$ of potentially
non-terminating points and verification of its properties relies on
Kronecker's theorem on simultaneous Diophantine
approximation~\citep[page 53--59]{cassels} and a result
of~\citet{Combot25} that allows computing all multiplicative relations
among the eigenvalues of the update matrix of a given loop.
To analyse termination over $\ints$ we also use Kinchine's Flatness
Theorem, which gives sufficient conditions for a convex set to contain
an integer point (see Section~\ref{sec:mathbg:lattices}).

The rest of this section is structured as follows:
Section~\ref{sec:classify} classifies the termination behaviour of
initial values; Section~\ref{sec:critical} discusses the termination
of affine \slc loops with a single guard;
Section~\ref{sec:multiple-guard} discuss the termination of affine
\slc loops with multiple guards; and finally
Section~\ref{sec:aff:related} overviews related work.

\subsection{Classifying Initial Values}
\label{sec:classify}

\subsubsection{Reduction to the Non-Degenerate Case}
\label{sec:degen}

Recall that the general form of an affine \slc loop with $n$ variables
is as follows:
\begin{gather}
  \while \; (g_1(\vec{x}) \geq  0 \wedge \cdots \wedge
  g_m(\vec{x}) \geq 0) \; \wdo \; \vec{x}:=f(\vec{x}) \, ,
\label{loop:single-path}
\end{gather}
where \(g_1,\ldots,g_m : \reals^n \to \reals\) and
\(f : \reals^n \to \reals^n\) are affine functions with
rational coefficients, that is, \(f(\vec{x})=A\vec{x} + \vec{a}\) for
\(A \in \rats^{n\times n}\) and \(\vec{a} \in \rats^n\), and
\(g_i(\vec{x}) = \vec{b}_i^\trans \vec{x} + c_i\) for
\(\vec{b}_i \in \rats^n\), \(c_i \in \rats\), and \(i=1,\ldots,m\).
Note that
\begin{gather}
  \begin{pmatrix} f(\vec{x}) \\ 1 \end{pmatrix} =
  \begin{pmatrix} A& \vec{a}\\ 0 & 1 \end{pmatrix}
  \begin{pmatrix} \vec{x} \\ 1 \end{pmatrix} \text{ and }
  g_i(\vec{x}) = (\vec{b}_i^\trans \; c_i)
  \begin{pmatrix} \vec{x} \\ 1 \end{pmatrix} 
  \label{eq:update}
\end{gather}
for all \(\vec{x} \in \reals^n\).

We first reduce the termination problem to the special case in which
the update map in the loop body is invertible.  We outline the
reduction for the version of the problem over the reals, but the same
construction works over the rationals or the integers.  To this end,
consider the Loop~\eqref{loop:single-path}.  The decreasing chain of
affine sub-spaces
$\mathbb R^n \supseteq f(\mathbb R^n) \supseteq f^2(\mathbb R^n) \supseteq 
\cdots$ stabilises in at most $n$ steps.  Let
$V \subseteq \mathbb R^n$ be the stabilising value.  Then $f$ restricts
to an invertible affine self-map of $V$.  Moreover, the
Loop~\eqref{loop:single-path} terminates on $\mathbb R^n$ if and only if
it terminates on $V$.  But the question of termination on $V$ can be
transformed into an instance of Loop~\eqref{loop:single-path} in which the
loop update map is invertible.  For this, we simply need to compute a
basis of $V$ and representations, with respect to this basis, of the
restrictions of the affine maps $f,g_1,\ldots,g_m$ to $V$.

With the above reduction in hand, we may henceforth assume without
loss of generality that the loop update map \(f\) in the termination
problem is invertible, and hence that zero is not an eigenvalue of the
update matrix
\(\begin{psmallmatrix} A& \vec{a}\\ 0 & 1 \end{psmallmatrix}\). We
furthermore say that \(f\) is \emph{non-degenerate} if no quotient of
two distinct eigenvalues of this matrix is a
root of unity.

We claim that the termination problem for affine \slc loops is
reducible to the special case of the problem for non-degenerate update
functions.  (We only need this reduction to handle termination over the integers.)
To prove the claim, consider an affine \slc loop, as described above,
whose update matrix has distinct eigenvalues
\(\lambda_1,\ldots,\lambda_s\).
Let \(L\) be the least positive integer such that whenever \(i\neq j\)
are such that \(\frac{\lambda_i}{\lambda_j}\) is a root of unity, we have
\((\frac{\lambda_i}{\lambda_j})^L=1\).
It is known that
\(L = 2^{O(n\sqrt{\log n})}\)~\citep[Section~1.1.9]{EPSW03}.
The update matrix corresponding to the affine map
\(f^L = f \circ \cdots \circ f\) ($L$ times) has eigenvalues
\(\lambda_1^L,\ldots,\lambda_s^L\) and hence is non-degenerate.
Moreover the original loop terminates if and only if the following
loop terminates:
  \begin{gather*}
    \while \; \bigwedge_{i=0}^{L-1} \left(g_1(f^i(\vec{x})) \geq 0 \wedge \cdots \wedge g_m(f^i(\vec{x})) \geq 0\right)
    \; \wdo \; \vec{x}:=f^L(\vec{x}) \, ,
  \end{gather*}
But this loop is non-degenerate and the argument is complete.

\subsubsection{Spectral Analysis}
\label{sec:spectral}

Let us focus now on the case of an affine \slc loop of the form
\begin{gather}
\label{eq:the-loop}
  \while \; (g(\vec{x})\geq 0)
  \;\wdo \; \vec{x} := f(\vec{x})
\end{gather}
with a single guard function \(g(\vec{x})=\vec{b}^\trans \vec{x} + c\)
and with non-degenerate update function
\(f(\vec{x})=A\vec{x}+\vec{a}\), with both maps having rational
coefficients.
We show that a spectral analysis of the matrix underlying the loop
update function suffices to classify almost all initial values of the
loop as either terminating or eventually non-terminating.  We isolate
a class of points called \emph{critical points} for the loop for which
the spectral analysis does not determine whether or not they are
terminating.

With respect to Loop~\eqref{eq:the-loop} we say that
\(\vec{x} \in \reals^n\) is \emph{terminating} if there exists
\(m \in \nats\) such that \(g(f^m(\vec{x})) < 0\).
We say that \(\vec{x}\) is \emph{eventually non-terminating} if the
sequence \(\langle g(f^m(\vec{x})) : m\in\nats\rangle\) is
\emph{ultimately positive}, \ie, there exists \(N\) such that for all
\(m\geq N\), \(g(f^m(\vec{x})) \geq 0\).
Let $\numdom$ be a sub-ring of $\reals$ that is preserved by $f$, that
is, such that $f(\numdom^n)\subseteq \numdom^n$.
Then there exists \(\vec{z} \in R^n\) that is non-terminating if and
only if there exists \(\vec{z} \in \numdom^n\) that is eventually
non-terminating. Thus we can regard the problem of deciding
termination on \(\numdom^n\) as that of searching for an eventually
non-terminating point in \(\numdom^n\). Note that $f$ certainly
preserves $\reals$ and $\rats$ and it moreover preserves $\ints$ if we
assume that the coefficients of $A$ and $\vec{a}$ are integer.

Let \(\lambda_1,\ldots,\lambda_s\) be the eigenvalues of
\(\begin{psmallmatrix} A& \vec{a}\\ 0 & 1 \end{psmallmatrix}\)
and let \(k_{\mathrm{max}}\) be the maximum algebraic multiplicity over all
eigenvalues.
Define a linear pre-order on
\(I:=\{0,\ldots,k_{\mathrm{max}}-1\} \times \{1,\ldots,s\}\) by
\((i_1,j_1) \preccurlyeq (i_2,j_2)\) if either
\begin{inparaenum}[\upshape(i\upshape)]
\item \(|\lambda_{j_1}| <|\lambda_{j_2}|\) or
\item \(|\lambda_{j_1}|=|\lambda_{j_2}|\) and \(i_1 \leq i_2\).
\end{inparaenum}
Write \((i_1,j_1) \prec (i_2,j_2)\) if
\((i_1,j_1) \preccurlyeq (i_2,j_2)\) and
\((i_2,j_2) \not\preccurlyeq (i_1,j_1)\).  Then we have
\begin{equation*}
  (i_1,j_1) \prec (i_2,j_2) \iff \lim_{m\to \infty}
  \frac{ \binom{m}{i_1}|\lambda_{j_1}|^m}{
    \binom{m}{i_2}|\lambda_{j_2}|^m } = 0 \, ,
\end{equation*}
that is, the preorder \(\preccurlyeq\) characterises the asymptotic
order of growth in absolute value of the terms
\(\binom{m}{i}\lambda^m_j\) for \((i,j) \in I\).  This preorder,
moreover, induces an equivalence relation \(\approx\) on \(I\) where
\((i_1,j_1) \approx (i_2,j_2)\) if and only if
\((i_1,j_1) \preccurlyeq (i_2,j_2)\) and
\((i_2,j_2) \preccurlyeq (i_1,j_1)\).

The following closed-form expression for \(g(f^m(\vec{x}))\)
will be the focus of the subsequent development.  The expression is
obtained from the Jordan-Chevalley decomposition of the affine map
$f$.

\begin{proposition}
\label{prop:ind}
There are affine functions
\(h_{i,j} : \reals^n \to \complex\) such that for all
\(\vec{x} \in \reals^n\) and all \(m \in \mathbb N\) we have
\(g(f^m(\vec{x})) = \sum_{(i,j)\in I} \binom{m}{i} \lambda_j^m \,
h_{i,j}(\vec{x})\).
\end{proposition}

Define \(\gamma_i = \frac{\lambda_i}{|\lambda_i|}\) for
\(i=1,\ldots,s\), that is, we obtain \(\gamma_i\) by normalising
the eigenvalues to have length \(1\).  Recall from
Section~\ref{subsec:grp_of_mul_rel} the definition of the group
\(L(\bgamma)\) of multiplicative relations that hold among
\(\gamma_{1},\ldots,\gamma_{s}\), namely,
\begin{equation*}
  L(\bgamma) = \{ (n_1,\ldots,n_s) \in \ints^s :
  \gamma_1^{n_1} \cdots \gamma_s^{n_s} = 1 \} \, .
\end{equation*}
Recall also that we have \(T(\bgamma) \subseteq \TT^s\), given by

\medskip
\scalebox{0.94}{$%
 T(\bgamma) = \{ (\mu_1,\ldots,\mu_s) \in
  \TT^s : \mu_1^{n_1} \cdots \mu_s^{n_s} = 1 \text{ for all }
  (n_1,\ldots,n_s) \in L(\bgamma) \}.%
$}

\medskip
Given an \(\approx\)-equivalence class \(L \subseteq I\), 
for all \((i_1,j_1),(i_2,j_2) \in L\) we have \(i_1=i_2\) and
\(|\lambda_{j_1}|=|\lambda_{j_2}|\).  Thus \(L\) determines a
common multiplicity, which we denote \(i_L\), and a set of eigenvalues
that all have the same absolute value, which we denote \(\rho_L\).

Given an \(\approx\)-equivalence class \(L\), define
\(\Phi_L : \reals^n \times T(\bgamma) \to \reals\)
by\footnote{That the function \(\Phi_L\) is real-valued follows from
  the fact that if eigenvalues \(\lambda_{j_1}\) and \(\lambda_{j_2}\)
  are complex conjugates then \(\gamma_{j_1}\) and \(\gamma_{j_2}\)
  are also complex conjugates, as are \(h_{i,j_1}(\vec{z})\) and
  \(h_{i,j_2}(\vec{z})\).}%
\begin{gather}
  \Phi_L(\vec{x},\vec{\mu}) = \sum_{(i,j)\in L}
  h_{i,j}(\vec{x}) \mu_j \, .
\label{def:Phi}
\end{gather}
From the above definition of \(\Phi_L\) we have
\begin{equation}
\begin{array}{rcl}
  \sum_{(i,j) \in L} \binom{m}{i}\lambda_j^m h_{i,j}(\vec{x}) &\,  = \, &
\binom{m}{i_L}\rho_L^m \sum_{(i,j)\in L}  h_{i,j}(\vec{x})\gamma_j^m  \\
                               & \, = \, &   \binom{m}{i_L}\rho_L^m \Phi_L(\vec{x},\vec{\gamma}^m) 
  \end{array}
  \label{eq:formula}
\end{equation}
for all \(\vec{x} \in \reals^n\) and all
\(m\in \nats\).

We say that an $\approx$-equivalence class \(E\) of \(I\) is
\emph{dominant} for \(\vec{x} \in \reals^n\) if, for all indices
$(i,j)$ belonging to an equivalence class $E' \succ E$, we have that
\(h_{i,j}(\vec{x})\) is identically zero. Equivalently, \(E\) is
dominant for \(\vec{x}\) if for all $E'\succ E $ we have that
\(\Phi_E(\vec{x},\cdot)\) is identically zero on \(T(\bgamma)\).
The equivalence of these two characterisations follows from the linear
independence of the functions \(\binom{m}{i}\lambda_j^m\) for
\((i,j) \in E\).

\begin{example}
  Define $f:\reals^3\rightarrow \reals^3$ and
  $g:\reals^3\rightarrow \reals$ by
  $f(\vec{x}):=(2x_1-3x_2,x_1,x_3+1)$ and $g(\vec{x}):=x_1+x_3$.  The
  eigenvalues of the update matrix corresponding to $f$ are
  $\lambda,\overline{\lambda},1$, where $\lambda=1+i\sqrt{2}$.
  These divide into two equivalence classes: a dominant equivalence class
  $\{\lambda,\overline{\lambda}\}$, and a non-dominant class $\{1\}$.  Since
\[ g(f^m(\vec{x})) = \lambda^m(ax_1+bx_2)+\overline{\lambda^m(ax_1+bx_2)}
  +x_3+m  \, ,\]
where 
$a:=\frac{1}{2}-\frac{i}{2\sqrt{2}}$ and $b:=\frac{3i}{2\sqrt{2}}$,
if $E$ denotes the dominant equivalence class, then we have
\[ \Phi_E(\vec{x},\vec{\mu}) = \mu_1(ax_1+bx_2)+ \mu_2(\overline{a}x_1+\overline{b}x_2)
\, .  \]
\label{ex:ind}
\end{example}

The following proposition shows how information about termination of
Loop~\eqref{eq:the-loop} on an initial value \(\vec{x}\in\reals^n\) can be
derived from properties of \(\Phi_E(\vec{x},\cdot)\).

\begin{proposition}
\label{prop:easy}
Consider Loop~\eqref{eq:the-loop}. Let \(\vec{x} \in \reals^n\) and
let \(E\) be an \(\approx\)-equivalence class that is dominant for
\(\vec{x}\). Then
\begin{enumerate}
\item\label{itm: easy non-term iv's} If
  \(\displaystyle\inf_{\vec{\mu} \in T(\bgamma)}
  \Phi_E(\vec{x},\vec{\mu}) > 0\), then \(\vec{x}\) is eventually
  non-terminating.%
\item\label{itm: easy term iv's} If
  \(\displaystyle\inf_{\vec{\mu} \in T(\bgamma)}
  \Phi_E(\vec{x},\vec{\mu}) < 0\), then \(\vec{x}\) is terminating.%
  \end{enumerate}
\end{proposition}

\begin{proof}
  By Proposition~\ref{prop:ind} and~\eqref{eq:formula} we have that for all
  \(m\geq n\),
  \begin{eqnarray}
    g(f^m(\vec{x}))
    &=& \sum_{(i,j) \in I} \binom{m}{i}\lambda_j^m h_{i,j}(\vec{x}) \notag \\
    &=& \binom{m}{i_E} \rho_E^m \Phi_E(\vec{x},\vec{\gamma}^m) +
        \sum_{(i,j) \in I\setminus E} \binom{m}{i}\lambda_j^m h_{i,j}(\vec{x}) 
  \label{eqn:dom}
  \end{eqnarray}
  Moreover by the dominance of \(E\) we have that
  \begin{gather} \lim_{m\to \infty}
    \frac{\binom{m}{i}|\lambda_j|^m}{\binom{m}{i_E}\rho_E^m} = 0
    \label{eq:dominate}
  \end{gather}
  for all \((i,j) \in I \setminus E\) such that
  \(h_{i,j}(\vec{x})\neq 0\). 

  We first prove Item~\ref{itm: easy non-term iv's}.  By assumption, in
  this case there exists \(\varepsilon>0\) such that
  \(\Phi_E(\vec{x},\vec{\mu}) \geq \varepsilon\) for all
  \(\vec{\mu} \in T(\bgamma)\).
  Together with~\eqref{eq:dominate}, this shows that the right-hand side of~\eqref{eqn:dom}
  is positive for \(m\) sufficiently large.  Hence 
  \(g(f^m(\vec{x}))\) is positive for \(m\)
  sufficiently large and so \(\vec{x}\) is eventually
  non-terminating.

  We turn now to Item~\ref{itm: easy term iv's}.  By assumption, there
  exists \(\varepsilon>0\) such that
  $\inf_{\vec{\mu} \in T(\bgamma)}                                                        
  \Phi_E(\vec{x},\vec{\mu}) < -\varepsilon$.
  By density of \(\{ \vec{\gamma}^m : m \in \nats\}\) in
  \(T(\bgamma)\) and continuity of $\Phi_E(\vec{x},\cdot)$,
  there exist infinitely many \(m\) such that
  \(\Phi_E(\vec{x},\vec{\gamma}^m) < -\varepsilon \).
  By~\eqref{eq:dominate}, if   \(\Phi_E(\vec{x},\vec{\gamma}^m) < -\varepsilon \)
  and $m$ is sufficiently large, then the right-hand side of~\eqref{eqn:dom}
  is negative.  We conclude that there are infinitely many $m$ such that 
    \(g(f^m(\vec{x}))<0\); hence \(\vec{x}\) is terminating.
\end{proof}

Given \(\vec{z} \in \ints^n\), since \(T(\bgamma)\) is an algebraic
subset of \(\TT^s\), the number
\(\displaystyle\inf_{\vec{\mu} \in T(\bgamma)}
\Phi_E(\vec{z},\vec{\mu})\) is algebraic (by quantifier elimination)
and its sign can be decided.
Note however that Proposition~\ref{prop:easy} does not completely
resolve the question of termination with respect to guard \(g\) from a
given initial value \(\vec{z}\).
Indeed, let us define \(\vec{z} \in \reals^n\) to be \emph{critical}
if
\(\displaystyle\inf_{\vec{\mu} \in E} \Phi_E(\vec{z},\vec{\mu}) = 0\),
where \(E\) is the dominant \(\approx\)-equivalence class for
\(\vec{z}\).  Then neither clause in the above proposition suffices to
resolve termination of Loop~\eqref{eq:the-loop} on such a
\(\vec{z}\).

In general, the question of whether a critical point is eventually
non-terminating is equivalent to the \emph{Ultimate Positivity
  Problem} for linear recurrence sequences: a longstanding and
notoriously difficult open problem in number theory, only known to be
decidable up to order 5 \citep{ACHOW18, OuaknineW14a}.
Fortunately in the setting of deciding loop termination we can
sidestep such difficult questions.  The following section is devoted
to handling critical points.
The idea is to show that if there is a non-terminating critical initial value then
there is another initial value that is eventually non-terminating and
whose eventual non-termination can be established by
Proposition~\ref{prop:easy}.

\begin{example}
\label{ex:critical}
Consider the loop:
\begin{align*}
\while\; (w-z \geq  0) \; \wdo \; \quad
\begin{pmatrix}
w\\x \\
y \\
z\\ \end{pmatrix}
\leftarrow
\begin{pmatrix}
-1 & 5 & 125& 0 \\
1 & 0 & 0 & 0\\
0 & 1 & 0 & 0 \\
0 & 0 & 0 & 2
\end{pmatrix}
\begin{pmatrix}
w\\
  x \\
y \\
z
\end{pmatrix}
\end{align*}
The idea is that the variables $(w,x,y)$ store consecutive values of
the order-3 linear recurrence sequence
\[ u_n = - u_{n-1} +5u_{n-2}+125u_{n-3} \]
while the variable $z$ stores values of the sequence $v_n=2v_{n-1}$.

The update matrix in the loop body has eigenvalues
\[ \lambda_1=5,\;
  \lambda_2=-3+4i,\; \lambda_3=-3+4i,
  \; \lambda_4=2 \, . \]
For $f : \reals^4 \to \reals^4$, the linear map
computed in the loop body, and
$g : \reals^4 \to \reals$, the map $g(w,x,y,z)=w-z$ in
the loop guard, and for the initial value $\vec{x} = (18,2,2,2)^\trans$ we
have
\begin{gather}
  g(f^m(\vec{x}))  =5^m + \frac{1}{2}(-3+4i)^m +
  \frac{1}{2}(-3-4i)^m - 2\cdot 2^{m} \, .
\label{eq:LRS}
\end{gather}

The first three eigenvalues form an $\approx$-equivalence class $E$ with respect
to the dominance preorder and together dominate the fourth eigenvalue.
Normalising the eigenvalues to have length one we obtain
\[ \gamma_1:=1,\; \gamma_2:=\frac{-3+4i}{5} ,\; \gamma_3 :=
  \frac{-3-4i}{5},\; \gamma_4=1 \, .\]
Given the multiplicative relations $\gamma_1=\gamma_4=1$ and
$\gamma_2\gamma_3=1$, we have
\[ T(\vec{\gamma}) = \left\{ \vec{\mu} \in \TT^4 :
  \mu_1=\mu_4=1, \mu_2\mu_3 =1 \right\} \, . \]

The coefficients of the dominant eigenvalues in the exponential-sum
expression~\eqref{eq:LRS} determine the map
$\Phi_E(\vec{x},\cdot) : T(\vec{\gamma})\to \reals$, leading
to
\begin{eqnarray*}
\inf_{\vec{\mu} \in T(\vec{\gamma})} \Phi_E(\vec{x},\vec{\mu}) 
&=&
\inf_{\vec{\mu} \in T(\vec{\gamma})} \mu_1 +
\frac{1}{2}\mu_2 + \frac{1}{2}\mu_3 \\
&=&
\inf_{\mu \in \TT} 1+\frac{1}{2}\mu + \frac{1}{2} \overline{\mu}
  \\
&=& 0 \, .
\end{eqnarray*}
We conclude that $\vec{x}$ is a critical point.
\end{example}

Example~\ref{ex:critical} helps illustrate the idea that critical
points are initial values for which termination involves considering
all eigenvalues of the loop update map, not just the dominant
eigenvalues.  The initial value $(18,2,2,2)^\trans$ is eventually
non-terminating if and only if the order-4 linear recurrence
sequence~\eqref{eq:LRS} is ultimately positive: The sum of the three
dominant terms in this expression is guaranteed to be non-negative,
but establishing ultimate positivity of the whole expression would
require a suitable lower bound on the contribution of the dominant
terms.  In the case at hand, ultimate positivity can be established
using Baker's Theorem on linear forms in logarithms~\citep{OuaknineW14a}.
However, as noted above, in general it is not known to determine
ultimate positivity of linear recurrences from order 6 onwards.

\subsection{Non-Termination for a Single Guard Affine \slc Loop}
\label{sec:critical}

In this section we continue to analyse termination of 
Loop~\eqref{eq:the-loop}, and refer to the notation established
so far.

\subsubsection{Non-Termination over the Reals and Rationals}

The following definition encompasses both non-terminating and critical
points:
\begin{definition}
\label{def:PN}  
  For Loop~\eqref{eq:the-loop}, we define the set $\mathit{PN}$ of
  \emph{potentially non-terminating points} by
\[ \mathit{PN}  := \left\{ \vec{x} \in \reals^n:  
    \inf_{\vec{\mu} \in T(\bgamma)}
    \Phi_E(\vec{x},\vec{\mu}) \geq  0,\, \text{where $E$
      is dominant for $\vec{x}$}\right\} \, . \]
\end{definition}

It is evident that $\mathit{PN}$ is convex.  The following proposition
implies that $\mathit{PN}$ is moreover an invariant of
Loop~\eqref{eq:the-loop}, that is, if $\vec{x} \in \mathit{PN}$ then
$f(\vec{x}) \in \mathit{PN}$.
\begin{proposition}
\label{prop:shift}
Let \(\vec{x} \in \reals^n\) and let \(E\subseteq I\) be an
$\approx$-equivalence class that is dominant for \(\vec{x}\).  Then
\(E\) is also dominant for \(f(\vec{x})\), and for all
\(\vec{\mu} \in T(\bgamma)\) we have
\(\Phi_E(f(\vec{x}),\vec{\mu}) = \rho_E \,
\Phi_E(\vec{x},\vec{\gamma} \vec{\mu})\), where the
product \(\vec{\gamma} \vec{\mu}\) is defined pointwise.
\end{proposition}

\begin{proof}
  By definition we have
  \(\Phi_E(\vec{x},\vec{\mu}) = \sum_{(i,j)\in E}
  h_{i,j}(\vec{x})\mu_j\), where the \(h_{i,j}\) satisfy
  \begin{gather}  
  (\vec{b}^\trans \; c) \begin{pmatrix}A& \vec{a} \\ 0 &
    1 \end{pmatrix}^m
  \begin{pmatrix} \vec{x} \\ 1 \end{pmatrix} = \sum_{(i,j) \in
    I} h_{i,j}(\vec{x}) \binom{m}{i}\lambda_j^m \,
  \label{eq:first}
\end{gather}
for all \(m\geq 0\).
Since $f(\vec{x})=A\vec{x}+\vec{b}$, substituting
$f(\vec{x})$ for $\vec{x}$ in~\eqref{eq:first} yields
\begin{gather}
  (\vec{b}^\trans \; c) \begin{pmatrix}A& \vec{a} \\ 0 & 1 \end{pmatrix}^{m+1}
  \begin{pmatrix} \vec{x} \\ 1 \end{pmatrix} = \sum_{(i,j) \in
    I} h_{i,j}(f(\vec{x})) \binom{m}{i}\lambda_j^m \, .
\label{eq:second}
\end{gather}
Combining~\eqref{eq:first} and~\eqref{eq:second} we have that
for all \(m\geq 0\),
\begin{eqnarray*}
  \sum_{(i,j) \in I} h_{i,j}(f(\vec{x})) \binom{m}{i} \lambda_j^m
  &=& \sum_{(i,j) \in I} h_{i,j}(\vec{x}) \binom{m+1}{i}\lambda_j^{m+1} \\
  &=& \sum_{(i,j) \in I} h_{i,j}(\vec{x})
      \left[\binom{m}{i}+\binom{m}{i-1}\right]\lambda_j \lambda_j^{m} \, .
\end{eqnarray*}

Now the collection of functions  \(m\mapsto \binom{m}{i}\lambda_j^m\)
for \((i,j) \in I\) is linearly independent over $\mathbb C$ (see
Section~\ref{sec:exp-poly}).  Fix $(i,j) \in E$.  
Since the function
\(\binom{m}{i}\lambda_j^m\) must have the  same coefficient on the left-hand and right-hand sides of
the above equation, using the fact that $h_{i+1,j}(\vec{x})=0$ for $(i,j)\in E$ by dominance of $E$, we have
\(h_{i,j}(f(\vec{x})) = \lambda_j h_{i,j}(\vec{x})  = \rho_E \gamma_j h_{i,j}\).
We conclude that
$\Phi_E(f(\vec{x}),\vec{\mu}) = \rho_E \Phi_E(\vec{x},\vec{\gamma}\vec{\mu})$
and that $E$ is dominant for $f(\vec{x})$.
\end{proof}

The next lemma is the key to the framework presented in this section.
It shows that the non-emptiness of $\mathit{PN}$ entails the existence
of an eventually non-terminating point.

\begin{lemma}
  If \(\vec{z} \in \mathit{PN}\), then all points in the relative
  interior of \(\conv(\{ f^m(\vec{z}) : m\in \nats \} )\) are
  eventually non-terminating.
  \label{lem:critical}
\end{lemma}

\begin{proof}
  Let \(E\) be the \(\approx\)-equivalence class that is dominant for
  \(\vec{z}\).  If $\Phi_E(\vec{z},\cdot)$ is identically zero, then by
  definition of dominance we must have that
  $\Phi_{E'}(\vec{z},\cdot)$ is identically zero for all
  $\approx$-equivalence classes $E'$.
  By Proposition~\ref{prop:shift} we have that
  $\Phi_{E'}(f^m(\vec{z}),\cdot)$ is identically zero for all
  $\approx$-equivalence classes $E'$ and all $m\in \nats$.  Hence
  $f^m(\vec{z})$ is eventually non-terminating for all $m\in \nats$.

  We thus suppose that $\Phi_E(\vec{z},\cdot)$ is non-negative and not
  identically zero on \(\bmu \in T(\bgamma)\).  Fix
  \(\bmu \in T(\bgamma)\).  We claim that there exists \(m \in \nats\)
  such that \(\Phi_E (f^m(\bz),\bmu) > 0\). If this were not the case,
  then by Proposition~\ref{prop:shift} for all \(m \in \nats\) we
  would have
  \(\Phi_E (f^m(\bz),\bmu) = \rho_E^m \, \Phi_E
  (\bz,\bgamma^m\bmu)=0\).
  But by Theorem~\ref{thm:dense}, the set \(\{ \bgamma^m\bmu : m \geq 0\}\)
  is dense in \(T(\bgamma)\) and hence we would have that
  \(\Phi_E (\bz,\cdot)\) is identically \(0\) on \(T(\bgamma)\),
  contradicting our initial assumption.  This establishes the claim.

  Now $T(\bgamma)$ is a closed subset of $\mathbb T^s$ and is therefore compact.
  The above claim shows that for all $\vec{\mu}\in T(\bgamma)$ there exists $m \in \mathbb N$
  such that
  $\Phi_E(f^m(\vec{z}),\cdot)$ is strictly positive on a neighbourhood of $\vec{\mu}$.
  Thus
  there exists $m_0 \in \nats$ such
  that for all $\vec{\mu} \in T(\bgamma)$ there exists $m\leq m_0$
  such that $\Phi_E(f^m(\vec{z}),\vec{\mu})>0$.
  By Proposition~\ref{prop:relative-int}, for all points \(\vec{x}\)
  lying in the relative interior of
  \[\conv(\{ \vec{z},f(\vec{z}),\ldots,f^{m_0}(\vec{z})\})\]
  there exist 
  \(\alpha_0,\ldots,\alpha_{m_0} > 0\) such that:
  \begin{inparaenum}[\upshape(i\upshape)]
  \item \(\sum_{m=0}^{m_0} \alpha_m=1\); and
  \item \(\vec{x} = \sum_{m=0}^{m_0} \alpha_i f^m(\vec{z})\).
    \end{inparaenum}
  Since \(\Phi_E\) is an affine map in its first variable, it follows
  that
  \(\Phi_E(\vec{x},\cdot) = \sum_{m=0}^{m_0} \alpha_m \Phi_E
  (f^m(\vec{z}),\cdot)\) is strictly positive on \(T(\bgamma)\).
  Hence \(\vec{x}\) is eventually non-terminating by
  Proposition~\ref{prop:easy}.
\end{proof}

The following Example illustrates Lemma~\ref{lem:critical}.

\begin{example}
  Consider the loop from Example~\ref{ex:critical}. Starting from the
  critical point $\vec{x}:=(18,2,2,2)^\trans$, after one execution of the
  loop body we arrive at $\vec{y}:=(242,18,2,4)^\trans$.
  By Proposition~\ref{prop:shift} the point $\vec{y}$ is also
  critical.  Consider the mid-point
  \[ \vec{z} := \frac{1}{2} \left(\vec{x}+\vec{y}\right) = (130,10,2,3)^\trans \]
  between $\vec{x}$ and $\vec{y}$.
  We claim that $\vec{z}$ is eventually non-terminating.  Indeed we
  have
\[ \Phi_E(\vec{z},\vec{\mu}) =\alpha_1 \mu_1 + \alpha_2 \mu_2 + \alpha_3 \mu_3 \]
where $\alpha_1,\alpha_2,\alpha_3$ are uniquely defined by the
requirement that the sequence 
\[ v_n:=\alpha_15^n+\alpha_2(-3+4i)^n+\alpha_3 (-3-4i)^n \]
have initial values $v_0=2,v_1=10,v_2=130$, respectively.
We thus obtain $\alpha_1= 3$, $\alpha_2:= -\frac{1}{2}+i$, and
$\alpha_3:=-\frac{1}{2}-i$.  Since for
$(\mu_1,\mu_2,\mu_3)\in T(\vec{\gamma})$ we have $\mu_1=1$ and
$\mu_2=\overline{\mu_3}$, we deduce that
\[ \inf_{\mu \in T(\vec{\gamma})} \Phi_E(\vec{z},\vec{\mu}) = \alpha_1-2|\alpha_2| = 3 - \sqrt{5} > 0 \, .\]
It follows from Proposition~\ref{prop:easy} that $\vec{z}$ is
eventually non-terminating.
\end{example}

From Lemma~\ref{lem:critical} we obtain the following effective
criterion for non-termination over both $\reals$ and $\rats$.

\begin{corollary}
  Loop~\eqref{eq:the-loop} is non-terminating over $\reals$ if and
  only if $\mathit{PN}$ is non-empty and is non-terminating over
  $\rats$ if and only if $\mathit{PN}$ contains a rational point.
\label{corl:RNT}
\end{corollary}

\begin{proof}
  Given $\vec{z} \in \mathit{PN}$, all points in the relative interior
  of \(\conv(\{ f^{m}(\vec{z}) : m\in \nats\}) \) are eventually
  non-terminating by Lemma~\ref{lem:critical}.
  Hence the loop is non-terminating over $\reals$.  If moreover
  $\vec{z}$ is rational, then the relative interior contains a rational
  point and hence the loop is non-terminating over $\rats$.
\end{proof}

\subsubsection{Non-Termination over the Integers}

We now refine the above analysis to obtain an effective criterion of
the existence of \emph{integer} non-terminating points. In particular,
fixing an initial value \(\vec{z}_0 \in \ints^n\), we show that for
\(m\) sufficiently large, the set
\( \conv(\{f^{m}(\vec{z}_0):m\in \nats \}) \) contains an integer
point in its relative interior.
Recall that when considering termination over integers we consider
that the coefficients of the functions $f$ and $g$ that define
Loop~\eqref{eq:the-loop} are integer.

Define \(V:= \aff(\{ f^m (\vec{z}_0) : m \in \nats\})\) and let
the vector subspace \(V_0 \subseteq \reals^n\) be the unique translate
of \(V\) containing the origin. Write \(n_0\) for the dimension of
\(V_0\) (equivalently the dimension of \(V\)).

\begin{proposition}
\label{prop:unbounded}
For all non-zero integer vectors \(\vec{v} \in V_0\) the set
\(\{ |\vec{v}^\trans f^{m}(\vec{z}_0)| : m \in \nats \}\) is
unbounded.
\end{proposition}

\begin{proof}
  Consider the sequence
  \(x_m:=\vec{v}^\trans f^m(\vec{z}_0) = \vec{v}^\trans
  \begin{psmallmatrix} A&\vec{a}\\ 0& 1 \end{psmallmatrix}^m
  \begin{psmallmatrix} \vec{z}_0 \\ 1 \end{psmallmatrix}\).
  If this sequence is constant, then $V_0=\{0\}$ and the proposition holds
  vacuously.  Assume that the sequence is not constant.
  Since the sequence is integer-valued and satisfies a
  non-degenerate linear recurrence of order at most \(n+1\) (see,
  \eg, \citet[Section~1.1.12]{EPSW03}), by the Skolem-Mahler-Lech
  Theorem we have that
  \(\{ |\vec{v}^\trans f^{m}(\vec{z}_0)| : m \in \nats \}\) is
  unbounded (see the discussion of growth of linear recurrence
  by \citet[Section~2.2]{EPSW03}).\footnote{The above argument actually
    establishes that \(\langle x_m : m \in\nats\rangle\) diverges
    to infinity in absolute value.  We briefly sketch a more
    elementary proof of mere unboundedness.  If the sequence
    \(\langle x_m : m \in\nats\rangle\) were bounded then by van
    der Waerden's Theorem, for all \(m'\) it would contain a constant
    subsequence of the form \(x_\ell,x_{\ell+p},\ldots,x_{\ell+m'p}\)
    for some \(\ell,p \geq 1\).  In particular, if \(m'=n\) then since
    every infinite subsequence \(y_m:=x_{\ell+pm}\) satisfies a linear
    recurrence of order at most \(m+1\),
    \(\langle x_m : m \in \nats \rangle\) would have an infinite
    constant subsequence \(\langle x_{\ell+pm} : m\in \nats \rangle\).
    If \(p=1\) then \(\langle x_m : m \in \nats \rangle\) is constant
    and if \(p>1\) then by \citet[Lemma 9.11]{SS78}
    \(\langle x_m : m \in \nats \rangle\) is degenerate.}
\end{proof}

\begin{proposition}
\label{prop:rel-int}
Given $\vec{z}_0 \in \ints^n$, the set
\(\conv(\{f^m(\vec{z}_0) : m\in \nats\}) \) contains an integer
point in its relative interior.
\end{proposition}

\begin{proof}
  Since \(V_0\) is spanned by integer vectors,
  \(\Lambda:=V_0\cap\ints^n\) is a lattice of rank \(n_0\) in \(\reals^n\).
  Define \(C:=\conv(\{ f^m(\vec{z}_0) : m\in \nats\}) \subseteq V\)
  and \(C_0:=C-f^n(\vec{z}_0) \subseteq V_0\).
  We may assume that $n_0 \geq 1$ since otherwise $V$ is a singleton,
  \ie, $\vec{z}_0$ is a fixed point of $f$ and the proposition
  is vacuously true (here, note that a singleton set is its own
  relative interior).

  Let \(\theta:\reals^n \to V_0\) be the orthogonal projection of
  $\reals^n$ onto $V_0$.  Then \(\theta(\Lambda)\) is a lattice in
  \(V_0\) of full rank.  We claim that the lattice width of
  \(\theta(C_0)\) with respect to \(\theta(\Lambda)\) is infinite.
  Indeed for any non-zero vector \(\vec{v} \in \theta(\Lambda)\)
  we have
  \begin{gather}
    \vec{v}^\trans (\theta(f^m(\vec{z}_0)) -
    \theta(f^n(\vec{z}_0))) = \vec{v}^\trans 
    (f^m(\vec{z}_0)-f^n(\vec{z}_0)) \, , 
    \label{eq:UB}
  \end{gather}

  But \(\vec{v}\) is a non-zero vector in \(V_0\) with rational
  coefficients and hence Proposition~\ref{prop:unbounded} entails that
  the absolute value of~\eqref{eq:UB} is unbounded as \(m\) runs over
  \(\nats\). Since $V_0$ has positive dimension, this proves the
  claim.

  Since \(\theta(C_0)\) is a full-dimensional convex subset of
  \(\reals^{n_0}\), by Theorem~\ref{thm:flatness} we have that
  \(\theta(C_0)\) contains a point of \(\theta(\Lambda)\) in its
  relative interior and hence \(C_0\) contains a point of \(\Lambda\)
  (necessarily an integer point) in its relative interior.  Since $C$
  is the translation of $C_0$ by an integer vector, we conclude that
  \(C\) also contains an integer point in its relative interior.
\end{proof}

The following theorem characterises when an affine \slc loop with a
single guard is terminating over the integers.

\begin{theorem}
  \label{thm:one-guard}
  Loop~\eqref{eq:the-loop} is non-terminating on \(\ints\) if and
  only if the set $\mathit{PN}$ contains an integer point $\vec{z}$.
\end{theorem}

\begin{proof}
  If no such \(\vec{z}\) exists, then the loop is terminating by
  Proposition~\ref{prop:easy}.(\ref{itm: easy term iv's}).  Conversely,
  if such a \(\vec{z}\) exists then the loop is non-terminating
  by Lemma~\ref{lem:critical} and Proposition~\ref{prop:rel-int}.
\end{proof}

We postpone the question of the effectiveness of the above
characterisation until we handle loops with multiple guards.

\subsection{Multiple Guards}
\label{sec:multiple-guard}

Next we present a decision procedure for a general
affine \slc loop
\begin{gather}
  \transitions\, : \while \; (g_1(\vec{x}) \geq 0 \wedge
  \ldots \wedge g_m(\vec{x}) \geq 0) \; \wdo \; \vec{x}:=f(\vec{x}) \, ,
  \label{eq:multi-loop}
\end{gather}
with multiple guards.
Associated to Loop~\eqref{eq:multi-loop} we consider \(m\)
single-guard loops with a common update function:
\begin{gather*}
  \transitions_i \, : \, \while \; (g_i(\vec{x}) \geq 0) \;
  \wdo \; \vec{x}:=f(\vec{x}) \, ,
\end{gather*}
for \(i=1,\ldots,m\).  Clearly Loop~\eqref{eq:multi-loop} non-terminating if and
only if there exists \(\vec{z} \in \ints^n\) such that each loop
\(\transitions_i\) is non-terminating on \(\vec{z}\).

\begin{theorem}
  Let $\mathit{PN}_i$ be the set of potentially non-terminating points
  for each loop $\transitions_i$ for $i \in \{1,\ldots,m\}$ and write
  $\mathit{PN}:=\bigcap_{i=1}^m \mathit{PN}_i$.
  Then loop $\transitions$ of~\eqref{eq:multi-loop} is non-terminating
  over $\reals$ if and only if $\mathit{PN}$ is non-empty and
  $\transitions$ is non-terminating over $\rats$ if and only if
  $\mathit{PN}$ contains a rational point.
  If all numerical constants in $\transitions$ are integer, then the loop
  is non-terminating over $\ints$ if and only if $\mathit{PN}$ contains an
  integer point.
\label{thm:MAIN}
\end{theorem}

Theorem~\ref{thm:MAIN} leads to the following procedure for deciding
termination of a given affine \slc loop $\transitions$, as shown
in~\eqref{eq:multi-loop}, over a ring $\numdom \in \{\reals,\rats,\ints\}$:
\begin{enumerate}
\item Compute the non-zero eigenvalues $\lambda_1,\ldots,\lambda_s$ of
  the matrix corresponding to the loop update function, as given
  in~\eqref{eq:update}.
  Let $\gamma_i := \frac{\lambda_i}{|\lambda_i|}$ for
  $i \in \{1,\ldots,s\}$.
\item Compute the dominance preorder \(\preccurlyeq\) among  eigenvalues.
\item Compute a basis of the group  \(L(\bgamma)\) of multiplicative relations among $\gamma_1,\ldots,\gamma_s$.
\item Compute the set $\mathit{PN}_i$ of potentially non-terminating
  points for each loop $\transitions_i$ using steps 2 and 3.
\item Return ``non-terminating'' if
  \( \mathit{PN}:=\bigcap_{i=1}^m \mathit{PN}_i\) contains a point with all coordinates in
  $\numdom$ and otherwise return ``terminating''.
\end{enumerate}

We briefly discuss the effectiveness of each step.  Step~1 involves
computing the roots of an integer polynomial.  These can be
represented by rational approximations of sufficient accuracy to
distinguish the roots from each other.  (The required accuracy is
determined by standard polynomial root separation bounds.)  Such
approximations can be computed in polynomial time in the loop
description.  Furthermore, the necessary computations on algebraic numbers (arithmetic,
testing equality, the less-than relation for real numbers) can be done in polynomial time.
These approximations can be used to determine the
dominance preorder in Step~2.  Step~3 can be accomplished in
polynomial time using the algorithm of~\citet{Combot25}.  Thus
Steps~1-3 can be carried out in polynomial time in the size of the
linear loop.  For Step~4 we describe the semi-algebraic set
$\mathit{PN}$, as given in Definition~\ref{def:PN}, by a
polynomial-size formula of first-order logic of with two quantifier
alternations.  Whether such a set contains a real, rational, or
integer point can be decided in exponential time in the size of the
formula~\citep[Theorem 1.1]{KP97}.  Thus the overall running time of
the procedure above is exponential in the size of the input linear
loop.  Note that in the case of termination over $\mathbb Z$, we
pre-process the input loop in the case of degeneracy, which may yield
exponentially many different problem instances.

We have thus established the main result of this section:

\begin{theorem}
  There is a procedure to decide termination of affine \slc loops over
  $\reals$, $\rats$, and $\ints$.
\end{theorem}

As a final comment, we note that all results presented in this section
hold also when the loop guard involve strict inequalities.

\subsection{Termination with Respect to Initial States}
\label{sec:aff:init}

There are not many results on the termination of an affine \slc loop
with respect to a given initial state (or set of initial states). This
is likely because the problem is very difficult;  it
subsumes \emph{Positivity Problem} for linear recurrence
sequences (\eg, see~\citet{KenisonNO023}).
This is the problem of determining whether all terms in a given
integer linear recurrence sequence are positive. Decidability of the
Positivity Problem is a longstanding open problem (going back at least
as far as the 1970s~\citep{RS94,Soi75}), and results by
\citet{OuaknineW14a} suggest that a solution to the problem will
require significant breakthroughs in number theory.

While decidability of the positivity problem is still open for the
general case, partial solutions for some special cases
exist~\citep{OuaknineW14,OuaknineW14b,AkshayBV17,KenisonNO023}.
Thus, the halting problem (termination \wrt a single initial state)
for any subclass of integer affine \slc loops whose corresponding
recurrence sequences fall in these special cases, is decidable.
For example, \citet{OuaknineW14} show that the positivity problem is
decidable for recurrences of order 5 or less, which implies
decidability of the halting problem for integer affine \slc loops with
at most 4 variables (we need an extra variable to eliminate the
constants in the guard and the update).
\citet{KincaidBCR19} show decidability of the halting
problem for integer affine
\slc loops where every eigenvalue of the  update matrix is
a radical of a rational number.

Another class of update matrices handled in previous work is the class of \emph{finite monoid affine relations}, this is the class of loops where matrix $A$ in the update $\vec{x}' = A\vec{x} +b$ has a finite
set of powers $\{A^i \mid i>0\}$. An alternative characterisation is that the matrix is diagonalisable with
eigenvalues in $\{0,1\}$. It was shown by \citet{Boigelot98,BIK10,FL02} that the iteration of such
an update has a closed form in Presburger arithmetic. Consequently, the termination problem for
loops with such an update is decidable over the integers as well as over the reals; moreover the set of non-terminating initial values can be precisely computed. This result can be extended to loops where the guard is a Presburger formula~\citep{BIK14}.

\citet{HarkFG20} show that the halting problem is decidable for affine
\slc loops with a triangular update matrix, over any ring
$\ints \subseteq \numdom \subseteq \areals$ (where $\areals$ is the ring of algebraic real numbers).
Their results go beyond simple linear loops, as they allow the loop
condition to be any Boolean formula over atoms of the form
$p(\vec{x}) \ge 0$ or $p(\vec{x}) > 0$, and the update can also
include polynomial assignments that respect the triangular condition,
which means that $x_i$ does not depend on $x_j$ for $j < i$,
and $x_i$ depends linearly on itself.

The core idea is that the truth value of the condition always
stabilises after some iterations, and since such loops have
(computable) closed forms, a bound on the number of iterations
to stabilisation can be computed.

A method for computing a subset of the non-terminating initial states
for affine \slc loops over the real numbers was presented by
\citet{Li17}. For linear homogeneous loops with only two program
variables (and a strict inequality in the guard), \citet{DaiX12}
provided a complete algorithm to compute the full set of
non-terminating initial states.

\begin{problem}
  Is termination of affine \slc loops \wrt to an initial value, or a
  (polyhedral) set of initial states, over $\reals$, $\rats$ or
  $\ints$ decidable?
\end{problem}

\subsection{Other Results Related to Affine \slc Loops}
\label{sec:aff:related}

\citet{Li14} gave an alternative algorithm to decide termination of
linear programs over \(\reals\).  Whereas the approach of
\citet{Tiwari:04} and \citet{Bra06} is based on searching for
eventually non-terminating initial values, Li's algorithm outputs, in
the case of non-termination, a genuinely non-terminating initial
value.

\citet{XiaYZZ11} show that the decision procedure of \citet{Tiwari:04}
suffers from imprecision when implemented using floating-point
arithmetic (to compute Jordan forms), and they fix this imprecision by
developing a symbolic implementation.

\citet{XuL13} show how to decide the termination over the
reals of an extension of \slc loops with solvable polynomial
assignments.  

\citet{FrohnG19} showed decidability of termination of linear loops
over \(\ints\) under the assumption that the loop update matrix is
upper-triangular, that is, all elements below the main diagonal are
zero. \citet{HarkFG25} extend the approach to loops with nonlinear updates (which is beyond
the scope of this \survey), and also generalise the loop guard to be any Boolean combination
of inequalities (\ie, not necessarily a convex polyhedron), while still showing decidability over $\reals$
and $\areals$ (the ring of algebraic real numbers). Moreover, in the same work, they consider affine
loops where the update matrix has rational spectrum, and show that its termination, over either the
integers, rational numbers or algebraic reals is \conpc. In the more general case of matrices
with a real spectrum, they show that termination over the algebraic reals is $\forall\reals$-complete; this
class includes problems reducible to validity of a universally quantified formula
of polynomial inequalities over the reals, and is contained in \pspace.

\citet{ZhuK20} explore how techniques for proving termination of
affine \slc loops can be used to prove termination of more realistic
programs.

Using techniques that ultimately rely on the $p$-adic Subspace Theorem
in Diophantine approximation, \citet{OPW15} gave an effective
characterisation of the set of all \emph{eventually non-terminating
  points}\footnote{A point is eventually non-terminating if it evolves
  into a non-terminating point after a finite number of iterations of
  the loop body, disregarding the loop guard.  The problem of
  determining whether a given point is eventually non-terminating for
  a given loop is equivalent to the \emph{Ultimate Positivity Problem
  } for linear recurrence sequence.  This asks to determine whether
  all but finitely many terms in a given linear recurrence sequence
  are positive.} for affine \slc loops whose update matrix is
diagonalisable.
This suffices to decide whether such a loop terminates over the
integers.
In contrast, the method presented in this section solves the
termination problem without giving an effective characterisation of
all non-terminating points (or eventually non-terminating points).

\section{Termination of  Single-path Linear-Constraint Loops}
\label{sec:dec:slc}

The case of general \slc loops constitutes an important open problem:

\begin{problem}
  Is termination of \slc loops, with rational or equivalently integer
  coefficients, over \reals, \rats, or \ints decidable?
\end{problem}

Attempts to solve this problem have lead to results for special cases
or extensions of \slc loops. Next we overview these results.

\citet{BGM12} considered \slc loops where irrational coefficients are
allowed (recall that \slc loops, as defined in
Section~\ref{sec:loops}, involve only rational coefficients).

\begin{theorem} 
\label{th:und-irrational}
Termination of \slc loops, where the coefficients are from
$\ints \cup \{r\}$, for a single arbitrary \emph{irrational} constant
$r \in \reals$, and variables range over integers, is undecidable.
\end{theorem}

The proof of this result shows that such loops can simulate a counter
program. The key idea is to use linear constraints that involve $r$ as
a coefficient to simulate the instruction
$x_j=\mathit{isPositive}(x_i)$, where $\mathit{isPositive}$ returns
$1$ if $x_i>0$ and $0$ otherwise.

\citet{BGM12} show that Petri nets can be simulated using integer \slc
loops, and thus provide an \expspace lower-bound on the hardness of
proving termination of integer \slc loops \wrt to polyhedral set of
initial states, even for deterministic \slc loops.
For nondeterministic \slc loop, a similar reduction from \citet{Ben-Amram14} proves that 
termination with a polyhedral set of initial states is \ackh (based on recent results on the hardness
of reachability in Vector Addition Systems~\citep{CzOr21,Leroux21}).

\citet{GLOW2024} consider \slc loops but in two dimensions only (\ie,
two variables) and prove that termination is decidable.

\subsection{Octagonal Loops}
\label{sec:dec:slc:oct}

\citet{BIK14} consider octagonal \slc loops, a special case of \slc
loops where the transition polyhedron is defined by inequalities of
the form $\pm x \le c$ or $\pm x \pm y \le c$. They prove that
termination over the integers is decidable in polynomial time, a
result that also holds for the reals. Furthermore, for loops that do
not terminate universally, they can compute a weakest precondition to
non-termination, which is an octagonal constraints system.

To present their results in a little more detail, we cite the following results about octagonal 
constraint loops:
\begin{itemize}
\item It is possible to represent a transition relation $R$, defined by octagonal constraints over
variables $x_1,\dots,x_n$, as a $4n\times 4n$ matrix.
\item  This matrix representation allows for polynomial-%
time computation of the composition $R \circ S$ of two such relations $R$ and $S$; and allows us to extract the pre-image
$\pre_R(\ints^n)$ of such a relation $R$~(as \citet{BIK14}, we concentrate on integer loops).
\item The matrix representation also allows us to efficiently test containment and equality between two
such pre-images, and to test whether a relation is satisfiable.
\end{itemize}

Since composition of relations is associative, the polynomial-time composition operator can be used
to efficiently compute powers $R^i$ using repeated squaring. This is used in
Algorithm~\ref{alg:recurset}, which computes the weakest non-termination precondition for an octagonal loop.

\begin{algorithm}[t]
\caption{Weakest non-termination precondition for Octagonal Relations.}
\label{alg:recurset}
\DontPrintSemicolon
\LinesNumberedHidden
\KwIn{An octagonal transition relation $R \subseteq \ints^n\times\ints^n$.}
\KwOut{An Octagonal constraint representing the non-termination set for $R$.}
\Begin{
\setcounter{AlgoLine}{0}
  \ShowLn $V \leftarrow R^{5^{2n}}$ \;
  \ShowLn $W \leftarrow R^{5^{2n}+1}$ \;
  \ShowLn \uIf{$W$ is empty \textbf{or} $\pre_{V}(\ints^n) \supsetneq \pre_{W}(\ints^n)$}{
  \ShowLn \textbf{return} $\textrm{false}$}
  \ShowLn \Else{ \ShowLn\textbf{return} $\pre_{V}(\ints^n)$}
}
\end{algorithm}

Where does the ``magic number'' $5^{2n}$ come from? To give an intuition about this crucial
ingredient of the algorithm, we first have to talk about octagonal constraints and \emph{difference
constraints}.  Difference constraints are a subclass of octagonal constraints, in which all inequalities
are of the form $x-y\le c$. Octagonal constraints are represented by \citet{BIK14} as difference
constraints using a set of $n$ auxiliary variables, hence the $2n$ in the exponent. For difference
constraints, the ``magic number'' is $5^n$. We focus on this case for simplicity of presentation.

It is possible to represent a difference constraint $R$ by a graph $\mathcal{G}_R$ with variables for $x_1,\dots,x_n$ and $x'_1,\dots,x'_n$. The graph has an arc $x \to y$ with weight $c$ if $x-y \ge c$ is in the constraint. Further, it is possible to represent the $m$-fold iteration of $R$ by concatenating $m$ copies of
$\mathcal{G}_R$~(see Figure~\ref{fig:BIK14}, (\subref{fig:BIK14:a})--(\subref{fig:BIK14:b})). We will denote this $m$-fold of $\mathcal{G}_R$ by $\mathcal{G}^m_R$.

If we want to find the implications of $R^m$ on $x_1,\dots,x_n$, in other words the inequalities that define $\pre_{R^m}(\ints^n)$, we look for paths in $\mathcal{G}_R^m$ from $x_i$ to $x_j$, for $1\le i,j\le n$.
For instance, in Figure~\ref{fig:BIK14:c} we see a path that implies $x_2 - x_4 \le -3$.
In order to analyse the evolution of such paths as $m$ grows, let us concentrate on one path.
We see that every node in the graph contributes to the path in one of five ways, which~\citet{BIK14}
label by:
\begin{description}
\item{$r$} --- the path crosses the node from left to right.
\item{$l$} --- the path crosses the node from right to left.
\item{$rl$} --- the path reaches the node from the left and turns left again.
\item{$lr$} --- the path reaches the node from the right and turns right again.
\item{$\bot$} --- the path does not reach this node.
\end{description}
See Figure~\ref{fig:BIK14:d} for an example.
Since there are 5 labels, if we form $\mathcal{G}_R^m$ for $m\ge 5^n$, there will be a column that
repeats (which happens in the figure).
The part of the path between such columns can be ``pumped'' as in the pumping lemma
for finite automata (\citet{BIK14}~actually use an automaton to analyse this behaviour).
It follows that the pre-image $\pre_{R^m}(\ints^n)$ will stabilise at or before $m=5^n$, or,
if it does not stabilise, it is because there is a segment of a path with negative weight that can
be pumped. The implication of this for the pre-image is that we have an inequality $x_i-x_j\le c_m$
for numbers $c_m$ that descend as $m$ grows.
Whatever the initial values of $x_i$ and $x_j$, the inequality will be
false for $m$ large enough. Therefore the iteration of the loop terminates. Note that if we have a cycle
of positive weight, it means that we can deduce $x_i-x_j\le c_m$ for an \emph{ascending} sequence
$c_m$, but then the net effect is just $x_i-x_j\le c_m$ for the smallest such $m$. If this is the case for 
all $i,j$ then the pre-image stabilises and is a non-terminating initial value set.

\newcommand{\symbolGone}[0]{
      \scalebox{0.85}{\begin{tikzpicture}
        \TermGridGenNoCapBoxed{0.0}{0.7}{2}{0.0}{0.4}{4}{0.4}{0.4}
        \foreach \ii in {1,...,1} {
          \pgfmathtruncatemacro\jj{\ii+1}
          \TermGridEdgeC{\ii}{1}{\jj}{3}{\tiny$0$}{above}
          \TermGridEdgeC{\jj}{4}{\ii}{4}{\tiny$0$}{above}
        }
      \end{tikzpicture}}
}
\newcommand{\symbolGtwo}[0]{
      \scalebox{0.85}{\begin{tikzpicture}
        \TermGridGenNoCapBoxed{0.0}{0.7}{2}{0.0}{0.4}{4}{0.4}{0.4}
        \foreach \ii in {1,...,1} {
          \pgfmathtruncatemacro\jj{\ii+1}
          \TermGridEdgeC{\ii}{3}{\jj}{2}{\tiny$0$}{above}
          \TermGridEdgeC{\jj}{4}{\ii}{4}{\tiny$0$}{above}
        }
      \end{tikzpicture}}
}
\newcommand{\symbolGthree}[0]{
      \scalebox{0.85}{\begin{tikzpicture}
        \TermGridGenNoCapBoxed{0.0}{0.7}{2}{0.0}{0.4}{4}{0.4}{0.4}
        \foreach \ii in {1,...,1} {
          \pgfmathtruncatemacro\jj{\ii+1}
          \TermGridEdgeC{\ii}{2}{\jj}{1}{\tiny$-\!1$}{above}
          \TermGridEdgeC{\jj}{4}{\ii}{4}{\tiny$0$}{above}
        }
      \end{tikzpicture}}
}
\newcommand{\symbolGfour}[0]{
      \scalebox{0.85}{\begin{tikzpicture}
        \TermGridGenNoCapBoxed{0.0}{0.7}{2}{0.0}{0.4}{4}{0.4}{0.4}
        \foreach \ii in {1,...,1} {
          \pgfmathtruncatemacro\jj{\ii+1}
          \TermGridEdgeC{\ii}{1}{\jj}{3}{\tiny$0$}{above}
          \TermGridEdgeC{\jj}{3}{\ii}{4}{\tiny$0$}{above}
        }
      \end{tikzpicture}}
}
\newcommand{\symbolGeps}[0]{
      \scalebox{0.85}{\begin{tikzpicture}
        \TermGridGenNoCapBoxed{0.0}{0.7}{2}{0.0}{0.4}{4}{0.4}{0.4}
        \foreach \ii in {1,...,1} {
          \pgfmathtruncatemacro\jj{\ii+1}
        }
      \end{tikzpicture}}
}
\newcommand{\symbolGfive}[0]{
      \scalebox{0.85}{\begin{tikzpicture}
        \TermGridGenNoCapBoxed{0.0}{0.7}{2}{0.0}{0.4}{4}{0.4}{0.4}
        \foreach \ii in {1,...,1} {
          \pgfmathtruncatemacro\jj{\ii+1}
          \TermGridEdgeC{\ii}{3}{\jj}{2}{\tiny$0$}{above}
          \TermGridEdgeC{\jj}{3}{\ii}{4}{\tiny$0$}{above}
        }
      \end{tikzpicture}}
}
\newcommand{\symbolGsix}[0]{
      \scalebox{0.85}{\begin{tikzpicture}
        \TermGridGenNoCapBoxed{0.0}{0.7}{2}{0.0}{0.4}{4}{0.4}{0.4}
        \foreach \ii in {1,...,1} {
          \pgfmathtruncatemacro\jj{\ii+1}
          \TermGridEdgeC{\ii}{2}{\jj}{1}{\tiny$-\!1$}{above}
          \TermGridEdgeC{\jj}{3}{\ii}{4}{\tiny$0$}{above}
        }
      \end{tikzpicture}}
}

\begin{figure}[htbp]
  \centering
  
  \begin{subfigure}[b]{\textwidth}
    \centering
    \mbox{\begin{minipage}{2.5cm}
        \scalebox{0.85}{\begin{tikzpicture}
            \TermGridGenDifferentCap{0.0}{1.0}{2}{0.0}{0.75}{4}{0.7}{0.5}{0}
            \foreach \ii in {1,...,1} {
              \pgfmathtruncatemacro\jj{\ii+1}
              \TermGridEdgeC{\ii}{2}{\jj}{1}{\tiny$-\!1$}{right}
              \TermGridEdgeC{\ii}{3}{\jj}{2}{\tiny$0$}{left}
              \TermGridEdgeC{\ii}{1}{\jj}{3}{\tiny$0$}{right}
              \TermGridEdgeC{\jj}{4}{\ii}{4}{\tiny$0$}{below}
              \TermGridEdgeC{\jj}{3}{\ii}{4}{\tiny$0$}{right}
            }
          \end{tikzpicture}}
      \end{minipage}}
    \caption{$\mathcal{G}_R$ -- the constraint graph of $R$}
    \label{fig:BIK14:a}
  \end{subfigure}

  \begin{subfigure}[b]{\textwidth}
    \centering
    \mbox{\begin{minipage}{8cm}
        \scalebox{0.85}{\begin{tikzpicture}
            \TermGridGen{0.0}{1.0}{9}{0.0}{0.75}{4}{0.7}{0.5}{0}
            \foreach \ii in {1,...,8} {
              \pgfmathtruncatemacro\jj{\ii+1}
              \TermGridEdgeC{\ii}{2}{\jj}{1}{\tiny$-\!1$}{right}
              \TermGridEdgeC{\ii}{3}{\jj}{2}{\tiny$0$}{left}
              \TermGridEdgeC{\ii}{1}{\jj}{3}{\tiny$0$}{right}
              \TermGridEdgeC{\jj}{4}{\ii}{4}{\tiny$0$}{below}
              \TermGridEdgeC{\jj}{3}{\ii}{4}{\tiny$0$}{right}
            }
          \end{tikzpicture}}
      \end{minipage}}
    \caption{$\mathcal{G}_R^8$ -- the 8-times iteration of $\mathcal{G}_R$}
    \label{fig:BIK14:b}
  \end{subfigure}

  \begin{subfigure}[b]{\textwidth}
    \centering
    \mbox{\begin{minipage}{8cm}
        \scalebox{1.0}{\begin{tikzpicture}
            \scriptsize
            \TermGridGen{0.0}{0.8}{9}{0.0}{0.4}{4}{0.7}{0.5}{0}
            \TermGridEdgeC{1}{2}{2}{1}{$-\!1$}{above}
            \TermGridEdgeC{2}{1}{3}{3}{$0$}{above}
            \TermGridEdgeC{3}{3}{4}{2}{$0$}{above}
            \TermGridEdgeC{4}{2}{5}{1}{$-\!1$}{above}
            \TermGridEdgeC{5}{1}{6}{3}{$0$}{above}
            \TermGridEdgeC{6}{3}{7}{2}{$0$}{above}
            \TermGridEdgeC{7}{2}{8}{1}{$-\!1$}{above}
            \TermGridEdgeC{8}{1}{9}{3}{$0$}{above}
            \TermGridEdgeC{9}{3}{8}{4}{$0$}{above}
            \TermGridEdgeC{8}{4}{7}{4}{$0$}{above}
            \TermGridEdgeC{7}{4}{6}{4}{$0$}{above}
            \TermGridEdgeC{6}{4}{5}{4}{$0$}{above}
            \TermGridEdgeC{5}{4}{4}{4}{$0$}{above}
            \TermGridEdgeC{4}{4}{3}{4}{$0$}{above}
            \TermGridEdgeC{3}{4}{2}{4}{$0$}{above}
            \TermGridEdgeC{2}{4}{1}{4}{$0$}{above}
          \end{tikzpicture}}
      \end{minipage}}
    \caption{A~path from $x_2$ to $x_4$ in $\mathcal{G}_R^8$}
    \label{fig:BIK14:c}
  \end{subfigure}

  \begin{subfigure}[b]{\textwidth}
    \centering
    \mbox{\begin{minipage}{12cm}
        \hspace{0cm}\mbox{\scalebox{1.0}{\begin{tikzpicture}
              \TermZrBase{1}{4}{4}{0.7}{0.65}{0.5}{1}{3}{1}
              \TermZrBase{5}{8}{4}{0.7}{0.65}{0.5}{4}{3}{1}
              \newarray\aCaption
              \readarray{aCaption}{$G_3$ & $G_1$ & $G_2$ & $G_3$}
              \TermZrCapGraph{1}{4}{4}{0.7}{0.65}{0.5}
              \readarray{aCaption}{$G_1$ & $G_2$ & $G_3$ & $G_4$}
              \TermZrCapGraph{5}{8}{4}{0.7}{0.65}{0.5}
              
              \TermZrStateElem{0}{1}{\bot}
              \TermZrStateElem{0}{2}{r}
              \TermZrStateElem{0}{3}{\bot}
              \TermZrStateElem{0}{4}{l}
              
              \TermZrStateElem{1}{1}{r}
              \TermZrStateElem{1}{2}{\bot}
              \TermZrStateElem{1}{3}{\bot}
              \TermZrStateElem{1}{4}{l}
              
              \TermZrStateElem{2}{1}{\bot}
              \TermZrStateElem{2}{2}{\bot}
              \TermZrStateElem{2}{3}{r}
              \TermZrStateElem{2}{4}{l}
              
              \TermZrStateElem{3}{1}{\bot}
              \TermZrStateElem{3}{2}{r}
              \TermZrStateElem{3}{3}{\bot}
              \TermZrStateElem{3}{4}{l}
              
              \TermZrStateElem{4}{1}{r}
              \TermZrStateElem{4}{2}{\bot}
              \TermZrStateElem{4}{3}{\bot}
              \TermZrStateElem{4}{4}{l}
              
              \TermZrStateElem{5}{1}{\bot}
              \TermZrStateElem{5}{2}{\bot}
              \TermZrStateElem{5}{3}{r}
              \TermZrStateElem{5}{4}{l}
              
              \TermZrStateElem{6}{1}{\bot}
              \TermZrStateElem{6}{2}{r}
              \TermZrStateElem{6}{3}{\bot}
              \TermZrStateElem{6}{4}{l}
              
              \TermZrStateElem{7}{1}{r}
              \TermZrStateElem{7}{2}{\bot}
              \TermZrStateElem{7}{3}{\bot}
              \TermZrStateElem{7}{4}{l}
              
              \TermZrStateElem{8}{1}{\bot}
              \TermZrStateElem{8}{2}{\bot}
              \TermZrStateElem{8}{3}{rl}
              \TermZrStateElem{8}{4}{\bot}
              
              {\scriptsize
                \TermZrEdgeFWLab{1}{2}{1}{-1}
                \TermZrEdgeBWLab{1}{4}{4}{0}
                
                \TermZrEdgeFWLab{2}{1}{3}{0}
                \TermZrEdgeBWLab{2}{4}{4}{0}
                
                \TermZrEdgeFWLab{3}{3}{2}{0}
                \TermZrEdgeBWLab{3}{4}{4}{0}
                
                \TermZrEdgeFWLab{4}{2}{1}{-1}
                \TermZrEdgeBWLab{4}{4}{4}{0}
                
                \TermZrEdgeFWLab{5}{1}{3}{0}
                \TermZrEdgeBWLab{5}{4}{4}{0}
                
                \TermZrEdgeFWLab{6}{3}{2}{0}
                \TermZrEdgeBWLab{6}{4}{4}{0}
                
                \TermZrEdgeFWLab{7}{2}{1}{-1}
                \TermZrEdgeBWLab{7}{4}{4}{0}
                
                \TermZrEdgeFWLab{8}{1}{3}{0}
                \TermZrEdgeBWLab{8}{3}{4}{0}
              }%
            \end{tikzpicture}}}
      \end{minipage}}
    \caption{Labeling the nodes to identify a portion of the path that may be pumped}
    \label{fig:BIK14:d}
  \end{subfigure}

  \caption{ Illustration (from \citet{BIK14}) of various notions for a
    difference bounds relation $R$ defoned by
    $\{ x_2- x'_1\leq -1 ,\, x_3- x'_2\leq 0 ,\, x_1- x'_3\leq 0 ,\,
    x'_4- x_4\leq 0 ,\, x'_3- x_4\leq 0\}$.}
\label{fig:BIK14}
\end{figure}

\section{Termination of  Multi-path Linear-Constraint Loops}
\label{sec:dec:mlc}

\citet{Tiwari:04} observed that termination of \mlc loops, and
therefore of general \cfgs, is undecidable over $\ints$, $\rats$ and
$\reals$.

\begin{theorem} 
\label{th:mlc:undec}
The termination problem, with and without initial stat\-es, is
undecidable for \mlc loops, over $\ints$, $\rats$ and $\reals$.
\end{theorem}

This undecidability is shown even for \mlc loops where every path is
defined by an affine \slc loop and the paths are mutually exclusive,
making the \mlc loop deterministic.
This is demonstrated by a reduction from counter programs, where a
counter program with $n$ counters is translated to an \mlc loop with
$n$ counter variables and a location variable $\mathit{pc}$, as
follows:
\begin{itemize}
\item Increment or decrement of counter $X_i$ at location $j$
  generates the path
  $\{\mathit{pc}=j, x_i'=x_i \pm 1, \mathit{pc}'=j+1\}$; and
\item Conditional statement
  ``$\mathit{if}~X_i>0~\mathit{then}~k_1~\mathit{else}~k_2$'' at
  location $j$ generates the paths
  $\{\mathit{pc}=j, x_i \ge 1, \mathit{pc}'=k_1\}$ and
  $\{\mathit{pc}=j, x_i \le 0, \mathit{pc}'=k_2\}$.
\end{itemize}
This reduction implies that termination of integer \mlc loops, with
and without initial states, is undecidable over $\ints$.
For undecidability over $\reals$ and $\rats$, \citet{Tiwari:04}
observes that the generated \mlc loop is terminating over $\ints$ if
and only if it is terminating over $\reals$ and $\rats$.
Furthermore, due to Theorem~\ref{thm:counter}, undecidability already
hold for $3$ variables.

\citet{BGM12} show that undecidability already holds when restricting the
\mlc loop to $2$ paths where each is an affine \slc loop.

\begin{theorem} 
\label{th:und-2piece}
The termination problem, with and without initial set of states,
is undecidable for loops of the following form
\begin{equation*}
\while\; (B\vec{x} \ge \vec{b}) \; \wdo \; \vec{x} := 
\left\{
    \begin{array}{lll}
      A_0\vec{x} &~~~& x_i\le 0\\
      A_1\vec{x} &~~~& x_i > 0
  \end{array}\right.
\end{equation*}
where the state vector $\vec{x}$ ranges over
$\ints^n$, 
$A_0, A_1\in {\ints}^{n\times n}$, $\vec{b}\in {\ints}^p$ for some
$p>0$, $B\in {\ints}^{p\times n}$, and $x_i \in \vec{x}$.
\end{theorem}

The proof of this result is by a reduction from $2$-counter
programs.

Another restricted form of \mlc loop for which termination is known to be undecidable is a deterministic loop
in two variables, of the form
\begin{equation*}
\while\; (x_1+x_2>0) \; \wdo \; (x_1,x_2) := f(x_1,x_2) 
\end{equation*}
where $f$ is piecewise-affine, whose pieces are defined by linear inequalities (thus defining the paths
of the \mlc loop). The termination of such loops is undecidable over the rationals and reals~\citep{BlondelBKPT01} as well as over integers~\citep{BenAmram15}.

We note however that Tiwari observes that the decidability of termination of linear loops
allows us to decide the termination of multi-path loops in the following favourable case.
Let us denote, as in Section \ref{sec:loops}, the paths of the loop as transition polyhedra $\poly{Q}_1,
\dots,\poly{Q}_k$, and consider each $\poly{Q}_i$ as a binary relation on $\reals^n$ (respectively,
$\rats^n$, $\ints^n$), so that $\poly{Q}_i\circ\poly{Q}_j$ denote the composition of relations.

\begin{theorem}
Let $\poly{Q}_1, \dots,\poly{Q}_k$ be a \mlc loop over the reals (respectively, the rationals or integers). Let $T = \bigcup_i \poly{Q}_i$ be the set of all loop
transitions.
 Assume that whenever $i<j$, it is the case that
$\poly{Q}_j\circ\poly{Q}_i \subseteq \poly{Q}_i\circ T^*$. Then, the \mlc loop terminates if and only if
each $\poly{Q}_i$ does.
\end{theorem}

\chapter{Ranking Functions}
\label{chp:rfs}

The use of ranking functions to prove termination goes back to~\citet{Turing48} 
and was subsequently popularised by~\citet{Floyd67}.

\begin{definition}
\label{def:rf}
Let $T\subseteq S \times S$ be a transition relation,
$S_0 \subseteq S$ a set of initial states, $\trres{T}{S_0}$ the restriction
of $T$ to the reachable states $\reach{T}{S_0}$, and
$\tuple{W,\preceq}$ a partially ordered set such that $\preceq$ is
well-founded.
We say that $\rho: S \to W$ is a ranking function for $T$ \wrt $S_0$,
if for every $(s,s') \in \trres{T}{S_0}$, $\rho(s) \succ \rho(s')$, where
$\succ$ is the strict order relation on $W$.
\end{definition}

Note that if $S_0=S$ then $\trres{T}{S_0}=T$, a fact used when we
consider universal termination.

The fact that $\rho$ proves termination of $T$ \wrt the set of initial
states $S_0$ is immediate from the definition: a non-terminating
computation staring in $s_0\in S_0$ would yield an infinite descending
chain in $W$, contradicting the well-foundedness assumption.
On the other hand, every terminating transition relation \wrt the set
of initial states $S_0$ has a ranking function. Let
$W = \reach{T}{S_0} \cup \{\bot\}$, ordered by the reachability
relation with a least element $\bot$, and let $\rho(s)=s$ if
$s \in \reach{T}{S_0}$, otherwise $\rho(s)=\bot$.%
\footnote{%
  There is some room for explanation regarding whether $W$ is
  partially or totally ordered. Our statement is easy to see if
  partial orders are allowed, but also holds if total orders are
  required, since the partial order can be extended to a total one.}

The last observation shows that to obtain practical methods
for proving termination one must restrict the search to a specific class of ranking
functions, otherwise the problem is as hard as termination itself.
Clearly, the choice of the class determines the decidability and
computational complexity of the resulting decision problems.

In this \chp, we are concerned with ranking functions that are based
on linear combinations of state variables, for the different kinds of
programs defined in Section~\ref{sec:programs}, and with or without
restricting the initial states, \ie, termination and universal
termination.

We begin, in Section~\ref{sec:lrf}, with \emph{linear ranking functions} (\lrfs);
we discuss the complexity of finding such ranking functions in various settings.
Then in Section~\ref{sec:llrf} we discusses \emph{lexicographic-linear ranking
functions} (\llrfs). This kind of ranking function appeared in the literature in various variants,
and our goal in this \survey is to present multiple variants in a unified manner as much as possible.
Finally, Section~\ref{sec:other_rfs} lists some references regarding other
kinds of ranking functions, which we do not expand upon.

\section{Linear Ranking Functions}
\label{sec:lrf}

In this section we survey algorithmic and complexity aspects of linear ranking functions
(briefly, \lrfs)
for \slc loops, \mlc loops, and the general case of \cfgs. The domain
of program variables is assumed, by default, to be the rationals, but
all results apply also to the case of real valued variables. The
integer case is discussed separately. For each case, we first consider
termination without any assumption on the input values, \ie,
universal termination, and then treat the case when a polyhedral set
of initial states is given.

Recall that an affine linear function $\rho:\rats^{n} \to \rats$ is a
function of the form
$\rho(\vec{x}) = \vect{\rfcoeff}\vec{x} + \rfcoeff_0$, where
$\vect{\rfcoeff}\in\rats^n$ is a row vector and
$\rfcoeff_0\in\rats$. For such a function, and a transition
${\vec{x''}} =\tr{\vec{x}}{\vec{x}'}$, we write
$\diff{\rho}(\vec{x}'')$ for the difference
$\rho(\vec{x}) - \rho(\vec{x}')$.

\begin{definition}[\lrf]
\label{def:lrf}
Given a rational \mlc loop
$\transitions_1, \ldots, \transitions_k \subseteq \rats^{2n}$, we say
that an affine linear function $\rho$ is an \lrf for the loop if the
following hold for every
$\vec{x}'' \in \transitions_1 \cup \cdots \cup \transitions_k$:
\begin{align}
 \rho(\vec{x})  \ge 0  \,, \label{eq:lrf:1}\\
 \diff{\rho}(\vec{x}'')  \ge 1 \,. \label{eq:lrf:2} 
\end{align}
\end{definition}

\begin{remark}
\label{rem:lift}
Note that the co-domain of $\rho$ is $\rats$ which is not well-founded
under the usual order. However, it is easy to see that such a function
proves termination, and it can be converted to match
Definition~\ref{def:rf} by considering
$\max(0,\lceil \rho+1 \rceil): \rats^{n}\to \nats$.
Such a consideration will apply to all the following definitions which
are based on this one.
\end{remark}

\begin{remark}
\label{rem:delta}
We could replace \eqref{eq:lrf:2} with
$\diff{\rho}(\vec{x}'') \ge \delta$ for an arbitrary constant
$\delta > 0$.  Indeed, it suffices to multiply $\rho$ by $\frac{1}{\delta}$ to
obtain the original condition of Definition~\ref{def:lrf}. This is
again an observation that we will take for granted when considering
variants of this definition.
Note that using $\diff{\rho}(\vec{x}'') \ge 1$ is important in
complexity analysis, since then an \lrf induces a linear bound on the
length of corresponding traces.
\end{remark}

\begin{remark}
\label{rem:weak}
When considering integer loops, we can use a strict inequality
$\diff{\rho}(\vec{x}'') > 0$ instead of \eqref{eq:lrf:2}, because we
may assume that $\rho$ used integer coefficients. This change is not
obviously safe when dealing with the rationals, so when we do use the strict inequality,
we refer cautiously to a \emph{weak} ranking function (versus a
\emph{strict} one).  Interestingly, in the case of \lrf and loops
given by polyhedra, it is easy to prove that a weak \lrf is also a
strict one, due to the fact that a bounded \lp minimisation problem
always attains its minimum (thus if $\diff{\rho}(\vec{x}'') > 0$ holds
over $\poly{Q}$, then there is $\delta>0$ such that
$\diff{\rho}(\vec{x}'') \ge\delta$ holds as well).
\end{remark}

The rest of this section is structured as follows:
Sections~\ref{sec:slc:rat} and~\ref{sec:slc:int} review results on the
\lrf problem for rational and integer \slc loops, respectively;
Section~\ref{sec:slc:int} reviews results on the \lrf problem for \mlc
loops; Section~\ref{sec:lrf:cfg} reviews results on the \lrf problem
for \cfgs; Section~\ref{sec:lrf:history} provides a historical
perspective on the \lrf problem; and finally,
Section~\ref{sec:lrf:conc} concludes.
Table~\ref{tbl:lrfs-summary} summarises the results that we present in
this Section.

\begin{table}[t]
  \begin{center}
  \begin{tabular}{|c|cc|}
    \hline
    Domain    & \lrf         & $\lrf_{\poly{S}_0}$  \\
    \hline
    $\reals$  & \ptime        & \pspaceh\\
    $\rats$   & \ptime        & \pspaceh\\
    $\ints$   & \conpc        & \ackh \\
    \hline
  \end{tabular}
  \end{center}
  \caption{Complexity of deciding existence of \lrfs (over $\reals$,
    $\rats$, and $\ints$) for \slc loops, \mlc loops, and \cfgs (with
    and without initial states). }
  \label{tbl:lrfs-summary}
\end{table}

\subsection{\lrfs Over the Rationals for \slc Loops}
\label{sec:slc:rat}

In what follows we assume a given \slc loop, specified by a transition
polyhedron $\transitions \subseteq \rats^{2n}$.
When variables range over the rationals, there is an algorithm to find
\lrfs which is \emph{complete} (always finds an \lrf if there is one)
and has \emph{polynomial time complexity}. This algorithm is based on
seeking inequalities of the form (\ref{eq:lrf:1},\ref{eq:lrf:2}) that
are entailed by the transition polyhedron $\transitions$, which can be
done using Farkas' Lemma.  Specifically, this approach involves
turning the conditions for an \lrf~(\ref{eq:lrf:1},\ref{eq:lrf:2}) into
a set of linear constraints where the variables are the coefficients
of $\rho$, and then solving these constraints using an \lp algorithm to
find values for the coefficients, if possible. Next we explain the
details of such an algorithm.

Let us write $\rho(\vec{x})$ as $\vect{\rfcoeff}\vec{x} + \rfcoeff_0$,
where $\vect{\rfcoeff}\in\rats^n$ is a row vector and
$\rfcoeff_0\in\rats$. Recall that the transition polyhedron can be
specified as $A''{\vec{x}}'' \le {\vec{c}}''$; then we have the
deduction problem (the entailed inequalities are rewritten to use
$\le$ instead of $\ge$):
\[
\begin{array}{cccr}
A''\vec{x}'' &\le& \vec{c} & \\
\cline{1-3}
 -\vect{\rfcoeff}\vec{x}-\vect{0}\vec{x}  & \le & \rfcoeff_0 & ~~\mbox{-- obtained from~\eqref{eq:lrf:1}}\\
 -\vect{\rfcoeff}\vec{x} + \vect{\rfcoeff}\vec{x}' & \le & -1 & ~~\mbox{-- obtained from~\eqref{eq:lrf:2}}\\
\end{array}
\]
Using Farkas' Lemma (see Section~\ref{sec:farkas}), synthesising the
two entailed inequalities can be done by solving the following \lp
problem, where $\vect{\mu}, \vect{\eta}$ are (row) vectors of
variables representing the Farkas' coefficients, and
$\vect{\rfcoeff}$ and $\rfcoeff_0$ are rational variables
representing the coefficients and constant of $\rho$:
\begin{eqnarray}
    \vect{\mu} A'' = (-\vect{\rfcoeff},\vect{0}),\  \vect{\mu} \vec{c} \le \rfcoeff_0,\  \vect{\mu}\ge 0 \label{eq:lrf-farkas:1} \\
    \vect{\eta} A'' = (-\vect{\rfcoeff},\, \vect{\rfcoeff}),\ \vect{\eta} \vec{c} \le -1,\  \vect{\eta}\ge 0 \label{eq:lrf-farkas:2} 
\end{eqnarray}
Any solution of~(\ref{eq:lrf-farkas:1},\ref{eq:lrf-farkas:2}) over the
reals (or rationals) defines a corresponding \lrf, and any \lrf yields
a corresponding solution
to~(\ref{eq:lrf-farkas:1},\ref{eq:lrf-farkas:2}).

\begin{example}
Consider the \slc loop:
\begin{equation}
\label{eq:ex:lrf:loop1}
\begin{array}{l}
\while\; ( x_1 \ge 0, x_2 \ge 1  ) \; \wdo \; x_1' \le x_1-x_2, x_2'\ge x_2
\end{array}
\end{equation}
and its corresponding matrix representations  $A''\vec{x}\le \vec{c}''$ where
\begin{align*}
  A''=
  \left(
  \begin{array}{rrrr}
    x_1 & x_2 & x_1' & x_2' \\
    \hline
     -1 &  0 &  0 &  0 \\
      0 & -1 &  0 &  0 \\
     -1 &  1 &  1 &  0 \\
      0 &  1 &  0 & -1 \\
  \end{array}
  \right)
  &
~
  &
    \vec{c}''=
  \left(
  \begin{array}{r}
       0 \\
      -1 \\
       0 \\
       0 \\
  \end{array}
  \right)
\end{align*}
Let $\rho(x_1,x_2)=\rfcoeff_1 x_1+\rfcoeff_2 x_2+\rfcoeff_0$ be an \lrf
template, \ie, $\rfcoeff_i$ are unknowns,
$\vect{\mu}=(\mu_0,\ldots,\mu_3)$ and
$\vect{\eta}=(\eta_0,\ldots,\eta_3)$. To synthesise an \lrf for
loop~\eqref{eq:ex:lrf:loop1}, we first
use~(\ref{eq:lrf-farkas:1},\ref{eq:lrf-farkas:2}) to generate the
constraint system
\begin{equation}
\label{eq:farkas:ex}
\begin{small} 
\begin{array}{l}
  -\mu_0-\mu_2 = -\rfcoeff_1,\, -\mu_1+\mu_2+\mu_3=-\rfcoeff_2,\; \mu_2 = 0, -\mu_3=0\\
  -\mu_1 \le \rfcoeff_0,\;
  \mu_0\ge0,\; \mu_1\ge0,\;\mu_2\ge0, \mu_3\ge 0 \\
  -\eta_0-\eta_2 = -\rfcoeff_1,\; -\eta_1+\eta_2+\eta_3=-\rfcoeff_2,\; \eta_2 = \rfcoeff_1, -\eta_3=\rfcoeff_2\\
  -\eta_1 \le -1,\;
  \eta_0\ge0,\; \eta_1\ge0,\;\eta_2\ge0, \eta_3\ge 0 
\end{array}
\end{small}
\end{equation}
The constraints in the first two lines come
from~\eqref{eq:lrf-farkas:1}, and the last two lines
from~\eqref{eq:lrf-farkas:2}. The following is a possible solution
for~\eqref{eq:farkas:ex}
\begin{equation}
\label{eq:farkas:ex:sol}
\begin{array}{l}
  \lambda_0\mapsto 0,\; \lambda_1\mapsto 1,\;  \lambda_2\mapsto 0,\; \\
  \mu_0\mapsto 1,\; \mu_1\mapsto0,\;  \mu_2\mapsto0,\;  \mu_3\mapsto0,\;  \\
  \eta_0\mapsto 0,\; \eta_1\mapsto1,\;  \eta_2\mapsto 1,\;  \eta_3\mapsto0\;  \\
\end{array}
\end{equation}
which means that $\rho(x_1,x_2)=x_1$ is an \lrf
for~\eqref{eq:ex:lrf:loop1}.
\end{example}

\citet{PodelskiR04a}
simplified~(\ref{eq:lrf-farkas:1},\ref{eq:lrf-farkas:2}) using the
fact that $A''=(A\ A')$ for some matrices $A, A'$ with $n$ columns
each, to the following equivalent one~(they eliminate
$\vect{\rfcoeff}$ and $\rfcoeff_0$ to reduce the number of variables
for efficiency):
\begin{equation}\label{eq:PR-lrf}
\begin{aligned}
\vect{\mu} A' &=  \vect{0}, \\
(\vect{\mu} - \vect{\eta}) A  &= \vect{0}, \\
\vect{\eta} (A + A') &= \vect{0},\\
\vect{\eta} \vec{c} &\le -1,\\
\vect{\mu},\vect{\eta} &\ge \vec{0} .\\
\end{aligned}
\end{equation}
Solving~\eqref{eq:PR-lrf} answers the existence question (\ie,
if~\eqref{eq:PR-lrf} has a solution then an \lrf exists) and
furthermore, the \lrf coefficients can be computed as
$\vect\rfcoeff = -\vect\mu A$ and $\rfcoeff_0$ can be any value
satisfying $\vect{\mu}\vec{c} \le \rfcoeff_0$ (in particular
$\rfcoeff_0 = \vect{\mu}\vec{c}$).

\begin{theorem}[\citealp{PodelskiR04a}]
\label{thm:pr04}
An \slc loop $\transitions$, specified by
$A''{\vec{x}}'' \le {\vec{c}}''$, has an \lrf if and only if the
linear program~\eqref{eq:PR-lrf} has a solution.
\end{theorem}

\begin{example}
\label{ex:lrf-are-not-enough}
Obviously, \lrfs are not sufficient for proving termination of (affine)
\slc loops. For example, the following loop~\citep{LeikeHeizmann15}
\[
\while\; (x\ge 1,\; y\ge 1,\; x\ge y) \; \wdo \; x'=2x,\; y'=3y
\]
does not admit an \lrf, while it can be shown terminating using the
techniques of Section~\ref{sec:dec:affine}. This loop, in particular,
does not even admit a more general form of (linear-based) ranking
functions, and some variations are discussed in
Section~\ref{sec:mlrfs}.
\end{example}

Let us now consider the case in which we seek an \lrf \wrt to a
polyhedral set of initial states $\poly{S}_0 \subseteq \rats^n$. We refer
to such \lrf as $\lrf_{\poly{S}_0}$.
As we have mentioned in Section~\ref{sec:programs}, it is enough to
consider the universal termination of
$\trres{\transitions}{\poly{S}_0}$ instead of termination of
$\transitions$ \wrt to $\poly{S}_0$.

\begin{example}
\label{ex:slc:init:0}
Consider the \slc loop
$\transitions = \{ x \ge 0, x' \le x - y, y' \ge y+1 \}$, and note
that $\max(0,x+1)$ is a ranking function, according to
Definition~\ref{def:rf} when restricting the initial states to
$\poly{S}_0=\{ y = 1 \}$.
However, $\transitions$ does not have an \lrf according to
Definition~\ref{def:lrf}, unless we apply it to
$\trres{\transitions}{\poly{S}_0} = \{\cbox{y\ge 1}, x \ge 0, x' \le
x' - y, y \ge y+1 \}$ instead of $\transitions$, which then admits
$\rho(x,y)=x$ as an \lrf.
\end{example}

This example suggests the following approach for seeking
\lrfs for loops with initial states:
\begin{inparaenum}[\upshape(1\upshape)]
\item compute the set of reachable states
  $\reach{\transitions}{\poly{S}_0}$ and use it to compute
  $\trres{\transitions}{\poly{S}_0}$; and
\item seek an \lrf for $\trres{\transitions}{\poly{S}_0}$.
\end{inparaenum}
However, there is a problem with this approach: we do not know, in
general, how to compute (or even express) the set of reachable states, and it is
certainly not guaranteed to be polyhedral.
To address this in practice, we over-approximate
$\reach{\transitions}{S_0}$ using a polyhedral invariant
$\pinv(\vec{x})$ (called a \emph{supporting invariant}) and then
analyse the transition relation
$\transitions'=\transitions(\vec{x},\vec{x}') \land \pinv(\vec{x})$.
This sacrifices completeness because $\transitions'$ is an
over-approximation of $\trres{\transitions}{\poly{S}_0}$.

Polyhedral invariants (more precisely, inductive polyhedral
invariants) can
be inferred either beforehand using dedicated tools~\citep{CH78}, or by
using a \emph{template-based}
approach~\citep{ColonSS03,BMS05a,LarrazORR13} to synthesise an \lrf and
a supporting polyhedral invariant simultaneously. This has the
advantage that the search for an invariant is ``automatically'' guided
by the requirements of the \lrf. Let us briefly explain this approach.

A template invariant $\pinv(\vec{x})$ is a conjunction of linear
inequalities over variables $\vec{x}$ where the coefficients are
unknowns, \eg,
$\pinv(x,y)= \{a_1x+a_2y \le a_0\}$
where $a_i$ represent the unknown coefficients.
Our interest is to seek a linear function
$\rho(x,y) = \vect{\rfcoeff}\vec{x} + \rfcoeff_0$ and values for
$a_i$, such that $\pinv(\vec{x})$ is an invariant for $\transitions$
\wrt the initials states $\poly{S}_0$ and $\rho$ is an \lrf for
$\transitions(\vec{x},\vec{x}')\wedge\pinv(\vec{x})$ which can be
stated as follows:
\begin{align}
 \poly{S}_0(\vec{x}) \implies &\; \pinv(\vec{x})  \,, \label{eq:lrfinv:1}\\
 \transitions(\vec{x},\vec{x}')\wedge\pinv(\vec{x}) \implies &\; \pinv(\vec{x}')  \,, \label{eq:lrfinv:2}\\
 \transitions(\vec{x},\vec{x}')\wedge\pinv(\vec{x}) \implies &\; \rho(\vec{x})  \ge 0  \,, \label{eq:lrfinv:3}\\
 \transitions(\vec{x},\vec{x}')\wedge\pinv(\vec{x}) \implies &\; \diff{\rho}(\vec{x}'')  \ge 1 \,. \label{eq:lrfinv:4} 
\end{align}
The first two formulas ensure that $\pinv(\vec{x})$ is an inductive
invariant for $\transitions$, while the remaining formulas ensure that
$\rho$ is an \lrf for
$\transitions(\vec{x},\vec{x}')\wedge\pinv(\vec{x})$, and therefore
an $\lrf_{\poly{S}_0}$ for $\transitions$.
This entire problem can be solved using Farkas' Lemma, which
transforms it into solving a corresponding system of constraints in
which, among others, $a_i$ and $\rfcoeff_i$ are variables.
However, since the template $\pinv(\vec{x})$ appears on the left-hand
side of the implications, the resulting constraints are non-linear,
and thus solving them is decidable \emph{over the reals} (decidability
is unknown for the rationals) but not guaranteed to be
polynomial-time~(it might be exponential, since the corresponding
decision problem is \pspace~\citep{Canny88}).
Note that such an algorithm is complete for a slightly different
problem in which we seek an invariant of a particular form: Is there a
polyhedral invariant $\pinv(\vec{x})$ for $\transitions$ and
$\poly{S}_0$, \emph{matching a given template}, such that the rational
loop $\transitions(\vec{x},\vec{x}')\wedge\pinv(\vec{x})$ has an \lrf?

\begin{example}
\label{ex:lrf:temp}
Let us apply the template based approach to the \slc loop
$\transitions = \{ x \ge 0, x' \le x - y, y' \ge y+1 \}$ and initial
condition $\poly{S}_0=\{y=1\}$ of Example~\ref{ex:slc:init:0}, and
a template invariant $\pinv(x,y)=\{a_1 x+a_2 y \le a_0\}$. We first
note that:
  \begin{align*}
    \poly{S_0}(x,y)\equiv
    &
  \left(
  \begin{array}{@{}rr@{}}
      0 &  -1 \\
      0 &   1 \\
  \end{array}
    \right)    
  \left(
  \begin{array}{@{}l@{}}
      x  \\
      y  \\
  \end{array}
    \right)
    \le
  \left(
  \begin{array}{@{}r@{}}
      -1  \\
       1  \\
  \end{array}
      \right)    \\
  \transitions(x,y,x',y') \wedge \pinv(x,y)\equiv &
  \left(
  \begin{array}{@{}rrrr@{}}
     -1 &  0 &  0 &  0 \\
     -1 &  1 &  1 &  0 \\
      0 &  1 &  0 & -1 \\
      a_1 & a_2& 0 & 0 \\
  \end{array}
  \right)  
  \left(
  \begin{array}{@{}l@{}}
      x  \\
      y  \\
      x' \\
      y' \\
  \end{array}
    \right)
    \le
  \left(
  \begin{array}{@{}r@{}}
       0 \\
       0 \\
      -1 \\
       a_0 
  \end{array}
  \right)
\end{align*}
Let $\rho(x,y)=\rfcoeff_1 x+\rfcoeff_2 y+\rfcoeff_0$ be an \lrf
template, \ie, $\rfcoeff_i$ are unknowns.
To synthesise an \lrf and an invariant simultaneously, we
translate~\eqref{eq:lrfinv:1}-\eqref{eq:lrfinv:4} into a set of
existential constraints using Farkas' lemma which results in
($\vect{\mu}, \vect{\eta}, \vect{\xi}, \vect{\alpha}$ are the Farkas'
coefficients):
\begin{equation*}
\begin{small} 
  \begin{array}{|r|l|}
  \hline
  \eqref{eq:lrfinv:1}&                  
  0=a_1,\; -\mu_0+\mu_1=a_2,\; -\mu_0+\mu_1 \le a_0,\; \mu_0\ge 0,\; \mu_1\ge0\\
  \hline
  \eqref{eq:lrfinv:2}&                  
  -\eta_0-\eta_1+\cbox{\eta_3 a_1} = 0,\;
  \eta_1+\eta_2+\cbox{\eta_3 a_2} = 0,\;
  \eta_1 = a_1,\;\\
  &-\eta_2 = a_2,\;
  -\eta_2 + \cbox{\eta_3 a_0} \le a_0,\;
  \eta_0\ge0,\; \eta_1\ge 0,\; \eta_2\ge 0,\; \eta_3\ge 0\\
  \hline
  \eqref{eq:lrfinv:3}&                  
  -\xi_0-\xi_1+\cbox{\xi_3 a_1} = -\rfcoeff_1,\,
  \xi_1+\xi_2+\cbox{\xi_3 a_2} = -\rfcoeff_2,\,
  \xi_1 = 0,\, \\
  &-\xi_2 = 0,\, 
  -\xi_2 + \cbox{\xi_3 a_0} \le -\rfcoeff_0,\,
  \xi_0\ge0,\, \xi_1\ge0,\, \xi_2\ge 0,\, \xi_3\ge 0\\
  \hline
  \eqref{eq:lrfinv:4}&                  
  -\alpha_0-\alpha_1+\cbox{\alpha_3 a_1} = -\rfcoeff_1,\,
  \alpha_1+\alpha_2+\cbox{\alpha_3 a_2} = -\rfcoeff_2,\,
  \alpha_1 = \rfcoeff_1,\,\\
  &-\alpha_2 = \rfcoeff_2,\,
  -\alpha_2 + \cbox{\alpha_3 a_0} \le -1,\,
  \alpha_0\ge0,\,  \alpha_1\ge0,\,  \alpha_2\ge0,\,  \alpha_3\ge0\\
   \hline
\end{array}
\end{small}
\end{equation*}
Note that they include nonlinear terms. 
Solving these constraints we find the following possible solution:
\begin{equation*}
\label{eq:farkas:ex:lrfinv:sol}
\begin{array}{l}
  \rfcoeff_0\mapsto 0,\; \rfcoeff_1\mapsto 1,\; \rfcoeff_2\mapsto 0,\; \\
  a_0\mapsto -1,\; a_1\mapsto0,\;  a_2\mapsto -1,\; \\
  \mu_0 \mapsto 1,\; \mu_1 \mapsto 0, \; \\
  \eta_0 \mapsto 0,\; \eta_1 \mapsto 0, \;  \eta_2\mapsto 1, \;  \eta_3 \mapsto 1, \; \\
  \xi_0 \mapsto 1,\; \xi_1 \mapsto 0, \;  \xi_2\mapsto 0, \;  \xi_3 \mapsto 0, \; \\
  \alpha_0 \mapsto 0,\; \alpha_1 \mapsto 1, \;  \alpha_2\mapsto 0, \;  \alpha_3 \mapsto 1. \; \\
\end{array}
\end{equation*}
Thus, $\rho(x,y)=x$ is an \lrf and $y \ge 1$ is a supporting
invariant.
\end{example}

\begin{problem}
\label{lrf:slc:init:decprob}
Is it decidable whether a given rational \slc loop $\transitions$ has
an \lrf \wrt to a polyhedral set of initial states $\poly{S}_0$ and, if
yes, what is the complexity of this problem?
\end{problem}

\citet{Ben-Amram14} provides a lower bound on the hardness of this
problem.

\begin{theorem}
Deciding if a given rational \slc $\transitions$ has an \lrf \wrt a
polyhedral set of initial states $\poly{S}_0$ is \pspaceh (even if
we know that the loop is terminating). 
\end{theorem}

\begin{problem}
\label{lrf:slc:init:invprob}
Are polyhedral invariants sufficient for deciding if an \lrf exists
for a given \slc loop $\transitions$ \wrt a polyhedral set of initial
states $\poly{S}_0$? That is, does $\transitions_{\poly{S}_0}$ have an
\lrf if and only if there exists a polyhedral invariant
$\pinv(\vec{x})$ such that
$\transitions(\vec{x},\vec{x}') \wedge \pinv(\vec{x})$ has an \lrf?
If the answer is no, a different question arises: Is it decidable
whether a polyhedral supporting invariant $\pinv(\vec{x})$ exists such
that $\transitions(\vec{x},\vec{x}') \wedge \pinv(\vec{x})$ has an
\lrf?
\end{problem}

\subsection{\lrfs Over the Integers for \slc Loops}
\label{sec:slc:int}

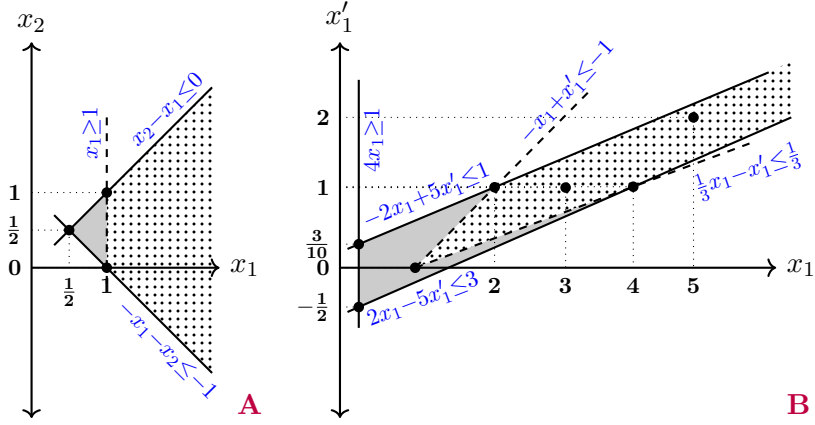
\begin{figure}

\begin{center}
\begin{tikzpicture}

\begin{scope}[shift={(0,0)}]
  \coordinate (a_1) at (0.3,2.3);
  \coordinate (a_2) at (2.4,4.4);
  \coordinate (b_1) at (0.3,2.7);
  \coordinate (b_2) at (2.4,0.6);
  \coordinate (c_1) at (1,4);
  \coordinate (c_2) at (1,1.75);

  \coordinate (c) at (intersection of a_1--a_2 and b_1--b_2);
  \coordinate (d) at (intersection of c_1--c_2 and a_1--a_2);
  \coordinate (e) at (intersection of c_1--c_2 and b_1--b_2);

  \fill[pattern=dots] (e) -- (b_2) -- (2.4,4.4) -- (a_2) -- (d) -- cycle;
  \fill[fill=black!20] (c) -- (d) -- (e) -- cycle;
  \draw [->,thick] (0,2) node  {} -- (2.5,2) node (xaxis) [right] {$x_1$};
  \draw [<->,thick] (0,0) node  {} -- (0,5) node (yaxis) [above] {$x_2$};

  \draw [thick]  (a_1) -- (a_2) {};
  \node[rotate=45] () at (1.8,4.1) {\scalebox{0.8}{\textcolor{blue}{$x_2{-}x_1{\le}0$}}}; 

  \draw [thick]  (b_1) -- (b_2) {};
  \node[rotate=-45] () at (1.8,0.9) {\scalebox{0.8}{\textcolor{blue}{${-}x_1{-}x_2{\le}{-}1$}}};

  \draw [dashed,thick]  (c_1) -- (c_2) {};
  \node[rotate=90] () at (0.8,3.8) {\scalebox{0.8}{\textcolor{blue}{$x_1{\ge}1$}}};
 
  \draw[dotted] (yaxis |- c) node[left] {\scalebox{0.8}{$\mathbf{\frac{1}{2}}$}}
        -| (xaxis -| c) node[below] {\scalebox{0.8}{$\mathbf{\frac{1}{2}}$}};
  \draw[dotted] (yaxis |- d) node[left] {\scalebox{0.8}{$\mathbf{1}$}} -- (d);
  \draw[dotted] (yaxis |- e) node[left] {\scalebox{0.8}{$\mathbf{0}$}}
        -| (xaxis -| e) node[below] {\scalebox{0.8}{$\mathbf{1}$}};

  \fill[] (c) circle (2pt);
  \fill[] (d) circle (2pt);
  \fill[] (e) circle (2pt);

  \node[] at (2.9,0.18) {\textcolor{purple}{\textbf{A}}};

\end{scope}

\begin{scope}[shift={(4,0)}]
  \coordinate (a_1) at (0.35,4.5);
  \coordinate (a_2) at (0.35,1.2);
  \coordinate (b_1) at (0.2,1.4111111);
  \coordinate (b_2) at (6.1,4);
  \coordinate (c_1) at (0.2,2.25);
  \coordinate (c_2) at (5.8,4.6);
  \coordinate (d_1) at (1.1,2);
  \coordinate (d_2) at (3.4,4.3333333);
  \coordinate (e_1) at (1.1,2);
  \coordinate (e_2) at (5.544444,3.653);
  \coordinate (a) at (intersection of a_1--a_2 and b_1--b_2);
  \coordinate (b) at (intersection of a_1--a_2 and c_1--c_2);
  \coordinate (c) at (intersection of d_1--d_2 and e_1--e_2);
  \coordinate (d) at (intersection of c_1--c_2 and d_1--d_2);
  \coordinate (e) at (intersection of b_1--b_2 and e_1--e_2);
  \coordinate (f) at (3.1,3.06666);
  \coordinate (g) at (4.8,4);

  \fill[pattern=dots] (c) -- (e) -- (b_2) -- (6.1,4.75) -- (c_2) -- (d) -- cycle;
  \fill[fill=black!20] (a) -- (e) -- (c) -- (d) -- (b) -- cycle;

  \draw [->,thick] (0.1,2) node  {} -- (5.9,2) node (xaxis) [right] {$x_1$};
  \draw [<->,thick] (0.1,0) node  {} -- (0.1,5) node (yaxis) [above] {$x_1'$};

  \draw [thick] (a_1) -- (a_2) {};
  \node[rotate=90] () at (0.55,3.7) {\scalebox{0.8}{\textcolor{blue}{$4x_1{\ge}1$}}};

  \draw [thick] (b_1) -- (b_2) {};
  \node[rotate=21,anchor=west] () at (0.35,1.25) {\scalebox{0.8}{\textcolor{blue}{$2x_1{-}5x_1'{\le}3$}}};

  \draw [thick] (c_1) -- (c_2) {};
  \node[rotate=23,anchor=west] () at (0.3,2.55) {\scalebox{0.8}{\textcolor{blue}{${-}2x_1{+}5x_1'{\le}1$}}};

  \draw [dashed,thick] (d_1) -- (d_2) {};
  \node[rotate=48,anchor=west] () at (2.45,3.6) {\scalebox{0.8}{\textcolor{blue}{${-}x_1{+}x_1'{\le}{-}1$}}};

  \draw [dashed,thick] (e_1) -- (e_2) {};
  \node[rotate=19,anchor=west] () at (4.7,3.0) {\scalebox{0.8}{\textcolor{blue}{$\frac{1}{3}x_1{-}x_1'{\le}\frac{1}{3}$}}};

  \draw[dotted] (yaxis |- c) node[left] {\scalebox{0.8}{$\mathbf{0}$}}
        -| (xaxis -| c) node[below] {};
  \draw[dotted] (yaxis |- d) node[left] {\scalebox{0.8}{$\mathbf{1}$}}
        -| (xaxis -| d) node[below] {\scalebox{0.8}{$\mathbf{2}$}};
  \draw[dotted] (yaxis |- f) node[left] {}
        -| (xaxis -| f) node[below] {\scalebox{0.8}{$\mathbf{3}$}};
  \draw[dotted] (yaxis |- e) node[left] {}
        -| (xaxis -| e) node[below] {\scalebox{0.8}{$\mathbf{4}$}};
  \draw[dotted] (yaxis |- g) node[left] {\scalebox{0.8}{$\mathbf{2}$}}
        -| (xaxis -| g) node[below] {\scalebox{0.8}{$\mathbf{5}$}};

  \draw[dotted] (yaxis |- a) node[left] {\scalebox{0.8}{$\mathbf{-\frac{1}{2}}$}} -- (a);
  \draw[dotted] (yaxis |- b) node[left] {\scalebox{0.8}{$\mathbf{\frac{3}{10}}$}} -- (b);

  \fill[] (a) circle (2pt);
  \fill[] (b) circle (2pt);
  \fill[] (c) circle (2pt);
  \fill[] (d) circle (2pt);
  \fill[] (e) circle (2pt);
  \fill[] (f) circle (2pt);
  \fill[] (g) circle (2pt);

  \node[] at (6.2,0.185) {\textcolor{purple}{\textbf{B}}};

 \end{scope}
\end{tikzpicture}
\end{center}
\caption{The polyhedra associated with two of our examples, projected
  to two dimensions: \textbf{(A)} corresponds to
  Loop~\eqref{eq:bg:loop1} on Page~\pageref{eq:bg:loop1}; \textbf{(B)}
  corresponds to Loop~\eqref{eq:bg:loop2} on
  Page~\pageref{eq:bg:loop2}.
  Dashed lines are added when computing the integer hull; dotted areas
  represent the integer hull; Gray areas are rational points
  eliminated when computing the integer hull (Figure from \citep{BG14}).}
\label{fig:intpoly}
\end{figure}

When variables range over integers, the \slc loop can still be
understood in terms of the transition polyhedron $\transitions \subseteq \rats^{2n}$, but
this time we are interested not in all the rational points in this
polyhedron but just in its integer points, \ie, in the set of
transitions $\intpoly{\transitions}$. This means that for $\rho$ to be
an \lrf we require~(\ref{eq:lrf:1},\ref{eq:lrf:2}) to hold only for
$\vec{x}''\in \intpoly{\transitions}$.

\begin{example}
Consider the following loop:
\begin{equation}
\label{eq:bg:loop1}
\begin{array}{l}
\while\; ( x_2{-}x_1 \le 0, x_1{+}x_2 \ge 1 ) \; \wdo \; x_2' = x_2{-}2x_1{+}1, x_1'=x_1
\end{array}
\end{equation}
When considered as an integer loop, it has the \lrf
$\rho(x_1,x_2) = x_1+x_2$. On the contrary, over rationals the loop
does not always terminate --- consider its computation from
$(\frac{1}{2},\frac{1}{2})$.
\end{example}

In the above example, the restriction to integers excludes the
non-terminating state $(\frac{1}{2},\frac{1}{2})$. So a natural step
towards analysing a loop over the integers is to reduce the polyhedron
to its \emph{integer hull}, since it eliminates all points that are
not convex combinations of points from $\intpoly{\transitions}$.
Indeed, the integer hull of Loop \eqref{eq:bg:loop1} is the following
loop, which adds the constraints $x_1 \ge 1$ to the guard (see
Figure~\ref{fig:intpoly}(a))
\begin{equation}
\label{eq:bg:loop1int}
\begin{array}{l}
\while\; (x_2-x_1 \le 0, x_1+x_2 \ge 1, \cbox{x_1\ge 1}) \; \wdo \; \\
\hspace*{4cm}x_2' = x_2-2x_1+1, x_1'=x_1
\end{array}
\end{equation}
and this loop has the \lrf mentioned above, since
$(\frac{1}{2},\frac{1}{2})$ is excluded by the guard. Similarly, Loop
\eqref{eq:bg:loop2} does not terminate over the rationals, \eg, for
initial point $(\frac{1}{4},1$), but terminates, and has an \lrf, over
integers~(see Figure~\ref{fig:intpoly}(b)).

Synthesising \lrfs over the integers, can be also reduced to seeking
implied inequalities of the form~(\ref{eq:lrf:1},\ref{eq:lrf:2}), but
using $\intpoly{\transitions}$ instead of $\transitions$. This can be
also be done using Farkas' lemma and $\inthull{\transitions}$, because
an inequality is entailed by $\intpoly{\transitions}$ if and only if
it is entailed by $\inthull{\transitions}$. This was observed
independently by several researchers~\citep{Feautrier92.1,CookKRW13,
  BG14}.

\begin{theorem}
\label{thm:inthull-lrf}
An integer \slc loop $\intpoly{\transitions}$ has an \lrf if and only
if its integer hull $\inthull{\transitions}$ has an \lrf (as a rational
loop).
\end{theorem}

This gives us a complete algorithm to solve the \lrf problem for
integer \slc loops: compute the integer hull of $\transitions$ and use
a polynomial-time \lrf algorithm.  The complexity of computing integer
hulls is, in general, exponential. \citet{BG14} list a number of
special cases which can be solved in polynomial time, since the
integer hull can be computed in polynomial time for these cases, but
also prove that in general, the \lrf problem over integers is
\conpc.

The exponential complexity of computing the integer hull, in the
general case, gives the correct intuition why the problem is hard.
For \emph{inclusion} in \conp, \citet{BG14} show that
$\intpoly{\transitions}$ does \emph{not} have an \lrf if and only if
there are finite sets
$X\neq\emptyset \subseteq \intpoly{\transitions}$ and
$Y\subseteq \intpoly{\ccone(\transitions)}$, of polynomial size, such
that the loop
$\convhull\{X\}+\cone\{Y\} \subseteq \inthull{\transitions}$ does not
have an \lrf, and that this last check can be done in polynomial time.

Let us now consider the case in which the initial states are
restricted to a polyhedral set $\poly{S}_0 \subset \rats^n$, and
recall that our interest is in the integer states
$\intpoly{\poly{S}_0}$.
The algorithmic aspects of this case are similar to the one of the
rational case (but using $\inthull{\transitions}$ instead of
$\transitions$), \ie, either we infer a supporting invariant
beforehand and add it to the transition polyhedron, or we use the
template approach to synthesise a supporting invariant and an \lrf
simultaneously.
However, there is one important difference regarding the problem of
inferring a supporting invariant (that matches a template) and an \lrf
at the same time: In the rational case the algorithm is complete, but
this does not hold for the integer case since
$\inthull{\transitions}(\vec{x},\vec{x}')\wedge\pinv(\vec{x})$ is not
necessarily an integer polyhedron, and we cannot compute its integer
hull because $\pinv(\vec{x})$ includes template parameters.

Problems~\ref{lrf:slc:init:decprob} and~\ref{lrf:slc:init:invprob} are
also still open for the integer case.  \citet{Ben-Amram14} provided
lower bounds on the hardness for related problems.

\begin{theorem}
Deciding whether a given integer \slc loop $\transitions$ has an \lrf \wrt a
polyhedral set of initial states $\poly{S}_0$ is \ackh\footnote{%
This follows from a reduction from \citet{Ben-Amram14} along with recent results on the hardness
of reachability in Vector Addition Systems~\citep{CzOr21,Leroux21}.%
}.
\end{theorem}

\begin{theorem}
Deciding whether a given integer \slc loop $\transitions$ has a polyhedral
inductive invariant $\pinv(\vec{x})$ \wrt a polyhedral set of initial
states $\poly{S}_0$ (not necessarily matching a template) such that
$\transitions(\vec{x},\vec{x}') \wedge \pinv(\vec{x})$ has an \lrf over
the integers is \pspaceh.
\end{theorem}

\subsection{\lrfs for \mlc Loops}

An \lrf for an \mlc loop $\transitions_1,\ldots,\transitions_k$, is a
function $\rho$ which is an \lrf for all its transitions
$T=\transitions_1\cup\cdots\cup\transitions_k$, that is all the
paths. The following complexity results follow quite easily.

\subsubsection{Polynomial-time Synthesis for Rational Loops}

We create for each path $\transitions_i$ a constraint system as
in~(\ref{eq:lrf-farkas:1},\ref{eq:lrf-farkas:2}), where each system uses
different $\vect{\mu}$ and $\vect{\eta}$, say $\vect{\mu}_i$ and
$\vect{\eta}_i$, but the same $(\vect{\rfcoeff},\rfcoeff_0)$. This
results in a bigger, still polynomial-sized \lp problem, and its
solutions define \lrfs that hold for all paths.
We can also do the same using~\eqref{eq:PR-lrf} instead of
~(\ref{eq:lrf-farkas:1},\ref{eq:lrf-farkas:2}), but in this case we
have to add constraints requiring the \lrf coefficients arising
from each of these sub-problems to coincide, namely
$\vect\rfcoeff = -\vect{\mu}_i A$ and
$\rfcoeff_0 \ge \vect{\eta}_i\vec{c}$ for each $\transitions_i$.

\begin{example}
\label{ex:mlc:1}
Consider the \mlc loop of Example~\ref{ex:mlc:0}, and note that $x_1$
is an \lrf for $\transitions_1$ and $x_2$ is an \lrf for
$\transitions_2$. However, the \mlc loop defined by both paths does
not have an \lrf.
Modifying the paths to
\begin{align*}
  \transitions_1&=\{ x_1 \geq 0, \cbox{x_2 \geq 0}, x_1' = x_1-1, \cbox{x_2'=x_2}\} \\
  \transitions_2&= \{\cbox{x_1 \geq 0}, x_2 \geq 0, x_1' \leq x_1, x_2'=x_2-1\}
\end{align*}
the loop has an \lrf $\rho(x_1,x_2)=x_1+x_2$.
\end{example}

\subsubsection{\lrfs Over the Integers for \mlc Loops}

For integer loops we get a complete algorithm by first computing the
integer hulls of all paths, namely
$\inthull{(\transitions_1)},\ldots,\inthull{(\transitions_k)}$, and
then applying the algorithm of the rational case. The completeness of
this method follows from the same considerations as the ones of \slc
loops. \citet{BG14} show that deciding if a given integer \mlc loop
has an \lrf is \conpc. The hardness is clear since it is already hard
for \slc loops. Inclusion in \conp is shown by generalising the
witnesses of the \slc case to cover all paths.

\begin{example}
\label{ex:mlc:1:int}
Let use consider an \mlc
$\transitions_1,\transitions_2,\transitions_3$, where the first two
paths are those of Example~\ref{ex:mlc:1}, and the last is that of the \slc
loop~\eqref{eq:bg:loop1}.
This loop does not have an \lrf over the rationals since
$\transitions_3$ does not have an \lrf over the rationals (when
considered separately as an \slc loop), however, over the integers it
has the \lrf $\rho(x_1,x_2)=x_1+x_2$.
To synthesise this \lrf we have to compute the integer hull of all
paths first (note that $\transitions_1$ and $\transitions_2$ are
already integral, and $\inthull{(\transitions_3)}$ is 
Loop~\eqref{eq:bg:loop1int}).
\end{example}

\subsubsection{\lrfs for \mlc Loops with Polyhedral Set of Initial States}

The same consideration for the case of \slc loop applies to \mlc loops as
well, both for the rational and the integer case. In particular we can
use the template based approach which in this case
requires~\eqref{eq:lrfinv:2}-\eqref{eq:lrfinv:4} for all paths.
As for the complexity of related problems (\eg,
problems~\ref{lrf:slc:init:decprob} and~\ref{lrf:slc:init:invprob}),
nothing is known for the \mlc case.

\subsection{\lrfs for \cfgs}
\label{sec:lrf:cfg}

In this section we discuss how the algorithmic and complexity aspects
of the \lrf problem extend to the case of \cfgs. In what follows, we
assume a given \cfg $P=(V,\numdom,L,\ell_{0},E)$ where $\numdom$ is
$\rats$ or $\ints$~(recall that the case of $\reals$ is the same as
that of $\rats$).
We first consider the case where the execution can start at any
location, and then restrict ourselves to locations $\ell_0$.

To generalise Definition~\ref{def:lrf} of an \lrf to \cfgs, all we
need is to require~(\ref{eq:lrf:1},\ref{eq:lrf:2}) to hold for any
$\vec{x}'' =(\vec x, \vec x')\in \transitions_{\ell,\ell'} \in E$, \ie, for all
transitions on all edges.
In such case, the \lrf $\rho$ guarantees universal termination,
meaning that an execution can start from any location, not just
$\ell_0$, and with any values $\vec{x}\in\numdom$ for the program
variables.
With this adjustment, all complexity and algorithmic aspects, of the
\lrf problem, previously discussed for universal termination of \mlc
loops also apply to \cfgs, both for rational and integer
variables.

However, due to their complex structure, \cfgs are unlikely to admit an
\lrf of this form.
For instance, a \cfg might include several (simple) loops, each
potentially having a distinct \lrf, and even if they shared the same
\lrf, the edges connecting these loops are not likely to satisfy
Condition~\eqref{eq:lrf:2}.
Moreover, a loop might be represented by several edges in the \cfg
where only in one of them the loop counter decreases, while in the
rest it stays the same (\ie, it is impossible to have a single
function that decreases on all these edges).

It is therefore desirable to use a more general definition, where  we allow each node to use a different function
$\rho_\ell$, and change~(\ref{eq:lrf:1},\ref{eq:lrf:2}) to require that
each $\vec{x}'' \in \transitions_{\ell,\ell'} \in E$ satisfy:
\begin{align}
\rho_\ell(\vec{x}) &\ge 0 \\
\rho_\ell(\vec{x})-\rho_{\ell'}(\vec{x}') &\ge 1\,.
\end{align}
Now an \lrf is a collection of linear functions, where each node is
assigned one. The algorithmic and complexity aspects of synthesising
such an \lrf are the same as in the case of \lrf for \mlc loops.

\begin{figure}[t]

\begin{center}
\begin{tikzpicture}[>=latex,line join=bevel,]

\begin{scope}[shift={(2.7,4)}]
  \node [ inner sep = 0pt, align=left, font=\ttfamily] at (0,0) {
  \begin{minipage}{5.25cm}
    \begin{lstlisting}[escapechar=\#]
assert(x>=0);
int y = 1;    
while(x >= 0) {
  if (nondet()) {
    y=2*y;
    if (nondet()) break;
  } else y++;
  x--;
}
x = y;
while (y>=0) {
  y--;
  x = 3*x;
}
\end{lstlisting}
\end{minipage}
};
\end{scope}

\begin{scope}[shift={(7,0.5)}]
  \node (l4) at (3,0) [loc] {$\ell_{4}$};
  \node (l3) at (1,0) [loc] {$\ell_{3}$};    
  \node (l2) at (2,0) [loc] {$\ell_{2}$};
  \node (l5) at (1,1) [loc] {$\ell_{5}$};
  \node (l1) at (2,1) [loc] {$\ell_{1}$};
  \node (l0) at (3,1) [inloc] {$\ell_0$};
  \node (l6) at (0,1) [loc] {$\ell_{6}$};
  \node (l7) at (0,0) [loc] {$\ell_7$};
  \draw [tre] (l0) to[] node[tr,above] {\scalebox{0.8}{$\transitions_0$}} (l1);
  \draw [tre] (l1) to[] node[tr,left] {\scalebox{0.8}{$\transitions_1$}} (l2);
  \draw [tre] (l1) to[] node[tr,above] {\scalebox{0.8}{$\transitions_2$}} (l5);
  \draw [tre] (l2) to[] node[tr,above] {\scalebox{0.8}{$\transitions_3$}} (l3);
  \draw [tre] (l2) to[] node[tr,above] {\scalebox{0.8}{$\transitions_4$}} (l4);
  \draw [tre] (l4) to[] node[tr,right] {\scalebox{0.8}{$\transitions_6$}} (l1);
  \draw [tre,bend right=25] (l3) to[] node[tr,below] {\scalebox{0.8}{$\transitions_5$}} (l4);
  \draw [tre] (l3) to[] node[tr,left] {\scalebox{0.8}{$\transitions_7$}} (l5);
  \draw [tre] (l5) to[] node[tr,above] {\scalebox{0.8}{$\transitions_8$}} (l6);
  \draw [tre] (l6) to[] node[tr,left] {\scalebox{0.8}{$\transitions_{10}$}} (l7);
  \draw [tre] (l6) to[out=160,in=200,loop] node[tr,left] {\scalebox{0.8}{$\transitions_9$}}  (l6);
\end{scope}

\begin{scope}[shift={(7.75,5)}]
  \node [inner sep = 0pt, align=left, font=\ttfamily\footnotesize] at (0.5,0) {
    \begin{minipage}{5cm}
      \begin{center}
\(
  \begin{array}{|@{}r@{\hskip 2pt}l@{}|}
    \hline
\transitions_0{:}&  \{x\ge 0,x'=x,y'=1\} \\
\transitions_1{:}&  \{x \ge 0,x'=x,y=y\} \\
\transitions_2{:}&  \{x \le -1, x'=x,y'=y\} \\
\transitions_3{:}&  \{x'=x,y'=2y\} \\
\transitions_4{:}&  \{x'=x,y'=y+1\} \\
\transitions_5{:}&  \{x'=x,y=y\} \\
\transitions_6{:}&  \{x'=x-1,y=y\} \\
\transitions_7{:}&  \{x'=x,y=y\} \\
\transitions_8{:}&  \{x'=y,y'=y\} \\
\transitions_9{:}&  \{y \ge 0,y'=y-1,x'=3x\} \\
\transitions_{10}{:}&  \{y \le -1, x'=x,y=y\} \\
    \hline
  \end{array}
\)

\(
\begin{array}{|@{}l@{}|}
\hline
  \poly{S}_0 = \pinv_{\ell_0} =\{\} \\
  \pinv_{\ell_1}= \pinv_{\ell_5}= \{ x \ge -1, y \ge 1\} \\
  \pinv_{\ell_2}= \pinv_{\ell_3}= \pinv_{\ell_4}= \{x\ge0, y\ge 1 \} \\
  \pinv_{\ell_6}= \pinv_{\ell_7}= \{x \ge 1, y \ge -1\} \\
\hline
\end{array}
\)
\end{center}
\end{minipage}
};
\end{scope}
\end{tikzpicture}
\end{center}

\caption{A program (taken from~\citep{ADFG:2010}), its corresponding \cfg, and invariants when starting at location $l_0$).}
\label{fig:cfg:lrf}

\end{figure}
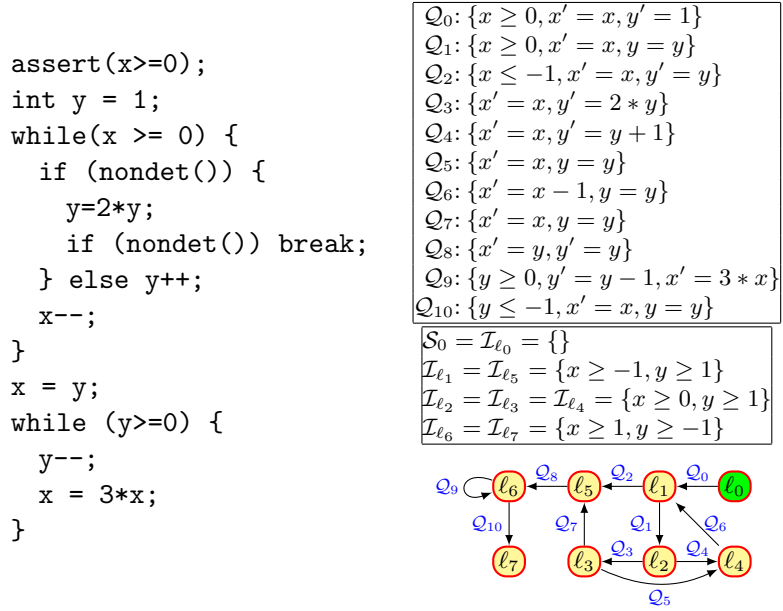

\begin{example}
\label{ex:lrf:cfg:1}
Consider the \cfg in Figure~\ref{fig:cfg:lrf}, and assume that
invariants have been added to the corresponding transitions (this is
what we usually do when starting from $\ell_0$, but we apply it here
to keep the example simple and meaningful). Let us also ignore the
second loop for now (and thus nodes $\ell_6$ and $\ell_7$); we will
consider it later.
If we seek an \lrf that assigns the same function $\rho$ to all nodes,
we will not find one, because in many transitions we have $x'=x$.
Instead, we look for an \lrf that assigns a (possibly) different
function $\rho_\ell$ to each node, and we find the following:
\[
  \begin{array}{ll}
    \rho_{\ell_0}(x,y) =  3x+5\\
    \rho_{\ell_1}(x,y) =  3x+4\\
  \end{array}
  ~
  \begin{array}{ll}
    \rho_{\ell_2}(x,y) =  3x+3\\
    \rho_{\ell_3}(x,y) =  3x+2\\
  \end{array}
  ~
  \begin{array}{ll}
    \rho_{\ell_4}(x,y) =  3x+2\\
    \rho_{\ell_5}(x,y) =  3x+1\\
  \end{array}
\]
These functions are only different in the constant, which means that
we could use templates for the different $\rho_\ell$ that are
different only in the constants. This would be more efficient in
practice since the corresponding \lp problems will have fewer
variables.
Note that from this \lrf (\ie, the collection of all $\rho_i$) we can
construct a ranking function as in Definition~\ref{def:rf}, namely:
$\rho(\ell,(x,y)) = \max(0,\lceil \rho_\ell(x,y)+1\rceil)$.
\end{example}

In the example above, we have limited ourselves to one loop, because
if we seek an \lrf for the whole \cfg, even when using different
functions for the different nodes, we would fail: while the \lrf of
the first loop is based on the loop counter $x$, the second is based
on the loop counter $y$.
Instead, we could analyse the strongly connected components (\sccs)
separately---note that for a termination proof this suffices: if there were an infinite execution,
it would eventually stay within a single \scc.
In this case it is not always possible to construct a global
``linear'' ranking function (there may be a global ranking function of a more complex form).

\begin{example}
\label{ex:lrf:cfg:2}
Let us analyse the \sccs of the \cfg of Figure~\ref{fig:cfg:lrf}
separately.
We start by seeking an \lrf for the \scc of $\transitions_1$,
$\transitions_3$, $\transitions_4$, $\transitions_5$ and
$\transitions_6$.
We find the same functions as in the previous example for the
corresponding nodes.
Next we continue with the \scc of $\transitions_9$, and we find
$\rho_{\ell_6}(x,y,z)=y$.
\end{example}

Let us now consider the case in which we seek an \lrf \wrt  a
polyhedral set of initial states $\poly{S}_0 \subseteq \rats^n$, and
starting at $\ell_0$.
Similarly to the case of \mlc loops, we can solve the problem by first
inferring supporting polyhedral invariants (for each location), add
them to the transition relations of corresponding outgoing edges, and
then use the algorithm of universal termination as described
above---this is what we have done in the examples above actually.
We can also simultaneously infer invariants and seek the functions
$\rho_\ell$ using the template approach, which is very similar to the
case of \slc and \mlc loops, except that here we have an invariant for
each location.
Also in this case we obtain a complete algorithm, for the rational
case, to the problem of deciding whether the template can be
instantiated such that the \cfg (or a given \scc) has an \lrf.
Finally, as for the complexity of related problems (\eg,
problems~\ref{lrf:slc:init:decprob} and~\ref{lrf:slc:init:invprob}),
nothing is known for the \cfg case.

\subsection{History of \lp-based  \lrfs Algorithms}
\label{sec:lrf:history}

Algorithms to find an \lrf for \slc loops have been proposed by
several
researchers~\citep{SG91,CS01,Feautrier92.1,PodelskiR04a,MS08}. All
these works, even if originating from an application where variables
are integer, relax the problem to the rationals.
\citet{BagnaraMPZ12} overview and compare the methods of
\citet{SG91,PodelskiR04a,MS08}.

It may be interesting to note that while most of these works concern
termination, \citet{Feautrier92.1} employs ranking functions for a
different purpose, solving a \emph{scheduling problem} for parallel
computation. It is also the only one among these works that discusses
the integer case and its complexity, and in doing so it precedes the
works of \citet{BG14,CookKRW13}. \citet{BMS05b} also studied \lrfs for
integer linear-constraint loops

\subsection{Other Approaches for \lrfs}
\label{sec:lrf:conc}

In contrast to work that is based on the use of Farkas' lemma,
\citet{LWF20} show that, in the rational case, one can compute a
witness against the existence of an \lrf in polynomial time. A
generalisation of this approach has been reported by
\citet{Ben-AmramDG19} for multiphase ranking functions (see
Section~\ref{sec:mlrfs}), and used to show the following result for
bounded \slc loops.

\begin{theorem}
Let $\transitions$ be an \slc loop such that the set of enabled states
$\proj{\vec{x}}{\transitions}$ is a bounded polyhedron, then: either
$\transitions$ is non-terminating and has a fixpoint
$\tr{\vec{x}}{\vec{x}}\in\transitions$, or it is terminating and has an \lrf.
\end{theorem}

\citet{MMP16} consider the problem of synthesising \lrfs for
floating-point \slc loops. They show that the decision
problem is at least \conph and provide an incomplete algorithm for
synthesising \lrfs for such loops.

\section{Lexicographic-Linear Ranking Functions}
\label{sec:llrf}

The notion of lexicographic ranking functions is ubiquitous in
termination analysis because they naturally arise when analysing
nested loops or programs with complex control flow, as in the
following example.

\begin{example}
\label{ex:llrf:intro}
Consider an \mlc loop defined by the following paths
\begin{equation}
\label{eq:llrf:intro:ex}
\begin{array}{rl}
\transitions_1&= \set{x_1 \ge 0,x_2 \ge 0, x_1'=x_1-1} \\
\transitions_2&= \set{x_1 \ge 0,x_2 \ge 0, x_2'=x_2-1, x_1'=x_1}
\end{array}
\end{equation}
In $\transitions_1$, $x_1$ decreases towards zero and $x_2$ is changed
unpredictably, since there is no constraint on $x_2'$; this could
arise, for instance, from $x_2$ being set to the result of an input
from the environment, an expression that cannot be modelled using
linear constraints, or a function call for which we have no
input-output summary.  In $\transitions_2$, $x_2$ decreases towards
zero and $x_1$ is unchanged.
Clearly, $\tuple{x_1,x_2}$ always decreases lexicographically, while
there can be no single \lrf for this loop.
Similarly, the same tuple decreases lexicographically for the \mlc
loop of Example~\ref{ex:mlc:0}, that does not have an \lrf as well.
\end{example}

Interestingly, Alan Turing's early demonstration~\citep{Turing48} of
how to verify a program used a lexicographic ranking function
for the termination proof. For the sake of developing
practical tools, and for studying properties of lexicographic ranking
functions, one typically restricts the form of functions allowed as
components.
A common such restriction considers components that are linear affine
functions, yielding \emph{lexicographic-linear ranking
  functions}~(\llrfs). In the rest of this section, we use $\rho_i$ to
denote a linear affine function that maps states to rational values,
as in the case of \lrfs.
The most general definition for an \llrf is the following.

\begin{definition}
\label{def:llrf}
Given a transition relation $T\subseteq\numdom^{2n}$, where
$\numdom\in\{\reals,\rats,\ints\}$, we say that
$\tau=\tuple{\rho_1,\dots,\rho_d}$ is an \llrf (of depth $d$) for $T$,
if for every $\vec{x}'' \in T$ there is an index $i$ such that:
\begin{alignat}{2}
                      && \rho_i(\vec{x}) &\ge 0          \,, \label{eq:llrf:1}\\
 \forall j < i \ .\   && \diff{\rho_j}(\vec{x}'') & \geq 0 \,, \label{eq:llrf:2}\\
                      && \diff{\rho_i}(\vec{x}'') & \geq 1 \,. \label{eq:llrf:3}
\end{alignat}
We say that $\vec{x}''$ is \emph{ranked by} $\rho_i$ (for the minimal such $i$).
\end{definition}

The justification that an \llrf implies termination uses the fact that
the lexicographic order over $\nats^d$ is well-founded. Given an \llrf
$\tuple{\rho_1,\ldots,\rho_d}$, we coerce the component $\rho_i$ to
$\max(0,\lceil\rho_i+1\rceil)$ and get a tuple
$\tuple{\max(0,\lceil\rho_1+1\rceil),\ldots,\max(0,\lceil\rho_d+1\rceil)}$
that decreases lexicographically over $\nats^d$. This works since each
$\rho_i$ decreases by at least $1$ on the transitions that it ranks.

\begin{remark}
  Replacing~\eqref{eq:llrf:3} by $\diff{\rho_i}(\vec{x}'') > 0$, we
  obtain a definition for a \emph{weak} \llrf. While weak \llrfs do
  not clearly imply termination (over the rationals or reals), they are useful to infer \llrfs as we will see later.
  Over the integers, weak \llrfs are equivalent to \llrfs since we may
  assume that all coefficients of $\rho_i$ are integer, and thus
  $\diff{\rho_i}(\vec{x}'') > 0$ means
  $\diff{\rho_i}(\vec{x}'') \ge 1$.
\end{remark}

It is easy to see that a given tuple $\tuple{\rho_1,\ldots,\rho_d}$ is
an \llrf for $T$ if and only if the following formula holds:
\begin{equation}
\label{eq:llrf:template}
 \left(\bigwedge_{i=1}^{d} (T _i(\vec{x},\vec{x}') \implies \diff{\rho_i}(\vec{x}'') \ge 0) \right) \wedge \left( T _{d+1}(\vec{x},\vec{x}') \implies \emph{false} \right)
\end{equation}
where
$T_i(\vec{x},\vec{x}')=T(\vec{x},\vec{x}') \wedge
(\wedge_{j=1}^{i-1}(\rho_j(\vec{x}) < 0 \vee \diff{\rho_j}(\vec{x}'')
< 1))$, \ie, we remove all transitions that are ranked by any component $\rho_j$
with $j<i$.

This formulation gives rise to the \emph{template based approach} for
synthesising an \llrf of a given depth~\citep{LeikeHeizmann15}.
We start from template functions
$\rho_i(\vec{x})=\vect{\rfcoeff}_i\vec{x}+\rfcoeff_{0,i}$, where
$\vect{\rfcoeff}_i$ and $\rfcoeff_{0,i}$ are variables~(``template
parameters''), and then using the Motzkin transposition theorem, which
is similar to Farkas' Lemma, we translate~\eqref{eq:llrf:template}
into a set of existential constraints over the template parameters
(and some other variables) that can be solved using off-the-shelf SMT
solvers, and thus get concrete values for the coefficients of each
$\rho_i$.

The resulting existential constraints, however, are non-linear since
the constraints that we add in each $T_i$ use template
parameters. They can be solved within polynomial space complexity
since the corresponding decision problem, over the reals, is
PSPACE~\citep{Canny88}. Note that we only propose this approach for loops over
the reals, and assuming that $T$ is given by polyhedra.
To decide existence of an \llrf, we can search iteratively for
increasing values of depth $d$, however if there is no \llrf this
method does not terminate.
Note also that we could incorporate inference of supporting
invariants, similarly to what we have done for \lrfs.

\begin{algorithm}[t]
\caption{Synthesizing Lexicographical Linear Ranking Functions }
\label{alg:llrfsyn_gen}
\DontPrintSemicolon
\LinesNumberedHidden
\SetKwFunction{procsyn}{LLRFSYN}
\SetKwFunction{procsynint}{LLRFint}
\procsyn{$T$}\;
{\small
\KwIn{A set of transition $T \subseteq \numdom^{2n}$, where $\numdom\in\{\reals,\rats,\ints\}$}
\KwOut{An \llrf $\llrfsym$ for $T$, if exists, otherwise $\nollrf$ }
\Begin{
\setcounter{AlgoLine}{0}
\ShowLn  $\llrfsym := \tuple{}$ \;  
\ShowLn  $T' := T$ \;  
\ShowLn  \While{$T'$ is not empty}{
\ShowLn  \uIf{ $T'$ has a non-trivial \textit{quasi}-\lrf $\rho$ \wrt $T$}{ \label{alg:llrf:qlrf}
\ShowLn     $T' = T' \setminus \{\vec{x}'' \in T' \mid \vec{x}'' \mbox{ is (weakly) ranked by } \rho\}$ \label{alg:llrf:red}\;
\ShowLn     $\tau = \tau::\rho$ \label{alg:llrf:app}
}
\ShowLn \Else
{
  \ShowLn  $\llrfsym=\nollrf$\; \label{alg:llrf:nollrf}
  \ShowLn \textbf{break} \label{alg:llrf:break}
}
}
\ShowLn \Return $\llrfsym$\label{alg:llrf:ret}
}
}
\end{algorithm}
 
An alternative and widely used approach for synthesising \llrfs is
based on a greedy algorithm
(Algorithm~\ref{alg:llrfsyn_gen}), which incrementally builds the \llrf
by seeking \textit{quasi}-\lrfs.
We first give the definition of a \textit{quasi}-\lrf, and then
explain the method, shown as Algorithm~\ref{alg:llrfsyn_gen}.

\begin{definition}
\label{def:qlrf}
We say that an affine linear function $\rho$ is \textit{quasi}-\lrf
(\qlrf for short) for $T' \subseteq T \subseteq \numdom^{2n}$ if the
following holds for all $\vec{x}''\in T'$:
\begin{align}
 \diff{\rho}(\vec{x}'') \ge 0 \label{eq:qlrf:1}
\end{align}
We say that it is \emph{non-trivial} if, in addition,
$\diff{\rho}(\vec{x}'') > 0$ and $\rho(\vec{x})\ge0$ for at least one
$\vec{x}'' \in T'$. We say that $\vec{x}''$ is (weakly) ranked by
$\rho$.
\end{definition}

This definition of \qlrfs will be specialised later by adding more
conditions; these variants correspond to variants of \llrfs, that are special cases of
Definition~\ref{def:llrf}. In some of these specialised definitions, the set $T$
(which is redundant in the above definition) will play a role.

Algorithm~\ref{alg:llrfsyn_gen} incrementally builds an \llrf, in each
iteration of the while loop, as follows: at Line~\ref{alg:llrf:qlrf}
it seeks a \qlrf $\rho$ for the current set of transitions $T'$, and
if it fails it exits the loop with $\tau=\nollrf$; at
Line~\ref{alg:llrf:red} it eliminates all transitions that are
(weakly) ranked by $\rho$ from $T'$, and then appends $\rho$ to
$\tau$.
When all transitions are eliminated from $T'$, it exits the loop and
returns $\llrfsym$ at Line~\ref{alg:llrf:ret} which can be an \llrf,
possibly weak, or $\nollrf$ in case of failure.

The \llrf is possibly weak because depending on the specific
definition of the \qlrf and the domain of variables, the transitions
that are eliminated at Line~\ref{alg:llrf:red} might be weakly ranked.
For example, if $T \subseteq \rats^{2n}$ and we eliminate all those
weakly ranked by $\rho$, \ie, the transitions on which $\rho$ is
decreasing ($\diff{\rho}(\vec{x}'')>0$) and non-negative
($\rho(\vec{x})\ge0$), then we get a weak \llrf which is not enough
for proving termination over \rats.
Some approaches solve this issue by converting the weak \llrf into an
\llrf (of the same depth) afterwards, other approaches guarantee that
transitions that are eliminated at Line~\ref{alg:llrf:red} actually
satisfy $\diff{\rho}(\vec{x}'') \ge 1$ and thus directly build an
\llrf.
Recall that over the integers, weak \llrfs are enough since we may
assume that all coefficients of $\rho$ are integer, and thus
$\diff{\rho}(\vec{x}'') > 0$ means $\diff{\rho}(\vec{x}'') \ge 1$.
Termination of the algorithm also depends on the choice of the \qlrf,
and on how transitions are eliminated from $T'$.

The following is a fundamental property that is used to prove
completeness of corresponding algorithms for synthesising \llrfs.

\begin{observation}
\label{obs:llrf:qlrf}
If $T \subseteq \numdom^{2n}$ has an \llrf
$\tuple{\rho_1,\ldots,\rho_d}$, then any subset of transitions $T'$
must have a non-trivial \qlrf, namely $\rho_j$ for
$j=\max\{i \mid \vec{x}'' \in T' \mbox{ is ranked by } \rho_i \}$.
\end{observation}

A natural question to ask, given a definition of a \qlrf, is whether
there is an optimal \qlrf $\rho$ that eliminates as many transitions
as possible (\ie, if $\vec{x}''$ is eliminated by some \qlrf $\rho'$,
then it is eliminated by $\rho$ as well). This has the following
consequence: if there is an optimal one, and it is picked in each
iteration of Algorithm~\ref{alg:llrfsyn_gen}, then the returned \llrf
is of minimal depth (the number of components of the
\llrf). Unfortunately, there does not have to be an optimal choice for
\qlrfs as in Definition~\ref{def:qlrf}. In certain variants of \qlrfs, as
we will see later, there actually is an optimal choice.

The minimal depth is of interest when \llrfs are used to infer bounds
on the number of execution steps, for example this is the case
in~\citet{ADFG:2010} where such bound is typically a polynomial of
degree $d$, where $d$ is the depth of the \llrf.
It is also natural to ask whether there is an \emph{a priori} upper
bound on the depth, in terms of parameters of the loop (such as the number of variables).
 Such an upper bound is useful, for
example, for fixing the template in the template-based approach, and
plays a role in analysing the complexity of corresponding algorithms.

The research problems we are interested in this context, for integer
and rational \mlc loops (and \cfgs), are:
\begin{itemize}
\label{Qs}
\item[\textbf{Q1}] Is there a complete algorithm for synthesising
  \llrfs? If so, what is its complexity.
\item[\textbf{Q2}] How difficult is it to decide if an \llrf exists for
  a given \mlc loop?
\item[\textbf{Q3}] Is there an \emph{a priori} bound on the depth, in
  terms of the number of variables and paths of a given \mlc loop?
\item[\textbf{Q4}] Is there a complete algorithm for synthesising
  \llrfs of a given depth? If so, what is its complexity.
\item[\textbf{Q5}] How difficult is it to find an \llrf of minimal
  depth, or as a relaxation of this optimisation problem, how
  difficult to decide if there exists an \llrf that satisfies a given
  bound on the depth?
\end{itemize}

All these problems are still open for \llrfs as in
Definition~\ref{def:llrf}. The only approach we are aware of for
synthesising such \llrfs, for \emph{integer} \mlc loops, is that of
\citet{LarrazORR13}. Their algorithm uses max-SMT to synthesise \qlrfs
as follows: they use Farkas' lemma to generate a set of constraints
whose solutions define all functions that satisfy~\eqref{eq:qlrf:1}
for all paths, but in addition they add \emph{soft constraints} that
require some paths to be ranked -- the idea is that the max-SMT solver
will try to maximise the number of soft constraints that are
satisfied. Moreover, in addition to the \qlrf, they infer a supporting
invariant which makes the generated constraints non-linear as we have
seen in the case of \lrfs.
Importantly, their algorithm is not complete, and they do not consider
any question related to complexity of the underlying decision
problems.

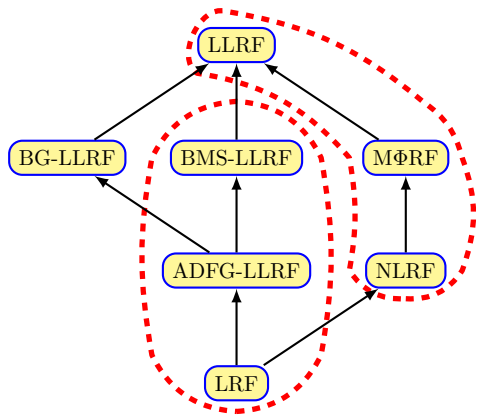
\begin{figure}[t]
  \begin{center}
    \begin{minipage}{6.25cm}

\begin{tikzpicture}[>=latex,line join=bevel, scale=0.75]

    \tikzstyle{class}=[draw=blue,thick,fill=yellow!50,rounded corners=5pt]
    \tikzstyle{conn}=[->,thick]
    \tikzstyle{grp}=[draw=red,dashed,line width=2pt]

\begin{scope}[shift={(0,0)}]
  \node (lrf) at (3,1) [class] {$\scalebox{0.75}{\lrf}$};
  \node (adfgllrf) at (3,3) [class] {$\scalebox{0.75}{\adfgllrf}$};
  \node (bmsllrf) at (3,5) [class] {$\scalebox{0.75}{\bmsllrf}$};
  \node (llrf) at (3,7) [class] {$\scalebox{0.75}{\llrf}$};
  \node (nlrf) at (6,3) [class] {$\scalebox{0.75}{\nlrf}$};
  \node (mlrf) at (6,5) [class] {$\scalebox{0.75}{\mlrf}$};
  \node (bgllrf) at (0,5) [class] {$\scalebox{0.75}{\bgllrf}$};

  \draw [grp] (3,0.5) to [bend right] (4.5,1.5) to [bend right=10] (4.5,5) to [bend right] (3,6) to [bend right] (1.5,5) to [bend right=10] (1.5,1.5) to [bend right] (3,0.5);

  \draw [grp] (6,2.5) to [bend right] (7,3) to [bend right=20] (7,5) to [bend right=20] (6.3,6) to [bend right=15] (2.8,7.6) to [bend right=80] (2.3,6.5) to [bend left=20] (5,5) to [bend left=10] (5,3) to [bend right] (6,2.5);

  \draw [conn] (lrf) -- (adfgllrf);
  \draw [conn] (lrf) -- (nlrf);
  \draw [conn] (adfgllrf) -- (bmsllrf);
  \draw [conn] (bmsllrf) -- (llrf);
  \draw [conn] (nlrf) -- (mlrf);
  \draw [conn] (mlrf) -- (llrf);
  \draw [conn] (adfgllrf) -- (bgllrf);
  \draw [conn] (bgllrf) -- (llrf);

\end{scope}

\end{tikzpicture}

    \end{minipage}
   ~
    \begin{minipage}{4.5cm}
    {\footnotesize
     \begin{tabular}{|l|l|l|}
       \hline
        \textbf{Class} & \textbf{Def.} & \textbf{Page} \\
       \hline
       \hline
      \lrf      & \ref{def:lrf}         & \pageref{def:lrf}\\
      \llrf     & \ref{def:llrf}        & \pageref{def:llrf}\\
      \bgllrf   & \ref{def:bg:llrf}     & \pageref{def:bg:llrf} \\
      \adfgllrf & \ref{def:adfg:llrf}   & \pageref{def:adfg:llrf} \\
      \bmsllrf  & \ref{def:bms:llrf}    & \pageref{def:bms:llrf}\\
      \mlrf     & \ref{def:mlrf}        & \pageref{def:mlrf} \\
      \nlrf     & \ref{def:nested-mlrf} & \pageref{def:nested-mlrf}\\
      \hline
    \end{tabular}}
  \end{minipage}
    \end{center}
    \caption{Classes of ranking functions for \mlc loop, ordered by
      their relative power. The classes surrounded by dashed lines
      become equivalent when restricted to \slc loops.}
  \label{fig:hierarchy}
\end{figure}

\begin{table}
  {\scriptsize
  \begin{center}
  \begin{tabular}{|c|l|ccccc|}
    \cline{3-7}
    \multicolumn{2}{c|}{} & Q1 & Q2 & Q3 & Q4 & Q5\\
    \hline
  \multirow{5}{*}{\rotatebox{90}{Over \rats}}
    & \llrf       & ?        & ?          & ?           & \exptime      & ?      \\
    &\bgllrf      & \ptime   & \ptime     & $n$         & \ptime        & \ptime \\
    &\adfgllrf    & \ptime   & \ptime     & $\min(n,k)$ & \ptime        & \ptime \\
    &\bmsllrf     & \ptime   & \ptime     & $k$         & \exptime      & \npc   \\
    &\mlrf        & ?        & ?          & ?           & \exptime      & ?      \\
    &\mlrf (\slc) & ?        & ?          & ?           & \ptime        & \ptime \\
    \hline
    \hline
  \multirow{5}{*}{\rotatebox{90}{Over \ints}}
    & \llrf       & ?        & ?     &  ?           & ?        & ? \\
    &\bgllrf      & \exptime & \conpc & $n$         & \exptime      & \conpc \\
    &\adfgllrf    & \exptime & \conpc & $\min(n,k)$ & \exptime      & \conpc \\
    &\bmsllrf     & \exptime & \conpc & $k$         & \exptime      & $\Sigma_2^P$\\
    &\mlrf        & ?        & ?             & ?           & ?             & ? \\
    &\mlrf (\slc) & ?        & ?             & ?           & \exptime      & \conpc \\
    \hline
  \end{tabular}
\end{center}
}
\caption{Summary of results, for the research questions \textbf{Q1-5}
  on Page \pageref{Qs}, for the different notions of \llrfs for \mlc
  loops (with $k$ paths and $n$ variables). For \cfgs, the results are
  the same as in the case of \mlc. For \slc loops all results are the
  same as in the case of \mlc, except for \mlrfs that we report
  explicitly in separated lines. The case of $\reals$ is the same as
  $\rats$.}
  \label{tab:llrf:summary}
\end{table}

Different researchers had come up with different variants of the
notion of \llrf for which there are answers to these questions. These
variants, and their relative power, are summarised in
Figure~\ref{fig:hierarchy}, and Table~\ref{tab:llrf:summary} includes
a summary of answers to the corresponding questions.

We note that a loop might have an \llrf according to one of these
variants but not another, for example the following \slc loop
\begin{equation}
  \label{eq:loop-llaraz}
\begin{array}{l}
  \while~(x \ge 0,\; y \le 10,\; z \ge 0,\; z\le 1)~\wdo\\
  ~~~~~~~~~~x'=x+y+z-10,\; y'=y+z,\; z'=1-z
\end{array}
\end{equation}
has the \llrf $\tuple{4y,4x-4z+1}$ according to
Definition~\ref{def:llrf}, but it is not admitted by any of the
variants that we will discuss.
In addition, it is possible for a loop to have \llrfs of all variants,
but such that the minimal depth is not the same in all of them (see
Example~\ref{ex:differentDim} in Section~\ref{sec:llrf:bms}).
Interestingly, all these variants can be described using
Algorithm~\ref{alg:llrfsyn_gen}, where the main differences between
them are: (1) the additional conditions they impose on \qlrfs; and (2)
the way (weakly) ranked transitions are eliminated.
We discuss the details in the next sections. For each variant, we
first discuss the case of \mlc (and \slc) loops without initial
states, then with initial states, and finally the case of \cfgs.
As in the case of \lrfs, by default we assume that variables range
over $\rats$, and the case of $\ints$ will always be discussed
separately. The case when variables range over $\reals$ is equivalent
of that of $\rats$.

\subsection{\bgllrfs}
\label{sec:llrf:bg}

The following definition of an \llrf is due to~\citet{BG14}, which is
obtained by strengthening~\eqref{eq:llrf:1} of
Definition~\ref{def:llrf} to require $\rho_j(\vec{x})\ge0$ for all
$j\le i$ -- this is reflected in~\eqref{eq:bg:llrf:1} of
Definition~\ref{def:bg:llrf}.

\begin{definition}
\label{def:bg:llrf}
Given an \mlc loop
$\transitions_1,\ldots,\transitions_k\subseteq\rats^{2n}$, we say that
$\tau=\tuple{\rho_1,\dots,\rho_d}$ is a \bgllrf (of depth $d$) for the
loop, if for every
$\vec{x}'' \in \transitions_1\cup\cdots\cup\transitions_k$ there is an
index $i$ such that:
\begin{alignat}{2}
 \forall j \le i \ .\ && \rho_j(\vec{x}) &\ge 0          \,, \label{eq:bg:llrf:1}\\
 \forall j < i \ .\   && \diff{\rho_j}(\vec{x}'') &\ge 0 \,, \label{eq:bg:llrf:2}\\
                      && \diff{\rho_i}(\vec{x}'') &\geq 1\,. \label{eq:bg:llrf:3} 
\end{alignat}
We say that $\vec{x}''$ is \emph{ranked by} $\rho_i$ (for the minimal such $i$).
\end{definition}

\begin{example}
\label{ex:kinds-of-lrfs-1}
Consider the \slc loop
\begin{equation}
\label{eq:llrf:loop:1}
\begin{array}{l}  
\while \; ( x_1 \ge 0, x_2 \ge 0, x_3 \ge -x_1 )\; \wdo \;\\
  \hspace*{2cm}x_2'= x_2-x_1,\; x_3'= x_3+x_1-2\,.
\end{array}
\end{equation}
This loop has a \bgllrf $\llrfsym=\tuple{x_2, x_3}$ as in
Definition~\ref{def:bg:llrf} (over both rationals and integers). Note
that when $x_2$ decreases, $x_3$ can be negative, \eg, for $x_1=1$,
$x_2=2$ and $x_3=-1$.
The \mlc of Example~\ref{ex:llrf:intro} has a \bgllrf
$\llrfsym=\tuple{x_1, x_2}$.
The \mlc loop of Example~\ref{ex:mlc:0} does not have a \bgllrf
(recall that it has a \llrf $\tuple{x_1, x_2}$).
\end{example}

Replacing~$\diff{\rho_i}(\vec{x}'') \geq 1$ by
$\diff{\rho_i}(\vec{x}'') > 0$ in~\eqref{eq:bg:llrf:3} we obtain a
class of functions that~\citet{BG14} call \emph{weak} \bgllrfs, which
are similar to weak \llrfs that we have discussed previously.
For integer loops, it is easy to see that weak and non-weak \bgllrfs
are equivalent for proving termination, since we may assume that all
$\rho_i$ have integer coefficients and thus
$\diff{\rho_i}(\vec{x}'') > 0$ means
~$\diff{\rho_i}(\vec{x}'') \geq 1$.
\citet{Ben-AmramG13} show that this equivalence is also true for
rational loops, and provide a polynomial-time algorithm for converting
a weak \bgllrf into a \bgllrf of the same depth.
We rely on this algorithm to convert the weak \llrf returned by
Algorithm~\ref{alg:llrfsyn_gen} to an \llrf.

\begin{definition}
\label{def:bg:qlrf}
Let $\transitions_1,\ldots,\transitions_k$ be an \mlc loop. We say that
an affine linear function $\rho$ is a \bg-\qlrf for
$\transitions'_1\cup\cdots\cup\transitions'_k \subseteq \rats^{2n}$,
where $\transitions'_i\subseteq \transitions_i$, if the following holds
for all $\vec{x}''\in \transitions'_1\cup\cdots\cup\transitions'_k$:
\begin{align}
 \rho(\vec{x}) \ge 0 \label{eq:bg:qlrf:1}\\
 \diff{\rho}(\vec{x}'') \ge 0 \label{eq:bg:qlrf:2}
\end{align}
We say that it is \emph{non-trivial} if, in addition,
inequality~\eqref{eq:bg:qlrf:2} is strict, \ie,
$\diff{\rho}(\vec{x}'') > 0$, for at least one
$\vec{x}'' \in \transitions'_1\cup\cdots\cup\transitions'_k$.
\end{definition}

When compared to \qlrfs as in Definition~\ref{def:qlrf}, the
difference is that $\rho$ is required to be non-negative on the set of
transitions under consideration and not only on the transitions for
which $\diff{\rho}(\vec{x}'') > 0$ holds. This is a stronger
requirement, however, it has the following consequence: any non-trivial
conic combination of \bg-\qlrfs $\rho_1$ and $\rho_2$ results in a
\bg-\qlrf that ranks all transitions ranked by $\rho_1$ and $\rho_2$,
which means that there exists an optimal \bg-\qlrf, given the loop.

\begin{example}
\label{ex:qlrfs:1}
Consider the \slc loop~\eqref{eq:llrf:loop:1}: 
$\rho(x_1,x_2,x_3)=x_2$ is a non-trivial \bg-\qlrf;
$\rho(x_1,x_2,x_3)=x_1$ is not because $x_1-x_1' \ge 0$ does not hold
for all transitions; and $\rho(x_1,x_2,x_3)=x_3$ is not because
$\rho(2,1,-1)=-1<0$.
For the \mlc loop of Example~\ref{ex:llrf:intro}:
$\rho(x_1,x_2)=x_1$ is a non-trivial \bg-\qlrf, while
$\rho(x_1,x_2)=x_2$ is not because $x_2-x_2'\ge 0$ does
not hold for all transitions.
The \mlc loop of Example~\ref{ex:mlc:0} does not have a \bg-\qlrf
because $x_1$ and $x_2$ can be arbitrarily negative.
\end{example}

\citet{Ben-AmramG13} provide a complete polynomial-time algorithm for
seeking an optimal non-trivial \bg-\qlrf
$\rho(\vec{x})=\vect{\rfcoeff}\vec{x}+\rfcoeff_0$ for a set of
transitions defined by an \mlc loop
$\transitions_1,\ldots,\transitions_k$.
The algorithm is as follows:
\begin{enumerate}[\upshape(1\upshape)]
\item Set up an \lp problem (using Farkas' Lemma) requiring all
  $\transitions_j$ to imply (\ref{eq:bg:qlrf:1},\ref{eq:bg:qlrf:2})
  for all $1 \le j \le k$. This generates a set of linear constraints
  over the variables $(\vect{\rfcoeff},\rfcoeff_0)$ and some other
  variables for the Farkas' coefficients; we denote the polyhedron
  specified by these constraints by $\poly{S}$;
\item Pick a point from the \emph{relative interior} of $\poly{S}$,
  which fixes values for $(\vect{\rfcoeff},\rfcoeff_0)$ and thus
  define $\rho$; and
\item If $\rho(\vec{x})>0$ holds for some
  $\vec{x}''\in\transitions_1\cup\cdots\cup\transitions_k$, then
  $\rho$ is an optimal \bg-\qlrf, otherwise there is no non-trivial
  \bg-\qlrf.
\end{enumerate}
The key point of this algorithm is that any
$(\vect{\rfcoeff},\rfcoeff_0)$ that comes from the relative interior
of $\poly{S}$ leads to an optimal \bg-\qlrf $\rho$.

When this algorithm is used within Algorithm~\ref{alg:llrfsyn_gen},
once $\rho$ has been found at Line~\ref{alg:llrf:qlrf}, we eliminate
all (weakly) ranked transitions by adding $\diff{\rho}(\vec{x}'')=0$
to each $\transitions_j$ at Line~\ref{alg:llrf:red}.
It easy to see that when the algorithm reaches
Line~\ref{alg:llrf:ret}, the tuple $\tau$ is a weak \bg-\llrf, and,
moreover, it is of minimal depth since we use optimal \bg-\qlrfs. As
we have mentioned before, $\tau$ can be always converted to a \bgllrf
of the same depth, in polynomial time.

Completeness is due to the following two properties:
\begin{inparaenum}[\upshape(1\upshape)]
\item The algorithm is guaranteed to terminate, because
$\transitions_j\wedge\diff{\rho}(\vec{x}'')=0$ is a proper face of
$\transitions_j$, and thus its dimension is smaller than that of
$\transitions_j$~(the dimension of the empty polyhedron is $-1$); and
\item When it returns $\nollrf$, then indeed there is no \bgllrf for
  the loop. This is because it has found a subset of transitions for
  which no non-trivial \bg-\qlrf exists, which would be impossible if
  the loop had a \bgllrf (see Observation~\ref{obs:llrf:qlrf}).
\end{inparaenum}

The complexity of Algorithm~\ref{alg:llrfsyn_gen} in this case is
polynomial since every iteration is polynomial.  In fact, this is not
immediate since reducing $\transitions_j$ to
$\transitions_j \wedge \diff{\rho}(\vec{x}'')=0$ might potentially
increase the bit-size of that path exponentially during the
iterations. However, \citet{BG14} show that this reduction can be done
by changing one of the inequalities of $\transitions_j$ to an equality
since $\transitions_j\wedge\diff{\rho}(\vec{x}'')=0$ is a face of
$\transitions_j$, and thus we at most double the size of the
constraint representation of $\transitions_j$ during \emph{all}
iterations.
The number of iterations is bounded by the maximum dimension of
$\transitions_1,\ldots,\transitions_k$.

\begin{theorem}[\citealp{BG14}]
  There is a complete polynomial-time algorithm for finding a \bgllrf
  of minimal depth, if one exists, for a given rational \mlc loop
  $\transitions_1,\ldots,\transitions_k$.
\end{theorem}

\begin{example}
\label{ex:llralg:int:1}
Let us demonstrate the algorithm on the \slc
loop~\eqref{eq:llrf:loop:1} of Example~\ref{ex:kinds-of-lrfs-1}, which
is defined by
\[
\transitions=\{ x_1 \ge 0,~ x_2 \ge 0,~ x_3 \ge -x_1,~ x_2'= x_2-x_1,~ x_3'= x_3+x_1-2\}.
\]
\procsyn is called with $\transitions$, and then, in the first
iteration of the while loop, at Line~\ref{alg:llrf:qlrf} it finds the
non-trivial \bg-\qlrf $\rho_1(x_1,x_2,x_3)=x_2$ for $\transitions$, at
Line~\ref{alg:llrf:red} it eliminates all transitions for which
$x_2-x'_2=0$, and appends $\rho_1$ to $\tau$.
In the next iteration, at Line~\ref{alg:llrf:qlrf} it finds the
non-trivial \bg-\qlrf $\rho_2(x_1,x_2,x_3)=x_3$ for
$\transitions\land x_2-x'_2=0$, at Line~\ref{alg:llrf:red} it
eliminates all transitions for which $x_3-x'_3=0$, which results in an
empty set, and appends $\rho_2$ to $\tau$.
Since the set of transitions is empty, we exit the while loop and
arrive at Line~\ref{alg:llrf:ret} with the weak \bgllrf
$\tuple{x_2, x_3}$. Converting it to an \llrf results in the same
tuple, as it is already a \bgllrf in this case.
\end{example}

\begin{example}
  Let us demonstrate the algorithm on the \mlc loops of
  Example~\ref{ex:llrf:intro}.
  \procsyn is called with $\transitions_1,\transitions_2$, and then,
  in the first iteration of the while loop, at
  Line~\ref{alg:llrf:qlrf} it finds the non-trivial \bg-\qlrf
  $\rho_1(x_1,x_2)=x_1$, at Line~\ref{alg:llrf:red} it eliminates all
  transitions for which $x_1-x'_1=0$, which eliminates
  $\transitions_1$ and leaves $\transitions_2$ unchanged, and appends
  $\rho_1$ to $\tau$.
  In the next iteration, at Line~\ref{alg:llrf:qlrf} it finds the
  non-trivial \bg-\qlrf $\rho_2(x_1,x_2)=x_2$ for $\transitions_2$, at
  Line~\ref{alg:llrf:red} it eliminates all transitions for which
  $x_2-x'_2=0$, which eliminates $\transitions_2$, and appends
  $\rho_2$ to $\tau$.
  Since both paths were eliminated, we exit the while loop and arrive
  at Line~\ref{alg:llrf:ret} with the weak \bgllrf $\tuple{x_1,
    x_2}$. Converting it to an \llrf results in the same tuple, as it
  is already a \bgllrf in this case.
  Applying \procsyn to the \mlc loop of Example~\ref{ex:mlc:0} fails
  in the first iteration, because $\transitions_1,\transitions_2$ does
  not have a \bg-\qlrf.
\end{example}

As for the upper bound on the depth of \bgllrfs, \citet{BG14} show
that it is $n$, the number of variables.

\begin{theorem}[\citealp{BG14}]
If there is a \bgllrf for a given \mlc loop
$\transitions_1,\ldots,\transitions_k$, then there is one with at most
$n$ components.
\end{theorem}

Let us now consider the integer case.
A complete algorithm for synthesising \bg-\qlrfs for
$\intpoly{\transitions_1},\ldots,\intpoly{\transitions_k}$ can be
obtained by applying the one of the rational case on the corresponding
integer hulls
$\inthull{(\transitions_1)},\ldots,\inthull{(\transitions_k)}$.

\begin{observation}[\citealp{BG14}]
  The integer \mlc loop
  $\intpoly{\transitions_1},\ldots,\intpoly{\transitions_k}$ has a
  \bgllrf of depth $d$, if and only if
  $\inthull{(\transitions_1)},\ldots,\inthull{(\transitions_k)}$ has a
  (weak) \bgllrf of depth $d$.
\end{observation}

Using this observation, synthesising \bg-\qlrfs for
$\intpoly{\transitions_1},\ldots,\intpoly{\transitions_k}$ can be done
by applying the algorithm of the rational case on the corresponding
integer hulls
$\inthull{(\transitions_1)},\ldots,\inthull{(\transitions_k)}$,
however, one needs to guarantee that when reducing
$\inthull{(\transitions_j)}$ to
$\inthull{(\transitions_j)}\wedge\diff{\rho}(\vec{x}'')=0$, we still
have an integer polyhedron.
This is indeed the case since
$\inthull{(\transitions_j)}\wedge\diff{\rho}(\vec{x}'')=0$ is a face
of $\inthull{(\transitions_j)}$. The runtime is (in the worst case) exponential since
computing the integer hull takes exponential time.

\begin{theorem}[\citealp{BG14}]
  There is a complete exponential-time algorithm for finding a \bgllrf
  of minimal depth, if one exists, for a given integer \mlc loop
  $\intpoly{\transitions_1},\ldots,\intpoly{\transitions_k}$.
\end{theorem}

The \emph{decision problem} for integer loops is \conpc, as for
\lrfs. This result follows from the following characterisation that is
related to Observation~\ref{obs:llrf:qlrf}:

\begin{theorem}
\label{thm:nollrf}
There is no \bgllrf for
$\intpoly{\transitions_1},\ldots,\intpoly{\transitions_k}$, if and
only if there is
$T \subseteq
\intpoly{\transitions_1}\cup\cdots\cup\intpoly{\transitions_k}$ for
which there is no non-trivial \bg-\qlrf.
\end{theorem}

This characterisation facilitates the construction of witnesses
against the existence of a \bgllrf. In fact they are witnesses against
the existence of a non-trivial \bg-\qlrf for a subset of the
transitions. The form of such witnesses is similar to what we have
shown for \lrfs.

The problem of seeking a \bgllrf when provided a polyhedral set of
initial states $\poly{S}_0$ is similar to what we have described for
the case of \lrfs.
Namely, we first infer a supporting invariant and add it to the
transition relations of the different paths, and then apply the
algorithm described above to find the different components of the
\bgllrf.
We could also use the template approach to infer a supporting
invariant and a \bg-\qlrf simultaneously. However, in this case, it is
important to note that the invariants should always consider the
original \mlc loop, and not just transitions that have not been
eliminated so far~\citep{BrockschmidtCF13,LarrazNORR14}.
Another obstacle for the template approach is that it is not clear how
to select an optimal \bg-\qlrf since the constraints are non-linear,
and thus, unlike the case of \lrfs, completeness is not guaranteed
even in the case of $\reals$.
Regarding the complexity of the related decision problems, nothing is
known beyond the lower bound results for \lrfs.

Inferring \bgllrfs for \cfgs can be done similarly to what we have
explained for the \mlc loop case, where in every iteration we find a
\bg-\qlrf for the transition relations of all remaining edges, and
then eliminate transitions that are ranked.
As explained with respect to \lrfs,  we can seek a \bg-\qlrf where each node is assigned a
(possibly) different $\rho_\ell$, or seek \bg-\qlrfs at the level of
\sccs.
The complexity of the related decision problems, in both approaches, and
without restricting to an initial state, are the same as the case of
\mlc loops.
Handling the case of initial states is done as explained above for
\mlc loops, in particular, the inference of invariants must always
consider the original \cfg~\citep{BrockschmidtCF13,LarrazNORR14} and
not only the parts that are currently under consideration.

\begin{example}
\label{ex:bgllrf:cfg}
Consider the \cfg depicted in Figure~\ref{fig:cfg:1}, and let us
demonstrate how to synthesise a \bgllrf. We first do it for the entire
\cfg and then at the level of \sccs. In both cases we assume that
invariants have been added to the corresponding transition relations.

In a first step, we consider all transition relations of the \cfg,
where each node is assigned a (template) function
$\rho_\ell(x,y,z) = \rfcoeff_{\ell,1} x + \rfcoeff_{\ell,2} y +
\rfcoeff_{\ell,3} z + \rfcoeff_{\ell,0}$. We find the following
optimal \bg-\qlrf:
\[
  \begin{array}{ll}
    \rho_{\ell_0}(x,y,z) =  2x+3\\
    \rho_{\ell_1}(x,y,z) =  2x+2\\
  \end{array}
  ~
  \begin{array}{ll}
    \rho_{\ell_2}(x,y,z) =  2x+2\\
    \rho_{\ell_3}(x,y,z) =  2x+1\\
  \end{array}
  ~
  \begin{array}{ll}
    \rho_{\ell_4}(x,y,z) =  2x+2\\
    \rho_{\ell_5}(x,y,z) =  2x+1\\
  \end{array}
\]
This \bg-\qlrf is decreasing on all transitions of $\transitions_0$,
$\transitions_2$, $\transitions_5$, $\transitions_7$, and
$\transitions_8$, and thus it eliminates the corresponding edges.
Seeking a \bg-\qlrf for what is left of the \cfg (\ie,
$\transitions_1$, $\transitions_3$, $\transitions_4$, and $\transitions_6$) we find
the following \bg-\qlrf that is decreasing on remaining transition
relations:
\[
  \begin{array}{ll}
    \rho_{\ell_1}(x,y,z) =  3y+2\\
    \rho_{\ell_2}(x,y,z) =  3y+1\\
  \end{array}
  ~
  \begin{array}{ll}
    \rho_{\ell_3}(x,y,z) =  z-y\\
    \rho_{\ell_4}(x,y,z) =  3y+3\\
  \end{array}
\]
Now we are left with no edges, and thus we have the following
\bgllrf~(a tuple for each node, where those of $\ell_0$ and
$\ell_5$ were complemented with $0$ components for clarity):
\[
  \begin{array}{rl}
    \ell_0 :& \tuple{2x+4,0} \\
    \ell_1 :& \tuple{2x+3,3y+2} \\
  \end{array}
  ~
  \begin{array}{ll}
    \ell_2 :& \tuple{2x+3,3y+1} \\
    \ell_3 :& \tuple{2x,z-y} \\
  \end{array}
  ~
  \begin{array}{ll}
    \ell_4 :& \tuple{2x+3,3y+3} \\
    \ell_5 :& \tuple{2x+2,0} \\
  \end{array}
\]

Let us now consider the approach that works at the level of the \sccs.
We start by seeking a \bg-\qlrf for the single \scc of
$\transitions_1$, $\transitions_2$, $\transitions_3$,
$\transitions_4$, $\transitions_5$, and $\transitions_6$.
We find the optimal \bg-\qlrf $\rho_1(x,y,z)=x+1$ which is decreasing
on all transitions of $\transitions_5$, and thus eliminates the
corresponding edge and splits the \scc into two: the one of
$\transitions_4$, and the one of $\transitions_1,\transitions_3$ and
$\transitions_6$.
For the first one we find the \bg-\qlrf $\rho_2(x,y,z)=z-y$ which
eliminates $\transitions_4$, and for the second we find the \bg-\llrf
$\rho_1(x,y,z)=y+1$ which eliminates $\transitions_3$ and leaves us
without cycles and thus we proved termination.
Note that when seeking \bg-\qlrfs at the level of \sccs, it is not
always needed to use different function for the different nodes, since
unlike \lrfs, \qlrfs are not required to decrease on all transitions.
\end{example}

\subsection{\adfgllrfs}
\label{sec:llrf:adfg}

The following definition of an \llrf is due to~\citet{ADFG:2010},
which is obtained%
\footnote{Chronologically, the work of \citet{ADFG:2010} was developed
  before that of~\citet{BG14}, but we present them in a reverse order
  for the sake of the systematic presentation.}  by strengthening the
one of \bgllrf to require all components to be non-negative on all
transitions---this is reflected in~\eqref{eq:adfg:llrf:1} of
Definition~\ref{def:adfg:llrf} when compared to~\eqref{eq:bg:llrf:1}
of Definition~\ref{def:bg:llrf}.

\begin{definition} 
\label{def:adfg:llrf}
Given an \mlc loop
$\transitions_1,\ldots,\transitions_k\subseteq\rats^{2n}$, we say that
$\tau=\tuple{\rho_1,\dots,\rho_d}$ is an \adfgllrf (of depth $d$) for
the loop, if for every
$\vec{x}'' \in \transitions_1\cup\cdots\cup\transitions_k$ there is an
index $i$ such that:
\begin{alignat}{2}
 \forall j \le d \ .\ && \rho_j(\vec{x}) &\ge 0          \,, \label{eq:adfg:llrf:1}\\
 \forall j < i \ .\   && \diff{\rho_j}(\vec{x}'') &\ge 0 \,, \label{eq:adfg:llrf:2}\\
                      && \diff{\rho_i}(\vec{x}'') &\geq 1\,. \label{eq:adfg:llrf:3} 
\end{alignat}
We say that $\vec{x}''$ is \emph{ranked by} $\rho_i$ (for the minimal such $i$).
\end{definition}

\bgllrfs are more powerful than \adfgllrfs.

\begin{example}
  Loop~\eqref{eq:llrf:loop:1} of Example~\ref{ex:kinds-of-lrfs-1} does
  not have an \adfgllrf, while it has a \bgllrf.
  The \mlc of Example~\ref{ex:llrf:intro} has an \adfgllrf
  $\llrfsym=\tuple{x_1, x_2}$.
  The \mlc loop of Example~\ref{ex:mlc:0} does not have an \adfgllrf
  (recall that it does not have a \bgllrf as well).
\end{example}

\begin{definition}
\label{def:adfg:qlrf}
Let $\transitions_1,\ldots,\transitions_k$ be an \mlc loop. We say that
an affine linear function $\rho$ is an \adfg-\qlrf for
$\transitions'_1\cup\cdots\cup\transitions'_k \subseteq \rats^{2n}$,
where $\transitions'_i\subseteq \transitions_i$, if the following holds:
\begin{alignat}{2}
 \forall\vec{x''}\in \transitions_1\cup\cdots\cup\transitions_k \ .\ && \rho(\vec{x}) &\ge 0            \,, \label{eq:adfg:qlrf1}\\
 \forall\vec{x''}\in \transitions'_1\cup\cdots\cup\transitions'_k \ .\ && \diff{\rho}(\vec{x}'') &\ge 0 \,, \label{eq:adfg:qlrf2}
\end{alignat}
We say that it is \emph{non-trivial} if, in addition,
$\diff{\rho}(\vec{x}'') > 0$ for at least one
$\vec{x}'' \in \transitions'_1\cup\cdots\cup\transitions'_k$.
\end{definition}

When compared to \bg-\qlrfs as in Definition~\ref{def:bg:qlrf}, the
difference is that $\rho$ is required to be non-negative on all
transitions and not only on the transitions under consideration. Note
that \adfg-\qlrfs also have the property that any nonzero conic
combination of \adfg-\qlrfs $\rho_1$ and $\rho_2$ results in an
\adfg-\qlrf that ranks all transitions that are ranked by $\rho_1$ and
$\rho_2$, which means that there exists an optimal \adfg-\qlrf.

\begin{remark}
  Interestingly, \citet{BG:15:CAV} show that all the results
  (complexity and algorithmic, both over rationals and integers) that
  we have discussed in Section~\ref{sec:llrf:bg} for \bgllrfs, hold
  also for \adfgllrf. The only (trivial) change required is in the
  procedure that synthesises the \qlrfs, to require the \qlrf to be
  non-negative on all transitions instead on those under consideration.
  However, the algorithmic aspects of \adfgllrfs as developed in the
  original work of \citet{ADFG:2010} are different, and shed light on
  some properties of such \llrfs. We discuss this in the rest of this
  section.
\end{remark}

\citet{ADFG:2010} provide a complete polynomial-time algorithm for
finding a non-trivial \adfg-\qlrf
$\rho(\vec{x})=\vect{\rfcoeff}\vec{x}+\rfcoeff_0$ for a set of
transitions defined by a given \mlc loop
$\transitions_1,\ldots,\transitions_k$.
The algorithm is as follows:
\begin{enumerate}
\item Set up an \lp problem (using Farkas' Lemma) requiring all paths
  of the input \mlc loop to entail $\rho(\vec{x}) \ge 0$, and each
  path $\transitions_j$ to entail
  $\diff{\rho}(\vec{x}'') \ge \delta_j$, where $0 \le \delta_j\le 1$
  is a variable.

\item Solve the \lp problem by maximising $\sum_{j=0}^k\delta_j$, which fixes
  values for all variables, including $(\vect{\rfcoeff},\rfcoeff_0)$.
\item If all $\delta_j$ are zero in the solution, the algorithm fails,
  otherwise $\rho$ ranks all paths $\transitions_j$ for which
  $\delta_j=1$~(each $\delta_j$ can be either $0$ or $1$, since when
  $0<\delta_j<1$ we can always scale $\rho$ up to obtain
  $\delta_j=1$).
\end{enumerate}
The run-time of this algorithm is polynomial since it is based on
solving a single \lp problem of polynomial size.

When the algorithm above is used within
Algorithm~\ref{alg:llrfsyn_gen}, once $\rho$ has been found at
Line~\ref{alg:llrf:qlrf}, \citet{ADFG:2010} eliminate at
Line~\ref{alg:llrf:red} all paths for which $\delta_j=1$.
This also means that the \adfgllrf is not weak.
The total run-time of Algorithm~\ref{alg:llrfsyn_gen} in this case is
polynomial, since it solves at most $k$ \lp problems (in at most $k$
iterations of the while loop) of polynomial size~(the bit-size of the
\mlc loop does not increase through the iterations, since we only
eliminate paths).

The algorithm for synthesising \adfg-\qlrfs that we described above is
clearly sound, however, its optimality is not clear. Moreover, at
Line~\ref{alg:llrf:red} of Algorithm~\ref{alg:llrfsyn_gen} we
eliminate only paths that are completely ranked by $\rho$, but there
might be transitions in other paths that are ranked by $\rho$ that are
not eliminated. Thus, completeness and optimality are not immediate to
see (\ie, the reason why Algorithm~\ref{alg:llrfsyn_gen} will find an
\adfgllrf of minimal depth if one exists).
\citet{ADFG:2010} show, in a quite elaborate proof, that
Algorithm~\ref{alg:llrfsyn_gen} is complete in this case, and will
find an \adfgllrf of \emph{minimal depth}, if one exists, \ie, it is
equivalent to using a procedure that synthesise an optimal \adfg-\qlrf
similar to that of \bg-\qlrfs.

\begin{theorem}[\citealp{ADFG:2010}]
  There is a polynomial-time algorithm for finding an \adfgllrf of
  minimal depth, if one exists, for a given rational \mlc loop
  $\transitions_1,\ldots,\transitions_k$.
\end{theorem}

\begin{example}
  Let us demonstrate the algorithm on the \mlc loops of
  Example~\ref{ex:llrf:intro} using the above algorithm for
  \adfg-\qlrfs.
  \procsyn is called with $\transitions_1,\transitions_2$, and then,
  in the first iteration of the while loop, at
  Line~\ref{alg:llrf:qlrf} it finds the non-trivial \adfg-\qlrf
  $\rho_1(x_1,x_2)=x_1$ that ranks $\transitions_1$, which is then
  eliminated at Line~\ref{alg:llrf:red}.
  In the next iteration, at Line~\ref{alg:llrf:qlrf} it finds the
  non-trivial \adfg-\qlrf $\rho_2(x_1,x_2)=x_2$ that ranks
  $\transitions_2$, which is then eliminated at
  Line~\ref{alg:llrf:red}.
  Since both paths were eliminated, we exit the while loop and arrive
  at Line~\ref{alg:llrf:ret} with the \adfgllrf $\tuple{x_1, x_2}$,
  which is not weak.
  Applying \procsyn to the \mlc loop of Example~\ref{ex:mlc:0} fails
  in the first iteration, as in the case of \bgllrfs, because
  $\transitions_1,\transitions_2$ does not have an \adfg-\qlrf.
\end{example}

As for the upper bound on the depth of \adfgllrfs, \citet{ADFG:2010}
show that it is $\min(n,k)$. This means that for \slc loops,
\adfgllrfs have the same power as \lrfs since $\min(n,1)=1$ is an
upper bound on the depth of the \adfgllrf in this case.

The problem of deciding existence of an \adfgllrf of a given depth is
simply solved by bounding the number of iterations of the while-loop
in Algorithm~\ref{alg:llrfsyn_gen}.

The problem of finding an \adfgllrf when starting from a polyhedral set of
initial state $\poly{S}_0$, and that of general \cfgs are the
same as in the case of \bgllrf. The difference is only in the kind of
\qlrf that we infer.

\begin{remark}
  Let us change the algorithm of \adfgllrf as described above, to
  require the \adfg-\qlrf to be non-negative only on the transitions
  under considerations instead of all transitions, but still work at
  the level of paths.
  We get a new kind of \llrfs that are weaker than \bgllrfs and
  stronger than \adfgllrfs. The definition would be like \bgllrfs, but
  requires each path to be completely ranked by some $\rho_i$.
  We believe that this definition of \llrfs has not been used by
  \citet{ADFG:2010}, despite of being more intuitive, because they
  wanted the \llrfs to satisfy additional properties that would allow
  them to construct a bound on the number of execution steps.
\end{remark}

\subsection{\bmsllrfs}
\label{sec:llrf:bms}

The next type of \llrfs is due to~\citet{BMS05a}, which is more
general than~\adfgllrfs and not comparable to \bgllrf~(\ie, there are
loops that have one kind of \llrf but not the other).

\begin{definition} 
\label{def:bms:llrf}
Given an \mlc
$\transitions_1\cup\cdots\cup\transitions_k\subseteq\rats^{2n}$, we
say that $\tau=\tuple{\rho_1,\dots,\rho_d}$ is a \bmsllrf (of depth
$d$) for the loop, if for every $\transitions_\ell$ there is
$1 \le i \le d$ such that the following hold for any
$\vec{x}'' \in \transitions_\ell$
\begin{alignat}{2}
 \forall j < i \ .\   && \diff{\rho_j}(\vec{x}'') &\ge 0 \,, \label{eq:bms:llrf1}\\
                      && \rho_i(\vec{x}) &\ge 0          \,, \label{eq:bms:llrf2}\\
                      && \diff{\rho_i}(\vec{x}'') &\geq 1\,. \label{eq:bms:llrf3} 
\end{alignat}
We say that $\transitions_\ell$ is \emph{ranked by} $\rho_i$ (for the minimal such $i$).
\end{definition}

Note that that it explicitly associates paths to components of the
\bmsllrf.  Recall that such association of paths and components was
implicit in \adfgllrf for \mlc loops (\ie, it is not explicit in
Definition~\ref{def:adfg:llrf}, but rather implied by the \adfg-\qlrf
algorithm of \citet{ADFG:2010}).

\begin{example}
\label{ex:bms-vs-bg}
Consider an \mlc loop $\transitions_1,\ldots,\transitions_4$ where:
\begin{equation}
\label{eq:loop:bms-vs-bg}
\begin{array}{rll}
\transitions_1= & \set{ x \ge 0, & x'\le x-1,  y'=y,  z'=z } \\
\transitions_2=           & \set{ x \ge 0, z \ge 0, & x' \le x-1,  y'=y,  z' \le z-1 } \\
\transitions_3=           & \set{ y \ge 0, z \ge 0, & x'=x, y'\le y-1,  z'\le z-1 } \\
\transitions_4=           & \set{ y \ge 0, & x'=x,  y'\le y-1, z'=z } \\

\end{array}
\end{equation}
It has the \bmsllrf $\tuple{x,y}$, but it has no \bgllrf, and thus no
\adfgllrf, due to the simple fact that there is no linear function
that is non-negative on all enabled states, and thus we cannot find a
corresponding \bg-\qlrf.
On the other hand, the loop of Example~\ref{ex:kinds-of-lrfs-1} has a
\bgllrf but not a \bmsllrf. This shows that these two kinds of \llrfs
have different power.
The loop of Example~\ref{ex:mlc:0} has the \bmsllrf $\tuple{x_1,x_2}$,
but not an \adfgllrf nor a \bgllrf.
The loop of Example~\ref{ex:llrf:intro} has the \bmsllrf
$\tuple{x_1,x_2}$, which is also an \adfgllrf and a \bgllrf.
\end{example}

\begin{definition}
\label{def:bms:qlrf}
Let $\transitions_1,\cdots,\transitions_k$ be an \mlc loop. We say that
an affine linear function $\rho$ is a \bms-\qlrf for
$\transitions'_1\cup\cdots\cup\transitions'_k \subseteq \rats^{2n}$,
where $\transitions'_i\subseteq \transitions_i$, if the following holds
for all $\vec{x}''\in \transitions'_1\cup\cdots\cup\transitions'_k$:
\begin{alignat}{2}
  \diff{\rho}(\vec{x}'') &\ge 0 \label{eq:bms:qllrf1}
\end{alignat}
We say that it is \emph{non-trivial} if for at least one
$\transitions'_\ell$ it is an \lrf.
\end{definition}

Unlike \bg- and \adfg-\qlrfs, existence of an optimal \bms-\qlrf is
not guaranteed, because a nonzero conic combination of \bms-\qlrfs
$\rho_1$ and $\rho_2$ is not guaranteed to rank all paths ranked by
$\rho_1$ and $\rho_2$.

\begin{example}
Considering all paths of Loop~\eqref{eq:loop:bms-vs-bg}:
$\rho(x,y,z)=x$, $\rho(x,y,z)=y$, and $\rho(x,y,z)=z$ are all
\bms-\qlrfs. However combinations such as $x+y$, $x+z$ or $x+y+z$ are
not, since they do not rank any complete path.  \end{example}

\citet{BMS05a} provide a complete polynomial-time algorithm for
finding a non-trivial \bms-\qlrf
$\rho(\vec{x})=\vect{\rfcoeff}\vec{x}+\rfcoeff_0$ for a set of
transitions defined by a given \mlc loop
$\transitions_1,\ldots,\transitions_k$ that, in brief, works as
follows: it iterates over all paths, and in each iteration checks
whether there is a non-trivial \bms-\qlrfs that ranks the current path
$\transitions_j$.
This is done by setting a \lp problem (using Farkas' Lemma) requiring
all paths to entail $\diff{\rho}(\vec{x}'') \ge 0$, and
$\transitions_j$ to entails $\diff{\rho}_i(\vec{x}'') \ge 1$ and
$\rho_i(\vec{x})\ge 0$; any solution to this problem fixes
$(\vect{\rfcoeff},\rfcoeff_0)$, and thus $\rho$.
If no such path is found the algorithm fails. The runtime of the
algorithm is polynomial since it solves at most $k$ \lp problems of
polynomial size \wrt to the size of the input \mlc loop.

When this algorithm is used within Algorithm~\ref{alg:llrfsyn_gen},
once $\rho$ has been found at Line~\ref{alg:llrf:qlrf}, \citet{BMS05a}
eliminate the path $\transitions_j$ (\ie, the one that is completely
ranked by $\rho$).
This also means that the \bmsllrf is not weak.

It is easy to see that if Algorithm~\ref{alg:llrfsyn_gen} returns a
tuple $\tau$, in this case, then it is a \bms-\llrf, and, moreover,
completeness is guaranteed because:
\begin{inparaenum}[\upshape(1\upshape)]
\item it terminates, since in each
  iteration we eliminate at least one path; and
\item when it returns $\nollrf$, then there is indeed no \bmsllrf for
  the loop because it has found a subset of transitions for which
  there is no \bms-\qlrf (see Observation~\ref{obs:llrf:qlrf}).
\end{inparaenum}
The overall runtime is still polynomial since we have at most $k$
iterations, and each iteration requires polynomial time to find a
non-trivial \bms-\qlrf. However, this algorithm is not guaranteed to
return a \bmsllrf of minimal depth, since there is no optimal choice
for \bms-\qlrfs.

\begin{example}
Consider Loop~\eqref{eq:loop:bms-vs-bg}. In the first iteration we could
use the \bms-\qlrf $\rho(x,y,z)=x$ to eliminate the paths
$\transitions_1$ and $\transitions_2$, and in the second iteration we
could use the \bms-\qlrf $\rho(x,y,z)=y$ to eliminate the remaining
paths $\transitions_3$ and $\transitions_4$. This results in the
\bmsllrf $\tuple{x,y}$.
Note that since there is no optimal \bms-\qlrf, this choice will
affect the length of the final \bmsllrf.
For example, if in the first iteration we choose the \bms-\qlrf
$\rho(x,y,z)=z$, we eliminate paths $\transitions_2$ and
$\transitions_3$; but then there is no single \bms-\qlrf that
eliminates both paths $\transitions_1$ and $\transitions_4$, so we
have to use $\rho(x,y,z)=x$ to eliminate $\transitions_1$ and
$\rho(x,y,z)=y$ to eliminate $\transitions_4$. This results in the
\bmsllrf $\tuple{z,x,y}$ which has a different length.
\end{example}

\begin{theorem}[\citealp{BMS05a}]
There is a polynomial-time algorithm for finding a \bmsllrf, if one exists,
for a rational \mlc loop.
\end{theorem}

Let us now consider the integer case.
First observe that a complete algorithm for synthesising \bms-\qlrfs
for $\intpoly{\transitions_1},\ldots,\intpoly{\transitions_k}$ can be
done by applying the one of the rational case on
$\inthull{(\transitions_1)},\ldots,\inthull{(\transitions_k)}$. Then,
the following observation helps us to adapt the overall algorithm for
rational loop to handle integer loops.

\begin{observation}[\citealp{BG:15:CAV}]
The integer \mlc loop $\intpoly{\transitions_1},\ldots,\intpoly{\transitions_k}$ has a
\bmsllrf of depth $d$, if and only if
$\inthull{(\transitions_1)},\ldots,\inthull{(\transitions_k)}$ has a
\bmsllrf of depth $d$.
\end{observation}

Using this observation, synthesising \bms-\qlrfs for the integer \mlc
loop $\intpoly{\transitions_1},\ldots,\intpoly{\transitions_k}$ can be
done by applying the algorithm of the rational case on
$\inthull{(\transitions_1)},\ldots,\inthull{(\transitions_k)}$.
Completeness is guaranteed since we eliminate a complete path in each
iteration, and thus all paths remain integral through the iterations
of the while-loop. The runtime is exponential since computing the
integer hull is exponential.

\begin{theorem}
There is an exponential-time algorithm for finding a \bmsllrf, if one
exists, for an integer \mlc loop
$\intpoly{\transitions_1},\ldots,\intpoly{\transitions_k}$.
\end{theorem}

\citet{BG:15:CAV} show that the corresponding \emph{decision problem}
for integer loops is \conpc, which results from a similar
characterisation of Theorem~\ref{thm:nollrf} for the case of \bgllrfs
and facilitates the construction of witnesses against the existence of
a \bms-\qlrf for a subset of the transitions.

An upper bound on the depth of \bmsllrfs is clearly given by $k$; the
number of paths.  Moreover, \citet{BG:15:CAV} show that this bound is
tight, \ie, there are $k$-path loops for which we need $k$
components. Moreover, they show that it is possible for a loop to have
\llrfs of all variants that we have seen so far, but such that the
minimal depths differ.

\begin{example} 
\label{ex:differentDim}
Consider an \mlc loop specified by the following paths
\begin{center}
\scalebox{\ifdefined\techreport 1.0 \else 0.97\fi}{$%
\begin{array}{r@{}l@{}l@{}l@{}l@{}l@{}l@{}l@{}l}
\multirow{2}{*}{$\transitions_1 = \Big\{$} 
               & r \ge 0, ~~~~~ 
               & %
               & t \ge 0, ~~~~~
               & x \ge 0, ~~~~~
               & %
               & z \ge 0, ~~~~~
               & w \ge 0, 
               & \multirow{2}{*}{$\Big\}$} \\[-0.5ex]
 & r' < r, & & t' < t, & & & & & \\
\multirow{2}{*}{$\transitions_2 = \Big\{$} 
               & r \ge 0, 
               & s \ge 0, ~~~~~
               & t \ge 0, 
               & x \ge 0,
               & %
               & z \ge 0, 
               & w \ge 0, 
               & \multirow{2}{*}{$\Big\}$}\\
 & r' = r, & s' < s, & t' < t, & & & & \\               
\multirow{2}{*}{$\transitions_3 = \Big\{$} 
               & r \ge 0, 
               & s \ge 0, 
               & t' = t,
               & x \ge 0, 
               & %
               & z \ge 0, 
               & w \ge 0, 
               & \multirow{2}{*}{$\Big\}$}\\
& r'=r, & s'=s, &  & x' < x, & & & \\
\multirow{2}{*}{$\transitions_4 = \Big\{$} 
               & r \ge 0, 
               & s \ge 0, 
               & t' = t,
               & x \ge 0, 
               & y \ge 0, ~~~~~
               & z \ge 0,  
               & w \ge 0, 
               & \multirow{2}{*}{$\Big\}$}\\[-0.5ex]
 & r' = r, & s' = s, & & x' = x, & y' < y, & z' < z, & \\
\multirow{2}{*}{$\transitions_5 = \Big\{$} 
               & r \ge 0, 
               & s \ge 0, 
               & t' = t,
               & x \ge 0, 
               & y \ge 0, 
               & z \ge 0,  
               & w \ge 0, 
               & \multirow{2}{*}{$\Big\}$}\\[-0.5ex]
& r' = r, & s' = s, & & x' = x, & y' < y, & z' = z, & w' < w \\
\end{array}$%
}
\end{center}
where, for readability, we use $<$ for the relation ``smaller at least
by $1$''.
This loop has the \bmsllrf $\tuple{t, x, y}$, which is neither a
\bgllrf or \adfgllrf because $t$ is not lower-bounded on all the
paths.
Its shortest \bgllrf is of depth $4$, \eg, $\tuple{r, s, x, y}$, which
is not an \adfgllrf because $y$ is not lower-bounded on all the
paths. Its shortest \adfgllrf is of depth $5$, \eg,
$\tuple{r, s, x, z, w}$.
This reasoning is valid for both integer and rational
variables.
\end{example}

Since Algorithm~\ref{alg:llrfsyn_gen} does not return a \bmsllrf of
minimal depth, \citet{BG:15:CAV} study the complexity of finding a
\bmsllrf that satisfies a given bound on the depth.

\begin{theorem}[\citealp{BG14}]
  Deciding whether there is a \bmsllrf of depth $d$ for a rational
  loop $\transitions_1,\ldots,\transitions_k$, is an \npc problem, and
  for an integer loop
  $\intpoly{\transitions_1},\ldots,\intpoly{\transitions_k}$, is a
  $\Sigma^P_2$-complete problem.%
  \footnote{The class $\Sigma^P_2$ is the class of decision problems
    that can be solved by a standard, non-deterministic computational
    model in polynomial time assuming access to an oracle for an
    \npc problem. \ie,
    $\Sigma^P_2 = \mbox{NP}^{\mbox{\scriptsize NP}}$.  This class
    contains both \np and \conp, and is likely to differ from them both
    (this is an open problem).}
\end{theorem}

The problem of finding a \bmsllrf when starting from a polyhedral set
of initial states $\poly{S}_0$, and that for general \cfgs, could be
addressed as in the case of \bgllrf. The difference is only in the
kind of \qlrf that we infer.

\subsection{\texorpdfstring{\mlrfs}{M(Phi)RFs}}
\label{sec:mlrfs}

An interesting special case of \llrfs is \emph{multiphase-linear
  ranking functions} (\mlrfs), which are defined as follows.

\begin{definition}[\mlrf]
\label{def:mlrf}
Given an \mlc loop
$\transitions_1,\ldots,\transitions_k\subseteq\rats^{2n}$, we
say that $\tau=\tuple{\rho_1,\dots,\rho_d}$ is an \mlrf (of depth $d$)
for the loop, if for every
$\vec{x}'' \in \transitions_1\cup\cdots\cup\transitions_k$ there is an
index $i$ such that:
\begin{alignat}{2}
\forall j \leq i \ .\   && \diff{\rho_j}(\vec{x}'') & \geq 1 \,, \label{eq:mlrf:1}\\
                        && \rho_i(\vec{x}) &\ge 0          \,    \label{eq:mlrf:2}
\end{alignat}
We say that $\vec{x}''$ is \emph{ranked by} $\rho_i$ (for the minimal
such $i$).
\end{definition}

When compared to \llrfs as in Definition~\ref{def:llrf}, the
difference is that all components $\rho_j$, with $j<i$ are decreasing
rather than non-increasing.
It is easy to see that this definition, for $d=1$, means that $\rho_1$
is an \lrf, and for $d>1$, it implies that $\rho_1$ is always
decreasing; as long as $\rho_1(\vec{x}) \ge 0$, transition $\vec{x}''$
must be ranked by $\rho_1$, and when $\rho_1(\vec{x}) < 0$,
$\tuple{\rho_2,\dots,\rho_d}$ becomes an \mlrf for the rest of the
execution. This agrees with the intuitive notion of ``phases.''

\begin{example}
Consider the following loop:
\begin{equation}
\label{eq:loop-xyz}
\while~(x \ge -z)~\wdo~x'=x+y,\; y'=y+z,\; z'=z-1
\end{equation}
Clearly, the loop goes through three phases --- in the first, $z$
descends, while the other variables may increase; in the second (which
begins once $z$ becomes negative), $y$ decreases; in the last phase
(beginning when $y$ becomes negative), $x$ decreases.  Note that since
there is no lower bound on $y$ or on $z$, they cannot be used in an
\lrf; however, each phase is clearly finite, as it is associated with
a value that is non-negative and decreasing during that phase. In
other words, each phase is linearly ranked. Formally, this loop has
the \mlrf $\tuple{z+1,y+1,x}$.
\end{example}

\begin{example}
Some loops have multiphase behaviour which is not so evident as in the
last example. Consider the following loop
\begin{equation}
\while~(x\ge 1,\; y\ge 1,\; x\ge y,\; 4 y\ge  x)~\wdo~x'=2x,\; y'=3y
\end{equation}
It has the \mlrf $\tuple{x-4y , x-2y, x-y}$.
\end{example}

\begin{definition}
\label{def:mlrf:qlrf}
Let $\transitions_1,\ldots,\transitions_k$ be an \mlc loop. We say that
an affine linear function $\rho$ is an \mphi-\qlrf for
$\transitions'_1\cup\cdots\cup\transitions'_k \subseteq \rats^{2n}$,
where $\transitions'_i\subseteq \transitions_i$, if the following holds
for all $\vec{x}''\in \transitions'_1\cup\cdots\cup\transitions'_k$:
\begin{align}
 \diff{\rho}(\vec{x}'') \ge 1 \label{eq:mlrf:qlrf}
\end{align}
We say that it is \emph{non-trivial} if, in addition,
$\rho(\vec{x}) \ge 0$, for at least one
$\vec{x}'' \in \transitions'_1\cup\cdots\cup\transitions'_k$.
\end{definition}

Unlike \bg- and \adfg-\llrfs, the existence of optimal \mphi-\qlrf is
not guaranteed because a non-zero conic combination of \mphi-\qlrfs
$\rho_1$ and $\rho_2$ is not guaranteed to rank all transitions ranked
by $\rho_1$ and $\rho_2$.

A polynomial-time algorithm for synthesising a \mphi-\qlrf
$\rho(\vec{x})=\vect{\rfcoeff}\vec{x}+\rfcoeff_0$ can be as follows:
\begin{enumerate}
\item Set up an \lp problem $\poly{S}_d$ (resp. $\poly{S}_p$), using
  Farkas' Lemma, requiring all paths to imply
  $\diff{\rho}(\vec{x}'')\ge 1$ (resp. $\rho(\vec{x}) \le 0$);  and
\item Choose a point $(\vect{\rfcoeff},\rfcoeff_0)$ from $\poly{S}_d$
  that is not in $\poly{S}_p$, which can be done by iterating over the
  inequalities $\vect{a}_i\vec{x}'' \le b_i$ of $\poly{S}_p$, and
  picking a point from $\poly{S}_d \wedge \vect{a}\vec{x}'' > b$ if it
  is not empty.
\end{enumerate}
Incorporating such a procedure at Line~\ref{alg:llrf:qlrf} of
Algorithm~\ref{alg:llrfsyn_gen}, and eliminating all transition for
which $\rho(x) > 0$ at Line~\ref{alg:llrf:red}, we obtain a sound
procedure for synthesising \mlrfs for \mlc loops, however completeness
in not guaranteed since the algorithm might not terminate. Note that
the \mlrf we build is not weak.

Unlike other kinds of \llrfs, that we have seen in the previous
sections, there are almost no results on complexity and algorithmic
aspects of \mlrfs for \mlc loops. However, when fixing the depth $d$,
\citet{LeikeHeizmann15} and \citet{li2016} propose complete solutions
for \mlrfs over \reals. Both rely on the template-based approach, that
we have described at the beginning of Section~\ref{sec:llrf}, which
turns the requirements of Definition~\ref{def:mlrf}, \emph{for a fixed
  $d$}, into a set of existential constraints -- this gives us a
\pspace upper bound, since the existential theory of the reals can be
decided in polynomial space~\citep{Canny88}.

For \slc loops, \citet{BG17} show that the template-based approach,
for seeking an \mlrf for a fixed $d$, can be performed in polynomial
time by avoiding the generation of non-linear constraints. This is
done by showing that \mlrfs and a further subclass of \mlrfs called
\emph{nested ranking functions} (\nlrfs), that was introduced by
\citet{LeikeHeizmann15} and can be synthesised in polynomial time,
have the same power for \slc loops, \ie, an \slc loop has an \mlrf of
depth $d$ if and only if it has an \nlrf of depth $d$.

\begin{definition}[\nlrf]
\label{def:nested-mlrf}
Given an \slc loop $\transitions \subseteq\rats^{2n}$, we say that
$\tau=\tuple{\rho_1,\dots,\rho_d}$ is a \emph{nested ranking function}
(of depth $d$) for $\transitions$ if the following requirements are
satisfied for all $\vec{x}'' \in \transitions$:
\begin{alignat}{2}
\rho_d(\vec{x}) \ge 0\label{eq:nlrf:1} & &\\
\diff{\rho_i}(\vec{x}'') + \rho_{i-1}(\vec{x}) \ge 1  &\quad\quad\quad& \mbox{for all } i=1,\dots,d.\label{eq:nlrf:2}
\end{alignat}
where for uniformity we let $\rho_0(\vec{x}) = 0$.
\end{definition}

It is easy to see that an \nlrf is an \mlrf. Indeed, $\rho_1$ is
decreasing, and when it becomes negative $\rho_2$ starts to decrease,
\etc. In addition, the loop must stop by the time that the last
component becomes negative, since $\rho_d$ is non-negative on all
enabled states. Note that the above definition extends also to \mlc
loops.

\begin{example}
\label{ex:loop1}
Consider Loop~\eqref{eq:loop-xyz}. %
It has the \mlrf $\tuple{z+1,y+1,x}$ which is not nested because,
among other things, the last component $x$ might be negative, \eg, for the
state $x=-1, y=0, z=1$.
However, it has the \nlrf $\tuple{z+1,y+1,z+x}$.
\end{example}

The above example shows that there are \mlrfs which are not \nlrfs,
however, for \slc loops \citet{BG17} provide a procedure to construct
an \nlrf from a given \mlrf.

\begin{theorem}[\citealp{BG17}]
\label{thm:slc-nested}
If a rational \slc loop $\transitions \subseteq \rats^{2n}$ has an
\mlrf of depth $d$, then it has an \nlrf of depth $d$.
\end{theorem}

This gives us a complete polynomial-time procedure to determine
whether a given \slc loop $\transitions$ has an \mlrf, which is done by
synthesising an \nlrf $\tau=\tuple{\rho_1,\dots,\rho_d}$, where
$\rho_i(\vec{x})=\vect{\rfcoeff}_i\vec{x}+\rfcoeff_{i,0}$, as
follows:
\begin{enumerate}
\item Set a \lp problem (using Farkas' Lemma) requiring $\transitions$
  to imply (\ref{eq:nlrf:1},\ref{eq:nlrf:2}), which generates a set of
  linear constraints over the variables
  $(\vect{\rfcoeff}_i,\rfcoeff_{i,0})$ and some other variables for
  the Farkas' coefficients; and
\item Any solution of this \lp problem fixes values for
  $(\vect{\rfcoeff},\rfcoeff_0)$ and thus define $\tau$. Moreover, if
  there is no solution, then $\transitions$ does not have an \nlrf.
\end{enumerate}
This give us the following theorem.

\begin{theorem}
There is a polynomial-time algorithm that, given an \slc loop
$\transitions$ and a depth-bound $d$, determines whether a depth-$d$
\mlrf exists for $\transitions$ and finds its coefficients if one
exists.
\end{theorem}

\citet{BG17} also show that, for the class of \slc loop, \nlrfs have
the same power as \llrfs of Definition~\ref{def:llrf}, and thus for
\llrfs, too, we have a complete solution in polynomial time (over the
rationals).

\begin{theorem}[\citealp{BG17}]
\label{thm:llrf2mlrf}
If $\transitions$ has an \llrf of depth $d$, it has an \mlrf of depth
$d$.
\end{theorem}

We next consider integer loops. The following results are by
\citet{BG17}.

\begin{theorem}
$\intpoly{\transitions}$ has an \mlrf of depth $d$ if and only if
$\inthull{\transitions}$ has an \mlrf of depth $d$ (as a rational
loop).
\end{theorem}

This gives us a solution of exponential time complexity, because
computing the integer hull requires exponential time. However, it is
polynomial for the cases in which the integer hull can be computed in
polynomial time~\citep[Section~4]{BG14}.
The next theorem shows that the exponential time complexity is
unavoidable for the general case (unless $\p=\np$).

\begin{theorem} 
\label{thm:coNP-complete}
Existence of an \mlrf of depth $d$ for a given integer \slc loop is a
\conpc problem.
\end{theorem}

We are not aware of a computable upper bound on the depth of \mlrf,
given the loop.  \citet{BG17} show that such a bound cannot depend
only on the number of variables or paths of the loop, but must also
take account of the coefficients and the constants used in the
inequalities defining the loop.

\begin{example}
For an integer $B>0$, \citet{BG14} show that the following \slc loop
\[
\while\; (x\ge 1,\; y\ge 1,\; x\ge y,\; 2^{B} y\ge  x) \; \wdo \; x'=2x,\; y'=3y
\]
needs at least $B+1$ components in any \mlrf, and that this bound
$B+1$ is tight and confirmed by the \mlrf
$\tuple{x-2^By , x-2^{B-1}y , x-2^{B-2}y , \ldots, x-y}$.
\end{example}

\citet{BG17} also discuss the consequence of existence of \mlrfs on
the number of iterations that an \slc loop can make, and show that it
is actually linear in the input values.

\begin{theorem}
An \slc loop that has an \mlrf terminates for an input $\vec{x}_0$ in a
number of iterations bounded by $O(\Vert \vec{x}_0 \Vert_{\infty})$.
\end{theorem}

In a subsequent work, \citet{Ben-AmramDG19} attempted to solve the
general \mlrf problem for \slc loops, \ie, without a given bound on
the depth. Although the problem remains open, this attempt yielded
several important observations. They first observe that if an \slc
loop has an irredundant \mlrf of depth $d$, then it has one of the
same depth in which the last component $\rho_d$ is non-negative over
all enabled states of $\transitions$. Using this observation they
propose an algorithm that builds a \mlrf recursively starting from the
last component, which always find an \mlrf if one exists, however, it
might not terminate in other cases. The algorithm can also, in some
cases, find witnesses for non-termination when it fails to find a
\mlrf.

\citet{BDG21} demonstrate the usefulness of the algorithm described
above for studying properties of \slc loops, in particular, it is used
to characterise kinds of \slc loops for which there is always an
\mlrf, if the loop is terminating, and thus has linear run-time
complexity. This is done for \emph{octagonal relations} and
\emph{affine relations with the finite-monoid property}---for both
classes, termination has been proven decidable~\citep{BIK14} (see
Section~\ref{sec:dec:slc:oct}). In addition, they provide a bound on
the depth of \mlrfs for these classes of \slc loops, which can be used
to make the above algorithm complete (related to the bound $5^{2n}$
discussed in Section~\ref{sec:dec:slc:oct}).

The problem of finding an \mlrf when starting from a polyhedral set of
initial states $\poly{S}_0$, and that for general \cfgs, are the same
as in the case of \bgllrf. The difference is only in the kind of \qlrf
that we infer.

\subsection{Other Approaches for \llrfs}
\label{sec:llrf:related}

The earliest work that we know that addressed the generation of \llrfs
is by \citet{Feautrier92.2}, where they are called
\emph{multidimensional schedules}.
\citet{CS02} use \lp methods based on the computation of polars.  The
\llrf is not constructed explicitly but can be inferred from the
results of their algorithm.
\citet{BMS05c} introduced the notion of \emph{polyranking principle}
which is based on lexicographic ranking functions where each component is an \nlrf of depth at
most $2$.
In another work, \citet{BMS05d} considered \mlc loops with polynomial
transitions and the synthesis of lexicographic-polynomial ranking
functions. All the works by this group actually tackle an even more
complex problem, since they also search for \emph{supporting
  invariants}, based on the transition constraints and on given
preconditions.

\citet{UrbanM14} compute lexicographic ranking functions using
abstract interpretation. %
\citet{GMR15} compute \llrfs, essentially \adfgllrfs, for complete
programs, including a computation of invariants. Their method is
designed to improve over both the efficiency and the effectiveness of
previous methods, such as \citet{ADFG:2010} and \citet{GZ2010}.
\citet{YLS2021} suggest an approach to the problem of bounding the
depth of \mlrfs.
\citet{ZLW2016} consider a type of \llrfs that combines \bmsllrfs with
the idea of ``phases''. It is a special case of general \llrf, but one
for which we have an (exponential) complete algorithm.

\section{Other Types of Ranking Functions}
\label{sec:other_rfs}

Another type of ranking function that may be interesting in the
context of linear programs is \emph{piecewise linear ranking
  functions} \citep{Urban13}.
We are not aware of complexity results for this type of functions, for
linear-constraint loops like the ones we address in this \survey.
Also beyond the scope of this \survey are \emph{polynomial ranking
  functions} \citep{NOW2020,SWYZ2013,CXYZZ07,Cousot05}.

\citet{ZhuK24} develop a synthesis algorithm for (lexicographic)
polynomial ranking functions that is complete relative to the theory
of Linear Integer/Real Rings~\citep{KincaidKZ23}. This allows for the
analysis of simple loops containing nonlinear constraints in their
description, significantly generalising the class of \slc loops.

\citet{LeikeHeizmann15} present a template-based approach to
synthesise many types of ranking functions, including \adfg-\llrfs,
piecewise-linear ranking functions and others.

\citet{DomenechGG19} use control-flow refinement to transform programs
with complex control-flow into equivalent simpler ones, which makes it
possible, for example, to use \lrfs instead \llrfs for proving
termination. For example, the loop on the left would be translated
into the loops on the right:
\begin{center}
\begin{minipage}{4.5cm}
\begin{lstlisting}
while(x >= 1)           
  if (y <= z-1) y++;  
  else x--;                
\end{lstlisting}
\end{minipage}
\vrule
\begin{minipage}{6.5cm}
\begin{lstlisting}
while(x >= 1 && y <= z-1) y++; 
while(x >= 1) x--;
$~$
\end{lstlisting}
\end{minipage}
\end{center}
The one on the left requires the \llrf $\tuple{z-y,x}$, while the one
on the right requires the \lrfs $z-y$ and $x$.
There are also examples that do not admit any kind of ranking function
(from those discussed in this \chp), while after the refinement they
do admit \llrfs.
\citet{BorrallerasBLOR17} develop a technique for proving
(conditional) termination, which is based on incrementally finding
conditional \lrfs for the different parts of the program.

Polynomial interpretations are used to prove the termination of term
rewriting systems, which are out of the scope of this \survey. They
are polynomials assigned to each function symbol such that they
decrease with every derivation. While they may seem similar to ranking
functions, their underlying problems are computationally harder. For
example, the problem of deciding whether a single rewriting rule
admits a linear interpretation is
undecidable~\citep{MitterwallnerMT24}.

\chapter{Transition Invariants and Size-Change Termination}
\label{chp:dti}

A key challenge of using ranking functions for termination proofs is
that it is not always possible to find a function from a tractable class
such as \llrfs, that strictly decreases with every
single step of a program's execution.
Instead of proving a decrease at every step, we can
resort to techniques that prove absence of infinite executions by
showing that in any infinite trace, there must be a sub-trace that
violates a well-foundedness property. These techniques often rely on
the use of \emph{Ramsey's theorem}.
This application of Ramsey's theorem was first
applied by \citet{Geser90}, and was later applied, in various
forms, by several other researchers, including
\citet{DoornbosK98,LJB01,DershowitzLSS01,CodishGBGV03,PodelskiR04}.
\citet[Page 2]{BlassY08} provide a brief history of this use of
Ramsey's theorem.

In this \chp, we will discuss \emph{disjunctive well-founded
  transition invariants} (\dti), a technique for proving termination
that applies Ramsey's theorem. This method has primarily emerged in
the context of linear-constraint programs.
We further present classes of linear-constraint programs for which
\dti provide a complete criterion for termination---specifically,
\dti based on \lrfs.
These classes (such as size-change terminating programs,
monotonicity-constraint programs, \etc) have been studied from
different viewpoints, but our presentation here aims to show how they
all fall under the \dti approach.

\paragraph{Organisation of this \Chp.}
We start with an overview of transition invariants in
Section~\ref{sec:dti:ti}. We then discuss several classes of programs:
$\delta$-size-change-termination (Section~\ref{sec:dti:dsct}),
size-change-termination (Section~\ref{sec:dti:sct}),
$\delta$-size-change-termination for fan-in free programs
(Section~\ref{sec:dti:FIFdsct}), monotonicity constraints
(Section~\ref{sec:dti:mc}), and gap constraints
(Section~\ref{sec:dti:gap}). We also examine the relation to ranking
functions (Section~\ref{sec:dti:mc:rfs}), the relative power of \dtis
(Section~\ref{sec:dti:slc:power}), and finally provide an overview of
other related works (Section~\ref{sec:dti:slc:related}).

\section{Transition Invariants}
\label{sec:dti:ti}

Given a transition relation $T \subseteq S\times S$, we define
$T^i = T^{i-1} \circ T$, for $i\ge1$, where $T^0 \subseteq S\times S$
is the identity relation and $\circ$ is the composition operation of
transition relations (see Section~\ref{sec:programs}).
The transitive closure of a relation $T$ is defined as
$T^+=\cup_{i\ge 1}T^i$.

The relation
$T^+$ provides crucial information about reachability: a
computation under  $T$ that starts in $s$ reaches $s'$ if and only if
$(s,s') \in T^+$.
This concept forms the basis for numerous applications in static
analysis and model checking, especially in termination analysis.
For termination with respect to an initial set of states, as we have
done previously, we assume that $T$ has been reduced to the set of
reachable states and then study universal termination.

Instead of working directly with $T^+$ (which is not always computable, or even representable in any
useful form),
termination tools resort to approximations known as
\emph{transition invariants}.

\begin{definition}[\citealp{PodelskiR04}]
\label{def:ti}
We say that $T_I \subseteq S\times S$ is a \emph{transition invariant} (\ti)
for $T \subseteq S\times S$, if and only if $T^+ \subseteq T_I$.%
\footnote{\citet{PodelskiR04} require 
  $T^+ \subseteq T_I \cap (\reach{T}{S_0}\times\reach{T}{S_0})$, because they
  consider a set $S_0$ of initial states.}
\end{definition}

Recall that a binary relation $T\subset S\times S$ is called
\emph{well-founded} if there is no infinite sequence $s_0,s_1,\ldots$
such that $(s_i,s_{i+1})\in T$ for all $i\ge0$, and that if $T$ is the
transition relation of a program, well-foundedness of $T$ is
equivalent to (universal) termination.

\begin{definition}
\label{def:dti}
Given a transition relation $T\subseteq S\times S$, and sets
$\tuple{T_1 , \dots , T_d}$ of transitions such that
$T_i \subseteq S\times S$, we say that $\tuple{T_1,\dots,T_d}$ is a
\emph{disjunctively well-founded transition invariant} (\dti) for $T$ if
$T^+ \subseteq {T_1 \cup\dots\cup T_d}$, and for each $1\le i\le d$,
$T_i$ is well-founded.
\end{definition}

\begin{theorem}[\citealp{PodelskiR04}]
\label{thm:dti}
$T \subseteq S\times S$ has a \dti if and only if $T$ is well-founded.
\end{theorem}

\begin{proof}
$\cbox{(\Rightarrow)}$ Assume that $T$ has the \dti $\tuple{T_1,\dots,T_d}$
and suppose, for a contradiction, that there is an infinite sequence
$s_0,s_1,\ldots$ such that $(s_i,s_{i+1})\in T$ for all $i\ge 0$.
For every pair $(s_i,s_j)$ with $i<j$ we must have $(s_i,s_j)\in T_k$
for some $T_k$.
Associating one such $k$ to the pair $(i,j)$ we obtain a colouring of
the infinite complete graph with $d$ colours; by Ramsey's theorem,
there is an infinite monochromatic clique.
This constitutes an infinite subsequence $s_{i_0},s_{i_1},\ldots$
where $(s_{i_j},s_{i_{j+1}}) \in T^k$ for all $j\ge 0$,
contradicting the well-foundedness of $T_k$.
$\cbox{(\Leftarrow)}$ This direction is trivial: if $T$ is well-founded then $T^+$ is
a \dti.
\end{proof}

To make \dti a practical tool for proving termination we need:
\begin{enumerate}
\item An effective way to show that the disjuncts are well-founded;
  and
\item An effective way to show that the disjuncts  cover the transitive closure
  of the transition relation.
\end{enumerate}
This clearly depends, among other things, on the state space $S$ and
on the way $T$ and each $T_i$ are specified.
In what follows we focus on \dtis for \cfgs, and thus assume that the
transition relation $T$ corresponds to a (linear-constraint) \cfg with
variables ranging over $\numdom \in \{\reals, \rats, \ints\}$.

\begin{remark}
\label{rem:dti:rest}
When a transition relation $T$ originates from a \cfg, we can relax
the requirements of Definition~\ref{def:dti} such that instead of
computing a \dti that over-approximates $T^+$, we compute one that
over-approximates
$T^+|_C = \{ ((\ell,\vec{x}),(\ell,\vec{x}')) \in T^+ \mid \ell \in C
\}$ where $C$ is any feedback vertex set (\ie, removing these vertexes
results in an acyclic graph).
This is true because $T^+|_C$ is transitively closed, and $T^+$ is
well founded if and only if $T^+|_C$ is well founded (we can
easily extend a \dti for $T^+|_C$ to a \dti for $T$).
If the \cfg originates from a structured program, $C$ could be the set
of locations corresponding to loop heads.
\end{remark}

In what follows we assume a given \cfg
$P=(V,\numdom,L,\ell_{0},E)$, where
$\numdom \in \{\reals, \rats, \ints\}$, and use $T_P$ to refer to the
corresponding transition relation.
In this context, and for transition relations specified by linear
constraints in general, it is common to restrict the \dti to a form in
which each $T_i$ is a well-founded convex polyhedron, \ie, a
terminating \slc loop.

\begin{definition}
\label{def:dti:poly}
$\tuple{T_1,\ldots,T_k}$ is a \emph{polyhedral} \dti for $T_P$ if it is
a \dti and each $T_i$ is of the form
$T_i=\{ ((\ell,\vec{x}),(\ell,\vec{x}')) \mid (\vec{x},\vec{x}') \in
\transitions \}$, where $\transitions$ is a convex
polyhedron and $\ell\in L$. We
sometimes write $T_i$ as $(\ell,\transitions,\ell)$.
\end{definition}

Intuitively, a polyhedral \dti is a termination proof that breaks the task of proving termination for
a complex program into a set of proofs for \slc loops.

There are \dti-based termination analysis
tools~\citep{LS97,CT99,AACGPZ08,SMP10}%
\footnote{They do not call them \dti, but they are conceptually the
  same.}. They work in two steps:
\begin{inparaenum}[\upshape(1\upshape)]
\item compute a candidate \dti $T_1\cup\cdots\cup T_d$ that over-approximates
  $T_P^+$, where each $T_i$ is polyhedral as in
  Definition~\ref{def:dti:poly}; and
\item check that for each $T_i=(\ell,\transitions,\ell)$, the \slc
  loop $\transitions$ is terminating by seeking a corresponding
  ranking function, \eg, \lrf.
\end{inparaenum}
This implies that $T_1\cup\ldots\cup T_d$ is a \dti.
\citet{CPR06} follows a different approach, and constructs a \dti
incrementally where each component is polyhedral, but has a restricted
form as in the following definition.

\begin{definition}
\label{def:lrfdti:poly}
A \emph{linear-ranking function based} \dti (\lrfdti for short), is a
polyhedral \dti as in Definition~\ref{def:dti:poly} where each
transition polyhedron $\transitions$ has an \lrf, specifically it
satisfies {$\rho_i(\vec{x}) \ge 0 \wedge \diff{\rho_i(\vec{x}'')}\ge 1$}
for some linear function $\rho_i$.
In what follows we use $T_{\rho_i}$ for the transition relation
relation
$\{ ((\ell,\vec{x}),(\ell,\vec{x}')) \mid \rho_i(\vec{x}) \ge 0,\; \diff{\rho_i(\vec{x}'')}\ge 1\}$ (the
location $\ell$ is not important, and will always be clear from the
context).
\end{definition}

The work of~\citet{CPR06} has several important observations that make
computing a \dti practical, and this paper was influential in
promoting the concept of \dti and the use of \slc loops as components
in a termination proof for a possibly complex program, relying (at
least in~\citet{CPR06}) on \lrfs, instead of using more complex
termination proofs such as \llrfs.
They describe a method, relying on a program transformation, to
compute an over-approximation of $T_P^+$ using off-the-shelf safety
checkers (such checkers are used to prove that a set of (error) states
is not reachable, and when they fail they usually provide a counter
example).
Unlike other algorithms in this \survey, this method is not complete
for the problem in any sense, but we describe it informally due to its
historical importance and as an illustration to how \dtis are used in
practice.  The rest of this subsection describes this method, while the following
subsections are independent of it.

Let us assume that during the execution we can non-deterministically
record the current state into (extra) program variables
$\pc_g,\vec{x}_g$, where $\pc_g$ is used to store the location and
$\vec{x}_g$ to store the value of the program variables $\vec{x}$.
Let us also assume that $\pc_g$ has a special value $\ell_\bot$ in the
initial state.
It is easy to see that when reaching a state $(\ell,\vec{x})$ and
$\pc_g\neq\ell_\bot$, it is guaranteed that
$((\pc_g,\vec{x}_g),(\ell,\vec{x})) \in T_P^+$.
Moreover, since the recording is done non-deterministically, the
opposite also holds: if $((\pc_g,\vec{x}_g),(\ell,\vec{x})) \in T_P^+$
then there is an execution that reaches the state $(\ell,\vec{x})$
where the recorded state is $(\pc_g,\vec{x}_g)$.
This means that state invariants of the program instrumented with this
recording mechanism induce transition invariants for the original
program, and thus we can use invariant inference tools to
over-approximate $T_P^+$.

At the level of a \cfg, this instrumentation can be done as follows.
First we add an extra program variable $\pc$, and for each
$(\ell_i,\transitions,\ell_j) \in E$ we add $\pc=i\wedge\pc'=j$ to
$\transitions$, \ie, variable $\pc$ simply tracks the location.
Next, we introduce a new set of \emph{ghost} variables
$\pc_g,\vec{x}_g$ (used to record a state), and split each edge
$(\ell_i,\transitions,\ell_j) \in E$ into two edges
$(\ell_i,\transitions_1,\ell_j)$ and $(\ell_i,\transitions_2,\ell_j)$
where:
\begin{inparaenum}[\upshape(1\upshape)]
\item
  $\transitions_1\equiv\transitions\wedge\pc_g'=\pc_g\wedge\vec{x}_g'=\vec{x}_g$,
  and;
\item $\transitions_2\equiv\transitions\wedge\pc_g'=\pc\wedge\vec{x}_g'=\vec{x}$.
\end{inparaenum}
The purpose of $\transitions_2$ is to non-deterministically record the
current state into $(\pc_g,\vec{x}_g)$.

The other observation of~\citet{CPR06} is that inferring a \dti can be
done using an off-the-shelf safety checker that is based on
\emph{counter example-guided abstraction refinement} approach~(\cegar).
We describe this in the next example.

\begin{example}
\label{ex:dti:1}
Let us consider the \cfg depicted in Figure~\ref{fig:cfg:1}, and we
start the execution at $\ell_0$ with
$\poly{S}_0 = \{x \ge 0, y \ge 0\}$.
Note that $\{\ell_1,\ell_3\}$ is a feedback vertex set (they
correspond to the loop heads of the program in
Figure~\ref{fig:cfg:1}).
Let us assume that the \cfg has been instrumented with the recording
mechanism as described above. Moreover, we add a new node $\errl$ that
represents an \emph{error location} that is not connected to the \cfg
yet.
We refer to the condition that allows us to move to $\errl$ as the
\emph{error condition}.

Next we will proceed iteratively, starting from an empty \dti, where
in each iteration:
\begin{inparaenum}[\upshape(1\upshape)]
\item we modify the \emph{error condition} (\ie, how $\errl$ is
  connected to the \cfg) to take into account the current \dti;
\item we use a safety checker to try to prove that $\errl$ is
  unreachable;
\item if we succeed, then as further explained below, this means that
  the current disjunction is indeed a \dti; otherwise, we use the
  counter example returned by the safety checker to add a new
  component $T_{\rho}$ to the disjunction, if possible, and repeat the
  process.
\end{inparaenum}

In the first step, since the current \dti is empty, we modify the \cfg
such that whenever $\ell_1$ (resp. $\ell_3$) is reached with
$\pc=\pc_g=1$ (resp. $\pc=\pc_g=3$), the execution can move to $\errl$
(\ie, we add corresponding edges with the corresponding condition).
These conditions simulate a situation where the execution visits
location $\ell_1$ (resp. $\ell_3$) at least twice.
Note that if $\errl$ is unreachable, it means that there are no loops
in the program and thus the empty disjunction is actually a valid \dti
because $\ell_1$ and $\ell_3$ form a feedback vertex set.
Applying a safety checker we get as a counter example an execution path
that passes through the nodes
$\ell_1, \ell_2, \ell_4$ and then $\ell_1$ again.
We treat this cycle as an \slc loop, namely:
$\{\pc = 1,\; \pc_g=1,\; x = x_g,\; x_g \ge 0,\; y_g\ge 0,\;
y=y_g-1,\; z_g=z\}$ (some invariants have been added).
Note that in this \slc loop, $(x_g,y_g,z_g)$ is the current state and
$(x,y,z)$ is the next state, and that it has the \lrf
$\rho_1(x_g,y_g,z_g)=y_g$.
This leads to adding
$T_1 \equiv \{\pc = 1,\; \pc_g=1, \; y_g \ge 0,\; y_g-1 \ge y\}$ to the
\dti.
The idea is that in the next iteration, this counter example, and
possibly others, are eliminated due to $T_1$.

In the second iteration, we refine the error condition to take $T_1$
into account, \ie, we allow moving from $\ell_1$ to $\errl$ if, in
addition to $\pc=\pc_g=1$, we have $y_g < 0$ or $y_g-1 < y$.
Applying a safety checker we get, as a counter example,
an execution path that passes through the nodes
$\ell_1, \ell_2, \ell_3, \ell_4$ and then $\ell_1$ again.
The \slc loop that correspond to this path is
$\{\pc = 1,\; \pc_g=1,\; x = x_g-1,\; x_g \ge 0,\; y_g\ge 0,\;
z_g=z\}$.
This leads to adding $T_2 \equiv \{\pc = 1,\; \pc_g=1,\; x_g \ge 0,\; x_g-1 \ge x\}$.

In the third iteration, we refine the error condition to take both
$T_1$ and $T_2$ into account. This means that we go to $\errl$ if, in
addition to $\pc=\pc_g=1$, we have both ($y_g < 0$ or $y_g-1 < y$) and
($x_g < 0$ or $x_g-1 < x$).
Applying a safety checker we get the following counter example (at
$\ell_3$):
$\{\pc=3,\; \pc_g=3,\; x = x_g,\; x_g \ge 0,\; y_g\le z_g,\;
y=y_g+1,\; z_g=z\}$.
It corresponds to looping at $\ell_3$, and it leads to adding
$T_3 \equiv \{\pc = 3,\; \pc_g=3,\; z_g-y_g \ge 0,\; z_g-y_g - 1 \ge z-y\}$.

In the forth iteration, we refine the error condition to take $T_3$
into account similarly to what we have done for $T_1$ and $T_2$ (this
time at $\ell_3$).
Now the safety checker succeeds in proving that $\errl$ is unreachable,
meaning that $T_1\cup T_2 \cup T_3$ is an invariant for the
instrumented \cfg (for $\ell_1$ and $\ell_3$), because otherwise there
must be an execution that leads to $\errl$, and thus a \dti for the
original \cfg.
\end{example}

\section{Wingspan of \lrf-\dti}

An easy observation is that \lrf-\dti{}s subsume \llrfs. This demonstrates the point that the \dti
approach breaks a complex termination proof into simple pieces.

Indeed, suppose that transition relation $T$ has the \llrf $\tuple{\rho_1,\dots,\rho_d}$.
Let $(s,s')\in T^+$. This means that there is a chain of transitions $(s=s_0,s_1),(s_1,s_2),\dots,
(s_{k-1},s_k=s')$, all in $T$. Each such transition is ranked by one of the $\rho_j$ (see Definition~\ref{def:llrf}). Let $m$ be the minimal such $j$. Then we have: $\diff{\rho_m}(s_i,s_{i+1}) \geq 0$
for all $i$, $\diff{\rho_m}(s_i,s_{i+1}) \geq 1$ for at least one $i$, and 
$\rho_m(s_0) \ge 0$ (since it is non-negative in at least one transition, and is decreasing throughout).
Thus $\rho_m$ ranks (as an \lrf) the transition $(s,s')$. It follows that $\{\rho_1,\dots,\rho_d\}$
constitutes a \lrf-\dti for $T$.

Next, we will describe a few types of programs (\ie, of linear-constraint \cfgs) for which \lrf-\dti{}s
provide a complete proof method for termination, and (in most of them) makes termination decidable.

\section{\texorpdfstring{$\delta$}{Delta}-Size-Change-Termination}
\label{sec:dti:dsct}

A $\delta$-Size-Change program (or \dsct program) is a \cfg where the
transition relations include only \emph{bound constraints} of the form
$y' \le x+\delta$, for state variables $x,y$ and $\delta\in \ints$; we
interpret such programs over the natural numbers (or assume $\ints$
and say that the constraints include $x\ge 0$, for every $x\in V$).
Note that $x$ and $y$ might be the same variable, \eg,
$x' \le x+\delta$, but the one on the left is primed and the other is
not.
The execution starts at $\ell_0$ with any values for the program
variables.

\begin{figure}[t]

\begin{center}

\begin{tikzpicture}[>=latex,line join=bevel,]

\begin{scope}[shift={(1,0.25)}]
  
  \node (l1) at (0,0) [inloc] {$\ell_{0}$};
  \node (l2) at (0,1.5) [loc] {$\ell_{1}$};
  
  \draw [tre,bend left=20] (l2) to[] node[tr,right] {\scalebox{0.8}{$\transitions_3$}} (l1);
  \draw [tre,bend left=20] (l1) to[] node[tr,left] {\scalebox{0.8}{$\transitions_5$}} (l2);
  \draw [tre] (l2) to[out=200,in=160,loop] node[tr,left] {\scalebox{0.8}{$\transitions_1$}}  (l2);
  \draw [tre] (l2) to[out=20,in=-20,loop] node[tr,right] {\scalebox{0.8}{$\transitions_2$}}  (l2);
  \draw [tre] (l1) to[out=20,in=-20,loop] node[tr,right] {\scalebox{0.8}{$\transitions_4$}}  (l1);
\end{scope}

\begin{scope}[shift={(7,1)}]
  \node [align=left, font=\ttfamily\footnotesize] at (0,0) {
\begin{minipage}{8.75cm}
\(
  \begin{array}{|@{}r@{\hskip 2pt}l@{}|}
    \hline
\transitions_1{:}&  \{x \ge 0, y\ge 0, z\ge 0, y' \le y-1, x' \le y+1,  \} \\
\transitions_2{:}&  \{x \ge 0, y\ge 0, z\ge 0, x' \le x-1, y' \le x-3 \} \\
\transitions_3{:}&  \{x \ge 0, y\ge 0, z\ge 0, z' \le x, z' \le y, x' \le x-1, y' \le y-1 \} \\
\transitions_4{:}&  \{x \ge 0, y\ge 0, z\ge 0, z' \le z-1, x' \le x, y' \le y \} \\
\transitions_5{:}&  \{x \ge 0, y\ge 0, z\ge 0, x' \le x, y' \le y \} \\
\hline                       
  \end{array}
  \)
\end{minipage}
};
\end{scope}

\end{tikzpicture}
\end{center}

\caption{A \cfg with \dsct transition relations.}
\label{fig:dsct:1}

\end{figure}
 
\begin{example}
Consider the \cfg depicted in Figure~\ref{fig:dsct:1}.
It is terminating, but it does not have an \llrf of any kind. This is
because a \qlrf cannot involve $x$ and $z$ due to $\transitions_1$,
and cannot involve $y$ due to $\transitions_2$.
This program, however, has an \lrfdti, as do terminating \dsct
programs in general, \eg, $T_{x} \cup T_{y} \cup T_{z}$ (using the
notation of Definition~\ref{def:lrfdti:poly}).
\end{example}

\dsct constraints are a special case of difference bound constraints, mentioned
in \Chp~\ref{chp:dec}, where we cited a decidability result for single-path loop with
such constraints. However, in this section we focus on \cfgs which are not a simple loop.
  
Next we state some properties of \dsct. For this, it is useful to view
a \dsct transition relation $\transitions$ as a weighted bipartite
graph $G_\transitions$.

\begin{figure}[t]

\begin{center}

\begin{tikzpicture}[>=latex,line join=bevel,]

\begin{scope}[shift={(0.6,1)}]

  \node (title) at (0.8,2.25) [scttitle] {$G_{\transitions_1}$};

  \node (x) at (0.25,1.75) [sctloc] {$x$};
  \node (y) at (0.25,1)   [sctloc] {$y$};
  
  \node (z) at (0.25,0.25) [sctloc] {$z$};
 
  \node (xp) at (1.25,1.75) [sctloc] {$x'$};
  \node (yp) at (1.25,1) [sctloc] {$y'$};
  \node (zp) at (1.25,0.25) [sctloc] {$z'$};

  \draw [scte] (x) to[] node[pos=0.3,sctw,above] {1} (yp);
  \draw [scte] (y) to[] node[pos=0.4,sctw,below] {-2} (xp);
  \draw [scte] (z) to[] node[sctw,above] {1} (zp);

  \draw[-] (0,0) -- (1.5,0) -- (1.5,2) -- (0,2) -- (0,0);
\end{scope}

\begin{scope}[shift={(2.6,1)}]

  \node (title) at (0.8,2.25) [scttitle] {$C_{\transitions_1}$};

  \node (x) at (0.25,1.75) [sctloc] {$x$};
  \node (y) at (0.25,1)   [sctloc] {$y$};
  \node (z) at (0.25,0.25) [sctloc] {$z$};
 
  \node (xp) at (1.25,1.75) [sctloc] {$x'$};
  \node (yp) at (1.25,1) [sctloc] {$y'$};
  \node (zp) at (1.25,0.25) [sctloc] {$z'$};

  \draw [scte] (x) to[] node[pos=0.3,sctw,above] {1} (yp);
  \draw [scte] (y) to[] node[pos=0.4,sctw,below] {-2} (xp);
  \draw [scte] (z) to[] node[sctw,above] {1} (zp);

  \draw [sctbe] (xp) to[] node[sctw,above] {0} (x);
  \draw [sctbe] (yp) to[] node[sctw,below] {0} (y);
  \draw [sctbe,bend left=15] (zp) to[] node[sctw,below] {0} (z);

  \draw[-] (0,0) -- (1.5,0) -- (1.5,2) -- (0,2) -- (0,0);
\end{scope}

\begin{scope}[shift={(4.6,1)}]

  \node (title) at (0.8,2.25) [scttitle] {$G_{\transitions_2}$};

  \node (x) at (0.25,1.75) [sctloc] {$x$};
  \node (y) at (0.25,1)   [sctloc] {$y$};
  
  \node (z) at (0.25,0.25) [sctloc] {$z$};
 
  \node (xp) at (1.25,1.75) [sctloc] {$x'$};
  \node (yp) at (1.25,1) [sctloc] {$y'$};
  \node (zp) at (1.25,0.25) [sctloc] {$z'$};

  \draw [scte] (x) to[] node[pos=0.5,sctw,below] {0} (xp);
  \draw [scte] (x) to[] node[pos=0.3,sctw,below] {-1} (yp);
  \draw [scte] (y) to[] node[pos=0.5,sctw,above] {0} (yp);
  \draw [scte] (z) to[] node[sctw,above] {1} (zp);

  \draw[-] (0,0) -- (1.5,0) -- (1.5,2) -- (0,2) -- (0,0);
\end{scope}

\begin{scope}[shift={(6.6,1)}]

  \node (title) at (0.8,2.25) [scttitle] {$C_{\transitions_2}$};

  \node (x) at (0.25,1.75) [sctloc] {$x$};
  \node (y) at (0.25,1)   [sctloc] {$y$};
  \node (z) at (0.25,0.25) [sctloc] {$z$};
 
  \node (xp) at (1.25,1.75) [sctloc] {$x'$};
  \node (yp) at (1.25,1) [sctloc] {$y'$};
  \node (zp) at (1.25,0.25) [sctloc] {$z'$};

  \draw [scte] (x) to[] node[pos=0.5,sctw,below] {0} (xp);
  \draw [scte] (x) to[] node[pos=0.3,sctw,below] {-1} (yp);
  \draw [scte] (y) to[] node[pos=0.5,sctw,above] {0} (yp);
  \draw [scte] (z) to[] node[sctw,above] {1} (zp);

  \draw [sctbe,bend right=15] (xp) to[] node[sctw,above] {0} (x);
  \draw [sctbe,bend left=15] (yp) to[] node[pos=0.3,sctw,below] {0} (y);
  \draw [sctbe,bend left=15] (zp) to[] node[sctw,below] {0} (z);

  \draw[-] (0,0) -- (1.5,0) -- (1.5,2) -- (0,2) -- (0,0);
\end{scope}

\begin{scope}[shift={(4.4,0.5)}]
  \node [align=center, font=\ttfamily\footnotesize] at (0,0) {
    \begin{minipage}{8.75cm}
      \(
  \begin{array}{|@{}r@{\hskip 2pt}l@{}|}
    \hline
\transitions_1{:}&  \{x\ge0, y\ge0, z\ge0, y' \le x+1, x' \le y-2, z' \le z+1  \} \\
\transitions_2{:}&  \{x\ge0, y\ge0, z\ge0, x' \le x, y' \le x-1, y'\le y, z' \le z+1  \} \\
\hline                       
  \end{array}
  \)
  \end{minipage}
};
\end{scope}

\end{tikzpicture}
\end{center}

\caption{\dsct transition relations, and their corresponding (circular) size change graphs.}
\label{fig:dsct:2}

\end{figure}
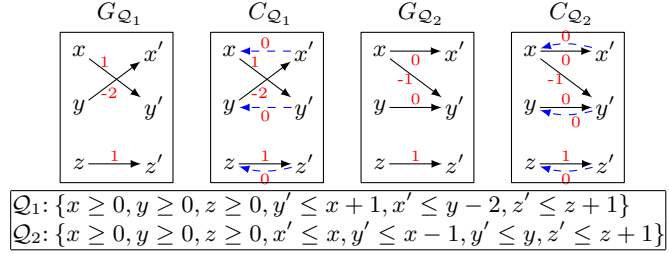
 
\begin{definition}
\label{def:dsct:graph}
For a \dsct transition relation $\transitions$, define the weighted
bipartite graph $G_\transitions$ with nodes
$\{x_1,\dots,x_n\}\cup \{x'_1,\dots,x'_n\}$ representing the state
variables before and after the transition, and arc $x \to y'$ with
weight $\delta$ whenever $y' \le x+\delta$ is in the transition
constraints.
This graph is called \emph{the size-change graph} for $\transitions$.
\end{definition}

\begin{definition} For a \dsct transition relation $\transitions$, the
\emph{circular size-change graph} $C_\transitions$ is obtained from
$G_\transitions$ by adding a zero-weight arc from every node $x'$ to
the corresponding node $x$.  These are called \emph{backward arcs}.
\end{definition}

\begin{example}
\label{ex:dsct:graphs}  
Figure~\ref{fig:dsct:2} includes two \dsct transition relations, and
their corresponding (circular) size change graphs.
Note that $\transitions_1$ is terminating and $\transitions_2$ is not.
\end{example}

The following theorem combines observations by~\citet{CLS05}
and~\citet{Moyen09}.

\begin{theorem}
\label{thm:dsct:equiv}
For a \dsct transitions relation $\transitions$, the following
statements are equivalent:
\begin{enumerate}
\item\label{thm:dsct:equiv:negcyc} $C_\transitions$ has a negative-weighted simple cycle.
\item\label{thm:dsct:equiv:lrf} $\transitions$ has an \lrf of the form
  $\rho(\vec{x}) = \sum_{i\in S}x_i$  for some $S\subseteq\{1,\dots,n\}$.
\item\label{thm:dsct:equiv:termin} The \slc loop $\transitions$ is terminating.
  
\item\label{thm:dsct:equiv:unsat} There is no solution to
  $\transitions \land (\vec{x} \le \vec{x'})\land (\vec{x} \ge
  \vec{0})$.

\end{enumerate}
\end{theorem}

\begin{proof}
  We show that each item implies the next one, and that the last
  implies the first.
  
\medskip
\noindent
$\cbox{1\Rightarrow 2}$:
A cycle in $C_\transitions$ must alternate regular (forward) arcs with
backward arcs. It is a ``zig-zag'' cycle (see Figure~\ref{fig:dsct:2}).
The set $S$ of variables and $S'$ of primed variable participating in
this cycle are counterparts, \ie, $x_i\in S \iff x_i'\in S'$.
For every $x_i\in S$ there is a single $x_j'\in S'$ such that
$x_j' \le x_i+\delta_{i}$.
This implies
$$\sum_{i\in S} x'_i \le \sum_{i\in S} (x_i +\delta_i)$$
and since we assume that the total weight of the cycle, which is
$\sum_{i\in S} \delta_i$, is negative, we have
$$\sum_{i\in S} x'_i < \sum_{i\in S} x_i$$
Thus we have our \lrf $\rho(\vec{x})=\sum_{i\in S} x_i$.

\medskip
\noindent
$\cbox{2\Rightarrow 3}$:
obvious, since \lrfs imply termination of \slc loops.

\medskip
\noindent
$\cbox{3\Rightarrow 4}$:
Assume, to the contrary, that there is a solution
$\tr{\vec{x}}{\vec{x}'}$ to
$\transitions\land (\vec{x} \le \vec{x}')\land (\vec{x} \ge\vec{0})$.
This solution satisfies every constraint $x_i' \le x_j + \delta$ of
$\transitions$, and since $x_i \le x_i'$, also the constraint
$x_i \le x_j + \delta$ is satisfied.
This means that $\tr{\vec{x}}{\vec{x}}\in \transitions$ and thus the
program is not terminating, contradicting $3$.

\medskip
\noindent
$\cbox{4\Rightarrow 1}$:
Suppose that $C_\transitions$ has \emph{no} negative-weight cycle. We
add an auxiliary node $y$ to $C_\transitions$ and connect it with
zero-weight arcs to all source nodes $x_i$. We can then compute the
weighted distance $\delta(y,\nu)$ for each node $\nu$. Note that:
\begin{inparaenum}[\upshape(1\upshape)]
\item these weights satisfy the constraints of $\transitions$, \eg,
  $x'_j \le x_i+\delta$, because this is the triangle inequality; and
\item they satisfy $x'_i \ge x_i$, because of the backward arcs.
\end{inparaenum}
We conclude that there is a solution to
$\transitions \land (\vec{x} \le \vec{x}')\land (\vec{x} \ge\vec{0})$.
\end{proof}

\begin{example}
\label{ex:dsct:equiv}  
Consider again the \dsct transition relations $\transitions_1$ and
$\transitions_2$ depicted in Figure~\ref{fig:dsct:2}.
For $\transitions_1$, it is easy to see that: $C_{\transitions_1}$
includes a negative weighed cycle; it has an \lrf $\rho(x,y,z)=x+y$; is
terminating; and $\transitions_1\wedge x\le x'\wedge y \le y', z \le z'$
is not satisfiable (since $x+y>x'+y'$).
For $\transitions_2$, it is easy to see that: $C_{\transitions_2}$
does not have a negative weighed cycle; it has no \lrf; is not
terminating; and 
$\transitions_2\wedge x\le x'\wedge y \le y', z \le z'$ is satisfied
by $x=x'=3$, $y=y'=2$, and $z=z'=0$.
\end{example}

\begin{corollary}
\label{cor:dti:dtilrf}
A \dsct \cfg terminates if and only if it has an \lrfdti. Moreover, the
form of the ranking functions used is as in
Theorem~\ref{thm:dsct:equiv}.%
\end{corollary}

\begin{proof}
First note that for \dsct transition relations $\transitions_1$ and
$\transitions_2$, the composition $\transitions_1\circ\transitions_2$
is also a \dsct transition relation.
Consider any $((\ell,\vec{x}),(\ell,\vec{x}'))\in T^+$, and note that
it corresponds to an execution trace where in each step it uses one of
the transition relations of the \cfg; let us say
$\transitions_1,\transitions_2,\ldots,\transitions_k$.
The composition of these transition relations is an \slc loop with
\dsct constraints that must be terminating, because otherwise we could
construct an infinite execution for the \cfg by repeating this
segment.
By Theorem~\ref{thm:dsct:equiv}, the composition has an \lrf of a
specific form (sum of variable), and there are a finite number of such
\lrfs.
This means that the \lrfdti induced by these \lrfs is a \dti for the
\cfg.
\end{proof}

Thus, the existence of a particular kind of termination
witness, namely \lrfdti, is equivalent to the termination problem for
\dsct programs.
We conclude that the crux of a termination analysis of (a class of)
\dsct programs is to obtain a finite description of all program cycles
as \slc loops.
If such a description is available we can check the \slc loops for
\lrfs. We indeed consider subclasses of \dsct programs, because the
whole class is too difficult:

\begin{theorem}[\citealp{B08}]
\label{th:dsct:undec}
The termination problem for \dsct programs is undecidable.
\end{theorem}

This undecidability result is obtained~\citep[Section~3]{B08} as a
reduction from a well-known undecidable problem, the halting problem
for inattentive (input-free) counter programs (the initial value of
all counters is 0).
Briefly, it is possible to translate a counter program into a \dsct
instance such that the \dsct transition system has an infinite run if
and only if the counter program does not.
The infinite runs of the \dsct instance corresponds to repeating the
finite execution of the counter program infinitely. If the counter
program is non-terminating, the \dsct transition system is
terminating.
This reversal agrees with the classification of \dsct termination as
\core, while counter program termination is \re.

\section{Size-Change-Termination}
\label{sec:dti:sct}

A Size-Change program (or \sct program) is the special case of \dsct
where the differences $\delta$ range over $\{0,-1\}$, or equivalently
$(-\infty,0]$ (the important thing is that there are no relations
$y'\ge x+\delta$ with $\delta>0$), and was developed by \citet{LJB01}
before \dsct.
Since we compute over the natural numbers, it means that we have two
types of constraints: $y'\le x$ and $y' < x$.
Note, for example, that the \cfg depicted in Figure~\ref{fig:dsct:1}
cannot be expressed using \sct constraints without affecting its
termination behaviour, since $x' \le y+1$ of $\transitions_2$ cannot
be exactly modelled using $x'<y$ or $x'\le y$, and thus would be
removed making the \cfg non-terminating.

\begin{example}
\label{ex:sct}
Consider an \mlc loop defined by the following paths:
\begin{align*}
\transitions_1&=\{ x' < y, y' < y\},\\
\transitions_2&=\{ x' < x, y' < x\},\\
\transitions_3&=\{ x' < y, y' \le x\}.
\end{align*}
It uses only \sct constraints, and it is terminating.
It does not have an \llrf of any kind, because we cannot have a \qlrf
that involves $x$ (due to $\transitions_1$) nor $y$ (due to
$\transitions_2$), but has an \lrfdti $T_{x} \cup T_{y} \cup T_{x+y}$.
\end{example}

If we express our constraints in this form, a natural way to define
the composition operation of two constraint sets
$\transitions_1$ and $\transitions_2$, that we denote by
$\transitions_1\sctcom\transitions_2$, is as follows:

\begin{enumerate}
  
\item $\transitions_1\sctcom\transitions_2$ includes $y' < x$ if and only if
  $\transitions_1$ includes $z' \bowtie_1 x$ and $\transitions_2$
  includes $y'\bowtie_2 z$, for some variable $z$, where at least one of
  the relations $\bowtie_i$ is $<$;
  
\item $\transitions_1\sctcom\transitions_2$ includes $y' \le x$ if and only
  if $\transitions_1$ includes $z' \le x$ and $\transitions_2$
  includes $y'\le z$, for some variable $z$, and Case~1 does not
  apply.

\end{enumerate}
That is, we ignore the fact that differences accumulate and express
all the constraints with the vocabulary of $<,\le$.
For \cfgs, the composition of two edges
$(\ell_i,\transitions_1,\ell_j)$ and $(\ell_j,\transitions_2,\ell_k)$
is $(\ell_i,\transitions_1\sctcom\transitions_2,\ell_k)$. Note that
the target node of the first edge must be equal to the source node of
the second edge.
We refer to $\sctcom$ by \sct-composition, to distinguish it from the
composition $\circ$.

\begin{example}
\label{ex:sct:comp}
Consider the \mlc loop of Example~\ref{ex:sct}. We have
$\transitions_4=\transitions_3\sctcom\transitions_3 = \{x' < x,
y'<y\}$. All other \sct-compositions yield one of the existing paths,
\eg, $\transitions_1\sctcom\transitions_2=\transitions_1$ and
$\transitions_2\sctcom\transitions_2=\transitions_2$.
\end{example}

Note that $\sctcom$ is an over-approximation of $\circ$. For example,
$\{x'<x\} \sctcom \{x'<x\} = \{x'<x\}$ while
$\{x'<x\} \circ \{x'<x\} = \{x' \le x-2\}$.

If we start from the set of all edges of the \cfg, and compute the
transitive closure using $\sctcom$, it is guaranteed that the computation
 terminates
since the set of possible constraint sets is finite. Thus $T_P^+$
can be symbolically over-approximated in finite time.
Moreover, this over-approximation does not lose any information that
may be necessary for the termination proof, \ie, the \cfg is
non-terminating if and only if there is $(\ell,\transitions,\ell)$ in
this transitive closure such that $\transitions$ is not well-founded.
The reason for that is that in a negative-weighted cycle (cf. Theorem~\ref{thm:dsct:equiv}),
if we change all negative weights to $-1$, the cycle will remain negative, since we have no
positive weights.

Thus we have \emph{the closure algorithm} for \sct:

\begin{enumerate}
\item Compute the transitive closure, \wrt \sct-composition, of the
  set of edges of the \cfg.
\item For every $(\ell,\transitions,\ell)$ in the transitive closure,
  check that $\transitions$ is well-founded which can be done by
  seeking corresponding \lrfs according to
  Theorem~\ref{thm:dsct:equiv}.
\end{enumerate}

\begin{example}
\label{ex:sct:closure}
The transitive closure of the \mlc loop of Example~\ref{ex:sct} adds
only $\transitions_4$ of Example~\ref{ex:sct:comp}. Then,
$\transitions_1$ has the \lrf $\rho(x,y)=x$, $\transitions_2$ has the
\lrf $\rho(x,y)=y$, $\transitions_3$ has the \lrf $\rho(x,y)=x+y$, and
$\transitions_4$ admits any of these functions as an \lrf.
\end{example}

Using the above algorithm (in a space-economic version) we obtain:

\begin{theorem}[\citealp{LJB01}]
\label{sec:dti:sct:pspace}
For \cfgs with \sct transition relations, termination is decidable in
\pspace.
\end{theorem}

\section{Fan-in Free  \texorpdfstring{$\delta$}{Delta}-Size-Change-Termination}
\label{sec:dti:FIFdsct}

We say that a \dsct transition polyhedron $\transitions$ \emph{has
  fan-in} if there are two constraints $y'\le x+\delta_x$,
$y'\le z+\delta_z$ which share the target variable $y'$.
Equivalently, if the corresponding size change graph
$G_{\transitions}$ has a node with in-degree greater than $1$.
Fan-in \emph{free} \dsct \cfg is a \dsct \cfg that does not have any
fan-in.

\begin{example}
\label{ex:dsct:fanin}
Consider the \dsct transition relations of Figure~\ref{fig:dsct:2}:
$\transitions_1$ is fan-in free and $\transitions_2$ has a fan-in on
the target variable $y'$.
\end{example}

\citet{B08} studied the class of \cfgs with fan-in \emph{free} \dsct
transition relations, and showed how to form a finite
over-approximation of $T_P^+$ that does not compromise information
that is important to termination.
The details are complex, so we will just give the result:

\begin{theorem}
\label{thm:dsct:faninfree}
The termination problem for \cfg with fan-in free \dsct transition
relations is decidable in \pspace.
\end{theorem}

\section{Monotonicity Constraints}
\label{sec:dti:mc}

A \emph{monotonicity constraint} (\mc) transition relation is a
conjunction of order constraints $x \bowtie y$ where
$x,y\in\{x_1,\dots,x_n,x_1',\dots,x_n'\}$, and
${\bowtie}\in\{>,\ge,=\}$.
It extends \sct by allowing order constraints between any pair of
variables, and, moreover, they are interpreted over $\ints$ instead of
$\nats$.
Note that $x=y$ is just syntactic sugar for $x\le y \land y\le x$. So
actually we have just two types of constraints.

\begin{figure}[t]

\begin{center}

\begin{tikzpicture}[>=latex,line join=bevel,]

\tikzgrid{gray!0}{8.5}{2}{0.1}

\begin{scope}[shift={(1,0.25)}]
  
  \node (l1) at (0,0) [inloc] {$\ell_{0}$};
  \node (l2) at (0,1.5) [loc] {$\ell_{1}$};
  
  \draw [tre,bend left=20] (l2) to[] node[tr,right] {\scalebox{0.8}{$\transitions_3$}} (l1);
  \draw [tre,bend left=20] (l1) to[] node[tr,left] {\scalebox{0.8}{$\transitions_5$}} (l2);
  \draw [tre] (l2) to[out=200,in=160,loop] node[tr,left] {\scalebox{0.8}{$\transitions_1$}}  (l2);
  \draw [tre] (l2) to[out=20,in=-20,loop] node[tr,right] {\scalebox{0.8}{$\transitions_2$}}  (l2);
  \draw [tre] (l1) to[out=20,in=-20,loop] node[tr,right] {\scalebox{0.8}{$\transitions_4$}}  (l1);
\end{scope}

\begin{scope}[shift={(7,1)}]
  \node [align=left, font=\ttfamily\footnotesize] at (0,0) {
    \begin{minipage}{6cm}
      \(
  \begin{array}{|@{}r@{\hskip 2pt}l@{}|}
    \hline
\transitions_1{:}&  \{y' < y, x' < y, y \ge z, z' \ge z\} \\
\transitions_2{:}&  \{x' < x, y' < x, x \ge z, z' \ge z\} \\
\transitions_3{:}&  \{x' \le x, y' \le y, z'> z, x \ge z, y \ge z\} \\
\transitions_4{:}&  \{x' \le y, y' \le x, z' > z, y \ge z \} \\
\transitions_5{:}&  \{x' \le x, y' \le y, z'> z, x \ge z, y \ge z\} \\
\hline
\end{array}
\)
\end{minipage}
};
\end{scope}

\end{tikzpicture}
\end{center}

\caption{A \cfg with \mc transition relations.}
\label{fig:mc:1}

\end{figure}
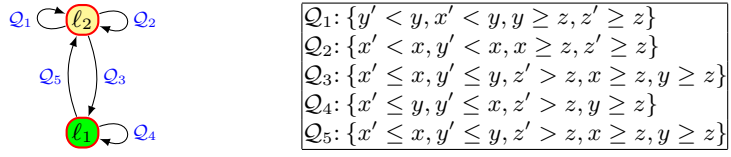
 
\begin{example}
\label{ex:dti:mc:1}
Consider the \mc \cfg depicted in Figure~\ref{fig:mc:1}:
it is terminating, and does not have an \llrf of any kind.
It cannot be modelled with \sct constraints since it includes
constraints like $x_2 \ge x_3$ and $x_3'>x_3$, which are not allowed
in \sct, and removing them would make it non-terminating.
\end{example}

Proving termination of \cfgs with \mc transition relations can be
done, as in the case of \sct, by computing the transitive closure of
the set of edges, and then check that every $(\ell,\transitions,\ell)$
in the closure is well-founded.
The composition of two $\mc$ transition relations $\transitions_1$ and
$\transitions_2$ is computed in a similar way to the case of \sct, but
considering all $x,y\in\{x_1,\dots,x_n,x_1',\dots,x_n'\}$ and
discarding results that are not satisfiable~(which is one of the
important differences from \sct).
Formally, the composition is defined as
\[
  \transitions_1\mccom\transitions_2=\{ x \bowtie y \mid x,y \in
  \vec{x}\cup\vec{x}',
  \transitions_1[\vec{x}'/\vec{z}]\land\transitions_2[\vec{x}/\vec{z}]
  \vdash x \bowtie y\}
\]
where $[\vec{x}'/\vec{z}]$ (resp. $[\vec{x}/\vec{z}]$) is the renaming
of $\vec{x}'$ (resp. $\vec{x}$) to $\vec{z}$.
The algorithm is as follows:
\begin{enumerate}
\item\label{alg:mc:closure} Compute the transitive closure of the set
  of edges of the \cfg.
\item\label{alg:mc:wf} If for every $(\ell,\transitions,\ell)$ in the transitive
  closure, $\transitions$ is well-founded then the \cfg is
  terminating, otherwise it is not.
\end{enumerate}
Like the case of \sct, the transitive closure can be computed in
finite time, and, moreover, the \cfg is non-terminating if and only if
there is $(\ell,\transitions,\ell)$ in the transitive closure such
that $\transitions$ is not well-founded.
Thus, to make the algorithm complete, we have to have find a complete
procedure for the well-foundedness check of Point~\ref{alg:mc:wf}
above.
Unlike the case of \sct, where \lrfs are enough for this check, a
complete procedure for \mc checks that $\transitions$ has an \mlrf of a
bounded depth.

\begin{lemma}
\label{lem:dti:mc:mlrf}
An \mc transition relation $\transitions$ is well-founded if and only
if it has an \mlrf of depth at most $5^{2n}$.
\end{lemma}

\begin{proof}
  It follows from results by \citet{BDG21} and~\citet{BG17}, in turn using
\citet{BIK14}.
  These results involve octagonal transition relations, where an
  octagonal polyhedron is one defined by constraints of either the
  form $\pm x \ge c$ or $\pm x \pm y \ge c$. Note that \mc transition
  relations are octagonal.  The first result shows that an octagonal
  transition relation, over the rationals, is well-founded if and only
  if it has an \mlrf of depth bounded by $5^{2n}$.
  The second shows that for \slc loops specified by \emph{integral}
  transition polyhedra, a tuple $\tuple{\rho_1,\ldots,\rho_d}$ is an
  \mlrf over the rationals if and only if it is over the integers.
  Since a set of \mc constraints is an octagonal relation and, unlike
  octagonal relations in general, is also an integral polyhedron, the
  statement of the lemma follows.
\end{proof}

This algorithm also implies that a terminating \cfg with \mc
transition relations has a \lrfdti, which is defined as a disjunction of
the components of the different \mlrfs.

\begin{theorem}
\label{thm:dti:mc:lrfdti}
  A \cfg with \mc transition relations is terminating if and only if
  it has an \lrfdti.
\end{theorem}

\begin{example}
\label{ex:dti:mc:1:closure}
Consider the \cfg depicted in Figure~\ref{fig:mc:1}. Computing the
transitive closure results in $16$ transition relations, including the
one already in \cfg.
The following (first column) are those important for termination, \ie,
source location equals to target location, and their corresponding
ranking functions (second column):
\[
  \small
  \begin{array}{|l@{}r@{}l@{}l|l|}
    \hline
(\ell_1,& \transitions_1& {=} \{y > y',y > x',y \ge z,z' \ge z\} &,\ell_1)          &   y-z   \\
(\ell_1,& \transitions_2& {=} \{x > x',x > y',x \ge z,z' \ge z\} &,\ell_1)          &   x-z  \\
(\ell_0,& \transitions_4& {=} \{x \ge y',y \ge x',y \ge z,z' > z\} &,\ell_0)        &   \tuple{x+y-2z , y-z}  \\
(\ell_1,& \transitions_6& {=} \{x > z,x \ge x',y > z,y \ge y',z' > z\} &,\ell_1)    &    x+y-z \\
(\ell_0,& \transitions_7& {=} \{x > z,x \ge x',y \ge z,y \ge y',z' > z\} &,\ell_0)  &   x+y-z    \\
(\ell_0,& \transitions_8& {=} \{x \ge z,y > z,y > x',y > y',z' > z\} &,\ell_0)      &   y-z   \\
(\ell_0,& \transitions_9& {=} \{x > z,x > x',x > y',y \ge z,z' > z\} &,\ell_0)      &   x-z   \\
(\ell_1,& \transitions_{10}& {=} \{x > z,x \ge y',y > z,y \ge x',z' > z\} &,\ell_1) &   x+y-z \\
    \hline
\end{array}
\]
The first three already appear in the \cfg, and the others were obtained
using the following compositions: 
$\transitions_6=\transitions_3 \mccom \transitions_5$,
$\transitions_7=\transitions_4 \mccom \transitions_4$,
$\transitions_8=\transitions_5 \mccom (\transitions_1 \mccom \transitions_3)$,
$\transitions_9=\transitions_3 \mccom (\transitions_2 \mccom \transitions_3)$,
$\transitions_{10}=\transitions_3 \mccom (\transitions_4 \mccom \transitions_5)$.
Note that all have \lrfs, except $\transitions_4$ that requires an \mlrf.
\end{example}

As for the case of \sct, using the closure algorithm (in a
space-economic version) we obtain:

\begin{theorem}[\citealp{B11}]
\label{sec:dti:mc:pspace}
The termination problem for \cfgs with \mc transition relations is in
\pspace.
\end{theorem}

\section{Gap Constraints}
\label{sec:dti:gap}

Gap constraints extend monotonicity constraints in two ways. First, a
non-negative ``gap'' may be added in inequalities, \ie, we have
constraints of the form $x\ge y+c$ with $c\in\nats$ (note that $c$
cannot be negative as allowed in $\dsct$ constraints). Here, too, $x$
and $y$ range over $\{x_1,\dots,x_n\}\cup \{x'_1,\dots,x'_n\}$.  In
addition, constraints of the form $x\ge a$ or $x\le a$ are allowed,
with $a\in \ints$.

\begin{theorem}[\citealp{BozzelliP14}]
\label{sec:dti:gap:pspace}
The termination problem for \cfg with gap constraint transition
relations is in \pspace.
\end{theorem}

We can prove termination of gap constraint programs using \lrfdtis
constructed similarly to the method for monotonicity constraints
described above.
This is because according to~\citet[Page 9]{BozzelliP14} (citing
\citet{Cerans94}, \citet{BozzelliG06}, and \citet{Revesz93}), the
reflexive transitive closure of the transition relation of a
gap-constraints system is definable using gap-constraints and is
effectively computable.
Therefore, from the same considerations of
Lemma~\ref{lem:dti:mc:mlrf}, a loop in this transitive closure is
terminating if and only if it has a \mlrf.

\section{Monotonicity Constraints and Ranking Functions}
\label{sec:dti:mc:rfs}

While \lrfdtis use a simple form of ranking functions to describe each
of the disjuncts in the \dti, it is not clear whether there is a
closed form for a \emph{global} ranking function, one that ranks every
transition of the program.
The case of \mc programs is an example where we have such a closed
form \citep{B11}. This form is more complex, however, than those
discussed in \Chp~\ref{chp:rfs}. Briefly, it is a \emph{piecewise
  lexicographic{-}linear ranking function}. The form is illustrated by
the following example:
\[ \rho_\ell(\vec x) = \left\{\begin{array}{cl}
        \langle 1,x_2-x_4,1,x_3-x_4\rangle & \mbox{if $x_2-x_4 > x_2-x_3$} \\
        \langle 1,x_2-x_4,0,x_3-x_4\rangle & \mbox{if $x_2-x_4 \le x_2-x_3$} 
\end{array}\right.
\]
Note that the function is indexed by the program location it is
associated with (see Section~\ref{sec:lrf:cfg}).
In comparison with \llrfs of \Chp~\ref{chp:rfs} we note the following
differences:
\begin{enumerate}
\item The function is \emph{piecewise}---each piece defined by a set
  of inequalities on differences of two variables.
\item The positions of the lexicographic tuple alternate between
  constants, and differences of pairs of variables.
\end{enumerate}

\begin{example}
  Consider the \cfg of Figure~\ref{fig:mc:1}. It has the following
  ranking function (the same for both locations):
  $ \rho(x,y,z) = \max(x,y) - z $ which in this case is just piecewise
  linear.
\end{example}

This result raises the following open problems:

\begin{problem}
  Is it decidable whether a general \cfg (or, for simplicity, an \mlc
  loop) has a piecewise \lrf? A piecewise \llrf? (Here we should
  allow \llrf positions to include arbitrary linear expressions in the
  program variables; and similarly for the conditions defining the
  pieces).
\end{problem}

\begin{problem}
  Is there a closed form for global ranking functions that works for
  all terminating fan-in free $\dsct$ programs?
\end{problem}

\section{The Power of Transition Invariants}
\label{sec:dti:slc:power}

The power of the \dti approach, even when restricted to \lrfdti, is
clear in the context of \cfgs, or even \mlc loops, because they have
branching and non-determinism that allows generating traces with
different properties.
\slc loops do not have branching, and have a limited form of
non-determinism that originates from the constraints specifying them.
Given this, it is natural to ask the following.

\begin{problem}
  Are there terminating \slc loops, deterministic or
  non-deterministic, whose termination can be shown using \lrfdti, but
  not using ranking functions as those of \Chp~\ref{chp:rfs}?
\end{problem}

The restriction to \lrfdti, is because the ranking functions of
\Chp~\ref{chp:rfs} are restricted to linear components. Moreover, we
can focus on \mlrfs since they are the most powerful, among those
discussed in \Chp~\ref{chp:rfs}, for \slc loop.
In what follows we discuss partial answers to this question, and 
state open problems.

For integer loops, the following deterministic \slc loop
\[
  \while\; (x\ge 0) \; \wdo \; x'=10-2x
\]
is terminating over the integers, and non-terminating over the
rationals~(\eg, for $3\frac{1}{3}$).
It has a \dti $(T_{x} \cup T_{10-x})$ over the integers, and does not
have an \mlrf.
This provides a positive answer for the above problem, for the integer
case, however, we note that this loop has a piecewise \lrf:
\[
  \rho(x) = \left\{
    \begin{array}{ll}
      x & x>3 \\
      10-x & \mbox{otherwise}
    \end{array}
  \right.
\]
This somehow introduces piecewise \lrfs (with polyhedral conditions)
into this discussion, and thus we can generalise the problem above to
the following one about the relative power of these termination
arguments.

\begin{problem}
  What is the relative power of piece-wise \lrfs (with polyhedral
  conditions), \lrfdtis, and \mlrfs, for \slc loops?
\end{problem}

To understand the power of \dti for \slc loops, for the rational case,
one might also study the need for $T^+$ for this class of loops.
In particular, study if the requirement
$T^+ \subseteq T_1\cup\cdots\cup T_k$, where each $T_i$ is
well-founded, can be relaxed to $T \subseteq T_1\cup\cdots\cup T_k$ for
\slc loops.
This is not true for integer \slc loops. For example,
$\transitions=\{x \ge 0, x'=1-x\} \subseteq T_{x} \cup T_{-x}$, but
the loop is non-terminating for $x=0$.

\begin{problem}
  For an \slc loop over the rationals, does
  $\transitions \subseteq T_i\cup\cdots\cup T_k$, where each $T_i$ is
  well-founded, \ie, a terminating \slc loop, imply termination of
  $\transitions$?
\end{problem}

For \lrfdtis we have the following conjuncture, which we know to be true for $k\le 3$.%

\begin{conjecture}
  If $\transitions \subseteq T_{\rho_1}\cup\cdots\cup T_{\rho_k}$, then
  $\transitions$ has an \mlrf.
\end{conjecture}

Finally, we note that there are terminating \slc loops that do not
have a polyhedral \dti at all. For example, the following \slc loop
\[
  \while\; (x\ge 1,\; y\ge 1,\; x\ge y) \; \wdo \; x'=2x,\; y'=3y
\]
which is terminating, and its termination can be shown using the
techniques of Section~\ref{sec:dec:affine}, or using non-linear
ranking functions such as $\rho(x,y) = \log_2(x)-\log_2(y)$ or
$\rho(x,y) = \frac{x}{y}$.

\section{Other Works Related to Transition Invariants}
\label{sec:dti:slc:related}

The practical application of \dti was also promoted by \citet{PodelskiR07},
who proposed a technique to generate transition invariants that are
\emph{inductive}, using predicate abstraction.
Two subsequent works \citep{HJP2010,Z2018} explore the connections of
this type of \dti termination proofs to \sct.
\citet{Z2015} constructs, for \emph{fan-out free} \sct programs, global
ranking functions which are still piecewise-lexicographic, as those
mentioned earlier, but are optimal in their depth (which is
interesting if the ranking functions are used to estimate execution
time, see our discussion of depth in Section~\ref{sec:llrf}).
\citet{CFM15} propose heuristics for discovering \dtis for \slc
loops.
\citet{KSTW10} proposed using \emph{compositional transition
  invariants}, which are transition invariants $T_I$ that satisfy
$T_I \circ T_I \subseteq T_I$. They show a heuristic for finding such
\dtis that performs better, empirically, than the method suggested
by \citet{CPR06}.
\citet{GantyG13} developed conditional termination analysis based on
\dtis. Their idea is to use \dtis to isolate the non-terminating part
of a given transition relation.

\chapter{Witnesses for Non-Termination}
\label{chp:nt}

By non-termination we mean the converse of termination, namely the
\emph{existence} of an infinite computation.  A \emph{non-termination
  witness} is an object whose existence proves that a given program,
or loop, is non-terminating.  Note that, in general, we cannot resort
to the easy answer ``present a non-terminating path'', as this is an
infinite object.
An algorithm that can decide the existence of a non-termination
witness of a given kind can serve as a partial solution to the
termination problem, and complement partial solutions that can only
confirm termination (\eg, ranking functions).
In this \chp we present non-termination witnesses, in particular
\emph{recurrent sets} of different forms.

\begin{definition}
\label{def:rset}
Given a transition relation $T \subseteq S\times S$, we say that a
non-empty set $G \subseteq S$ is a recurrent set for $T$ if and only if
$\forall s \in G. \exists s' \in G.\; (s,s') \in T$.
\end{definition}

A recurrent set clearly implies non-termination of $T$,
since we can construct an infinite execution that uses only states
from $G$, but also the inverse holds: if $T$ is non-terminating, then
the set of states that participate in (any subset of) its infinite
executions is a recurrent set.
Thus, recurrent sets constitute a complete criterion for
non-termination.

To establish non-termination \wrt a set of initial states
$S_0 \subset S$, we seek a recurrent set $G$ such that
$S_0 \cap G \neq \emptyset$.
This is still a complete criterion for non-termination, \wrt a given
set of initial states, because if a recurrent set $G$ is reachable
from $S_0$ only indirectly using an execution path
$s_0,s_1,\ldots,s_k$ where $s_0\in S_0$ and $s_k\in G$, then
$G '= G \cup \{s_0,\ldots,s_k\}$ is a recurrent set too and satisfies
$S_0 \cap G' \neq \emptyset$ (we could also seek a recurrent set for
$\trres{T}{S_0}$; the restriction of $T$ to states reachable from $S_0$).
However, requiring $S_0 \cap G \neq \emptyset$ might be too
restrictive in practice because we typically seek recurrent sets of a
particular form, \eg, polyhedral, and thus instead we require that $G$
is reachable from $S_0$.

\paragraph{Organisation of this \Chp.}
In the rest of this \chp we will discuss non-termination analysis
using polyhedral recurrent sets.
Section~\ref{sec:nt:rset:slc} discusses the inference of recurrent
sets for \slc loops;
Section~\ref{sec:nt:mrset} discusses an alternative definition of
recurrent sets that focuses on transitions instead of states;
Section~\ref{sec:nt:gnta} discusses the notion of \emph{Geometric
  Non-Termination Arguments}, and show that it is a special form of
recurrent sets;
Section~\ref{sec:nt:rset:cfg} explains how these notions extend to
non-termination of \cfgs;
Section~\ref{sec:nt:unbound} discusses the notion of unbounded
executions and its relation to non-termination;
and Section~\ref{sec:nt:other} discusses other approaches to
non-termination.

\section{Recurrent Sets for Single-path Linear-Constraint Loops}
\label{sec:nt:rset:slc}

In this section we discuss the inference of polyhedral recurrent sets
for \slc loops, first without any assumption on the input states and then
assuming a given polyhedral set of initial states.
Moreover, we first assume that variables range over the reals, and
then discuss the rational and integer cases.
Let us start by defining the notion of a recurrent set in this
context, which is equivalent to Definition~\ref{def:rset} but
more adequate for inferring them automatically.

\begin{definition}[\citealp{GuptaHMRX08}]
\label{def:recset}
A polyhedral set $\poly{G} \subseteq \reals^n$ is a recurrent set for an
\slc loop $\transitions \subseteq \reals^{2n}$ if and only if:
\begin{align}
\exists \vec{x} \in \reals^n &.\ \poly{G}(\vec{x}) \label{eq:rset:nonempty}\\
\forall \vec{x} \in \reals^n \:\exists \vec{x}' \in \reals^n &.\ \poly{G}(\vec{x}) \to \transitions(\vec{x}, \vec{x}') \land \poly{G}(\vec{x}'). \label{eq:rset:consec}
\end{align}
\end{definition}

Condition~\eqref{eq:rset:nonempty} forces $\poly{G}$ to be non-empty,
and Condition~\eqref{eq:rset:consec} forces any $\vec{x}\in\poly{G}$
to have a successor $\vec{x}' \in \poly{G}$.
The domain of variables is explicitly chosen as $\reals$. If we are
interested in $\ints$ or $\rats$, we require $\vec{x}$ and $\vec{x}'$
to range over the respective domain
in~(\ref{eq:rset:nonempty},\ref{eq:rset:consec}).
This is a subtle issue in automatic inference of recurrent sets, and
will be discussed later in detail.

Since $\poly{G}$ is polyhedral, \ie, defined by a finite set of
inequalities, inferring a recurrent set for $\transitions$ can be
based on the template-based approach.
We start from a template recurrent set $\poly{G}$, where the
coefficients and constants of its inequalities are parameters, and
then find values for these parameters such
that~(\ref{eq:rset:nonempty},\ref{eq:rset:consec}) hold.
In principle, this can be achieved using quantifier elimination when
considering loops over the reals, however, this approach is usually
not practical and does not apply to the rational and integer cases
that we will consider later.
A more practical approach would be to base such inference on Farkas'
Lemma, as we did in \Chp~\ref{chp:rfs} for \lrfs and \llrfs, however,
this is not immediately applicable due to the $\forall\exists$
quantifier alternation in~\eqref{eq:rset:consec}.
If we succeed to eliminate $\exists \vec{x}'$
from~\eqref{eq:rset:consec}, then we can apply Farkas' lemma since we
are left with a $\exists\forall$ formula (the $\exists$ here is over
the template parameters of $\poly{G}$).
This is clearly not possible in general, however, \citet{GuptaHMRX08}
show that this can be done for some cases of \slc loops, in particular
affine \slc loops as in~\eqref{eq:slc-lin-loop}.

Let us assume that $\transitions$ is given as
$A''\tr{\vec{x}}{\vec{x}'} \le \vec{c}$, and that $\poly{G}$ is a
template of the form $B \vec{x} \le \vec{b}$, where $B$ and $\vec{b}$
include template parameters.
To eliminate $\exists \vec{x}'$ of \eqref{eq:rset:consec},
\citet{GuptaHMRX08} assume that $\transitions$ includes (or implies)
equations of the form $\vec{x}'=A\vec{x}+\vec{d}$, \ie, the
variables are updated deterministically.
Then,  we eliminate $\exists \vec{x}'$ by replacing
occurrences of $\vec{x}'$ by $A\vec{x}+\vec{d}$.
This leaves us with a formula of the form
\begin{align}
\exists \vec{x} \in \reals^n &.\  B \vec{x} \le \vec{b}, \label{eq:rset:linloop:nonempty}\\
\forall \vec{x} \in \reals^n \in \reals^n &.\  B \vec{x} \le \vec{b} \to A''\tr{\vec{x}}{A\vec{x}+\vec{d}} \le \vec{c} \land B (A\vec{x}+\vec{d}) \le \vec{b}, \label{eq:rset:linloop:consec}
\end{align}
in which both sides of the implication are linear inequalities with
template parameters. Thus, we can use Farkas' lemma to
translate~\eqref{eq:rset:linloop:consec} into a non-linear formula
$\Psi_{\eqref{eq:rset:linloop:consec}}$ over the template parameters
and some other variables representing the Farkas' coefficients
(non-linearity is due to the template parameters on the left-hand side
of the implication).
Solving $\Psi_{\eqref{eq:rset:linloop:consec}}$ in conjunction
with~\eqref{eq:rset:linloop:nonempty} we obtain values for the
template parameters, in $B$ and $\vec{b}$, for
which~(\ref{eq:rset:nonempty},\ref{eq:rset:consec}) are satisfied, and
thus $B \vec{x} \le \vec{b}$ is a recurrent set for $\transitions$.
Note that if $\transitions$ is directly given as a linear loop of the
form
\[
  \while\; (G\vec{x} \le \vec{g}) \; \wdo \; \vec{x}' = A\vec{x}+\vec{d}
\]
then $A''\tr{\vec{x}}{A\vec{x}+\vec{d}} \le
\vec{c}$ in~\eqref{eq:rset:linloop:consec} become $G\vec{x} \le
\vec{g}$.

\begin{example}
\label{ex:nt:rset}
Consider the following \slc loop $\transitions$ and a corresponding
template recurrent set $\poly{G}$ ($b_1,\ldots,b_6$ are the
parameters):
\begin{align}
\transitions = &\{-x_1 + x_2 \le -1,\; x_1' = -x_1+x_2,\; x_2' = x_2-1\} \\
\poly{G}= & \{b_1 x_1+b_2 x_2 \le b_3,\; b_4 x_1+b_5  x_2 \le b_6\}\label{eq:ex:rset}
\end{align}
Note that $x_1$ and $x_2$ are updated as required
in~\eqref{eq:rset:linloop:consec}.
Rewriting~\eqref{eq:rset:linloop:consec} using this context we get:
\begin{equation}
  \label{eq:rset:exp}
  \begin{array}{@{}l@{}}
  \exists\vec{b}\in \reals^6.\,\forall\vec{x}\in \reals^2.\,\\
  \hspace*{0.25cm}
  \begin{array}{|l|}
    \hline
    b_1  x_1+b_2  x_2 \le b_3 \wedge\\
    b_4 x_1+b_5 x_2 \le b_6 \wedge\\
    \hline
  \end{array}    
      \to
  \begin{array}{|l|}
    \hline
    -x_1 + x_2 \le -1 \wedge\\
    \hline
    \hline
    -b_1 x_1+(b_1+b_2) x_2 \le b_3+b_2 \wedge\\
    -b_4 x_1+(b_4+b_5)  x_2 \le b_6+b_5 \\
    \hline
  \end{array}    
  \end{array}    
\end{equation}
The left-hand side is $\poly{G}(\vec{x})$; the first inequality in the
right-hand side is $\transitions(\vec{x},A\vec{x}+\vec{d})$; and the
rest correspond to $\poly{G}(\vec{x},A\vec{x}+\vec{d})$.

Using Farkas' lemma we can translate~\eqref{eq:rset:exp} into the
following set of non-linear constraints
\begin{equation}
  \label{eq:rset:farkas}
  \left\{
  \begin{array}{@{}l@{}}
    \hline
    \mu_{1} b_1+\mu_{2} b_4 = -1,\,
    \mu_{1} b_2+\mu_{2} b_5 = 1,\,\\
    \mu_{1} b_3+\mu_{2} b_6 \le -1,\,
    \mu_{1}\ge 0,\, \mu_{2} \ge 0,\, \\
    \hline
    \xi_{1} b_1+\xi_{2} b_4 = -b_1,\,
    \xi_{1} b_2+\xi_{2} b_5 = b_1+b_2,\,\\
    \xi_{1} b_3+\xi_{2} b_6 \le b_3+b_2,\,
    \xi_{1}\ge 0,\, \xi_{2} \ge 0,\,\\
    \hline
    \eta_{1} b_1+\eta_{2} b_4 = -b_4,\,
    \eta_{1} b_2+\eta_{2} b_5 = b_4+b_5,\,\\
    \eta_{1} b_3+\eta_{2} b_6 \le b_6+b_5,\,
    \eta_{1}\ge 0,\, \eta_{2} \ge 0,\,\\
    \hline
  \end{array}
    \right\}
\end{equation}
where each block corresponds to translating, using Farkas' lemma, one
constraint from the right-hand side of~\eqref{eq:rset:exp}.
Solving~\eqref{eq:rset:farkas} together with~\eqref{eq:ex:rset}, to require $\poly{G}$
to be non-empty, we get the following possible solution:
\[
  b_1 \mapsto 1,\; b_2\mapsto0,\; b_3 \mapsto 0,\; b_4 \mapsto -1,\; b_5\mapsto1,\; b_6\mapsto-1,\;
\]
which defines the recurrent set $\{x_1 \le 0, -x_1+x_2 \le -1\}$.
\end{example}

Let us now consider the case where the domain of the variables is the
integers, \ie, replacing $\reals$ by $\ints$
in~(\ref{eq:rset:linloop:nonempty},\ref{eq:rset:linloop:consec}).
The use of Farkas' lemma in this case is not immediate because a loop
might be non-terminating over $\reals$ but terminating over
$\ints$.
Thus, unlike for the case of \lrfs and \llrfs, relaxation of the
problem from $\ints$ to $\reals$ is not sound.
However, such a relaxation is sound if we guarantee that:
\begin{inparaenum}[\upshape(1\upshape)]
\item $\poly{G}$ has at least one integer state; and
\item for every integer state in $\poly{G}$, there is an integer
  successor in $\poly{G}$.
\end{inparaenum}
The first condition can be achieved by
requiring~\eqref{eq:rset:linloop:nonempty} to hold over $\ints$, and
the second is guaranteed to hold if we assume the update
$\vec{x}'=A\vec{x}+\vec{d}$ has only integer coefficients and
constants.
Similar argument holds for the case of $\rats^n$.

To summarise this approach, in terms of decidability of the underlying
problems, what we have described above is a complete procedure for
seeking recurrent sets, matching a given template, for affine \slc
loops over $\reals$ (because non-linear polynomial constraints can be
solved in polynomial space~\citep{Canny88}). The method is not
complete over $\ints$ and $\rats$, because solving non-linear
polynomial constraints is not decidable over $\ints$ and its
decidability over $\rats$ is unknown.

Next we present an alternative definition for recurrent sets, which is
more restrictive than the general case, but allows using Farkas' lemma
smoothly, even for nondeterministic \slc loops. This notion was
introduced by~\citet{ChenCFNO14}.

\begin{definition}
\label{thm:crecset}
Let $\transitions \subseteq \reals^{2n}$ be an \slc loop and
$\poly{B}=\proj{\vec{x}}{\transitions} \subseteq \reals^n$ be its set
of enabled states.
A polyhedral set $\poly{G} \subseteq \reals^n$ is a \emph{closed}
recurrent set for $\transitions$ if and only if:
\begin{align}
\exists \vec{x} \in \reals^n  &.\ \poly{G}(\vec{x})  \label{eq:crset:nonempty}\\
\forall \vec{x} \in \reals^n  &.\ \poly{G}(\vec{x}) \to  \poly{B}(\vec{x})  \label{eq:crset:enabled}\\
\forall \vec{x},\vec{x}' \in \reals^n  &.\ \poly{G}(\vec{x}) \land  \transitions(\vec{x}, \vec{x}')  \to\poly{G}(\vec{x}'). \label{eq:crset:consec}
\end{align}
\end{definition}

Note that~\eqref{eq:crset:nonempty} is required to guarantee that
$\poly{G}$ is not empty, and~\eqref{eq:crset:enabled} is required to
guaranties that $\poly{G}$ is a subset of the enabled states, and thus
for any $\vec{x}\in\poly{G}$ we can make progress.

The advantage of this definition over Definition~\ref{def:recset} is
that it allows using Farkas' lemma directly, however, it is more
restrictive in general since it requires all the successors of
$\vec{x}\in\poly{G}$ to be also in $\poly{G}$.
For deterministic \slc loops, this definition is equivalent to
Definition~\ref{def:recset} since in such case each enabled state
$\vec{x}$ has a single successor.
Moreover, if a transition relation $T$ that has a recurrent set, then
there exists a transition relation $T' \subseteq T$ that has a closed
recurrent set~\citep{ChenCFNO14}.

\begin{example}
\label{ex:nt:crset}
The loop of Example~\ref{ex:nt:rset} is deterministic, and thus the
recurrent set we inferred there is also closed.
The \slc loop $\transitions_1=\{ x \ge 0, x' = x-y, y' \le y \}$ is
non-deterministic, and has the closed recurrent set
$\poly{G}_1=\{x \ge 0, y \le 0 \}$. It also has the recurrent set
$\poly{G}_1'=\{x \ge 0, x \ge y\}$ which is not closed because
$\tr{1}{1}\in \poly{G}_1'$ has a successor
$\tr{0}{1} \not\in\poly{G}_1'$.
The loop $\transitions_2=\{ x \ge 0, x' \le x-y, y' \le y \}$ is
non-deterministic, and has the recurrent set
$\poly{G}_2=\{x \ge 0, y \le 0 \}$ but does not have a closed one.
\end{example}

Let us now consider the case when variables range over the
rationals. Requiring the solution (\ie, the coefficients
in~\eqref{eq:crset:consec} and~\eqref{eq:crset:enabled} together
with~\eqref{eq:crset:nonempty}) to be rational is sound.
This is true since if the polyhedron $\poly{G}$ uses only rational
coefficients in its inequalities, and satisfies
\eqref{eq:crset:nonempty}--\eqref{eq:crset:consec} then it is a
recurrent set over the rationals.
This, however, is not sound when variables range over the integers,
because it is not guaranteed that every integer state
$\vec{x}\in\poly{G}$ has an integer successor in $\poly{G}$ (the
successor might be non-integer).

\begin{example}
\label{ex:rset:intvsrats}  
The \slc loop $\transitions=\{x \ge 2, 2x' = 3x\}$ is non-terminating
over the rationals, and is terminating over the integers (because
$(\frac{3}{2})^ix$ is eventually non-integer).
The set $\poly{G}=\{x \ge 2\}$ is a recurrent over the rationals. Over
the integers, both~\eqref{eq:crset:nonempty}
and~\eqref{eq:crset:enabled} are satisfied, but the integer state
$x=3$, for example, does not have an integer successor.
\end{example}

This problem can also appear for non-deterministic loops.

\begin{example}
\label{ex:rset:intvsrats:nondet}
Consider the following (nondeterministic) \slc loop%
\footnote{This loop was constructed by taking the convex-hull of the
following transitions: $((\frac{1}{2},\frac{1}{3}),(1,1)), ((1,1),(\frac{1}{2},\frac{1}{3})), ((\frac{1}{3},\frac{1}{2}),(1,1)), ((1,1),(\frac{1}{3},\frac{1}{2}))$ and
$((1,2),(1,1))$.}
which is terminating over the integers but not over the reals (and
rationals):
\[
  \transitions=\left\{
    \begin{array}{l}
      -6x-6y-6x'-6y' \le -17,\; 4x'-3y' \le 1,\; \\
      70x-21y+18x'+18y' \le 64,\; -3x'+4y' \le 1\\
      -63x+28y-24x'-24y' \le -55
\end{array}
\right\}
\]
The only enabled integer states are $(1,1)$ and $(1,2)$, and the
transitions involving these states are $((1,2),(1,1))$,
$((1,1),(\frac{1}{2},\frac{1}{3}))$, 
$((\frac{1}{2},\frac{1}{3}),(1,1))$,
$((\frac{1}{3},\frac{1}{2}),(1,1))$, and $((1,1),(\frac{1}{3},\frac{1}{2}))$.
It is easy to see that these transitions can form an infinite
execution over the reals (and rationals), but not over the integers.
The following polyhedral set (which is the projection of
$\transitions$ on $x$ and $y$)
\[
  \poly{G}=\{ -6x+6y \le -5, 4x-3y \le 1, -3x+4y \le 1\}
\]
is a closed recurrent set over the reals,
however, the state $\tr{1}{1} \in \poly{G}$ does not have an integer
successor in $\poly{G}$ (nor in $\transitions$).
Note that $\poly{G}$ is closed because it is a superset of the
projection of $\transitions$ on $(x',y')$ which is
$\{ -6x+6y \le -5, 4x-3y \le 1, -9x+4y \le -1, x \le 1\}$.
\end{example}

To solve this problem, \ie, make the relaxation to the reals sound, we
can add template inequalities of the form $\vec{x}'=A\vec{x}+\vec{d}$
to $\transitions$, where $A$ and $\vec{d}$ are parameters, and
synthesise (integer) values for them together with a closed recurrent
set.
In addition, we have to require
\begin{align}
  \exists \vec{x},\vec{x}' \in \reals^n.~ \transitions(\vec{x},\vec{x}') \land \vec{x}'=A\vec{x}+\vec{d}\\
  \forall \vec{x},\vec{x}' \in \reals^n.~\transitions(\vec{x},\vec{x}') \land \vec{x}'=A\vec{x}+\vec{d} \rightarrow \poly{B}(\vec{x})
\end{align}
The first guarantees that the restriction of $\transitions$ is not
empty, and the second guarantees that the update does not block any of
the enabled states.
\citet{LarrazNORR14} introduced this technique for analysing the
non-termination of \cfgs, and we will discuss it later in
Section~\ref{sec:nt:rset:cfg}.
This techniques can also be used to make the approach described
in~(\ref{eq:rset:linloop:nonempty},\ref{eq:rset:linloop:consec})
applicable for nondeterministic \slc loops as well.

To summarise this approach, in terms of decidability of the underlying
problems, what we have described above is a complete procedure for
seeking closed recurrent sets, of a given template, for \slc loops over
$\reals$ (because non-linear polynomial constraints can be solved in
polynomial space~\citep{Canny88}).

Inferring a recurrent set for an \slc loop $\transitions$ \wrt a
polyhedral set of initial state $\poly{S}_0$ can be done by requiring
$\poly{S}_0(\vec{x})$ to hold as well in~\eqref{eq:rset:nonempty}
and~\eqref{eq:crset:nonempty}, \ie, require the recurrent set to
include a state from $\poly{S}_0$.
The decidability of the resulting problems is still the same as we
have described above, for both kinds of recurrent sets.  We note that
the requirement that $\poly{S}_0$ intersects the recurrent set is, in
some sense, non-restrictive: if the recurrent set $\poly{G}$ is
reachable using a finite sequence of states
$s_0\in \poly{S}_0, s_1,\dots,s_k\in\poly{G}$, then the convex hull of
$\poly{G}$ and $s_0,\dots,s_k$ is also a recurrent set. So there is a
recurrent set including $s_0$ (caveat: this recurrent set may have a
more complex description than $\poly{G}$).

\begin{example}
\label{ex:nt:rset:init}
Let us analyse non-termination of the \slc loop $\transitions$ of
Example~\ref{ex:nt:rset}, \wrt to the initial set of states
$\poly{S}_0=\{ x_1 \le -1, x_2=0\}$.
Solving~\eqref{eq:rset:farkas} together with~\eqref{eq:ex:rset} and
$\poly{S}_0$ fails, because the loop terminates after one iteration
for these initial states.
On the other hand, for $\poly{S}_0=\{ x_1 \le -1, x_2 \le -1\}$ we
succeed since it intersects the recurrent set
$\poly{G}=\{x_1\le 0, -x_1+x_2 \le -1\}$.
\end{example}

We finish this section with some open problems.

\begin{problems}\
  \begin{itemize}
  \item Is there an algorithm to decide the existence of a polyhedral
  recurrent set (Definition~\ref{def:recset}) for (special cases of)
  \slc loops, over $\reals$, $\rats$ or $\ints$?
  \item Is there an algorithm
  that decides the existence of a recurrent set matching a given a
  template for general \slc loops over $\reals$, $\rats$ or $\ints$?
  \end{itemize}
\end{problems}

An intriguing question is whether polyhedral recurrent sets suffice
for proving non-termination of \slc loops.

\begin{problem}
  Does every non-terminating \slc loop (perhaps, of a particular form)
  have a polyhedral recurrence set?
\end{problem}

\section{Monotone Recurrent Sets of Transitions}
\label{sec:nt:mrset}

As noted in Section~\ref{sec:mlrfs}, \citet{Ben-AmramDG19} developed
an incomplete algorithm to synthesise \mlrfs for \slc loops such that,
when it fails to find such a ranking function, in some cases it
identifies witnesses for non-termination. These witnesses, termed
\emph{monotone recurrent sets}, are the focus of this section.

The notion of recurrent sets used by~\citet{Ben-AmramDG19} focuses on
transitions instead of states.

\begin{definition}
\label{def:rset:tr}
A polyhedral set $\poly{G} \subseteq \rats^{2n}$ is a \emph{recurrent
  set of transitions} for an \slc loop
$\transitions \subseteq \rats^{2n}$ if and only if
$\poly{G} \subseteq \transitions$ and
$\proj{\vec{x}'}{\poly{G}} \subseteq \proj{\vec{x}}{\poly{G}}$.
\end{definition}

This notion is equivalent to polyhedral recurrent sets of states of
Definition~\ref{def:recset}, because:
\begin{inparaenum}[\upshape(1\upshape)]
\item if $\poly{G} \subseteq \rats^{2n}$ is a polyhedral recurrent set
  of transitions, then $\proj{\vec{x}}{\poly{G}}$ is a polyhedral
  recurrent set of states; and
\item if $\poly{G} \subseteq \rats^{n}$ is a recurrent sets of
  states, then $\transitions \cap (\poly{G} \times \poly{G})$ is a
  polyhedral recurrent set of transitions.
\end{inparaenum}

Given a polyhedron $\poly{P} \subseteq \rats^n$, we define
$\poly{P}^\# = \{ (\vect{a},b) \in \rats^{n+1} \mid \vect{a}\vec{x}+b
\ge 0 \mbox{ for all } \vec{x}\in \poly{P}\}$.
Intuitively, $\poly{P}^\#$ includes all inequalities implied by
$\poly{P}$, which also means that $\poly{P}$ is equal to the
intersection of all half-spaces $\vect{a}\vec{x}+b \ge 0$ where
$(\vect{a},b)\in\poly{P}^\#$.
Note that $\poly{P}^\#$ is a polyhedral cone, and thus has a finite
set of generators.

The algorithm of \citet{Ben-AmramDG19} is based on the following
operator $F:\rats^{2n} \mapsto \rats^{2n}$
\begin{equation}
\label{mrset:bkwop}
F(\transitions) =  \transitions \land \vect{a}_1\vec{x}-\vect{a}_1\vec{x}'\le0\land\cdots\land\vect{a}_l\vec{x}-\vect{a}_l\vec{x}'\le 0
\end{equation}
where $(\vect{a}_1,b_1),\ldots,(\vect{a}_l,b_l)$ are the generators of
$\proj{\vec{x}}{\transitions}^\#$. Each inequality
$\vect{a}_i\vec{x}-\vect{a}_i\vec{x}'\le0$ is actually a
simplification of $\vect{a}_i\vec{x}+b_i \le \vect{a}_i\vec{x}'+b_i$.

For an \slc loop $\transitions$, the algorithm computes
$F^i(\transitions)$ iteratively which might result in three possible
outcomes:
\begin{inparaenum}[\upshape(1\upshape)]
\item It reaches $F^i(\transitions)=\emptyset$ in a finite number of
  steps, in which case they show how to construct a \mlrf of optimal
  depth $i$ for $\transitions$;
\item It reaches $F^i(\transitions) = F^{i+1}(\transitions)$ in a
  finite number of steps, in which case they show that $\transitions$
  is non-terminating over the rationals, which is the case that we are
  interested in; and
\item $F^{i}(\transitions) \supset F^{i+1}(\transitions)$ for all
  $i\ge1$, \ie, the algorithm does not terminate, which is the reason
  for incompleteness.
\end{inparaenum}
For the second case, in which non-termination is proven, they show
that $F^i(\transitions)$ is actually a \emph{monotone recurrent set
  of transitions}.
In what follows we use the notation
$\transitions_{i+1}=F(\transitions_i)$ where
$\transitions_0=\transitions$ to refer to this iterative process.

The intuition to why we get a recurrent set in the second case is as
follows.
First recall that $\proj{\vec{x}}{\transitions_i}$ is the guard of
$\transitions_i$, and that it is equal to the intersection of all
half-spaces $\vect{a}\vec{x}+b \ge 0$ where
$(\vect{a},b)\in\proj{\vec{x}}{\transitions_i}^\#$.
By adding a constraint
$\vect{a}_j\vec{x}+b_j \le \vect{a}_j\vec{x}'+b_j$ to $\transitions_i$
for each generator $(\vect{a}_j,b_j)$ of
$\proj{\vec{x}}{\transitions_i}^\#$, we actually eliminate, among
others, all transitions $\tr{\vec{x}}{\vec{x}'} \in \transitions_i$
such that $\vec{x}'$ does not satisfy the guard of $\transitions_i$
(these are terminating transitions).
This means that when this process stabilises in a finite number of
steps, it is guaranteed that
$\proj{\vec{x}'}{\transitions_i} \subseteq
\proj{\vec{x}}{\transitions_i}$ which is the definition of a recurrent
set of transitions. Moreover, since we have required
$\vect{a}_i\vec{x}+b_i\le\vect{a}_i\vec{x}'+b_i$, this recurrent set
has a monotonicity property.

\begin{definition}
\label{def:mrset}
A polyhedral recurrent set of transitions
$\poly{G} \subseteq \rats^{2n}$ for an \slc loop
$\transitions \subseteq \rats^{2n}$ is called \emph{monotone}, if for
any $\tr{\vec{x}}{\vec{x}'}\in \poly{G}$ and
$(\vect{a},b) \in \proj{\vec{x}}{\transitions}^\#$ we have
$\vect{a}\vec{x}+b \le \vect{a}\vec{x}'+b$.
\end{definition}

\begin{example}
\label{ex:mrset}
Let us apply the algorithm to the \slc loop
$\transitions=\{-x_1 + x_2 \le -1,\; x_1' = -x_1+x_2,\; x_2' =
x_2-1\}$ of Example~\ref{ex:nt:rset}.
The following table includes all intermediate steps for computing
$\transitions_{i+1}=F(\transitions_i)$:
\begin{footnotesize}
\[
\begin{array}{|l|l|}
\hline
\multicolumn{1}{|c|}{\transitions_i}  & \multicolumn{1}{c|}{\mbox{Generators of }\proj{\vec{x}}{\transitions_i}^\#} \\
\hline \hline
\transitions_0{=}\transitions  & (\mathbf{(1,-1)},-1),((0,0),1)\\
\hline
\transitions_1{=}  \transitions_0 \wedge (x_1-x_2)-(x_1'-x_2') \le 0 & (\mathbf{(-2,1)},0),((1,-1),-1),((0,0),1)\\
\hline
  \transitions_2{=}  \transitions_1 \wedge (-2x_1+x_2)-(-2x_1'+x_2')\le 0 & (\mathbf{(2,-1),0)},(\mathbf{(-1,0)},-1),\\
  &((-2,1),0),((0,0),1)\\
\hline
\transitions_3{=} \transitions_2 \wedge (2x_1-x_2)-(2x_1'-x_2')\le 0 \wedge & ((2,-1),0),((-1,0),-1),\\
  \hspace*{3.25cm}(-x_1)-(-x_1')\le 0 &((-2,1),0),((0,0),1)\\
\hline
\transitions_4{=} \transitions_3 &\\
\hline
\end{array}
\]
\end{footnotesize}%
The column on the left includes $\transitions_i$, and the one on the
right includes the generators
$(\vect{a}_1,b_1),\ldots,(\vect{a}_l,b_l)$ of
$\proj{\vec{x}}{\transitions_i}^\#$. Those in bold are the ones used
to generate the next $\transitions_i$ (those that appear syntactically
in a previous iteration are ignored).
Since $\transitions_4=\transitions_3$, we conclude that
$\transitions_3$ is a monotone recurrent set of transitions and that
the loop is non-terminating.
Projecting $\transitions_3$ on $x_1$ and $x_2$ we get
$\{-2x_2 \ge 3, 4x_1 -2x_2=1\}$, which is the corresponding recurrent
set of states.
Note that it is a subset of the recurrent set
$\poly{G}=\{x_1 \le 0, -x_1+x_2 \le -1\}$ that we inferred in
Example~\ref{ex:nt:rset}, and that
$\transitions\cap (\poly{G}\times\poly{G})$ is not monotone.
\end{example}

The above example shows that an \slc loop can have both monotone and
non-monotone recurrent sets at the same time, and thus it is natural
to ask the following question.

\begin{problem}
  \label{mrset:open:1}
  Does every non-terminating rational \slc loop has a monotone
  (polyhedral) recurrent set?
\end{problem}

The algorithm of~\citet{Ben-AmramDG19} will always find a \mlrf if one
exists, \ie, it reaches $F^i(\transitions)=\emptyset$ in a finite
number of steps, but it is not known if it will find a monotone
polyhedral recurrent set of transitions if one exists.

\begin{problem}
  Does the algorithm of~\citet{Ben-AmramDG19} always find a monotone
  polyhedral recurrent set of transitions for a rational loop
  $\transitions$, if one exists?
\end{problem}

Handling non-termination \wrt to an initial set of states $\poly{S}_0$
can be done by checking that
$\proj{\vec{x}}{\poly{G}}\cap\poly{S}_0\neq\emptyset$.

The above discussion focused on rational loops. The algorithm will not
find a recurrent set if the loop is non-terminating only over the
reals, since if the process stabilises it would have found a recurrent
set over the rationals (recall that all polyhedra involved in the
process use rational coefficients).
To apply the algorithm to the integer case, we must guarantee that the
computed monotone recurrent set of transitions $\poly{G}$ is valid,
\ie, that $\intpoly{\poly{G}}$ is a recurrent sets of transitions. If
the loop is affine (in this case we assume that the update has integer coefficients), we can do this by checking that
$\proj{\vec{x}}{\poly{G}}$ has an integer point.
We also note that for integer loops, the answer to
Problem~\ref{mrset:open:1} is negative. For example, the \slc loop
$\{x\ge 0, x'=1-x\}$ has a unique integer recurrent set of transitions
$\{(0,1),(1,0)\}$ which is not monotone.

\begin{remark}
  Let us replace each $\vect{a}_i\vec{x}-\vect{a}_i\vec{x}'\le0$ by
  $\vect{a}_i\vec{x}' + b_i \ge 0$ in~\eqref{mrset:bkwop}, and call the
  new operator $\hat{F}$.
  Then, reaching
  $\hat{F}^i(\transitions) = \hat{F}^{i+1}(\transitions)$ in a finite
  number of steps means that $\hat{F}^i(\transitions)$ is a recurrent
  set of transitions~(not necessarily monotone).
  In addition, reaching $\hat{F}^i(\transitions)=\emptyset$ in a
  finite number of steps, means that $\transitions$ is terminating
  (not necessarily has a \mlrf as in the case of using $F$).
  Applying this to the \slc loop of Example~\ref{ex:nt:rset}, we get
  the recurrent set of transitions
  $\poly{G}=\{-x_1 + x_2 \le -1,\; x_1' = -x_1+x_2,\; x_2' = x_2-1,
  -x'_1 + x'_2 \le -1\}$, and its projection on $x_1$ and $x_2$ yields
  exactly the same recurrent set that we have inferred in
  Example~\ref{ex:nt:rset}.
  The reason why~\citet{Ben-AmramDG19} consider $F$ instead of
  $\hat{F}$, is that their main interest was in inferring \mlrfs and
  the non-termination results are a byproduct.
\end{remark}

Note that \citet{Ben-AmramDG19} do not study the case in which their
algorithm does not terminate, \ie, when
$F^{i}(\transitions) \supset F^{i+1}(\transitions)$ for all $i\ge1$,
which leaves us with the following open problems.

\begin{problems}
  When $F^{i}(\transitions) \supset F^{i+1}(\transitions)$ for all
  $i\ge1$:
\begin{itemize}
\item Does $\bigcap_{i \ge 1} F^i(\transitions) = \emptyset$ imply
  termination of $\transitions$ over the rationals or reals?
\item Does $\bigcap_{i \ge 1} F^i(\transitions) = S \neq\emptyset$
  imply that the closed convex set $S$ is a~(monotone) recurrent set
  of transitions over the rationals or reals?
\end{itemize}
\end{problems}

\section{Geometric Non-Termination Arguments}
\label{sec:nt:gnta}

The concept of Geometric Non-Termination Arguments (\gnta) is due to
\citet{LeikeH18}, and is intended for proving non-termination of
\slc loops.
Although \gntas are not formulated as recurrent sets by
\citet{LeikeH18}, we show that they directly correspond to
polyhedral recurrent sets.
We will also see that \gntas have a clear algorithmic advantage over
the approaches described in Section~\ref{sec:nt:rset:slc}, in
particular for integer loops.

\citet{LeikeH18} observed an infinite execution pattern, in
which variables have a geometric growth, of the form
\begin{equation}
  \label{eq:gnta:pat:1}
  \vec{x}_0,\; \vec{x}_0+\sum_{i=0}^0\vec{y}\lambda^i,\; \vec{x}_0+\sum_{i=0}^1\vec{y}\lambda^i, \; \vec{x}_0+\sum_{i=0}^2\vec{y}\lambda^i,\; \ldots
\end{equation}  
where $\vec{y} \in \reals^n$ is the direction in which the execution
moves, and is related to the recession cone of the loop, and
$\lambda>0$ is the speed at which it is moving.

\begin{example}
\label{ex:gnta:1}
Consider the \slc loop
$\transitions = \{ x_1+x_2 \ge 3, x_1'=3x_1+1 \}$, which has the
following infinite execution:
\begin{equation}
  \tr{2}{1}, \tr{7}{1}, \tr{22}{1}, \tr{67}{1}, \ldots
\end{equation}
It can be generated using~\eqref{eq:gnta:pat:1} with
$\vec{x}_0=\tr{2}{1}$, $\vec{y}=\tr{5}{0}$, and $\lambda=2$. Note that
$\vec{y} \in \ccone(\proj{\vec{x}}{\transitions})$.
\end{example}

\citet{LeikeH18} generalised~\eqref{eq:gnta:pat:1} to handle
cases in which variables grow in different directions, and at different
speeds, to the following form (it resembles pointwise sum of geometric
series)
\begin{equation}
  \label{eq:gnta:pat:2}
  \vec{x}_0,\; \vec{x}_0+\sum_{i=0}^0YU^i\vec{1},\; \vec{x}_0+\sum_{i=0}^1YU^i\vec{1}, \; \vec{x}_0+\sum_{i=0}^2YU^i\vec{1},\; \ldots
\end{equation}  
where for some $k>0$, $\vec{1} \in \reals^k$ is a column vector of
$1$'s, $Y\in {\reals}^{n\times k}$ is a matrix such that its columns
$\vec{y}_1,\ldots,\vec{y}_k$ are the directions in which the execution
moves, and are related to the recession cone of $\transitions$, and
$\vec{U}\in {\reals}^{k\times k}$ is a matrix
\[
U = \left(
\begin{matrix}
\lambda_1 & \mu_1 & 0 & \ldots & 0 & 0 \\
0 & \lambda_2 & \mu_2 & \ldots & 0 & 0\\
\vdots &  & & \ddots & & \vdots \\
0 & 0 & 0 & \ldots & \lambda_{k-1} & \mu_{k-1} \\
0 & 0 & 0 & \ldots & 0 & \lambda_k
\end{matrix}
\right)
\]
with $\lambda_1,\dots,\lambda_k,\mu_1,\dots,\mu_{k-1} \ge 0$,
representing the speed of growth.

\begin{example}
\label{ex:gnta:2}
Consider the \slc loop
$\transitions = \{ x_1+x_2 \ge 4, x_1'=3x_1+x_2, x_2'=2x_2 \}$, which
has the following infinite execution:
\begin{equation}
  \tr{3}{1}, \tr{10}{2}, \tr{32}{4}, \tr{100}{8}, \ldots
\end{equation}
It can be generated using~\eqref{eq:gnta:pat:2} with
\[
  \vec{x}_0=\tr{3}{1}, Y=\begin{pmatrix}4 & 3 \\ 0 & 1\end{pmatrix}, \mbox{ and } U=\begin{pmatrix}3 & 1 \\ 0 & 2\end{pmatrix}
\]
Note that the columns of $Y$ are in
$\ccone(\proj{\vec{x}}{\transitions})$.
\end{example}

A \gnta consists of $\vec{x}_0$, $Y$ and $U$ that yield
an infinite execution as in~\eqref{eq:gnta:pat:2}. The following
definition states how a \gnta is extracted from $\transitions$.

\begin{definition}[\citealp{LeikeH18}]
\label{def:gnta}
Let $\transitions$ be an \slc loop specified by
$A''\tr{\vec{x}}{\vec{x}'} \le \vec{c}$.
A tuple
$\tuple{\vec{x}_0,\vec{y}_1,\ldots,\vec{y}_k,\lambda_1,\ldots,\lambda_k,\mu_1,\ldots,\mu_k}$
is a \emph{geometric non-termination argument} (\gnta) of size $k$
for $\transitions$ if and only if the following holds
\begin{align}
&\vec{x}_0,\vec{y}_1,\ldots,\vec{y}_k \in \reals^n, \lambda_1,\ldots,\lambda_k,\mu_1,\ldots,\mu_k \ge 0\label{eq:gnta:dom}\\
&A''\tr{\vec{x}_0}{\vec{x}_0+\Sigma_i\vec{y}_i} \le \vec{c}\label{eq:gnta:point}\\
&A''\tr{\vec{y}_1}{\lambda_1\vec{y}_1} \le \vec{0} \mbox{ and } A''\tr{\vec{y}_i}{\lambda_i\vec{y}_i+\mu_{i-1}\vec{y}_{i-1}} \le \vec{0} \mbox{ for } 1 < i \le k.\label{eq:gnta:rays}
\end{align}
Note that~\eqref{eq:gnta:point} requires a specific transition to be in
$\transitions$, while~\eqref{eq:gnta:rays} requires specific rays to
be in the recession cone of
$\transitions$. Condition~\eqref{eq:gnta:dom} fixes the domain of the
elements of a \gnta, and is useful when seeking \gntas over the
integers as we will see later.
\end{definition}

\begin{theorem}[\citealp{LeikeH18}]
\label{thm:gnta}
If an \slc loop has a \gnta
$\tuple{\vec{x}_0,\vec{y}_1,\ldots,\vec{y}_k,\vect{\lambda},\vect{\mu}}$,
then there is an infinite execution that starts at state $\vec{x}_0$.
\end{theorem}

\begin{proof}
The idea is to construct an execution of the
form~\eqref{eq:gnta:pat:2}, and show that every pair of consecutive states is
a transition in $\transitions$, namely
\begin{equation}
\label{eq:gnta:tr}
\begin{pmatrix}
\vec{x}_0 + \sum_{j=0}^{i-1} Y U^j \vec1 \\
\vec{x}_0 + \sum_{j=0}^i Y U^j \vec1
\end{pmatrix} \in Q
\text{ for all $i \ge 0$}.
\end{equation}
This can be done by induction. It holds for $i=0$ due
to~\eqref{eq:gnta:point}. Assume it holds for $i=t>0$, then for $i=t+1$ we
can rewrite~\eqref{eq:gnta:tr} as
\begin{equation}
  \label{eq:gnta:tr:ind}
\begin{pmatrix}
\vec{x}_0 + \sum_{j=0}^{t-1} Y U^j \vec1 \\
\vec{x}_0 + \sum_{j=0}^{t} Y U^j \vec1
\end{pmatrix}
+
\begin{pmatrix}
 Y U^{t} \vec1 \\
 Y U^{t+1} \vec1
\end{pmatrix}
\end{equation}
The term on the left is in $\transitions$ by the induction hypothesis,
and the one on the right is a non-negative combination of the rays
defined in~\eqref{eq:gnta:rays}, and thus the sum is in $\transitions$.
Note that multiplication on the right by $\vec{1}$ is equivalent to
adding together the columns of the multiplied matrix.

\end{proof}

\begin{observation}
\gntas induce polyhedral recurrent sets.
\end{observation}

\begin{proof}
Consider the \slc loop $\transitions'\subseteq \reals^{2n}$ built from
the points and rays in~(\ref{eq:gnta:point},\ref{eq:gnta:rays}) as
follows
\[
  \convhull\{\tr{\vec{x}_0}{\vec{x}_0+\Sigma_i\vec{y}_i}\} + \cone\{\tr{\vec{y}_1}{\lambda_1\vec{y}_1}, \tr{\vec{y}_2}{\lambda_2\vec{y}_2+\mu_{1}\vec{y}_{1}},\ldots,\tr{\vec{y}_k}{\lambda_k\vec{y}_k+\mu_{k-1}\vec{y}_{k-1}}\}
\]
and note that $\transitions' \subseteq \transitions$. Clearly
$\proj{\vec{x}'}{\transitions'} \subseteq
\proj{\vec{x}}{\transitions'}$, which means that $\transitions'$ is a
recurrent set of transitions for $\transitions$, and thus
$\proj{\vec{x}}{\transitions'}$ is a recurrent set for $\transitions$.
\end{proof}

A complete algorithm for finding a \gnta of size $k$, in practice,
amounts to solving the constraints
\eqref{eq:gnta:dom}--\eqref{eq:gnta:rays}; this is a system of
quadratic equations, and can be solved in polynomial
space~\citep{Canny88}.
Note that bounding the size of the \gnta to $k$ is critical. In
general, we do not know a bound on the size of the \gnta that a loop
might have.
So in practice we have to settle for an incomplete solution and
arbitrarily set a bound.
However, \citet{LeikeH18} also identified special cases for
which \gnta is a complete non-termination criterion and such bound
exists.

\begin{theorem}[\citealp{LeikeH18}]
\label{thm:gnta:affine}
If an affine \slc loop $\while\; (B\vec{x} \le \vec{b})\; \wdo \;
\vec{x}'=A\vec{x}+\vec{c}$, with $n$ variables, is non-terminating, and
$A$ has only non-negative real eigenvalues, then there is a \gnta for
the loop, of size at most $n$.
\end{theorem}

In the discussion above we have considered the case in which variables
range over the reals, however, the case in which variables range over
the integers (resp. rationals) is similar: we need only to require
$\vec{x}_0,\vec{y}_i,\lambda_i$, and $\mu_i$ in~\eqref{eq:gnta:dom}
to be integer (resp. rational).
This is a clear advantage of the \gnta approach over those we
discussed in Section~\ref{sec:nt:rset:slc}.
However, completeness is not guaranteed over the integers and rationals
since the constraints induced by~\eqref{eq:gnta:rays} are non-linear.

\begin{theorem}
\label{thm:gnta:int}
A \gnta where all components are integers (resp. rationals), implies
that the corresponding loop has an infinite computation over the
integers (resp. rationals).
\end{theorem}

To handle non-termination \wrt a polyhedral set $\poly{S}_0$ of
initial states, we only need is to require $\poly{S}_0(\vec{x}_0)$ to
hold in Definition~\ref{def:gnta}.

\begin{problems}
\
\begin{itemize}
\item Is there a more efficient algorithm for finding a \gnta, or deciding its existence? 
\item Is there a (terminating) algorithm that does not need to be provided with the size of the \gnta? 
\item Do \gntas suffice for a larger class of loops?
\end{itemize}
\end{problems}

\section{Non-Termination of Control-Flow Graphs}
\label{sec:nt:rset:cfg}

In this section we turn our attention to proving non-termination of
\cfgs. We overview several techniques that are based on different
kinds of recurrent sets to detect non-terminating loops, and also
different approaches to prove that the loop is actually reachable.

\subsection{Lasso Loops Techniques}
\label{sec:nt:rset:cfg:lasso}

The technique of~\citet{GuptaHMRX08} is based on enumerating
\emph{lasso} loops, which are common in termination and
non-termination analysis, from the \cfg and then try to prove that
they are non-terminating. The work of \citet{VelroyenR08} is based on
similar ideas---\citet{GuptaHMRX08} mention that it was developed
independently at the same time.
A \emph{lasso} loop can be viewed as a \cfg of the form
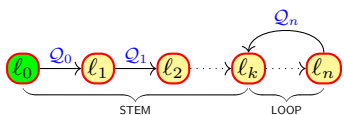
\begin{center}
\begin{tikzpicture}

\begin{scope}[shift={(0,0)}]
  \node (l0) at (0,0) [inloc] {$\ell_0$};
  \node (l1) at (1,0) [loc] {$\ell_{1}$};
  \node (l2) at (2,0) [loc] {$\ell_{2}$};
  \node[align=center] (lk) at (3,0) [loc] {$\ell_{k}$};
  \node (ln) at (4,0) [loc] {$\ell_{n}$};

  \draw [decorate, decoration={calligraphic brace, amplitude=5pt, mirror}] ([yshift=-0.05cm]l0.south) -- ([yshift=-0.05cm,xshift=-0.01cm]lk.south) node[midway, below, yshift=-0.1cm] {\scalebox{0.4}{\textsf{STEM}}};
  \draw [decorate, decoration={calligraphic brace, amplitude=5pt, mirror}] ([yshift=-0.05cm,xshift=0.01cm]lk.south) -- ([yshift=-0.05cm]ln.south) node[midway, below, yshift=-0.1cm] {\scalebox{0.4}{\textsf{LOOP}}};

  \draw [tre] (l0) to[] node[tr,above] {\scalebox{0.8}{$\transitions_0$}} (l1);
  \draw [tre] (l1) to[] node[tr,above] {\scalebox{0.8}{$\transitions_1$}} (l2);
  \draw [tre,dotted] (l2) to[] node[tr,left] {} (lk);
  \draw [tre,dotted] (lk) to[] node[tr,left] {} (ln);
  \draw [tre] (ln) to[out=90,in=90] node[tr,above] {\scalebox{0.8}{$\transitions_n$}} (lk);
\end{scope}
\end{tikzpicture}
\end{center}
and it is typically extracted from the original \cfg, in this case, by
starting at the initial location $\ell_0$ and following some path.
The nodes $\ell_0,\cdots,\ell_n$ are not necessarily different (in the
original \cfg), which allows the \textsf{STEM} and the loop to include
unrolling of loops of the original \cfg.
Clearly, non-termination of a lasso loop implies non-termination of
the original \cfg.

A lasso loop is basically an \slc loop with a polyhedral set of
initial states:
$\poly{S}=\transitions_0(\vec{x}_0,\vec{x}_1)\land\transitions_1(\vec{x}_1,\vec{x}_2)\land\cdots\land\transitions_{k-1}(\vec{x}_{k-1},\vec{x}_k)$
can be projected onto $\vec{x}_{k}$ to obtain a polyhedral set of
initial of states, and
$\poly{P}=\transitions_k(\vec{x}_k,\vec{x}_{k+1})\land\cdots\land\transitions_{n}(\vec{x}_{n},\vec{x}_{n+1})$
can be projected onto $(\vec{x}_k,\vec{x}_{n+1})$ to obtain an \slc
loop.
Thus, the techniques of
sections~\ref{sec:nt:rset:slc}--\ref{sec:nt:gnta} can be (indirectly)
used for proving non-termination of lasso loops.
It is also straightforward, and indeed done in practice, to adapt
those techniques to work directly on $\poly{P}$ and $\poly{S}$
(variables other than $(\vec{x}_k$ and $\vec{x}_{n+1})$ are considered
existential when using Farkas' lemma).

\begin{figure}[t]
\begin{center}
\begin{tikzpicture}[>=latex,line join=bevel,]

\tikzgrid{gray!0}{11.5}{6}{0.1}

\begin{scope}[shift={(3,2.8)}]
\node [align=left, font=\ttfamily, inner sep=0,outer sep=0] at (0,0) {
  \begin{minipage}{6cm}
\begin{lstlisting}[escapechar=\#]
assume(x >= 0 &&
        i >= 1 &&
        y >= 1);
while (i>=0 && nondet()) {
  y=y-1;
  i=i-1;
}
while (x >= 0) {
  if (nondet()) i=i+1;
  x = x-y-1;
}
\end{lstlisting}
\end{minipage}
};
\end{scope}

\begin{scope}[shift={(6,3.5)}]
  \node (l0) at (4,0) [inloc] {$\ell_0$};
  \node (l1) at (3,0) [loc] {$\ell_{1}$};
  \node (l2) at (2,0) [loc] {$\ell_{2}$};
  \node (l3) at (2,1) [loc] {$\ell_{3}$};
  \node (l4) at (1,0) [loc] {$\ell_{4}$};
  \draw [tre] (l0) to[] node[tr,below] {\scalebox{0.8}{$\transitions_0$}} (l1);
  \draw [tre] (l1) to[out=70,in=110,loop] node[tr,above] {\scalebox{0.8}{$\transitions_1$}}  (l1);
  \draw [tre] (l1) to[] node[tr,below] {\scalebox{0.8}{$\transitions_2$}} (l2);
  \draw [tre] (l2) to[bend left=40] node[tr,left] {\scalebox{0.8}{$\transitions_3$}}  (l3);
  \draw [tre] (l2) to[bend right=40] node[tr,right] {\scalebox{0.8}{$\transitions_4$}} (l3);
  \draw [tre] (l3) to[] node[tr] {\scalebox{0.8}{$\transitions_5$}} (l2);
  \draw [tre] (l2) to[] node[tr,below] {\scalebox{0.8}{$\transitions_6$}} (l4);
  \node [align=left,fill=blue!10] (l0) at (5,0) [] {\texttt{CFG}};
\end{scope}

\begin{scope}[shift={(3.8,5.25)}]
  \node (l0) at (6,0) [inloc] {$\ell_0$};
  \node (l1) at (5,0) [loc] {$\ell_{1}$};
  \node (l11) at (4,0) [loc] {$\ell_{1}$};
  \node (l12) at (3,0) [loc] {$\ell_{1}$};
  \node (l2) at (2,0) [loc] {$\ell_{2}$};
  \node (l3) at (1,0) [loc] {$\ell_{3}$};
  \draw [tre] (l0) to[] node[tr,above] {\scalebox{0.8}{$\transitions_0$}} (l1);
  \draw [tre] (l1) to[] node[tr,above] {\scalebox{0.8}{$\transitions_1$}} (l11);
  \draw [tre] (l11) to[] node[tr,above] {\scalebox{0.8}{$\transitions_1$}} (l12);
  \draw [tre] (l12) to[] node[tr,above] {\scalebox{0.8}{$\transitions_2$}} (l2);
  \draw [tre] (l2) to[] node[tr,above] {\scalebox{0.8}{$\transitions_3$}} (l3);
  \draw [tre] (l3) to[bend right=40] node[tr,below] {\scalebox{0.8}{$\transitions_5$}}  (l2);
  \node (l0) [align=center,fill=blue!10] at (7,0) [] {\texttt{LASSO}};

\end{scope}

\begin{scope}[shift={(8.25,1.7)}]
  \node [align=center, font=\ttfamily\footnotesize, inner sep=0,outer sep=0] at (0,0) {
    \begin{minipage}{6.5cm}
\[
  \begin{array}{|@{}r@{\hskip 2pt}l@{}|}
    \hline
\transitions_0{:}&  \{x\ge 0,i\ge 1, y\ge 1, x'=x,i'=i,y'=y\} \\
\transitions_1{:}&  \{i\ge 0, x'=x, i'=i-1,y'=y-1\} \\
\transitions_2{:}&  \{x'=x,i'=i,y'=y\} \\
\transitions_3{:}&  \{x\ge0, x'=x,i'=i,y'=y\} \\
\transitions_4{:}&  \{x\ge0, x'=x,i'=i+1,y'=y\} \\
\transitions_5{:}&  \{x\ge0, x'=x-y-1,i'=i,y'=y\} \\
\transitions_6{:}&  \{x\le-1, x'=x,i'=i,y'=y\} \\
    \hline
  \end{array}
\]
\end{minipage}
};
\end{scope}
\end{tikzpicture}
\end{center}

\caption{A program, its corresponding \cfg, and a corresponding lasso loop.}
\label{fig:cfg:lasso}

\end{figure}
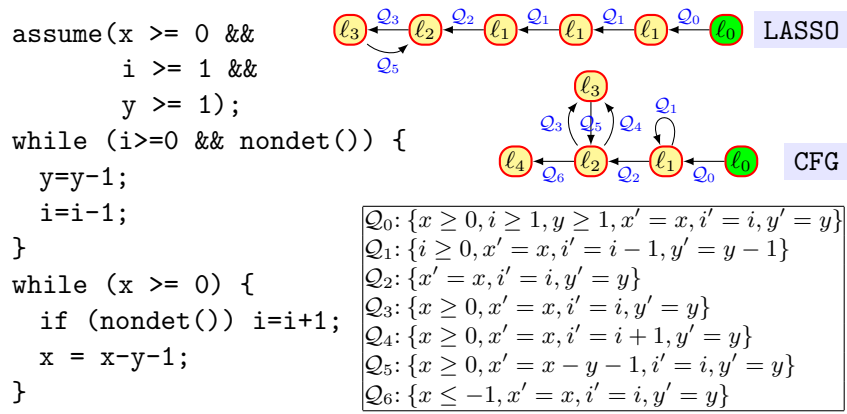

\begin{example}
\label{ex:rset:cfg}
Consider the program and the corresponding \cfg depicted in
Figure~\ref{fig:cfg:lasso}. The first loop is terminating, the second
loop does not terminate when $y$ is negative.
The initial value of $y$ is at least $1$, and the first loop decreases
its value at most $i+1$ times.
When the second loop is reached, the value of $y$ can be negative if
the first loop is executed at least two iterations (for initial value
$y=1$).
To expose this behaviour, \citet{GuptaHMRX08} unfold the first loop
twice and obtain the lasso loop shown in Figure~\ref{fig:cfg:lasso} as
well.
Now we can prove the non-termination of this lasso loop, because it is
like proving non-termination of the \slc loop
$\transitions=\{x \ge 0, x'=x'-y-1, y'=y, i'=i\}$ with the set of
initial states $\poly{S}_0=\{ y \ge -1, x\ge 0, i \ge -1 \}$.
Note that we can produce several terminating lasso loops before producing
the desired one.
\end{example}

\subsection{Quasi -Invariants Techniques}
\label{sec:nt:rset:cfg:qi}

The approach of~\citet{LarrazNORR14} is based on finding a strongly
connected sub-graph (\scsg) that is non-terminating when considered
separately, and then proving that it is reachable from the initial
location. This is done by enumerating all \scsgs until finding the
desired one.
The main advantage over the lasso based approach is that the number of
\scsgs is finite, while the number of lassos is infinite. One can also employ
various heuristics for reachability analysis~\citep{BeyerK11,AsadiC0GM21}.

Proving termination of a given \scsg is based on a concept
that~\citet{LarrazNORR14} call \emph{quasi}-invariants. These are
properties that once they hold at the locations of the \scsg, they
will continue to hold.
This notion can be seen as a generalisation of closed recurrent sets
to involve several locations.
In what follows, we will present the basic ideas of this approach, but
will not strictly follow the definitions as presented
by~\citet{LarrazNORR14}, since much of the details are added to obtain
a practical implementation.
We also note that~\citet{LarrazNORR14} assume that \cfgs satisfy some
properties, that we mostly skip, which can be easily achieved by
simple program transformations, and are useful for practical reasons.
The property that is important to our presentation is that we can
always make progress, except from the terminal locations, \ie, there
are no blocking states.

Let $P'$ be an \scsg of a \cfg $P$, and let
$\ell_{i_1},\ldots,\ell_{i_k}$ be its locations.  We say that
$\poly{I}_{i_1},\ldots,\poly{I}_{i_1} \subseteq \reals^n$ is a
(polyhedral) quasi-invariant for $P'$ if the following are satisfied:
\begin{align}
 \exists \vec{x},\vec{x}'.\; \poly{I}_{\ell_i}(\vec{x}) \land \transitions(\vec{x},\vec{x}') &\mbox{\small{ for all }} (\ell_i,\transitions,\ell_j) \in P' \label{eq:qi:nonempty}\\
 \forall \vec{x},\vec{x}'.\;  \poly{I}_{\ell_i}(\vec{x}) \land \transitions(\vec{x},\vec{x}') \rightarrow \poly{I}_{\ell_j}(\vec{x}') & \mbox{\small{ for all }} (\ell_i,\transitions,\ell_j) \in P' \label{eq:qi:consec}\\
  \forall \vec{x},\vec{x}'.\; \poly{I}_{\ell_i}(\vec{x}) \land \transitions(\vec{x},\vec{x}') \rightarrow \mathit{false} & \mbox{\small{ for all }} \ell_i {\in} P', (\ell_i,\transitions,\ell_j) \not\in P'\label{eq:qi:enabled}
\end{align}
Lets us explain the meaning of these formulas:
\eqref{eq:qi:nonempty} guarantees that all components of the
quasi-invariant are not empty, and is similar
to~\eqref{eq:crset:nonempty} of closed recurrent sets;
\eqref{eq:qi:consec} guarantees that when progressing from a state
within the quasi-invariant we remain within the quasi-invariant, and
is similar to~\eqref{eq:crset:consec} of closed recurrent sets;
and~\eqref{eq:qi:enabled} states that executions within the
quasi-invariant cannot escape from the \scsg, which is similar
to~\eqref{eq:crset:enabled} of closed recurrent sets.
Clearly, $P'$ does not terminate when starting the execution at
location $\ell_i \in P'$ with $\vec{x}\in \poly{I}_{\ell_i}$. Moreover, if the state
$(\ell_i,\vec{x})$ is reachable in $P$, then $P$ is
non-terminating.

\begin{example}
\label{ex:qi}
For the \cfg of Figure~\ref{fig:cfg:lasso}, \citet{LarrazNORR14}
consider the \scsg of nodes $\ell_2$ and $\ell_3$ and all edges that connect
them. Then they infer
$\poly{I}_{\ell_2}=\poly{I}_{\ell_3}=\{x \ge 0, y\le -1\}$, and then
separately prove that $l_2$ is reachable with some
$\vec{x} \in \poly{I}_{\ell_2}$.
\end{example}

We have seen in Section~\ref{sec:nt:rset:slc} that non-determinism
might prevent \slc loops to have a closed recurrent set. This is also
true for quasi-invariants.
The solution that was suggested in the context of \slc loops is to try
make $\transitions$ deterministic, by adding more constraints, while
seeking a closed recurrent set.
This solution was actually proposed by~\citet{LarrazNORR14} for
inferring quasi-invariants.
The most common way to do this is by adding parametric constraints of the form
$\vec{x}'=A\vec{x}+\vec{c}$, which are also useful for handling the
integer case when forcing $A$ and $\vec{c}$ to be integer (since any
integer enabled state will have an integer successor).

\begin{example}
\label{ex:qi:nondet}
Consider again the \cfg of Figure~\ref{fig:cfg:lasso}, and assume
$\transitions_3$ has $x' \le x-y-1$ instead of $x'=x-y-1$. With this
change, it is not possible to infer a quasi-invariant
satisfying~\eqref{eq:qi:nonempty}-\eqref{eq:qi:enabled} (there is no
closed recurrent set).
\citet{LarrazNORR14} automatically add $x'=x-y-1$ to $\transitions_3$,
which makes it possible to infer the quasi-invariant of
Example~\ref{ex:qi}.
\end{example}

\subsection{Loop Acceleration Techniques}
\label{sec:nt:rset:cfg:accel}

\citet{FrohnG19a} use \emph{loop acceleration} to prove non-termination
of integer \cfgs.
The core idea of this approach is that, instead of unfolding a loop a
finite number of times to generate a candidate lasso, we can
accelerate the loop which leaves the number of necessary unfoldings as
a parameter, $k$, within the accelerated loop's term.
A constraint solver can later determine the value of $k$ needed to
prove the reachability of a non-terminating simple loop (the loops
they consider are single-path, like affine \slc loops, but the guard
can have polynomial inequalities and the update is of the form
$x'_i = p(\vec{x})$ where $p$ is a polynomial).
Proving non-termination of a simple loop, however, still relies on the
concept of recurrent sets even if inferring such sets is done slightly
in a different way.
Note that their approach extends beyond linear-constraint \cfgs
because it allows using polynomial expression in the guard and the
update (even when analysing linear-constraint \cfgs, it might generate
transition relations with non-linear constraints).

Next we briefly describe the basics of the algorithm
of~\citet{FrohnG19a}, for more precise details the reader is referred
to~\citet{FrohnG19a}.
The algorithm is based on iteratively repeating a series of operations
until some conditions are satisfied:
\begin{enumerate}
\item\label{fg19:nt} \emph{Prove Non-Termination of Simple Loops}: The
  algorithm attempts to prove non-termination for each simple loop
  $(\ell_i, \transitions, \ell_i)$. If successful, the loop's edge is
  replaced by $(\ell_i, \transitions, \ell_\omega)$, with
  $\ell_\omega$ indicating non-termination. Non-termination is proven
  by a variety of techniques, one of them checks if the guard is a
  recurrent set (or ``simple invariant'' in their terms). While a
  guard may not initially be a recurrent set, a later step strengthens
  it with additional constraints to achieve this goal. In principle,
  any technique for proving simple loop non-termination can be used
  here, as long as the guard is strengthened with conditions that
  ensure non-termination.
\item\label{fg19:acc} \emph{Accelerate Simple Loops}: If certain
  conditions are met by a simple loop
  $(\ell_i, \transitions, \ell_i)$, it is replaced by its accelerated
  equivalent. This is done by adding an edge
  $(\ell_j, \transitions' \circ \transitions_a, \ell_i)$, for every
  incoming edge $(\ell_j, \transitions', \ell_i)$ with
  $\ell_j \neq \ell_i$, where $\transitions_a$ is the result of the
  acceleration.
  A single transition using these new edges represents the execution
  of $k > 0$ iterations of the original loop, where $k$ is a new
  variable in $\transitions_a$ that is existentially quantified
  (Alternatively, we could add $k$ as a program variable, in which
  case its value would be automatically chosen since it is not
  assigned).
  The conditions that must be met for acceleration ensure that if the
  loop guard holds after $k$ applications of the update, then it also
  holds for all previous applications.
  While these conditions may not be initially satisfied, a later step
  in the algorithm strengthens the guard with additional constraints
  to make this possible.
\item\label{fg19:inv} \emph{Strengthen Guards of Simple Loops}:
  Special kind of invariants (different from the standard notion of
  invariants) are added to the guards of simple loops. The purpose is
  to make acceleration or non-termination proofs possible for these
  loops.
\item\label{fg19:ch} \emph{Chaining}: Consecutive edges, such as
  $(\ell_i, \transitions_1, \ell_j)$ and
  $(\ell_j, \transitions_2, \ell_k)$, are replaced by a single,
  chained edge
  $(\ell_i, \transitions_1 \circ \transitions_2, \ell_k)$.  Chaining
  has multiple purposes, including simplifying complex loops into
  simple ones.
\end{enumerate}
The process concludes when the \cfg is reduced to a set of edges all
originating from the initial node $\ell_0$, or when no progress in
made.
Then, if an edge $(\ell_0, \transitions, \ell_\omega)$ exists and
$\transitions$ is satisfiable, the \cfg is proven to be
non-terminating.
Note that while we use the notation $\transitions$ for transition
relations, in practice, these can include polynomial constraints due
to acceleration.

Let us demonstrate some steps of this algorithm on the \cfg in
Figure~\ref{fig:cfg:lasso}.

\begin{example}
\label{ex:accel}
The algorithm starts by trying to prove non-termination of the simple
loop $(\ell_1, \transitions_1, \ell_1)$ and fails. Then it tries to
accelerate it and succeed with
$\transitions_1'=\{i'=i-k, x'=x, y'=y-k, i-k+1 \ge 0, k \ge 1\}$. Note
that the acceleration in this case resulted in linear expressions, but
it might be polynomial as well.
To reflect this acceleration in the \cfg, we remove the original edge
and add a new edge
$(\ell_0,\transitions_0\circ\transitions_1',\ell_1)$. When this edge
is taken with $k=n$ it simulates $n$ iterations of the original
loop.
Note that if we take the edge $(\ell_0,\transitions_0,\ell_1)$ then we
are not executing the loop, \eg, when the guard is not satisfied right
from the beginning.

There are no more simple loops, so the algorithm applies chaining
which converts the complex loop at $\ell_2$ into an \mlc loop with two
paths (simple loops): $(\ell_2,\transitions_3',\ell_2)$ and
$(\ell_2,\transitions_4',\ell_2)$ where
$\transitions_3'=\transitions_3\circ\transitions_5$ and
$\transitions_4'=\transitions_4\circ\transitions_5$.
In addition, it reduces the paths from $\ell_0$ to $\ell_2$ by
connecting $\ell_0$ to $\ell_2$, \ie, it generates
$(\ell_0,\transitions_0\circ\transitions_2,\ell_2)$ and
$(\ell_0,\transitions_0\circ\transitions_1'\circ\transitions_2,\ell_2)$.

In the next iteration, it attempts to prove non-termination
of these loops but fails because their guards are not recurrent
sets. Additionally, the loops cannot be accelerated.
The process then moves on to strengthen the guards with the constraint
$y \le -1$ (let us assume it is added to $\transitions_3'$ and
$\transitions_4'$).
In the subsequent iteration, this allows the algorithm to successfully
prove non-termination for both loops, as their strengthened guards are
now recurrent sets. As a result, it replaces the corresponding edges
with $(\ell_2, \transitions_3', \ell_\omega)$ and
$(\ell_2, \transitions_4', \ell_\omega)$.

Chaining now creates, among others, the edge
$(\ell_0,\transitions_0\circ\transitions_1'\circ\transitions_2\circ\transitions_3',\ell_\omega)$
whose transition relation is satisfiable for any $k \ge 2$, \ie,
execution of the first loop at least two iterations, and thus the \cfg
does not terminate.
\end{example}

In a subsequent work, \citet{FrohnG23} developed an approach to allow
for the acceleration of more complex loops, such as those with
disjunctions in their transition relations. The details are complex,
so we refer the reader to~\citet{FrohnG23} for more details.

\subsection{Safety Prover Techniques}
\label{sec:nt:rset:cfg:safey}

\citet{ChenCFNO14} present a method for proving non-termination of a
\cfg by reducing the problem to a series of safety-proving tasks. The
approach iteratively refines an under-approximation of the original
program using counterexamples from a safety prover. The ``never
terminates'' property is encoded as a safety violation, and this
refinement process ultimately produces an under-approximation of the
\cfg, that also induces a closed recurrent set.
Note that under-approximations, in this context, means restricting the
input values as well as the values of non-deterministic choices, and
select an execution path from the initial location to the loop under
consideration.

The algorithm by \citet{ChenCFNO14} is formalised on a slightly
different (though equivalent) notion of \cfgs. For the sake of
simplifying the presentation, will explain the basic idea using
\clang-like programs, like the one in Figure~\ref{fig:cfg:lasso}.
For this explanation, we slightly modify the meaning of the
instruction \lstinline{nondet()}. We assume it is of the form
\lstinline{nondet($\psi$)}, where $\psi$ is a boolean condition that
involves a variable $r$ that refers to the value returned by the
function. For example, \lstinline{nondet($r<0$)} would produce a
negative number. The original instruction \lstinline{nondet()} is
syntactic sugar for \lstinline{nondet($\mathit{true}$)}.

The algorithm by \citet{ChenCFNO14} is designed to prove
non-termination for a given loop within a given program. To do this,
it first \emph{instruments} the program with two instructions: an
\lstinline{assume($\mathit{false}$)} statement immediately after the
loop's exit to simulate an error state, and an
\lstinline{assume($\mathit{true}$)} statement at the program's beginning to
restrict the set of input values.
The core of the approach is to show that the
\lstinline{assume($\mathit{false}$)} statement is unreachable. If this
can be proven, the loop is guaranteed to be non-terminating (assuming
there are no blocking states).
However, proving this for all possible inputs is unlikely, as a loop
typically terminates for some inputs but not for others. The algorithm
therefore focuses on finding a specific subset of inputs and
non-deterministic choices for which non-termination holds.

The process works as follows:
\begin{inparaenum}[\upshape(1\upshape)]
\item The instrumented program is passed to a safety prover;
\item If the prover proves that \lstinline{assume($\mathit{false}$)}
  is unreachable, the algorithm succeeds (up to a post-processing step
  that we discuss below); otherwise
\item The prover returns a counterexample, which is then used to
  strengthen the \lstinline{assume} instruction that restricts the
  input and the choices of \lstinline{nondet(.)}. This strengthening
  eliminates the counterexample, and the process is repeated.
\end{inparaenum}

At the end of this process, we remain with a restriction of the
original program. We then need to prove that the loop is reachable
from the initial location, which is done by inserting
\lstinline{assume($\mathit{false}$)} just before the loop and passing
it to a safety prover, if it returns a counterexample it means that
the loop is reachable, and this counterexample is used as a stem for
the loop.
Finally, we have to prove that the program that consists of the stem
and the loop is non-blocking, \ie, that whenever
\lstinline{nondet($\psi$)} is reached it is possible to pick a value
that satisfies $\psi$ and does not block the execution.
If we succeed then non-termination is proven. Moreover, if we consider
the transition relation induced by the restricted program, then it has
a closed recurrent set since all executions are non-terminating.

\begin{example}
\label{ex:nt:safety}
Let us see how to prove non-termination for the second loop in the
program of Figure~\ref{fig:cfg:lasso}.
We first instrument the program by adding the instruction
\lstinline{assume($\mathit{false}$)} immediately after the loop's
exit. We do not need to add a separate
\lstinline{assume($\mathit{true}$)} instruction at the beginning, as
we will use the existing one to further restrict the input.
When passing the instrumented program to a safety prover, it returns
the following counterexample:
\begin{lstlisting}
  nondet()<=0; x>=0; x=x-y-1; x<0
\end{lstlisting}
To eliminate this trace, we can strengthen \lstinline{nondet()} to
\lstinline{nondet($r\ge 1$)}.
In the next iteration, we get the following counterexample:
\begin{lstlisting}
  i>=0 && nondet($r\ge 1$)>=0; i=i-1; y=y-1;
  i>=0 && nondet($r\ge 1$)>=0; i=i-1; y=y-1;
  i<0; x>=0; x=x-y-1; x<0  
\end{lstlisting}
To eliminate this trace, we could add \lstinline{y<=1} to the
\lstinline{assume} instruction at the beginning.
Now the safety prover proves that \lstinline{assume($\mathit{false}$)}
is unreachable, because at the beginning \lstinline{y} is always $1$,
\lstinline{x} at least $0$, and \lstinline{i} at least $1$. Thus, we
reach the second loop with \lstinline{y} at most $-1$ and the second
loop does not terminate.
Next we have to prove that the second loop is reachable. This is done
by adding \lstinline{assume($\mathit{false}$)} before the second loop,
and passing it to a safety prover. The prover returns the following
counterexample, confirming the loop's reachability via this trace:
\begin{lstlisting}
  i>=0 && nondet($r\ge 1$)>=0; i=i-1; y=y-1;
  i>=0 && nondet($r\ge 1$)>=0; i=i-1; y=y-1; i<0 
\end{lstlisting}
The restricted program now consists of this trace as a stem leading to
the second loop. This program represents a valid restriction of the
original one. Finally, it is easy to check that
\lstinline{nondet($r\ge 1$)} does not block any execution. Thus we
have proven non-termination.
\end{example}

\section{Non-terminating vs. Unbounded States}
\label{sec:nt:unbound}

We say that a transition relation $T$ is unbounded in a state
$\vec{x} \in \numdom^n$, with $\numdom \in \{\reals,\rats,\ints\}$, if
it is possible to make executions of arbitrary length starting from
$\vec{x}$. We say that $T$ is unbounded if it is unbounded in some
state.

\begin{example}
\label{ex:unbound}
Consider the \mlc loop~\eqref{ex:mlc:0}.
For any input state $(x_1,x_2)$ with $x_1=0$, we can take the first
path to reset $x_2$ to $n\in\nats$, and then use the second path to make a
terminating execution of length $n$ (in total $n+1$).
Thus, this loop is unbounded in any such input state, despite being terminating.
\end{example}

It seems clear that the situation in the example is due to non-determinism. It is easy to see that 
a deterministic loop is bounded if and only if it is terminating. For \slc loops we have an intriguing open problem. 

\begin{problem}
  Is there a terminating, yet unbounded, \slc loop?
\end{problem}

\section{Other Approaches to Non-Termination}
\label{sec:nt:other}

\citet{BrockschmidtSOG11} present an approach for detecting
non-termination in Java Bytecode programs using \emph{termination
  graphs}, which are finite representations of all potential program
executions. While originally developed for Java, the underlying
techniques, that have some similarities to seeking recurrent sets,
have been used in the non-termination component of
T2~\citep{BrockschmidtCIK16}.
The approach identifies variables that influence control flow based on
refinements made during graph construction, and analyses candidate
cycles using two distinct SMT-based techniques.
The first technique detects \emph{looping non-termination},
identifying infinite evaluations where the values of critical
variables remain unchanged across iterations. It generates an SMT
formula where satisfiability proves that these variables are
unmodified by the loop body, allowing the program to return to an
identical state infinitely.
The second technique addresses \emph{non-looping non-termination},
focusing on loops where variables change in every iteration while
still satisfying a loop condition. This method uses an SMT solver to
prove that if the loop condition is satisfied once, every subsequent
iteration will also satisfy it, meaning the program can never escape
the loop.

\citet{BIK14} present a complete method for inferring non-termination
preconditions for octagonal \slc loops and for affine \slc loops whose
update matrix generates a finite monoid (see
Section~\ref{sec:dec:slc:oct}).

\citet{CookFNO14} investigate the conditions under which abstractions
can be used to prove non-termination. Specifically, they explore when
a non-terminating abstract transition relation, $T^\alpha$ (an
over-approximation of a concrete relation $T$), guarantees that the
concrete relation $T$ is also non-terminating.
They introduce a class of abstractions, that they call \emph{live},
for which closed recurrent sets are preserved. This means that if the
abstract relation $T^\alpha$ has a closed recurrent set, then the
concrete relation $T$ is guaranteed to have one as well.
This finding simplifies the search for a non-termination proof, as one
can seek a closed recurrent set for the abstract relation $T^\alpha$,
which is typically easier to analyse.
Surprisingly, many of the linear-constraint abstractions used in
termination analysis fall into this category, as intuitively, the only
requirement is: if $f$ is a final concrete state, and it is in the
concrete states described by an abstract one $g$, then $g$ is also a
final abstract state.
The authors demonstrate how these abstractions can be applied to
analyse programs with non-linear arithmetic and heap manipulation.

\citet{LeQC15} propose a unified, modular framework that analyses and
proves both termination and non-termination simultaneously. The core
of this method involves using second-order termination constraints and
accumulating a set of relational assumptions on them via a Hoare-style
verification.

\citet{BakhirkinBP15} present a method for detecting, using a purely
forward abstract interpretation, non-terminating loops in imperative
programs. The analysis searches for a recurrent set by building and
analysing a graph of abstract states. In a subsequent work,
\citet{BakhirkinP16} present an abstract interpretation-based analysis
for finding recurrent sets, which combines an approximate backward
analysis to identify a candidate recurrent set with an
over-approximate forward analysis to check and refine it.

A method for computing a subset of the non-terminating initial states
for affine \slc loops over the reals was presented by
\citet{Li17}. For homogeneous linear loops over the reals with only
two program variables, \citet{DaiX12} provided a complete algorithm to
compute the full set of non-terminating initial states.

\citet{LeikeH18} proved that if an \slc loop over $\reals$ has a
non-terminating execution in which each state $\vec{x}_i$ satisfies
${\parallel}{\vec{x}_i} {\parallel} \le c$, for some norm
${\parallel} {\cdot} {\parallel}$ and $c\in\reals$, then it has a
fixpoint transition $\tr{\vec{x}}{\vec{x}}$.

\chapter{Conclusions}
\label{chp:conc}

Termination analysis has received considerable attention in recent
decades, and today several powerful tools exist for the automatic
termination analysis of different programming languages and
computational models.
This practical advancement would not have been possible without
corresponding theoretical progress, which aims to explore the limits
of proving termination and to provide algorithms for specific proof
techniques---\eg, ranking functions---along with corresponding
complexity classifications for the underlying problems.

In this \survey we provided a comprehensive overview of the
state-of-the-art in \emph{termination and non-termination analysis of
  linear-constraint programs}, a field that has seen significant
progress over the last three to four decades and whose results are
intensively used in practice.
At the core of this research is a trade-off between the expressive
power of a technique, \ie, the class of programs it can handle, and
the computational complexity of the associated decision problems. The
\survey systematically explored various research directions, from
decidability results for specific program classes to a wide range of
termination and non-termination witnesses.
Despite the significant volume of work in this field, many challenging
problems remain open, some of which we stated explicitly in the body
of this \survey. The answers to these problems will not only advance
the theoretical understanding of program termination but may also
impact the development of more powerful and automated termination
analysis tools.

Our discussion began with the fundamental problem of deciding
termination for different classes of linear-constraint programs,
including \emph{single-path}~(\slc) and \emph{multipath}~(\mlc)
Linear-Constraint loops.
We presented a uniform framework for \emph{affine} \slc loops, showing
that termination is decidable for variables over the reals, rationals,
and integers, a problem that had proven to be a long-standing
challenge.
We also highlighted key undecidability results for more general
classes, such as \mlc loops, which underscore the inherent difficulty
of the problem in its most general form.
There are still several major open problems in this direction:
\begin{inparaenum}[\upshape(1\upshape)]
\item The decidability of termination for general \slc loops, whether
  over real, rational, or integer domains, remains an important open
  question.
\item The decidability of termination for \slc loops \wrt a given set
  of initial states is also an unsolved problem even for affine \slc
  loop.
\end{inparaenum}
This latter question is closely related to the well-known, and
long-standing, \emph{Positivity Problem} for linear recurrence
sequences.
A good starting point for tackling the general termination problem for
\slc loops would be to first address simpler sub-problems. This could
involve focusing on deterministic loops that are not necessarily
affine or on loops that allow a small, controlled degree of
non-determinism.

A major part of this \survey was dedicated to ranking functions, a
classic and powerful method for proving termination. We covered a
spectrum of ranking function types, from simple linear ranking
functions~(\lrfs) to more expressive lexicographic-linear ranking
functions~(\llrfs) and multiphase-linear ranking functions~(\mlrfs).
For each type, we examined the algorithmic and complexity aspects of
their synthesis for the different kinds of programs we consider,
distinguishing between rational and integer domains.
There are still several major open problems in this direction:
\begin{inparaenum}[\upshape(1\upshape)]
\item Unlike other kinds of \llrfs that we considered, there are no
  decidability results or complete algorithms for \mlrfs without a
  given bound on the depth, not even for affine \slc loops.
\item The problem of synthesising ranking functions \wrt a given set
  of initial states has not received much attention, possibly due to
  its inherent difficulty, apart from partial solutions based on
  inductive invariants.
\end{inparaenum}
A good starting point for tackling these problems is by considering
simpler sub-problems, such as affine \slc loops or even those where the
update matrix is diagonalisable.

We also explored the concept of disjunctive well-founded transition
invariants (\dtis), which offers an alternative to ranking functions
for proving termination. This approach, which is based on Ramsey's
theorem, is particularly effective for programs with complex control
flow where a single ranking function, within the classes we consider,
might not exist.
We showed that several well-known termination analysis methods, such
as size-change termination and monotonicity constraints, can be
understood as applications of the \dti principle. We provided
decidability results for these classes. 
Note that these classes have been originally studied from different
viewpoints, but in this \survey we have shown how they all fall under
the \dti approach.
The link between \dtis and ranking functions was also discussed.
A major open problem in this area is to characterise classes of
programs for which \dtis are not more powerful than \llrfs.
A good starting point for tackling this problem is to consider \slc
loops, where non-determinism does not arise from branching. One could
begin with special cases, such as affine or deterministic \slc loops,
before moving to the general case.

We have also discussed witnesses for non-termination, such as
polyhedral recurrent sets and geometric non-termination
arguments. These witnesses provide a concrete object that proves a
program will not halt, complementing the techniques for proving
termination. We reviewed algorithms for their synthesis and
highlighted the challenges, particularly when dealing with
non-deterministic or integer-based programs.
Unlike other topics in this \survey, decidability results and complete
algorithms for non-termination proofs are very limited: for \slc loops
one has to provide a limit on the size of the geometric
non-termination argument (\gnta), and for affine \slc loops one has to
provide a template recurrent set. These also work only over the reals.
Addressing these problems is a major challenge, and one could start by
characterising subclasses of \slc loops for which polyhedral recurrent
sets or geometric non-termination arguments are sufficient.
 
\unless\ifdefined\techreport
\begin{acknowledgements}

  The authors would like to thank the anonymous reviewers for their
  helpful comments and feedback, which greatly strengthened the
  overall manuscript.

  Samir Genaim was funded partially by the Spanish MCI, Comunidad de
  Madrid, AEI and FEDER (EU) projects PID2021-122830OB-C41,
  PID2024-157044OB-C31, TEC-2024/COM-235, and the BOVIR project
  PR17/24-31926.

  Jo\"{e}l Ouaknine is also affiliated with Keble College, Oxford as
  \href{https://www.emmy.network}{emmy.network} Fellow; he was supported by ERC Grant DynAMiCs~(ID. 101167561)
  and by DFG grant 389792660 as part of TRR 248 (see
  \href{https://perspicuous-computing.science}{https://perspicuous-computing.science}).

  James Worrell was supported by UKRI Fellowship EP/X033813/1.
\end{acknowledgements}
\fi

\backmatter  %

\printbibliography
\ifdefined\techreport
  \addcontentsline{toc}{chapter}{\bibname}
\fi

\end{document}